\documentclass[oneside,final,letter]{ucr}
\usepackage{anysize}
\usepackage[left=1.5in,top=1.5in,right=1in,bottom=1in]{geometry}
\usepackage{amssymb,amsmath,amsthm,mathrsfs,xspace,multicol,stmaryrd}
\usepackage[mathscr]{eucal}
\usepackage{amscd}

\usepackage{graphicx}
\usepackage{pdfsync}
\usepackage{fncylab}
\usepackage[usenames]{color}
\usepackage{multirow}
\usepackage{hyperref}
\usepackage[mathscr]{eucal}
\usepackage[all]{xy}
\usepackage{epsfig}

\input{epsf}

\newtheorem{theorem}{Theorem}[chapter]

\newtheorem{corollary}[theorem]{Corollary}

\newtheorem{definition}[theorem]{Definition}

\newtheorem{lemma}[theorem]{Lemma}

\newtheorem{prop}[theorem]{Proposition}

\newtheorem{proposition}[theorem]{Proposition}
\newtheorem*{theorem*}{Theorem}
\newtheorem*{definition*}{Definition}
\newtheorem*{lemma*}{Lemma}
\newtheorem*{corollary*}{Corollary}
\newtheorem*{proposition*}{Proposition}
\newtheorem*{example*}{Example}
\newtheorem*{conjecture*}{Conjecture}
\newtheorem*{remark*}{Remark}
\newtheorem*{notation*}{Notation}
\newtheorem*{convention*}{Convention}

\theoremstyle{definition}
\newtheorem{example}[theorem]{Example}
\theoremstyle{remark}
\newtheorem{remark}[theorem]{Remark}

\newcommand{\blankbrac}{[ \cdot,\cdot ]}

\newcommand{\brac}[2]{\left \{ #1,#2 \right\}}
\newcommand{\bbrac}[2]{\left \llbracket #1,#2 \right \rrbracket}

\newcommand{\sbrac}[2]{\left [ #1,#2 \right ]}
\newcommand{\lbrac}[2]{\left [ #1,#2 \right ]_{L}}
\newcommand{\lpbrac}[2]{\left [ #1,#2 \right ]_{L'}}

\newcommand{\cbrac}[2]{\left [ #1,#2 \right ]_{C}}
\newcommand{\tcbrac}[2]{\left [ #1,#2 \right ]_{C}}
\newcommand{\abrac}[2]{\left [#1,#2 \right]_{A}}
\newcommand{\aibrac}[2]{\left [#1,#2 \right]_{i}}
\newcommand{\tabrac}[2]{\left [#1,#2 \right]_{A}}
\newcommand{\dbrac}[2]{\left \llbracket #1,#2 \right \rrbracket_{C}}
\newcommand{\tdbrac}[2]{\left \llbracket #1,#2 \right  \rrbracket_{C}}

\newcommand{\innerprodp}[2]{\bigl \langle #1  ,  #2 \bigr \rangle^{+}}
\newcommand{\innerprodm}[2]{\bigl \langle #1 ,  #2 \bigr \rangle^{-} }

\newcommand{\innerprod}[2]{\bigl \langle #1  ,  #2 \bigr \rangle}

\newcommand{\ham}{\Omega^{1}_{\mathrm{Ham}}(M)}
\newcommand{\hamn}[1]{\Omega^{#1}_{\mathrm{Ham}}\left(M\right)}
\newcommand{\X}{\mathfrak{X}}
\newcommand{\Xham}{\mathfrak{X}_{\mathrm{Ham}}(M)}
\newcommand{\XGham}{\mathfrak{X}_{\mathrm{Ham}}(G)}
\newcommand{\hamG}[1]{\Omega^{#1}_{\mathrm{Ham}}\left(M\right)^G}
\newcommand{\Gham}{\Omega^{1}_{\mathrm{Ham}}\left(G\right)}
\newcommand{\hamL}{\Omega^{1}_{\mathrm{Ham}}\left(G\right)^{L}}
\newcommand{\cinf}{C^{\infty}}

\newcommand{\cinfL}{C^{\infty}\left(G\right)^{L}}
\newcommand{\lie}[2]{\mathcal{L}\,_{v_{#1}}\, #2}
\renewcommand{\L}{\mathcal{L}}

\newcommand{\icinf}{C^{\infty}_{\mathrm{Im}}}

\renewcommand{\d}{\mathbf{d}}

\newcommand{\ip}[1]{\iota_{v_{#1}}}

\newcommand{\pback}[1]{{\mu_{#1}}^{\ast}}
\newcommand{\pfor}[1]{{\mu_{#1}}_{\ast}}

\newcommand{\F}{\mathsf{F}}
\newcommand{\G}{\mathsf{G}}
\renewcommand{\S}{\mathsf{S}}

\newcommand{\VectH}{\mathfrak{X}_{\mathrm{Ham}}}

\newcommand{\alphak}[1]{\alpha_{1} \tensor \cdots \tensor \alpha_{#1}}
\newcommand{\alphadk}[1]{\alpha_{1},\hdots,\alpha_{#1}}
\newcommand{\alphask}[1]{\alpha_{\sigma(1)} \tensor \cdots \tensor \alpha_{\sigma(#1)}}
\newcommand{\alphasdk}[1]{\alpha_{\sigma(1)}, \hdots,\alpha_{\sigma(#1)}}
\newcommand{\vk}[1]{v_{\alpha_{1}} \wedge \cdots \wedge  v_{\alpha_{#1}}}

\newcommand{\LX}{\mathfrak{X}^{\wedge {\bullet}}}

\newcommand{\Sh}{\mathrm{Sh}}
\newcommand{\Sn}{\mathcal{S}}
\renewcommand{\deg}[1]{\left \lvert #1 \right \rvert}

\newcommand{\Lie}{L_{\infty}}

\newcommand{\Co}{\mathrm{Co}}
\newcommand{\sh}[1]{\underline{#1}}

\newcommand{\cU}{\mathcal{U}}

\newcommand{\Op}{\mathsf{Open}}
\newcommand{\Cat}{\mathsf{Cat}}
\newcommand{\Set}{\mathsf{Set}}
\newcommand{\Tor}{\mathsf{Tor}}
\newcommand{\Shf}{\mathsf{Sh}}

\newcommand{\HVB}{\mathsf{Bund}}
\newcommand{\Quant}{\mathsf{Quant}}

\newcommand{\Des}{\mathsf{Des}}
\newcommand{\cOmega}{\Omega_{\mathrm{cl}}}

\DeclareMathOperator{\hol}{\mathrm{hol}}
\DeclareMathOperator{\Leib}{\mathrm{Leib}}

\newcommand{\N}{\mathbb{N}}
\newcommand{\R}{\mathbb{R}}
\newcommand{\C}{\mathbb{C}}
\newcommand{\Z}{\mathbb{Z}}
\newcommand{\Zi}{\mathbb{Z}(1)}
\newcommand{\Ri}{\mathbb{R}(1)}
\newcommand{\Cx}{\mathbb{C}^{\times}}

\newcommand{\inp}[2]{\langle #1,#2 \rangle}
\newcommand{\g}{\mathfrak{g}}

\newcommand{\su}{\mathfrak{su}}
\renewcommand{\u}{\mathfrak{u}}

\newcommand{\CP}{\mathbb{C}\mathrm{P}}
\newcommand{\U}{\mathrm{U}}
\newcommand{\SU}{\mathrm{SU}}

\DeclareMathOperator{\ad}{\mathrm{ad}}
\DeclareMathOperator{\Ad}{\mathrm{Ad}}
\newcommand{\tensor}{\otimes}
\renewcommand{\perp}{\bot}
\newcommand{\conn}{\nabla}

\newcommand{\maps}{\colon}
\renewcommand{\i}{\sqrt{-1}}
\newcommand{\del}{\partial}
\newcommand{\half}{\frac{1}{2}}

\newcommand{\epi}{\twoheadrightarrow}
\newcommand{\iso}{\stackrel{\sim}{\to}}
\newcommand{\xto}[1]{\xrightarrow{#1}}
\newcommand{\embed}{\hookrightarrow}
\renewcommand{\ss}{\subseteq}
\newcommand{\vphi}{\varphi}

\newcommand{\dlog}{d_{\log}}
\renewcommand{\b}{\bullet}

\newcommand{\noi}{\noindent}
\DeclareMathOperator{\tpi}{2\pi}
\DeclareMathOperator{\Aut}{\mathrm{Aut}}
\DeclareMathOperator{\Hom}{\mathrm{Hom}}
\DeclareMathOperator{\End}{\mathrm{End}}
\DeclareMathOperator{\id}{\mathrm{id}}
\DeclareMathOperator{\im}{\mathrm{im}}

\DeclareMathOperator{\cp}{\mathrm{c.\!p}}

\begin{document}

% Declarations for Front Matter

\title{Higher Symplectic Geometry}
\author{Christopher Lee Rogers}
\degreemonth{June}
\degreeyear{2011}
\degree{Doctor of Philosophy}
\chair{Professor John C.\ Baez}
\othermembers{Professor Julia Bergner\\
Professor Yat-Sun Poon}

\numberofmembers{3}
\field{Mathematics}
\campus{Riverside}

\maketitle
\copyrightpage{}
\approvalpage{}

\degreesemester{Spring}

\begin{frontmatter}

\begin{acknowledgements}
I wish to thank my advisor John Baez for his encouragement and guidance in
completing this thesis. I would also like to thank Julie Bergner and
Yat-Sun Poon for serving on my committee.

I wish to acknowledge the following individuals for
helpful discussions, comments, and suggestions:
Maarten Bergvelt, Henrique Bursztyn, Jim Dolan, Vasily Dolgushev,
Yael Fregier, Alex Hoffnung, Allen Knutson, Dmitry Roytenberg, Urs Schreiber,
Jim Stasheff, Danny Stevenson, Thomas Strobl,  Alan Weinstein, and Marco Zambon.

Part of this work was completed while I was a Junior Research Fellow at
the Erwin Schr\"{o}dinger International Institute for Mathematical
Physics. I thank them for their hospitality and
financial support. Additional support for this work was provided by
NSF grants PHY-0653646 and DMS-0856196, and FQXi grant RFP2-08-04.
\end{acknowledgements}

\begin{abstract}
%Abstract should be double-spaced and limited to 350 words or 2,450
%characters.
 
  In higher symplectic geometry, we consider generalizations of symplectic manifolds
  called $n$-plectic manifolds. We say a manifold is
  $n$-plectic if it is equipped with a closed, nondegenerate form of
  degree $(n+1)$. We show that certain higher algebraic and geometric
  structures naturally arise on these manifolds. These structures can be
  understood as the categorified or homotopy analogues of important
  structures studied in symplectic geometry and geometric
  quantization. Our results imply that higher symplectic geometry
  is closely related to several areas of current interest including string
  theory, loop groups, and generalized geometry.

  We begin by showing that, just as a symplectic manifold gives a
  Poisson algebra of functions, any $n$-plectic manifold gives a Lie
  $n$-algebra containing certain differential forms which we call
  Hamiltonian.  Lie $n$-algebras are examples of strongly homotopy Lie
  algebras.  They consist of an $n$-term chain complex equipped with a
  collection of skew-symmetric multi-brackets that satisfy a
  generalized Jacobi identity.

  We then develop the machinery necessary to geometrically quantize
  $n$-plectic manifolds. In particular, just as a prequantized
  symplectic manifold is equipped with a principal $\U(1)$-bundle with
  connection, we show that a prequantized 2-plectic manifold is equipped with a
  $\U(1)$-gerbe with 2-connection. A gerbe is a
  categorified sheaf, or stack, which generalizes the notion of a principal
  bundle. Furthermore, over any 2-plectic manifold there is a vector
  bundle equipped with extra structure called a Courant
  algebroid.  This bundle is the 2-plectic analogue of the
  Atiyah algebroid over a prequantized symplectic manifold. Its space of
  global sections also forms a Lie 2-algebra. We use this Lie
  2-algebra to prequantize the Lie 2-algebra of Hamiltonian forms.
  
  Finally, we introduce the 2-plectic analogue of the Bohr-Sommerfeld
  variety associated to a real polarization, and use this to
  geometrically quantize 2-plectic manifolds. 
  For symplectic manifolds, the output from quantization is a Hilbert space of quantum
  states. Similarly, quantizing a 2-plectic manifold gives a category of quantum states. We consider a particular
  example in which the objects of this category can be identified with representations
  of the Lie group $\SU(2)$.
\end{abstract} 

\tableofcontents
%\listoffigures
%\listoftables
\end{frontmatter}

% \part{First Part}

\chapter{Introduction}
Higher symplectic geometry is a generalization of symplectic geometry which
begins with considering manifolds equipped with a closed nondegenerate form of
higher degree. This thesis explains how such a differential form gives
rise to algebraic and geometric structures which act as the higher
analogues of important structures found in symplectic geometry and geometric quantization.
Indeed, a recurring theme in this work is the idea that basic results in
symplectic geometry are specific instances of more general theorems which hold for a 
much larger class of structures. 

In particular, we focus on manifolds equipped with a closed nondegenerate
3-form. We call such manifolds `2-plectic'.  In this case, we see that higher symplectic geometry is intimately
related to string theory. We use ideas from higher category theory and homotopical algebra 
 to develop a geometric quantization procedure for 2-plectic manifolds. In doing so, we encounter 
structures known to play important roles in other string-inspired areas of current interest. These include
the theory of $L_{\infty}$-algebras, loop groups, gerbes, and generalized geometry. Our results shine new light on these structures, and suggest new relationships among the above fields. We invite the reader who has some familiarity with these ideas to skip ahead and browse Table \ref{intro_table}. There we list examples of such structures and the roles they play in the quantization of 2-plectic manifolds.

We wish to provide in this introductory chapter a gentle overview of the basic ideas behind higher symplectic geometry, and describe, with some detail, the main results of this thesis. We begin with a brief survey of symplectic geometry and geometric quantization which emphasizes the role played by classical and quantum mechanics.
Higher symplectic geometry is then introduced as a consequence of combining two known approaches to studying classical field theory: multisymplectic geometry and higher gauge theory. We conclude by providing a chapter-by-chapter summary of our main results.  

\section*{Symplectic geometry and geometric quantization}
Symplectic geometry is the study of manifolds equipped with a closed
nondegenerate 2-form. nondegeneracy, in this context,
means that the 2-form gives an isomorphism between the space of tangent vectors
and the space of 1-forms by contraction or ``lowering indices''. 
Such a 2-form produces a variety of interesting algebraic and
geometric structures. Symplectic manifolds
appear in many branches of mathematics and these structures often provide
useful characterizations of important phenomena. 
In particular, symplectic manifolds play a crucial role in 
classical mechanics and representation theory.

The origins of symplectic geometry, in fact, lie in classical
mechanics. In classical mechanics, one studies the physics of a system of
point-like particles. For many systems of interest, the state of the system at
any time is uniquely determined by specifying the position
and momentum of each particle. This state can be interpreted as a
point in a manifold called the `phase space' of the system.  The
time evolution of the system is therefore represented by a smooth path in
this manifold, which is a solution to an ordinary differential
equation called `Hamilton's equation'. Physical observables of the
system are smooth functions on the manifold. Measurement of an
observable corresponds to evaluating the function at a particular
a point of phase space.
Remarkably, the structures needed to guarantee a
solution to Hamilton's equation, and also to describe how measurements
change in time, are provided by equipping the manifold with a
symplectic 2-form.

For example, the nondegeneracy of the symplectic 2-form guarantees that
Hamilton's equations have, at least for some interval of time, a
solution. More interestingly, the symplectic structure
makes the space of functions on the manifold into a special kind of
Lie algebra called a Poisson algebra. The fact that the symplectic
2-form is closed implies that the corresponding bracket satisfies the
Jacobi identity. This Lie bracket is used to compute the time evolution of observables.

%\section{Geometric quantization}
There are many systems of interest, however, which must be studied by
using quantum mechanics, instead of classical
mechanics.  In these cases, classical mechanics can be
understood as a very rough approximation to the true physical behavior of
the system. In their attempts to understand such quantum systems,
physicists developed a process called `quantization' in which
one first considers a system classically, and then replaces these structures with their
quantum analogues. Roughly speaking, in quantum mechanics 
the states of the system no longer correspond to points on a manifold, but
rather to vectors in a Hilbert space. Observables no longer correspond to
functions on a manifold, but rather to linear operators on the Hilbert space.
The time evolution of a system is given by
a solution to a partial differential equation called `Schr\"{o}dinger's
equation', rather than Hamilton's equation. The time evolution of
observables is now determined by the commutator bracket of operators,
rather than the Poisson bracket of functions. 

Hence, within the context of symplectic geometry, the physicists' findings
suggests that quantization is a procedure which involves assigning to
a symplectic manifold a Hilbert space, and to the Poisson algebra a
representation as linear operators on this space. This is, in fact, the first step
of a rigorous procedure called `geometric quantization' developed by
Kirillov \cite{Kirillov:2004}, Kostant \cite{Kostant:1970}, and
Souriau \cite{Souriau:1967} (KKS) in the 1960's. It is based on the
following facts: If a symplectic 2-form satisfies a certain
integrality condition, then it must be the curvature of a principal
$\U(1)$-bundle equipped with a connection living over the manifold. Such a
symplectic manifold is called `prequantizable'. Certain global
sections of the associated Hermitian line bundle form a Hilbert space
whose inner product is given by the symplectic structure.
The connection on the bundle then determines a faithful
representation of the Poisson algebra as operators on this prequantum Hilbert space. 

However, in practice, this Hilbert space is ``too large''.
The second step in the KKS procedure involves choosing an additional
structure on the manifold called a `polarization'. Roughly speaking, a
polarization on a symplectic manifold is a special kind of
integrable distribution \cite{Sniatycki:1980,Woodhouse:1991}. 
The size of the Hilbert space is reduced by
considering only those sections that are covariantly constant in the
directions given by vectors contained in the distribution. 
This smaller space is called the `quantum Hilbert space', or `space of
quantum states'.

Geometric quantization may appear, at first sight, to be a rather
mysterious procedure with limited applicability. Not every symplectic
manifold is prequantizable, and not every prequantized symplectic
manifold admits a polarization. Even when such structures do exist,
there are several non-canonical choices to be made. Furthermore,
the presence of certain topological obstructions often implies that additional
fine-tuning is required. Regardless, the KKS procedure is very powerful and has led
to a large number of important results, for example, in the representation theory of
Lie groups. Here, one typically studies the symmetry group of a geometric
object by first understanding the algebraic representation theory of the
group. Kirillov and Kostant's original motivation for developing
geometric quantization was, in some sense, the converse: to construct the
representations of groups as geometric objects.  Indeed, the central
tenet of Kirillov's orbit method \cite{Kirillov:2004} is, roughly, that an irreducible
representation of a Lie group corresponds to a particular symplectic manifold
equipped with an action of the group. The representation itself is
recovered as the quantum Hilbert space obtained from geometric
quantization.

%The Poisson algebra can be identified as the Lie algebra of invariant vector
%fields on the bundle's total space that act as infinitesimal connection-preserving
%automorphisms of the bundle. In this way one obtains 
%\section{Higher degree forms}

\section*{Higher degree, higher dimension, and higher structure}
After digesting all of this, the curious reader
might ask a simple question: What is so special about 2-forms? 
After all, many manifolds admit interesting closed forms
of higher degrees, and some of these, such as volume forms, are
``nondegenerate''. It is also reasonable to ask how
much, if any, of the above story involving symplectic geometry and quantization
carries over to manifolds equipped with such forms. 
The main goal of this thesis is to address these questions.

At its most basic level, higher symplectic geometry involves studying
manifolds equipped with a closed, nondegenerate form of higher degree.
We call such a manifold `$n$-plectic' if the form has degree $(n+1)$, so
that a 1-plectic manifold is a symplectic manifold.
Here, nondegeneracy means that the $n$-plectic form injectively
maps the space of tangent vectors into the space of $n$-forms, again by contraction. In
contrast with the symplectic case, this injection is not necessarily an isomorphism.
Many examples of $n$-plectic manifolds appear ``in nature''. 
These include orientable manifolds, exterior powers of cotangent bundles, and compact simple Lie groups.
%Multisymplectic geometry has undergone a
%fair amount of development, most of which is based on the work of
%Cantrijn, Ibort, and de Le\'{o}n. Generalizations of the important types of submanifolds
%considered in symplectic geometry, such as Lagrangian submanifolds,
%can be found in multisymplectic geometry. Similarly, the polarizations
%discussed earlier in the context of geometric quantization
%have been appropriately generalized to the multisymplectic case.
%However, as we will see later on, there are many interesting differences between the symplectic
%and multisymplectic cases. 

Usually, $n$-plectic manifolds go by the name of
multisymplectic manifolds \cite{Cantrijn:1999}. 
Just as symplectic geometry has its origins in the classical mechanics
of particles, multisymplectic geometry was initially
developed to study higher-dimensional classical field theories.
Let us briefly explain what this means. As previously mentioned, 
the time evolution of a point-like particle is described by a path which
depends on one variable: time. So, the `world-line' of a
zero-dimensional object is determined by a map from a one-dimensional
manifold. A physicist might call classical mechanics a
$(0+1)$-dimensional field theory. However, describing the behavior
of a higher-dimensional object, such as a string, requires
more variables. The amplitude of a vibrating string depends on both
time and the position along the string. Hence, the time evolution of the
one-dimensional string is described by a map from a 2-dimensional manifold or
`world-sheet'.  In this way, string theory  is a $(1+1)$-dimensional field theory.
In general, the physics of a $(n-1)$-dimensional object, or `brane',
is described by a $n$-dimensional field theory.

The basic ideas in multisymplectic geometry can be found in 
Weyl's 1935 work on the calculus of variations \cite{Weyl:1935}.
%\footnote{It was Weyl who coined the term ``symplectic''.} 
It was further developed in the 1970's mainly by the Polish school
of mathematical physics. The work of Kijowski \cite{Kijowski:1973gi},
Tulczyjew \cite{KT:1979}, and others \cite{Roman-Roy:2009}
showed that, just as symplectic manifolds can by used as phase spaces
for $(0+1)$-dimensional field theories, multisymplectic manifolds can
be used as `multiphase' spaces for higher-dimensional field
theories. Specifically, the multiphase space used to describe the
physics of an $(n-1)$-dimensional object is an $n$-plectic manifold.
A solution to a partial differential equation called the de
Donder-Weyl equation corresponds to a particular $n$-dimensional
submanifold of this space. The data encoded by these submanifolds
include the value of the field as well as the value of its
`multi-momentum' at each point in space and time.  The multi-momentum
is a quantity that is related to the time and spatial derivatives of
the field, in a manner similar to the relationship between the
velocity of a point particle and its momentum. This formalism has
several attractive mathematical features, but it still needs further development
before it can replace more common frameworks used by physicists to
study field theories.

The work of Baez and Schreiber \cite{BaezSchreiber:2005}, Freed \cite{Freed:1994},
Schreiber \cite{Schreiber:2005}, Sati, Schreiber, and Stasheff
\cite{SSS:2009} suggests that structures found in classical
mechanics can be generalized by using higher category and homotopy
theory and then applied to the study of higher-dimensional field
theories. 
%Examples of such structures which are relevant to this thesis include Lie $n$-algebras, 
%$n$-bundles, and $n$-stacks. 
So far this viewpoint has been most fruitful in studying the
string and brane-theoretic generalizations of gauge
theory. Although the details are quite technical, the basic philosophy
behind higher gauge theory is very simple. While a
classical particle has a position nicely modelled by an element 
of a set, namely a point in space:
\[          \bullet  \]
the position of a classical string is better modelled by a morphism in a 
category, namely an unparametrized path in space:
\[
\xymatrix{
   \bullet \ar@/^1pc/[rr]^{}
&& \bullet
}
\]
Similarly, the time evolution of a particle 
can be thought of as a morphism, while the time evolution of a string 
can be thought of as a 2-morphism, or 2-cell:
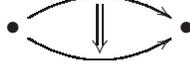
\[
\xymatrix{
   \bullet \ar@/^1pc/[rr]_{}="0"
           \ar@/_1pc/[rr]_{}="1"
           \ar@{=>}"0";"1"^{}
&& \bullet 
}
\]

So, both higher degree forms on manifolds and higher
structures can be used to study higher-dimensional field theories. 
Motivated by this idea, we suspect 
that the higher analogues of well-known structures on
symplectic manifolds should naturally arise on $n$-plectic manifolds. 
The work presented in this thesis confirms this hunch, and
we understand higher symplectic geometry as the formalism which completes the following
diagram:
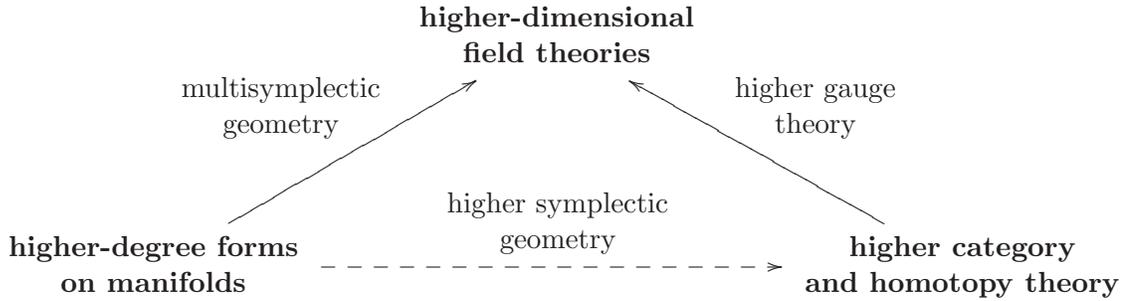
\vspace{.5cm}
\[
\xymatrix{
& \text{ \begin{tabular}{c} \textbf{higher-dimensional}\\ \textbf{field theories} \end{tabular}}&\\
&&&\\
\ar[uur]^{\scriptstyle \text{\begin{tabular}{c} multisymplectic\\geometry \end{tabular}}} 
\text{\begin{tabular}{c} \textbf{higher-degree forms}\\\textbf{on
      manifolds} \end{tabular}}
\ar @{-->}[rr]^{\scriptstyle \text{\begin{tabular}{c} higher symplectic\\geometry \end{tabular}}} 
& & 
\ar[uul]_{\scriptstyle \text{\begin{tabular}{c} higher gauge\\theory \end{tabular}}} 
\text{\begin{tabular}{c} \textbf{higher category} \\ \textbf{and homotopy theory} \end{tabular}}
}
\]
\vspace{.5cm}\\

\begin{table} 
\begin{tabular}{c c c}
\cline{2-3}
& \multicolumn{1}{|c|}{{\bf symplectic geometry}} &
\multicolumn{1}{|c|}{{\bf 2-plectic geometry}} \\ \hline
\multicolumn{1}{|c|}{degree of} & \multicolumn{1}{|c|}{\multirow{2}{*}{2}} &
\multicolumn{1}{|c|}{\multirow{2}{*}{3}} \\ 
\multicolumn{1}{|c|}{differential form} & \multicolumn{1}{|c|}{} & \multicolumn{1}{|c|}{} \\ \hline
\multicolumn{1}{|c|}{\multirow{3}{*}{examples}} &
\multicolumn{1}{|c|}{cotangent bundles} &
\multicolumn{1}{|c|}{exterior square of cotangent bundles} \\ \cline{2-3}
 \multicolumn{1}{|c|}{} & \multicolumn{1}{|c|}{coadjoint orbits} &
 \multicolumn{1}{|c|}{\multirow{2}{*}{compact simple Lie groups}} \\
 \multicolumn{1}{|c|}{} &\multicolumn{1}{|c|}{of Lie
  groups}&\multicolumn{1}{|c|}{}\\ \hline \hline
\multicolumn{3}{|c|}{classical field theory} \\ \hline \hline
\multicolumn{1}{|c|}{physical objects} &
\multicolumn{1}{|c|}{particles} &
\multicolumn{1}{|c|}{strings} \\ \hline
\multicolumn{1}{|c|}{\multirow{2}{*}{observables}} & \multicolumn{1}{|c|}{\multirow{2}{*}{Lie algebra
of functions}} & \multicolumn{1}{|c|}{Lie 2-algebra of}\\
\multicolumn{1}{|c|}{} & \multicolumn{1}{|c|}{} &
\multicolumn{1}{|c|}{Hamiltonian 1-forms} \\ \hline
\multicolumn{1}{|c|}{\multirow{2}{*}{measurement}} & \multicolumn{1}{|c|}{$x=\text{point in phase space}$} & 
\multicolumn{1}{|c|}{$\gamma=\text{path in multiphase space}$} \\
\multicolumn{1}{|c|}{} & \multicolumn{1}{|c|}{$x \mapsto f(x)$} &
\multicolumn{1}{|c|}{$\gamma \mapsto \int_{\gamma} \alpha$} \\ \hline \hline
\multicolumn{3}{|c|}{prequantization} \\ \hline \hline
\multicolumn{1}{|c|}{\multirow{5}{*}{prequantum structure}} &
\multicolumn{1}{|c|}{principal $\U(1)$-bundle} &
\multicolumn{1}{|c|}{$\U(1)$-gerbe}  \\
\multicolumn{1}{|c|}{} & \multicolumn{1}{|c|}{with connection} &  \multicolumn{1}{|c|}{with 2-connection}\\
\multicolumn{1}{|c|}{}& \multicolumn{1}{|c|}{\bf or} &
\multicolumn{1}{|c|}{\bf or} \\
\multicolumn{1}{|c|}{} & \multicolumn{1}{|c|}{Hermitian line bundle} &
\multicolumn{1}{|c|}{2-line stack} \\
\multicolumn{1}{|c|}{} & \multicolumn{1}{|c|}{with connection} &
\multicolumn{1}{|c|}{with 2-connection}\\ \hline
\multicolumn{1}{|c|}{local data for} &
\multicolumn{1}{|c|}{Deligne 1-cocycle:} & \multicolumn{1}{|c|}{Deligne
2-cocycle:} \\
\multicolumn{1}{|c|}{prequantum} & \multicolumn{1}{|c|}{transition functions,}&
\multicolumn{1}{|c|}{transition functions,}\\ 
\multicolumn{1}{|c|}{structure} & \multicolumn{1}{|c|}{1-forms} & \multicolumn{1}{|c|}{1-forms, 2-forms}\\ \hline
\multicolumn{1}{|c|}{infinitesimal symmetries} & \multicolumn{1}{|c|}{Atiyah algebroid} & \multicolumn{1}{|c|}
{Courant algebroid}\\
\multicolumn{1}{|c|}{of prequantum structure} & \multicolumn{1}{|c|}{(Lie algebroid)} & \multicolumn{1}{|c|}{(Lie 2-algebroid)} \\ \hline \hline
\multicolumn{3}{|c|}{quantization} \\ \hline \hline
\multicolumn{1}{|c|}{\multirow{2}{*}{example}} &
\multicolumn{1}{|c|}{$\R^{2}\setminus \{0\}$,} &
\multicolumn{1}{|c|}{$\R^{3} \setminus \{0\}$,} \\ 
\multicolumn{1}{|c|}{} & \multicolumn{1}{|c|}{$\omega = d\theta$} & \multicolumn{1}{|c|}{$\omega=dB$} \\ \hline
\multicolumn{1}{|c|}{polarization} & \multicolumn{1}{|c|}{concentric circles} & \multicolumn{1}{|c|}{concentric spheres} \\ \hline
\multicolumn{1}{|c|}{Bohr-Sommerfeld} & \multicolumn{1}{|c|}{\multirow{2}{*}{$\int_{S^{1}} \theta \in 2\pi i \Z$}} & 
\multicolumn{1}{|c|}{\multirow{2}{*}{$\int_{S^{2}} B \in 2\pi i \Z$}} \\
\multicolumn{1}{|c|}{condition} & \multicolumn{1}{|c|}{} & \multicolumn{1}{|c|}{} \\ \hline
\multicolumn{1}{|c|}{quantum} & \multicolumn{1}{|c|}{wavefunctions of} & \multicolumn{1}{|c|}{representations} \\ \multicolumn{1}{|c|}{ states} & \multicolumn{1}{|c|}{harmonic oscillator} & \multicolumn{1}{|c|}{of $\SU(2)$} \\\hline
\end{tabular}
\caption{Examples of structures found in symplectic geometry and higher symplectic geometry (for the 2-plectic case).
Comparisons of their roles in field theory, prequantization, and quantization are listed.}
\label{intro_table}
\end{table}

\section*{Overview of main results}
We now describe the main results in this thesis.
In Table \ref{intro_table}, we list 
some particular examples to keep in mind while reading this section.

We first present some basic facts about $n$-plectic manifolds in Chapter \ref{nplectic_geometry}. 
We say an $(n+1)$-form
$\omega$ on $M$ is $n$-plectic if $\omega$ is closed and
nondegenerate. By nondegenerate, we mean the contraction
\begin{equation} \label{intro_eq_1}
\begin{array}{c}
TM \to \Lambda^{n}T^{\ast}M\\
v \mapsto \omega(v,-)
\end{array}
\end{equation}
is injective. For the most part, we follow 
Cantrijn, Ibort, and de Le\'{o}n's work on multisymplectic manifolds \cite{Cantrijn:1999}. In
particular, we use their generalizations of the familiar notions of
Lagrangian submanifolds and real polarizations found in symplectic
geometry.

Next, in Chapter \ref{algebra_chapter}, we extend the algebraic structures
found in symplectic geometry to the $n$-plectic setting.
Given an $n$-plectic
manifold $(M,\omega)$, we show that the $n$-plectic structure naturally induces a
skew-symmetric bracket on a particular subspace of $(n-1)$-forms,
which we call Hamiltonian. An $(n-1)$- form $\alpha$ is Hamiltonian if
there exists a vector field $v$ such that
\[
d\alpha =-\omega(v,-).
\]
The vector field $v$ is called the Hamiltonian vector field associated
to $\alpha$. In the 1-plectic/symplectic case, we see that every
0-form is Hamiltonian, and our bracket reduces to the Poisson bracket of
functions. However, for higher values of $n$, the bracket
only satisfies the Jacobi identity up to an exact form. This
leads us to the notion of a Lie $n$-algebra. Lie $n$-algebras (equivalently,
$n$-term $L_{\infty}$-algebras \cite{LS}) are higher analogs of differential
graded Lie algebras. They consist of a graded vector space
concentrated in degrees $0,\ldots,n-1$, and are equipped
with a collection of skew-symmetric $k$-ary brackets, for $1 \leq k
\leq n+1$, that satisfy a generalized Jacobi identity.
In particular, the $k=2$ bilinear bracket behaves
like a Lie bracket that only satisfies the ordinary Jacobi identity up
to higher coherent chain homotopy. 
In Theorem \ref{main_thm}, we prove that, given an $n$-plectic
manifold, one can explicitly construct a Lie $n$-algebra on a
complex consisting of Hamiltonian $(n-1)$-forms and arbitrary
$p$-forms for $0 \leq p \leq n-2.$ The bilinear bracket, as well as
all higher $k$-ary brackets, are completely determined by the
$n$-plectic structure.  

We consider an important example of this
construction in Chapter \ref{lie_group_chapter}: the Lie 2-algebra
arising from a compact simple Lie group. Every such Lie group has a
1-parameter family of canonical 2-plectic structures generated by the
`Cartan 3-form'. These 3-forms are used to build central extensions of, and line bundles on, 
the corresponding loop group \cite{PressleySegal}. They also play a key
role in the theory of gerbes on Lie groups \cite{Meinrenken:2003} and the quantization of
conjugacy classes \cite{Mohrdieck:2004}. We show how the Lie 2-algebra of Hamiltonian 1-forms
on a compact simple Lie group $G$ relates to the `string Lie
2-algebra' of $G$ \cite{HDA6}. It is known that the string Lie 2-algebra can be
integrated to a `Lie 2-group' \cite{Henriques:2008}. This Lie 2-group can be geometrically realized as
a topological group which appears in the study of spin structures on
loop spaces.

Since geometric quantization has seen so much success in symplectic
geometry, we wish to extend it to the $n$-plectic
setting. In symplectic geometry, prequantization involves equipping the
manifold with a principal $\U(1)$-bundle with a
connection, whose curvature is the symplectic 2-form.
Therefore, in Chapter \ref{stacks_chapter}
we consider `stacks', the 2-plectic analogue of
bundles. A stack on a manifold can be
thought of as a categorified sheaf i.e.\ an assignment of a category to
each open neighborhood of the manifold. In particular, the higher
analogue of a principal $\U(1)$-bundle is a special kind of stack
called a `$\U(1)$-gerbe'. Just as a section of a $\U(1)$-bundle locally
looks like a $\U(1)$-valued function, a section of a $\U(1)$-gerbe locally
looks like a principal $\U(1)$-bundle. 

We then review Brylinski's
theory of `2-connections' for $\U(1)$-gerbes \cite{Brylinski:1993}. To understand what a
2-connection is, first recall that a $\U(1)$-bundle with
connection can be described by local transition functions and 1-forms
satisfying certain compatibility conditions. This local data
represents a degree 1 class in `Deligne cohomology', which can be
thought of as a refinement of the usual classification of 
bundles by \v{C}ech cohomology. Similarly, a 
$\U(1)$-gerbe  equipped with a 2-connection  can be described by local
transition functions, 1-forms, and 2-forms. This local data gives a
degree 2 class in Deligne cohomology. Just as the curvature of a
connection on a principal bundle is a 2-form, the `2-curvature' of a
2-connection is a 3-form. 
In general, we define a prequantized $n$-plectic manifold to be an $n$-plectic manifold
equipped a Deligne $n$-cocycle whose $n$-curvature is, up to sign, the $n$-plectic form.
As in the symplectic case, we show in Propositions \ref{curvature=integral} and
\ref{integral=curvature} that only those $n$-plectic
manifolds which satisfy an integrality condition can be prequantized.

In the remainder of the thesis, we focus on developing a quantization
scheme for 2-plectic manifolds. For prequantized symplectic manifolds,
the prequantum Hilbert space is obtained by considering global
sections of the Hermitian line bundle associated to the
$\U(1)$-bundle. We generalize this to 2-plectic manifolds by
constructing the `2-line stack' associated to a $\U(1)$-gerbe. 
Sections of the 2-line stack locally look like Hermitian vector
bundles. In Section \ref{2-line_stack_section}, we use some basic
ideas from `2-bundle theory' to explain why 2-line stacks are a
natural generalization of line bundles. We also present a formalism by
Carey, Johnson, and Murray \cite{Carey:2004} which generalizes the notion of holonomy to $\U(1)$-gerbes equipped with a 2-connection. We shall use this `2-holonomy' in our quantization procedure for 2-plectic manifolds.

In Chapter \ref{prequant_chapter}, we consider prequantization for
2-plectic manifolds in detail. In order to understand our results, it is, again, helpful
to momentarily return to the symplectic case. For a prequantized symplectic manifold, the
connection on the principal bundle determines a representation of the Poisson algebra
as linear operators on the prequantum Hilbert space. This representation 
identifies the Poisson algebra with certain $\U(1)$-invariant vector
fields on the bundle's total space. These vector fields are
characterized by the fact that their flows are connection-preserving automorphisms of the
bundle. Therefore, the Poisson algebra acts as linear differential
operators on the space of smooth complex-valued functions on the total
space. The prequantum Hilbert space is built using global
sections of the associated Hermitian line bundle, and there is a way to
interpret these sections as functions on the total space of the
principal bundle. Hence, the Poisson algebra acts as operators on this
Hilbert space.

This process of representing the Poisson algebra as operators can be
nicely explained in terms of the Atiyah sequence associated to a
principal bundle. Over any prequantized symplectic manifold, there is a
special kind of vector bundle called the `Atiyah algebroid' \cite{CdS-Weinstein:1999}. The global
sections of this vector bundle are the $\U(1)$-invariant vector
fields on the total space of the principal $\U(1)$-bundle. Hence, the space of
sections form a Lie algebra under the Lie bracket of vector
fields. In fact, the Atiyah algebroid is an example of a more general
structure called a `Lie algebroid'. The representation we described in
the previous paragraph corresponds to an injective Lie algebra
morphism embedding the Poisson algebra into the global sections of the
Atiyah algebroid.

We define a prequantized 2-plectic manifold to be an integral
2-plectic manifold equipped with a $\U(1)$-gerbe with
2-connection. A construction given by Hitchin \cite{Hitchin:2004ut} associates to any such gerbe on a manifold, a
vector bundle called a `Courant algebroid'. Its
space of global sections is equipped with a skew-symmetric bracket
which gives it the structure of a Lie 2-algebra. Hence, the Courant algebroid can be understood as a 
`Lie 2-algebroid'. This `Courant bracket' plays an important role in
generalized complex geometry \cite{Gualtieri:2007}
and Poisson geometry \cite{Liu:1997}. Beginning in Section \ref{courant_sec}, we show how the Courant
algebroid associated to a $\U(1)$-gerbe is the higher analogue of
the Atiyah algebroid associated to a $\U(1)$-bundle. Such an analogy
was conjectured to exist by Bressler and Chervov \cite{Bressler-Chervov} as well as others.
Our main result in this chapter is Theorem \ref{main_courant_thm}. It implies that the
2-connection of a gerbe on a prequantized 2-plectic manifold induces
an injective morphism from the Lie 2-algebra of Hamiltonian 1-forms
into the Lie 2-algebra of global sections of the Courant algebroid.
In this way, we obtain a prequantization of the Hamiltonian 1-forms, in
complete analogy with the symplectic case.

Finally, in Chapter \ref{quantization_chapter}, we use the 2-plectic
analogue of `real polarizations' to fully geometrically quantize
2-plectic manifolds. A real polarization on a prequantized symplectic manifold is a
certain kind of foliation. Over any leaf of the polarization, the prequantum
bundle restricts to a flat bundle.
The prequantum Hilbert space of global sections is cut down by considering only those
sections covariantly constant along the leaves of the
polarization. However, there are topological obstructions to 
obtaining a non-trivial Hilbert space from this process. For example, if the leaves
of the polarization are not simply-connected, then we are forced to
consider only the leaves on which the restricted bundle has trivial
holonomy. The collection of all such leaves is called the
`Bohr-Sommerfeld variety' associated to the polarization \cite{Sniatycki:1980}. 
The space of quantum states is built using certain sections which are covariantly constant
on the leaves contained in the variety.
As the name suggests, there is a relationship between this construction and the old Bohr-Sommerfeld
quantization rules from physics.

Before we go to the 2-plectic case, we review a well-known example in
symplectic geometry in Section \ref{harm_osc_sec}.
We quantize the punctured plane $M=\R^{2}\setminus \{0\}$, equipped with a volume-form $\omega=d\theta$, as the phase space of the
`simple harmonic oscillator'. Here $\theta$ is not the angular
coordinate on $M$, but rather a global 1-form which is related to the energy of the oscillator. 
We prequantize 
$M$ using the trivial principal $\U(1)$-bundle with connection $\theta$. The associated Hermitian line bundle is the trivial line bundle. We choose the polarization given by concentric circles about the origin. The corresponding Bohr-Sommerfeld
variety is a countable subset of these circles. We find sections of the
prequantum line bundle over the Bohr-Somerfeld variety which are
covariantly constant along the circles contained in the variety. This is
equivalent to finding solutions to the Schr\"{o}dinger wave equation. 
After applying a small correction, the radii of the circles in the variety correspond to the discrete energy levels for
the quantized oscillator. 

We generalize this entire construction to the 2-plectic case in
Section \ref{cat_quant_section}. We start with a prequantized
2-plectic manifold equipped with a Deligne 2-cocycle. We consider the
associated 2-line stack with 2-connection whose 2-curvature is the
2-plectic structure. The 2-plectic analogue of the prequantum Hilbert
space is the category of global sections of the 2-line stack, i.e.\ the
category of twisted Hermitian vector bundles on the manifold.

We quantize the manifold by choosing a real polarization as defined in
Chapter \ref{nplectic_geometry}. Over any leaf of the polarization,
the 2-line stack restricts to a `flat stack' i.e.\ the 2-curvature
vanishes. The Bohr-Sommerfeld variety associated to the polarization
is made up of those leaves on which the restricted 2-line stack has
trivial 2-holonomy. Here, we use the 2-holonomy formalism for Deligne
2-cocycles which we described in Chapter \ref{stacks_chapter}. The
2-plectic analogue of the space of quantum states is the category of
quantum states. Its objects are twisted vector bundles over the
Bohr-Sommerfeld variety whose restriction to each leaf in the variety
is `twisted-flat'. This twisted-flat condition replaces the covariantly constant
condition used in the symplectic case.

As an example of 2-plectic quantization, we consider the space
$M=\R^{3} \setminus \{0\}$ equipped with a particular volume form
$\omega=dB$. 
We prequantize the space using the trivial $\U(1)$-gerbe
whose 2-connection is given by the global 2-form $B$. The
associated 2-line stack in this case is equivalent to the stack of Hermitian vector
bundles equipped with connection over $M$. (There is no twisting since
the Deligne 2-cocycle is just a global 2-form.) We choose the
polarization given by concentric spheres about the origin.

A sphere centered about the origin in $\R^{3}$ is a coadjoint orbit of the Lie group $\SU(2)$. This can easily seen by identifying $\R^3$ with $\su(2) \cong \su(2)^{\ast}$. It turns out that the restriction of $B$ to any such sphere
gives the famous KKS symplectic form used in Kirillov's orbit method \cite{Kirillov:2004}.
By definition, a sphere is included in the Bohr-Sommerfeld variety if
the Deligne 2-cocycle given by $B$ has trivial
2-holonomy. Requiring trivial 2-holonomy is equivalent to the KKS
symplectic form satisfying an integrality condition, which further
implies that it is the curvature of a line bundle. We use some basic
facts about the orbit method to pass from bundles to
representations. We show that, in this example, the category of quantum states obtained
from our quantization process is closely related to the category of finite-dimensional representations of
$\SU(2)$. This suggests that, in some sense, 2-plectic quantization categorifies Kirillov's
orbit method. Interestingly, the process fails to produce 
representations whose decomposition into irreducibles contains
the trivial representation of $\SU(2)$. However, this is somewhat expected, since
it is well known that the analogous quantization procedure for the harmonic oscillator in
symplectic geometry requires an additional correction in order to obtain
the correct space of quantum states. 

We conclude the thesis in Chapter\ref{conclusions} by providing a
technical summary of the main results, and by
discussing some open problems and future directions for research. 

\section*{Previous work}
We have recently published some of the results presented here. 
Theorem\ \ref{main_thm} in Chapter
\ref{algebra_chapter} and Proposition \ref{n-plectic_Leibniz}
in Appendix \ref{leibniz_appendix}  appear in \cite{Rogers:2010}. 
Theorem \ref{string_Lie_Thm} in Chapter \ref{lie_group_chapter}
appears in \cite{Baez:2009uu}, which was co-authored with J.\ Baez. 
The other results in Chapter \ref{lie_group_chapter} generalize or improve upon
those of \cite{Baez:2009uu}.
Chapter \ref{prequant_chapter} is based on a recent preprint
\cite{Rogers:2011}, which has been submitted for publication. 
Finally, a different proof of Theorem \ref{isomorphism_L2A_thm} in
Appendix \ref{leibniz_appendix} appears in \cite{Baez:2008bu}, which
was co-authored with J.\ Baez and A.\ Hoffnung.
\chapter{$n$-Plectic geometry} \label{nplectic_geometry}
Our basic geometric objects of interest are $n$-plectic manifolds:
manifolds equipped with a closed, nondegenerate form of degree $n+1$.
Hence, a 1-plectic manifold is a symplectic manifold. $n$-Plectic
manifolds are also called multisymplectic manifolds.
Multisymplectic geometry originated in covariant Hamiltonian formalisms for classical
field theory, just as symplectic geometry originated in classical
mechanics. However, multisymplectic manifolds can be found outside
the context of classical field theory, and are interesting from a purely geometric point of view.
A few different definitions for multisymplectic structures exist in
the literature. We adopt the formalism developed by Cantrijn, Ibort,
and de Le\'{o}n \cite{Cantrijn:1999}, since it provides the simplest
generalization of symplectic structures, and also encapsulates a wide
variety of interesting examples.

\section{Linear theory}
We begin by introducing multisymplectic/$n$-plectic structures on vector
spaces. For the most part, we only present those aspects of the theory needed for
subsequent chapters. For more details, we refer the reader to \cite{Cantrijn:1999}.  
\begin{definition}
\label{n-plectic_vspace}
An $(n+1)$-form $\omega$ on a vector space is 
{\boldmath $n$}{\bf-plectic} iff it is nondegenerate:
\[
    \forall v \in V \ \iota_{v} \omega =0 \Rightarrow v =0.
\]
If $\omega$ is an $n$-plectic form on $V$, then we call the pair $(V,\omega)$ 
an {\boldmath $n$}{\bf-plectic vector space}.
\end{definition}

Note that a $1$-plectic vector space is simply a symplectic vector
space. A straightforward exercise in linear algebra shows that
$n$-plectic structures do not exist on vector spaces of dimension
$n+2$.  For the $n=1$ case, there is the stronger result that every
finite-dimensional symplectic vector space has even dimension.
Conversely, any even-dimensional vector space $V$ admits a symplectic
form $\omega$, which can be put into a normal form by choosing a
particular basis. Hence, $\mathrm{GL}(V)$ acts transitively on the
space of symplectic structures on a symplectic vector space
$(V,\omega)$. In contrast, it has been shown that if $\dim V \geq 6$,
then $n$-plectic structures on $V$ are generic for $ 2 \leq n \leq
\dim V -4$ \cite{Martinet:1970}.  Furthermore, 2-plectic structures on
real vector spaces $V$ with $\dim V \leq 7$ have been classified. In
these cases, the action of $\mathrm{GL}(V)$ is not transitive. If
$\dim V =6$, then there are 2 equivalence classes, and if $\dim V
=7$, then there are 8 classes \cite{Martinet:1970}. In general, the
classification of $n$-plectic structures remains an open problem
\cite{Cantrijn:1999}.

Next, we consider several natural generalizations of the orthogonal
complement associated to a bilinear form.
\begin{definition}[\cite{Cantrijn:1999}]
\label{ortho_comp}
Let $(V,\omega)$ be an $n$-plectic vector space and $W \subseteq V$ be a subspace.
The {\boldmath $k$}{\bf-orthogonal complement of} $W$ is the subspace
\[
W^{\perp,k}= \left \{ v \in V ~ \vert ~ \omega(v,w_{1},w_{2},\ldots,w_{k})=0 ~
\forall w_{1}, w_{2},\ldots,w_{k} \in W \right \}.
\]
\end{definition}
\noindent Hence, there is a filtration of orthogonal complements:
\[
W^{\perp,1} \ss W^{\perp,2} \ss \cdots \ss W^{\perp,n}.
\]
\begin{definition}[\cite{Cantrijn:1999}]
\label{subspaces}
A subspace $W$ of an $n$-plectic vector space $(V,\omega)$ is 
{\boldmath $k$}{\bf-isotropic} iff $W \subseteq W^{\perp,k}$, 
and {\boldmath $k$}{\bf-Lagrangian} iff $W=W^{\perp,k}$.
\end{definition}
\noindent For convenience, if $W$ is an $n$-isotropic or $n$-Lagrangian subspace of an $n$-plectic
vector space, then we will say $W$ is \textbf{isotropic} or
\textbf{Lagrangian}, respectively. The notion of a $k$-co-isotropic subspace
exists as well, but we will not need it here.

Obviously, every 1-dimensional
subspace of an $n$-plectic vector space is 1-isotropic.
Hence, the next proposition guarantees the existence of $k$-Lagrangian
subspaces for all $k \geq 1$.
\begin{prop} \label{lagrangians_exist}
Let $(V,\omega)$ be an $n$-plectic vector space.
If $W \subseteq V$ is a $k$-isotropic subspace, then for all $k' \geq k$ there
exists a $k'$-Lagrangian subspace containing $W$. 
\end{prop}
\begin{proof}
See Proposition 3.4 (iii) in the 
paper by Cantrijn, Ibort, and de L\'{e}on \cite{Cantrijn:1999}.
\end{proof}
\noindent In contrast with the symplectic case, two $k$-Lagrangian subspaces need not
have the same dimension. However, if the $n$-plectic vector space is
$(n+1)$-dimensional, then it is simply a vector space
equipped with a volume form and we have:
\begin{prop} \label{lagrangian_prop}
If $(V,\omega)$ is an $n$-plectic vector space with $\dim V =n+1$, then
a subspace $W \subseteq V$ is $n$-Lagrangian if and only if $\dim W = n$.
\end{prop}
\begin{proof}
First suppose $W=W^{\perp,n}$. Then $\dim W=k \leq n$. Let $e_1,\ldots,e_k$
be a basis for $W$, and let $e_1,\ldots,e_k ,e_{k+1},\ldots,e_{n+1}$
be its extension to a basis for $V$. Let
$\theta^{1},\ldots,\theta^{n+1}$ be the dual basis with $\theta^{i}(e_{j})=\delta^{i}_{j}$.
The $n$-plectic form can be written as
\[
\omega = r \cdot \theta^{1} \wedge \cdots \wedge \theta^{n+1},
\]
with $\deg{r} > 0$. If $w_{1},\ldots,w_{n}$ are elements of $W$, with
$w_{i}=\sum_{j=1}^{k}c_{ij}e_{j}$, and $\dim W$ is strictly less than
$n$, then
\[
\omega(v,w_{1},w_{2},\ldots,w_{n})=0
\]
for all $v \in V$. Hence, we must have $\dim W=n$.

Now suppose $W$ has dimension $n$ with basis $e_1,\ldots,e_n$.
Let $e_1,\ldots,e_n,e_{n+1}$ be the extended basis of
$V$. It is easy to see that $W \ss W^{\perp,n}$. If $v \in W^{\perp,n}$
is not an element in $W$, then its contraction with the dual basis
element $\theta^{n+1}$ is non-zero. However, we have:
\[
0=\omega(v,e_1,e_2,\ldots,e_{n})=\pm \omega(e_1,e_2,\ldots,e_{n},v)=
\pm r\cdot \theta^{n+1}(v),
\]
giving a contradiction. Hence no such $v$ exists, and therefore
$W = W^{\perp,n}$.
\end{proof}

\section{$n$-Plectic manifolds}
We now turn to the global theory. Our first definition generalizes the
definition of a symplectic manifold.
\begin{definition}
\label{n-plectic_def}
An $(n+1)$-form $\omega$ on a smooth manifold $M$ is 
{\boldmath $n$}{\bf-plectic}, or more specifically
an {\boldmath $n$}{\bf-plectic structure}, if it is both closed:
\[
    d\omega=0,
\]
and nondegenerate:
\[
\forall x \in M  ~  \forall v \in T_{x}M,\ \iota_{v} \omega =0 \Rightarrow v =0
\]
If $\omega$ is an $n$-plectic form on $M$ we call the pair $(M,\omega)$ 
an {\boldmath $n$}{\bf-plectic manifold}.
\end{definition}

\begin{remark}
In general, $n$-plectic manifolds are much more abundant than symplectic
manifolds. On a finite-dimensional manifold $M$, $n$-plectic
structures are generic for $2 \leq n \leq \dim M-4$ (i.e.\ the set of $n$-plectic
structures is comeager in $\Gamma(\Lambda^{n+1} T^{\ast}M$) by Thm.\ II
2.2 and Prop.\ II 4.2 in \cite{Martinet:1970}). Also, the remarks
made after Def.\ \ref{n-plectic_vspace} imply that no
Darboux-like theorem holds for $n$-plectic structures.
\end{remark}

Clearly, an $n$-plectic structure on an $(n+1)$-dimensional manifold
$M$ is a non-vanishing section of the top-exterior power of the
cotangent bundle. Hence, orientable manifolds equipped with a volume
form provide simple examples of $n$-plectic manifolds.
Below, we describe some other interesting examples of $n$-plectic manifolds. 

\begin{example}[Compact simple Lie groups]\label{Lie_group_example}
Every compact simple Lie group admits a 1-parameter family of canonical 2-plectic
structures. These structures have been discussed in the multisymplectic
geometry literature \cite{Cantrijn:1999,Ibort:2000}, and play an
important role in several branches of mathematics connected to string
theory.

Recall that if $G$ is a compact Lie group, then its Lie algebra $\g$ 
admits an inner product $\innerprod{\cdot}{\cdot}$ that is invariant
under the adjoint representation $\Ad \maps G \to \Aut\left(\g
\right)$.  
For any nonzero real number $k$, we can define a trilinear form
\[ \omega_{k}(x,y,z)= k \innerprod{x}{[y,z]}\]
for any $x,y,z \in \g$. 
Since the inner product is invariant under the adjoint representation,
it follows that the linear transformations $\ad_y \maps \g \to \g$
given by $\ad_y(x) = [y,x]$ are skew adjoint. 
That is,
$\innerprod{\ad_{y}(x)}{z}= -\innerprod{x}{\ad_{y}(z)}$ for all
$x,y,z \in \g$. Hence, $\omega_{k}$ is totally antisymmetric.
Moreover, $\omega_k$ is invariant under the adjoint representation since
$\left[\Ad_{g}(x),\Ad_{g}(y) \right] = \Ad_{g}\left([x,y] \right)$.

Let $L_{g} \maps G \to G$ and $R_{g} \maps G \to G$ denote left and
right translation by $g$, respectively. Let $\theta_{L} \in
\Omega^{1}(G,\g)$ denote the left-invariant Maurer-Cartan form, which
sends a vector $v \in T_{g}G$ to $L_{g^{-1}\ast} v \in \g$.
Using left translation, we can
extend $\omega_{k}$ to a left invariant 3-form $\nu_{k}$ on $G$:
%More precisely, given $g \in G$ and $v_{1},v_{2},v_{3} \in T_{g} G$ define
%a smooth section $\nu_{k}$ of $\Lambda^3 T^{\ast}G$ by
%\[ \nu_{k} \vert_{g} \left(v_{1},v_{2},v_{3} \right) = 
%\omega_{k} \left(\theta(v_{1}),\theta(v_{2}),\theta(v_{3}) \right).\] 
\begin{align*}
\nu_{k} &= \omega_{k} \left(\theta_L,\theta_L,\theta_L \right)\\
&=k \inp{\theta_L}{\bigl [ \theta_L, \theta_L \bigr ]}.
\end{align*} 
It is straightforward to show that $\nu_{k}$ is also a right invariant
3-form. Indeed, since $\Ad_{g} =L_{g \ast} \circ R_{g^{-1} \ast}$, the
invariance of $\omega_{k}$ under the adjoint representation implies 
$R^{\ast}_{g}\nu_{k} =\nu_{k}$. From the left and right invariance we
can conclude
\[d\nu_{k}=0,\] 
since any $p$-form on a Lie group that is both left and right 
invariant is closed. 

Now suppose that $G$ is a compact simple Lie group. Then $\g$ is
simple, so it has a canonical invariant inner product: the Killing 
form (up to a choice of normalization).  With this choice of inner product, 
the trilinear form $\omega_{k}$ is nondegenerate in the sense of 
Definition \ref{n-plectic_vspace}.
\begin{proposition} \label{nondegenerate}
If $G$ is a compact simple Lie group, then $(G,\nu_{k})$ is a 2-plectic manifold.
\end{proposition}
\begin{proof}
  We just need to show that $\omega_{k}$ is nondegenerate i.e.\
  if $x \in \g$ and $\omega_{k}(x,y,z)=0$ for all $y,z \in \g$ then $x=0$.
  Recall that if $\g$ is simple, then it is equal to its derived
  algebra $\bigl[\g,\g \bigr]$. Hence we may write $x= \sum^{n}_{i=1}
  [y_{i},z_{i}]$. Therefore
  \[k\innerprod{x}{x}=k\sum^{n}_{i=1}\innerprod{x}{[y_{i},z_{i}]}=\sum^{n}_{i=1}
  \omega_{k}(x,y_{i},z_{i})=0,\]  implies $x=0$ since
  $\innerprod{\cdot}{\cdot}$ is an inner product.
\end{proof}
\end{example}

\begin{example}[Exterior powers of cotangent bundles] \label{cotangent_bundles}
This next example generalizes the well-known fact that
cotangent bundles are symplectic manifolds.
Suppose $M$ is a smooth manifold, and let $X = \Lambda^n T^\ast M$ be the 
$n$-th exterior power of the cotangent bundle of $M$.  Then there
is a canonical $n$-form $\theta$ on $X$ given as follows: 
\[  \theta(v_1, \dots, v_n) \vert_{x} = x(\pi_{\ast}(v_1), \dots, \pi_{\ast}(v_n))  \]
where $v_1, \dots v_n$ are tangent vectors at the point $x \in X$, and
$\pi \maps X \to M$ is the projection from the bundle $X$ to the base
space $M$.  

We claim the $(n+1)$-form
\[           \omega = d\theta   \]
is $n$-plectic.  This can be seen by explicit computation.  Let $q^1, \dots,
q^d$ be coordinates on an open set $U \subseteq M$.  Then there is a basis
of $n$-forms on $U$ given by $dq^I = dq^{i_1} \wedge \cdots \wedge dq^{i_n}$
where $I = (i_1, \dots, i_n)$ ranges over multi-indices of length $n$.
Corresponding to these $n$-forms there are fiber coordinates $p_I$ which
combined with the coordinates $q^i$ pulled back from the base give a
coordinate system on $\Lambda^n T^* U$.  In these coordinates we have
\[    \theta = p_I dq^I , \]
where we follow the Einstein summation convention to sum over
repeated multi-indices of length $n$.  It follows
that
\[    \omega = dp_I \wedge dq^I  .\]
Using this formula one can check that $\omega$ is indeed $n$-plectic.
\end{example}

\begin{example}[Hyper-K\"{a}hler manifolds]
Let $(M,g)$ be a Riemannian manifold which admits two anti-commuting, almost
    complex structures $J_{1},J_{2} \maps TM \to TM$, i.e.\ $J_{1}^2 =
    J_{2}^2=-\id$ and $J_{1}J_{2}=-J_{2}J_{1}$. Then
    $J_{3}=J_{1}J_{2}$ is also an almost complex structure. If $J_{1},J_{2},J_{3}$
    preserve the metric $g$, then one can define the 2-forms $\theta_{1},\theta_{2},\theta_{3}$,
    where $\theta_{i}(v_{1},v_{2})=g(v_{1},J_{i}v_{2})$. If each $\theta_{i}$ is closed,
    then $M$ is called a hyper-K\"{a}hler manifold \cite{Swann:1991}. 
    Given such a manifold, one can construct the 4-form: 
   \[
   \omega=\theta_{1}\wedge \theta_{1} + \theta_{2}\wedge \theta_{2} + \theta_{3}\wedge \theta_{3}.
   \]
   Clearly, $\omega$ is closed. It is also straightforward to show 
   nondegeneracy. Indeed, suppose there existed a vector field $v$ such that
   $\omega(v,\cdot,\cdot,\cdot)=0$. A calculation shows that $\omega(v,
   J_{1}v,J_{2}v,J_{3}v)=0$ implies that $g(v,v)^{2}=0$. Since $g$ is
   Riemannian, we must have $v=0$. Hence a hyper-K\"{a}hler manifold is a 3-plectic
   manifold.
\end{example} 

\section{$k$-Lagrangian submanifolds and $k$-polarizations}
We return to our presentation of the general theory and describe some
geometric structures that will play important roles in the geometric
quantization of $n$-plectic manifolds.
\begin{definition}[\cite{Cantrijn:1999}]
\label{submanifolds}
A submanifold $N$ of an $n$-plectic manifold 
$(M,\omega)$ is {\boldmath $k$}{\bf-isotropic}
({\boldmath $k$}{\bf-Lagrangian}) iff for all $x \in N$,
$T_{x}N$ is a $k$-isotropic ($k$-Lagrangian) subspace of the $n$-plectic vector space
$(T_{x}M,\omega\vert_{x})$. 
\end{definition}
\noi As in the linear case, if $N$ is an $n$-isotropic or $n$-Lagrangian submanifold of an $n$-plectic
manifold, then we say $N$ is \textbf{isotropic} or
\textbf{Lagrangian}, respectively.
Of course, we recover the usual definitions when $n=1$.

In symplectic geometry, polarizations are defined as integrable
maximally isotropic sub-bundles of the complexified tangent bundle of
a symplectic manifold. They are used in geometric quantization to cut
down the size of the Hilbert space associated to the symplectic
manifold. Certain polarizations called ``real polarizations'' can be
understood as integrable distributions living in the real tangent
bundle rather than its complexification. We currently do not know what
an ``$n$-complex structure'' should be.
Therefore, we are only able to generalize real
polarizations to the $n$-plectic case.
\begin{definition}\label{k-polarization}
A  foliation $F$ of an $n$-plectic manifold $(M,\omega)$ is 
a {\bf real} {\boldmath $k$}{\bf-polarization} iff the leaves of $F$
are immersed $k$-Lagrangian submanifolds of $M$. 
\end{definition}
\noi For brevity, we call a real $n$-polarization on an $n$-plectic
manifold simply a \textbf{polarization}. 
We conclude with an example which we will use in Chapter \ref{quantization_chapter}.
\begin{example} \label{Rn_example}
A volume form on $M=\R^{n+1} \setminus \{0\}$ is an $n$-plectic form.
Let $F$ be the foliation of $M$ by
$n$-spheres centered about the origin. Since each leaf has codimension
1, it follows from Prop.\ \ref{lagrangian_prop} that $F$ is a real polarization of $M$. 
\end{example}
 
\chapter{Algebraic structures on $n$-plectic manifolds} \label{algebra_chapter}
From the algebraic point of view, the fundamental object in symplectic
geometry is the Poisson algebra of smooth functions whose bracket is
induced by the symplectic form. The nondegeneracy of a symplectic
2-form on $M$ induces an isomorphism from $TM$ to
$T^{\ast}M$. Hence, for every function $f$ there exists a unique
vector field $v_{f}$ such that $df =-\omega(v_{f},\cdot)$. This
assignment gives the Poisson bracket:
\begin{equation} \label{Poisson_bracket}
\brac{f}{g}=\omega(v_{f},v_{g}), \quad \forall f,g \in \cinf(M).
\end{equation}
This bracket is skew-symmetric and satisfies the Jacobi
identity. Hence, the space of smooth functions on a symplectic
manifold is a Lie algebra.\footnote{The Poisson bracket also obeys an additional 
Leibniz-like rule: $\brac{f}{gh}=\brac{f}{g}h + g\brac{f}{h}$.}
In classical mechanics, the functions play the role of the 
`observables', or measurements, of a physical system of point particles.
The Poisson bracket is used to describe how these measurements change
as the system evolves in time.

Certain complications arise if we try to repeat the above construction for an arbitrary
$n$-plectic manifold $(M,\omega)$. The nondegeneracy
of the $n$-plectic form gives an injection $TM \to \Lambda^{n}T^{\ast}M$
that is not necessarily onto. Therefore, only a subspace of the
$(n-1)$-forms on $M$ have the property that there exists a unique vector field
$v_{\alpha}$ such that 
\[
d\alpha = -\omega(v_{\alpha},\cdots).
\]
We call such $(n-1)$-forms `Hamiltonian'. Hence, we can copy the
definition of the Poisson bracket given above and
define a skew-symmetric bracket on the Hamiltonian $(n-1)$-forms
\[
\brac{\alpha}{\beta}=\omega(v_{\alpha},v_{\beta},\cdots).
\]
However, as we will see in Lemma \ref{tech_lemma}, this bracket only satisfies the
Jacobi identity up to an exact form:
\begin{equation} \label{intro_jacobi}
    \brac{\alpha}{\brac{\beta}{\gamma}} -
   \brac{\brac{\alpha}{\beta}}{\gamma} 
    -\brac{\beta}{\brac{\alpha}{\gamma}}
    =-d\bigl (\omega(v_{\alpha},v_{\beta},v_{\gamma},\cdots) \bigr).
\end{equation}
\noi Therefore, it is not necessarily a Lie bracket for $n>1$.

Roughly speaking, we can imagine the Hamiltonian forms as being
part of a complex $L$ whose boundary operator is the de Rham
differential, and interpret the left-hand side of 
Eq.\ \ref{intro_jacobi} as the difference of two chain maps:
\[
\brac{\cdot}{\brac{\cdot}{\cdot}} \maps L \tensor L \tensor L \to L,
\]
and
\[
\brac{\brac{\cdot}{\cdot}}{\cdot} + \brac{\cdot}{\brac{\cdot}{\cdot}}
\maps L \tensor L \tensor L \to L.
\]
From this point of view, the right-hand side of Eq.\
\ref{intro_jacobi} suggests that we interpret the evaluation of
$\omega$ on three Hamiltonian vector fields as a chain homotopy. This
leads us to consider an algebraic structure called a Lie $n$-algebra.

Lie $n$-algebras are
higher analogs of differential graded Lie algebras (DGLAs).
They consist of a graded vector space concentrated in degrees $0,\ldots,n-1$ and are
equipped with a collection of skew-symmetric $k$-ary brackets, for $1
\leq k \leq n+1$, that satisfy a generalized Jacobi identity
\cite{Lada-Markl,LS}. In particular, the $k=2$ bilinear bracket
behaves like a Lie bracket that only satisfies the ordinary Jacobi
identity up to `higher coherent' chain homotopy. When $n=1$, we
recover the definition of an ordinary Lie algebra. 
For $n=\infty$, we obtain the more
general notion of an $L_{\infty}$-algebra, which was first
discovered by Schlessinger and Stasheff \cite{Schlessinger:1979}.
The definition of a Lie $n$-algebra may seem at first
rather artificial. However, they are ubiquitous in
mathematical physics and in certain areas of algebraic topology. In fact,
there is an alternative definition of an $L_{\infty}$-algebra, based on 
a construction of Quillen \cite{Quillen:1969}, 
which shows that it is an obvious and quite natural generalization of a
DGLA.

The main result of this chapter is Theorem \ref{main_thm}. Given an $n$-plectic
manifold, we explicitly construct a Lie $n$-algebra on a complex
consisting of the Hamiltonian $(n-1)$-forms and arbitrary $p$-forms for $0 \leq p \leq
n-2.$ The bilinear bracket, as well as all higher $k$-ary brackets, are
specified by the $n$-plectic structure. For $n=1$, the Lie 1-algebra
we obtain from this construction is the underlying Lie algebra of the
Poisson algebra of a symplectic manifold. For a 2-plectic manifold
representing the `multi-phase' space of a bosonic string, we showed in
our work with Baez and Hoffnung that the Lie 2-algebra of Hamiltonian
1-forms contains the physical observables used in string theory \cite{Baez:2008bu}. 
Hence, we often refer to the Lie $n$-algebra arising from an
$n$-plectic manifold as the ``algebra of observables''.

In Appendix \ref{leibniz_appendix}, we consider other algebraic
structures which naturally arise in
higher symplectic geometry: dg Leibniz algebras and Roytenberg's weak
Lie 2-algebras.

\section{Hamiltonian forms}
In this section, we equip the space of Hamiltonian $(n-1)$-forms
on an $n$-plectic manifold with a bilinear skew-symmetric bracket,
and note some of its properties. In order to aid our computations, we
introduce some notation and review the Cartan calculus involving multivector
fields and differential forms. We follow the notation and sign
conventions found in Appendix A of the paper by Forger, Paufler, and
R\"{o}mer \cite{Forger:2002ak}. 

Let $\X(M)$ be the $\cinf(M)$-module of
vector fields on a manifold $M$ and
let 
\[
\LX(M)=\bigoplus^{\dim M}_{k=0}\Lambda^{k} \left(\X(M) \right)
\]
be the graded commutative algebra of multivector fields. On $\LX(M)$ there is a
$\R$-bilinear map $[\cdot,\cdot] \maps \LX(M) \times \LX(M) \to \LX(M)$ called
the \textbf{Schouten bracket}, which gives $\LX(M)$ the structure of a
Gerstenhaber algebra. This means the Schouten bracket is a
degree $-1$ Lie bracket which satisfies the graded Leibniz rule with respect to the
wedge product. The Schouten bracket of two decomposable multivector fields
$u_{1} \wedge \cdots \wedge u_{m}, v_{1} \wedge \cdots \wedge v_{n}
\in \LX(M)$ is
\begin{multline} \label{Schouten}
\left [ u_{1} \wedge \cdots \wedge u_{m}, v_{1} \wedge \cdots \wedge
  v_{n} \right] 
= \\\sum_{i=1}^{m} \sum_{j=1}^{n} (-1)^{i+j} [u_{i},v_{j}]
\wedge u_{1} \wedge \cdots \wedge \hat{u}_{i} \wedge  \cdots \wedge
u_{m}\\
\quad \wedge v_{1} \wedge \cdots \wedge \hat{v}_{j} \wedge \cdots \wedge v_{n},
\end{multline}

where $[u_{i},v_{j}]$ is the usual Lie bracket of vector fields.

Given a form $\alpha \in \Omega^{\bullet}(M)$, the \textbf{interior product} of a decomposable
multivector field $v_{1} \wedge \cdots \wedge v_{n}$ with $\alpha$ is
\begin{equation} \label{interior}
\iota(v_{1} \wedge \cdots \wedge v_{n}) \alpha = \iota_{v_{n}} \cdots
\iota_{v_{1}} \alpha,
\end{equation}
where $\iota_{v_{i}} \alpha$ is the usual interior product of vector
fields and differential forms. The interior product of an arbitrary
multivector field is obtained by extending the above formula by $\cinf(M)$-linearity. 

The \textbf{Lie derivative} $\L_{v}$ of a differential form along a multivector field $v \in
\LX(M)$ is defined via the graded commutator of $d$ and $\iota(v)$:
\begin{equation} \label{Lie}
\L_{v} \alpha =  d \iota(v) \alpha - (-1)^{\deg{v}} \iota(v) d\alpha,
\end{equation}
where $\iota(v)$ is considered as a degree $-\deg{v}$ operator.

The last identity we will need involving multivector fields is for the graded commutator of
the Lie derivative and the interior product. Given $u,v \in
\LX(M)$, it follows from Proposition A3 in \cite{Forger:2002ak} that
\begin{equation} \label{commutator}
\iota([u,v]) \alpha = (-1)^{(\deg{u}-1)\deg{v}} \L_{u} \iota(v)  \alpha - \iota(v)\L_{u} \alpha.
\end{equation}

We return now to $n$-plectic geometry. Our first definition is:
\begin{definition} \label{hamiltonian}
Let $(M,\omega)$ be an $n$-plectic manifold.  An $(n-1)$-form $\alpha$
is {\bf Hamiltonian} iff there exists a vector field $v_\alpha \in \X(M)$ such that
\[
d\alpha= -\ip{\alpha} \omega.
\]
\noi We say $v_\alpha$ is the {\bf Hamiltonian vector field} corresponding to $\alpha$. 
The set of Hamiltonian $(n-1)$-forms and the set of Hamiltonian vector
fields on an $n$-plectic manifold are both vector spaces and are denoted
as $\hamn{n-1}$ and $\VectH \left(M \right)$, respectively.
\end{definition}

The Hamiltonian vector field $v_\alpha$ is unique if it exists, but
there may be $(n-1)$-forms having no Hamiltonian vector field.  Note
that if $\alpha \in \Omega^{n-1}(M)$ is closed, then it is Hamiltonian
and its Hamiltonian vector field is the zero vector field. 
 
An elementary, yet important, fact is that the flow of a Hamiltonian
vector field preserves the $n$-plectic structure.
\begin{lemma}\label{L_thm}
If $v_{\alpha}$ is a Hamiltonian vector field, then $\L_{v_{\alpha}}\omega =0$.
\end{lemma}
\begin{proof}
\[
\L_{v_{\alpha}} \omega = d \ip{\alpha} \omega+ \ip{\alpha} d \omega
=-d d \alpha =0
\]
\end{proof}

We now formally define the bracket on $\hamn{n-1}$, which we described earlier
in the introduction. One motivation for considering this
bracket comes from its appearance in the multisymplectic formulations of
classical field theories \cite{Helein:2002wf,Kijowski:1973gi}.
\begin{definition}
\label{bracket_def}
Given $\alpha,\beta\in \hamn{n-1}$, the {\bf bracket} $\brac{\alpha}{\beta}$
is the $(n-1)$-form given by 
\[  \brac{\alpha}{\beta} = \iota_{v_{\beta}}\iota_{v_{\alpha}}\omega .\]
\end{definition}

When $n=1$, this bracket is the usual Poisson bracket of smooth
functions on a symplectic manifold. These next propositions show that
for $n>1$ the bracket of Hamiltonian forms has several properties in
common with the Poisson bracket. However,
unlike the case in symplectic geometry, we see that the bracket
$\brac{\cdot}{\cdot}$ does not need to satisfy the Jacobi identity for $n >1$.

\begin{prop}\label{bracket_prop} Let $\alpha,\beta \in
  \hamn{n-1}$ and
$v_{\alpha},v_{\beta}$ be their respective Hamiltonian
vector fields.  The bracket $\brac{\cdot}{\cdot}$ has the following properties:
  \begin{enumerate}
\item{The bracket is skew-symmetric:
\[
\brac{\alpha}{\beta}=-\brac{\beta}{\alpha}.
\]
}
\item{ The bracket of Hamiltonian forms is Hamiltonian:
  \[
    d\brac{\alpha}{\beta} = -\iota_{[v_{\alpha},v_{\beta}]} \omega,
  \]
and in particular we have 
\[     v_{\brac{\alpha}{\beta}} = [v_{\alpha},v_{\beta}]  .\]
}
\end{enumerate}
\end{prop}
\begin{proof} 
The first statement follows from the antisymmetry of $\omega$. To
prove the second statement, we use Lemma \ref{L_thm}:
\begin{align*}
d\brac{\alpha}{\beta} & = d\ip{\beta}\ip{\alpha} \omega\\
&= \left (\L_{v_{\beta}}-
\iota_{v_{\beta}} d \right) \iota_{v_{\alpha}}\omega \\
&=\L_{v_{\beta}} \iota_{v_{\alpha}}\omega + \iota_{v_{\beta}} d  d\alpha \\
&=\iota_{[v_{\beta},v_{\alpha}]}\omega +
\iota_{v_{\alpha}} \L_{v_{\beta}} \omega\\
&=-\iota_{[v_{\alpha},v_{\beta}]}\omega.
\end{align*}
\end{proof}

\begin{prop}\label{no_jacobi}
The bracket $\brac{\cdot}{\cdot}$ satisfies the Jacobi identity up to an exact $(n-1)$-form:
\[
    \brac{\alpha_{1}}{\brac{\alpha_{2}}{\alpha_{3}}} -
    \brac{\brac{\alpha_{1}}{\alpha_{2}}}{\alpha_{3}} 
    -\brac{\alpha_{2}}{\brac{\alpha_{1}}{\alpha_{3}}} =-d\iota(v_{\alpha_{1}}
      \wedge v_{\alpha_{2}} \wedge v_{\alpha_{3}}) \omega.
\]
\end{prop}

The proof of Proposition \ref{no_jacobi} follows from the next lemma. 
We will also use this lemma in the proof of Theorem \ref{main_thm} in
Section \ref{main}.

\begin{lemma}\label{tech_lemma}
If $(M,\omega)$ is an $n$-plectic manifold and $v_{1},\hdots,v_{m} \in
\VectH(M)$ with $m \geq 2$ then
\begin{multline} \label{big_identity}
d \iota(v_{1} \wedge\cdots \wedge v_{m}) \omega = \\(-1)^{m}\sum_{1 \leq i < j \leq
  m} (-1)^{i+j} \iota([v_{i},v_{j}] \wedge v_{1} \wedge \cdots
  \wedge \hat{v}_{i} \wedge \cdots \wedge \hat{v}_{j} \wedge \cdots \wedge v_{m})
\omega. 
\end{multline}
\end{lemma}
\begin{proof}
We proceed via induction on $m$. For $m=2$:
\[d\iota(v_{1} \wedge v_{2})\omega=d\brac{\alpha_{1}}{\alpha_{2}},
\]
 where $\alpha_{1},\alpha_{2}$
are any Hamiltonian $(n-1)$-forms whose Hamiltonian vector
fields are $v_{1},v_{2}$, respectively. Then Proposition
\ref{bracket_prop} implies Eq.\ \ref{big_identity} holds.

Assume Eq.\ \ref{big_identity} holds for $m-1$. 
Since $\iota(v_{1} \wedge\cdots \wedge v_{m})=\iota_{v_{m}}\iota(
v_{1} \wedge\cdots \wedge v_{m-1})$, Eq.\ \ref{Lie} implies:
\begin{equation} \label{step1}
d \iota(v_{1} \wedge\cdots \wedge v_{m})\omega= \L_{v_{m}}\iota(v_{1} \wedge
  \cdots \wedge v_{m-1}) \omega- \iota_{v_{m}} d \iota(v_{1} \wedge
  \cdots \wedge v_{m-1}) \omega. 
\end{equation}
Consider the first term on the right hand side. Using Eq.\
\ref{commutator} we can rewrite it as
\begin{align*}
\L_{v_{m}}\iota(v_{1} \wedge  \cdots \wedge v_{m-1}) \omega &= 
\iota([v_{m},v_{1} \wedge \cdots \wedge v_{m-1}]) \omega \\
& \quad +\iota(v_{1} \wedge \cdots \wedge v_{m-1})
\L_{v_{m}} \omega \\
&=\iota([v_{m},v_{1} \wedge \cdots \wedge v_{m-1}]) \omega,
\end{align*}
where the last equality follows from Lemma \ref{L_thm}.

The definition of the Schouten bracket given in Eq.\
\ref{Schouten} implies
\[
[v_{m},v_{1} \wedge \cdots \wedge v_{m-1}] =\sum_{i=1}^{m-1}
(-1)^{i+1} [v_{m},v_{i}] \wedge v_{1} \wedge \cdots \wedge \hat{v}_{i}
\wedge \cdots \wedge v_{m-1}.
\]
Therefore we have
\begin{align*}
\L_{v_{m}}\iota(v_{1} \wedge  \cdots \wedge v_{m-1}) \omega
&=\iota([v_{m},v_{1} \wedge \cdots \wedge v_{m-1}]) \omega \\
&=\sum_{i=1}^{m-1}
(-1)^{i} \iota([v_{i},v_{m}] \wedge v_{1} \wedge \cdots \wedge \hat{v}_{i}
\wedge \cdots \wedge v_{m-1})\omega.
\end{align*}
Combining this with the second term in Eq.\ \ref{step1} and using the
inductive hypothesis gives
\begin{align*}
\begin{split}
d \iota(v_{1} \wedge\cdots \wedge v_{m}) \omega 
 = \sum_{i=1}^{m-1}
(-1)^{i} \iota([v_{i},v_{m}] \wedge v_{1} \wedge \cdots \wedge \hat{v}_{i}
\wedge \cdots \wedge v_{m-1})\omega   
\end{split}
\\
& \quad -(-1)^{m-1} \sum \limits_{1 \leq i < j \leq
  m-1} (-1)^{i+j} \iota_{v_{m}}\iota([v_{i},v_{j}] \wedge v_{1} \wedge\cdots \\
& \quad \wedge \hat{v}_{i} \wedge \cdots  \wedge \hat{v}_{j} \wedge \cdots \wedge v_{m-1})
\omega \\
&= (-1)^{m}\left (\sum_{i=1}^{m-1}
(-1)^{i+m} \iota([v_{i},v_{m}] \wedge v_{1} \wedge \cdots \wedge \hat{v}_{i}
\wedge \cdots \wedge v_{m-1})\omega \right.  \\
& \quad \left. + \sum \limits_{1 \leq i < j \leq
  m-1} (-1)^{i+j} \iota([v_{i},v_{j}] \wedge v_{1} \wedge \cdots \wedge
  \hat{v}_{i} \wedge \cdots \wedge \hat{v}_{j} \wedge \cdots \wedge v_{m})
\omega \right)\\
&=(-1)^{m}\sum_{1 \leq i < j \leq
  m} (-1)^{i+j} \iota([v_{i},v_{j}] \wedge v_{1} \wedge \cdots \wedge
  \hat{v}_{i} \wedge \cdots \wedge \hat{v}_{j} \wedge \cdots \wedge v_{m})
\omega. 
\end{align*}
\end{proof}

\begin{proof}[Proof of Proposition \ref{no_jacobi}]
Apply Lemma \ref{tech_lemma} with $m=3$, and use the fact 
that $v_{\brac{\alpha_{i}}{\alpha_{j}}}=[v_{\alpha_{i}},v_{\alpha_{j}}]$.
\end{proof}

\section{$L_{\infty}$-algebras and Lie $n$-algebras}
We begin this section by recalling some basic graded linear algebra. Let $V$ be a graded vector space. 
Let $x_{1},\hdots,x_{n}$ be elements of $V$ and $\sigma \in \Sn_n$ a permutation. The \textbf{Koszul sign} $\epsilon(\sigma)=\epsilon(\sigma ; x_{1},\hdots,x_{n})$ is defined by the equality
\[
x_{1} \wedge \cdots \wedge x_{n} = \epsilon(\sigma ;
x_{1},\hdots,x_{n}) x_{\sigma(1)} \wedge \cdots \wedge x_{\sigma(n)},
\]
which holds in the free graded commutative algebra generated by
$V$. Given $\sigma \in \Sn_n$, let $(-1)^{\sigma}$
denote the usual sign of a permutation. Note that $\epsilon(\sigma)$ does
not include the sign $(-1)^{\sigma}$.  

We say $\sigma \in \Sn_{p+q}$ is a {\bf $\mathbf{(p,q)}$-unshuffle}
iff $\sigma(i) < \sigma(i+1)$ whenever $i \neq p$.  The set of
$(p,q)$-unshuffles is denoted by $\Sh(p,q)$. For example, $\Sh(2,1) =
\{ (1), (23), (123) \}$.

If $V$ and $W$ are graded
vector spaces, a linear map $f \maps V^{\tensor n} \to W$ is
\textbf{skew-symmetric} iff
\[
f(v_{\sigma(1)},\hdots,v_{\sigma(n)}) = (-1)^{\sigma}\epsilon(\sigma)
f(v_{1},\hdots,v_{n}),
\]
for all $\sigma \in \Sn_{n}$. The degree of an element $x_{1} \tensor \cdots
\tensor x_{n} \in V^{\tensor \bullet}$ of the graded tensor algebra
generated by $V$ is defined to be $\deg{x_{1} \tensor \cdots
\tensor x_{n}}=\sum_{i=1}^{n} \deg{x_{i}}$. 

Proposition \ref{no_jacobi} implies that we should not
expect $\hamn{n-1}$ to be a Lie algebra unless $n=1$. However, the
fact that the Jacobi identity is satisfied modulo boundary terms
suggests we consider what are known as strongly homotopy Lie algebras,
or $L_{\infty}$-algebras \cite{Lada-Markl,LS}.

\begin{definition} \label{Linfty} An
{\boldmath $L_{\infty}$}{\bf-algebra} is a graded vector space $L$
equipped with a collection
\[\left \{l_{k} \maps L^{\tensor k} \to L| 1
  \leq k < \infty \right\}\]
of
skew-symmetric linear maps with  $\deg{l_{k}}=k-2$ such that
the following identity holds for $1 \leq m < \infty :$
\begin{align} \label{gen_jacobi}
   \sum_{\substack{i+j = m+1, \\ \sigma \in \Sh(i,m-i)}}
  (-1)^{\sigma}\epsilon(\sigma)(-1)^{i(j-1)} l_{j}
   (l_{i}(x_{\sigma(1)}, \dots, x_{\sigma(i)}), x_{\sigma(i+1)},
   \ldots, x_{\sigma(m)})=0.
\end{align}
\end{definition}

\begin{definition} \label{LnA} An $L_{\infty}$-algebra $(L,\{l_{k} \})$
  is a {\bf Lie} {\boldmath $n$}{\bf -algebra} iff the underlying
  graded vector space $L$ is concentrated in degrees $0,\hdots,n-1$.
\end{definition}
Note that if $(L,\{l_{k} \})$ is a Lie $n$-algebra, then by degree counting $l_{k} =0 $ for $k > n+1$. 

The identity satisfied by the maps in Definition \ref{Linfty}
can be interpreted as a `generalized  Jacobi identity'. 
Indeed, using the notation $d=l_{1}$ and $[\cdot,\cdot] = l_{2}$, Eq.\ \ref{gen_jacobi} implies
\begin{align*}
d^2&=0\\
d [x_{1},x_{2}] &= [dx_{1},x_{2}] + (-1)^{\deg{x_{1}}}[x_{1},dx_{2}].
\end{align*}
Hence the map $l_{1} \maps L \to L$ can be
interpreted as a differential, while the map $l_{2} \maps L\tensor L
\to L$ can be interpreted as a bracket. The bracket is, of course, skew
symmetric:
\[
[x_{1},x_{2}] = -(-1)^{\deg{x_{1}} \deg{x_{2}}} [x_{2},x_{1}],
\]
but does not need to satisfy the usual Jacobi identity. In fact,
Eq.\ \ref{gen_jacobi} implies:
\begin{multline*}
(-1)^{\deg{x_{1}}\deg{x_{3}}}[[x_{1},x_{2}],x_{3}] + (-1)^{\deg{x_{2}} \deg{x_{3}}}[[x_{3},x_{1}],x_{2}] + (-1)^{\deg{x_{1}}\deg{x_{2}}} [[x_{2},x_{3}],x_{1}] \\= 
(-1)^{\deg{x_{1}}\deg{x_{3}}+1} \bigl (dl_{3}(x_{1},x_{2},x_{3}) + l_{3}(dx_{1},x_{2},x_{3}) \\
+ (-1)^{\deg{x_{1}}} l_{3}(x_{1},dx_{2},x_{3})  + (-1)^{\deg{x_{1}} +
  \deg{x_{2}}} l_{3}(x_{1},x_{2},dx_{3})  \bigr ).
\end{multline*}
Therefore one can interpret the traditional Jacobi identity as a
null-homotopic chain map from $L\tensor L \tensor L$ to $L$. The
map $l_{3}$ acts as a chain homotopy and is referred to as the
\textbf{Jacobiator}. Eq.\ \ref{gen_jacobi} also implies that $l_{3}$
must satisfy a coherence condition of its own. From the above
discussion, it is easy to see that a Lie 1-algebra is an ordinary Lie
algebra, while an $L_{\infty}$-algebra with $l_{k} \equiv 0$ for all $k
\geq 3$ is a differential graded Lie algebra.

\begin{remark}[Morphisms of $L_{\infty}$-algebras] \label{morphism_remark}
There is a more elegant way to define an $L_{\infty}$-algebra using
the language of graded coalgebras. This is inspired by the Quillen
construction \cite{Quillen:1969} for DGLAs, which realizes any DGLA structure on a
graded vector space $V$ as a codifferential on the cofree,
cocommutative coalgebra (without counit) generated by the suspension of
$V$.  One can then define an
$L_{\infty}$-structure on $V$ to simply be \textit{any} codifferential
on this coalgebra \cite{Lada-Markl}. The fact that a codifferential squares to
zero is equivalent to Eq.\ \ref{gen_jacobi}.
The reader unfamiliar with coalgebras
is probably quite confused by these remarks. We only mention this
alternative definition, since it provides a natural definition 
of morphism between $L_{\infty}$-algebras. Such a morphism is
just a morphism between the corresponding graded coalgebras which
respects the codifferentials. In this thesis, we will only consider morphisms
between Lie 2-algebras (Def.\ \ref{homo}).
\end{remark}

\subsection{Lie 2-algebras}\label{L2A_section}
Since we will be focusing specifically on 2-plectic manifolds in later chapters,
we discuss here the theory of Lie 2-algebras in more detail.
As $L_{\infty}$-algebras, Lie 2-algebras are relatively easy to work
with, since the underlying complex is concentrated in only two degrees.
In this case, one can write out the axioms explicitly using elementary
homological algebra.
\begin{proposition} \label{L2A}
A {\bf Lie 2-algebra} is a $2$-term chain complex of vector spaces
$L = (L_1\stackrel{d}\rightarrow L_0)$ equipped with:
\begin{itemize}
\item skew-symmetric chain map $\blankbrac\maps L \tensor
  L\to L$ called the {\bf bracket};
\item an skew-symmetric chain homotopy $J \maps L \tensor L \tensor L
  \to L$ 
from the chain map
\[     \begin{array}{ccl}  
     L \tensor L \tensor L & \to & L   \\
     x \tensor y \tensor z & \longmapsto & [x,[y,z]],  
  \end{array}
\]
to the chain map
\[     \begin{array}{ccl}  
     L \tensor L \tensor L& \to & L   \\
     x \tensor y \tensor z & \longmapsto & [[x,y],z] + [y,[x,z]]  
  \end{array}
\]
called the {\bf Jacobiator},
\end{itemize}
such that the following equation holds:
\begin{equation} \label{big_J}
\begin{array}{c}
  [x,J(y,z,w)] + J(x,[y,z],w) +
  J(x,z,[y,w]) + [J(x,y,z),w] \\ + [z,J(x,y,w)] 
  = J(x,y,[z,w]) + J([x,y],z,w) \\ + [y,J(x,z,w)] + J(y,[x,z],w) + J(y,z,[x,w]).
\end{array}
\end{equation}
\end{proposition}
\begin{proof}
See Lemma 33 in Baez and Crans \cite{HDA6}. Note that the Jacobiator $J$
is the map $l_{3}$ in Definition \ref{LnA}. 
\end{proof}

For Lie 2-algebras, it is easy to write down the definition of a
morphism without using coalgebras. (See Remark \ref{morphism_remark}.)
\begin{definition}[\cite{HDA6}]
\label{homo}
Given Lie $2$-algebras $L=(L,\blankbrac,J)$ and
$L'=(L',{\blankbrac}^{\prime},J')$ a {\bf morphism} from
$L$ to $L'$ consists of:
\begin{itemize}
\item{a chain map $\phi \maps L \to L'$, and}
\item{a chain homotopy $\Phi \maps L \tensor L \to L'$ from the chain
  map
\[     \begin{array}{ccl}  
     L \tensor L & \to & L'   \\
     x \tensor y & \longmapsto & \phi \left( [x,y] \right)
  \end{array}
\]
to the chain map
\[     \begin{array}{ccl}  
     L \tensor L & \to & L'   \\
     x \tensor y & \longmapsto & \left [ \phi(x),\phi(y) \right]^{\prime},
  \end{array}
\]
}
\end{itemize}
such that the following equation holds:
\begin{equation}
\begin{array}{l}
\phi_1(J(x,y,z))- J^{\prime}(\phi_0(x),\phi_0(y), \phi_0(z)) = \\
\Phi(x,[y,z]) -\Phi([x,y],z) - \Phi(y,[x,z]) - [\Phi(x,y),\phi_0(z)]^{\prime}\\
+ [\phi_0(x), \Phi(y,z)]^{\prime}- [\phi_0(y),\Phi(x,z)]^{\prime}.
\end{array}
\end{equation}
We say a morphism is {\bf strict} iff $\Phi=0$.
\end{definition}

Typically, isomorphism is too strong of an equivalence to use for $L_{\infty}$-algebras.
Instead we use:
\begin{definition}
A Lie 2-algebra morphism $(\phi,\Phi) \maps L \to L'$
is a  {\bf quasi-isomorphism} iff the chain map $\phi$
induces an isomorphism on the homology of the underlying chain
complexes of $L$ and $L'$.
\end{definition}
\noi Since every vector space is free, quasi-isomorphism in our case
is the same thing as chain homotopy equivalence, or categorical
equivalence in the sense of Baez and Crans \cite{HDA6}.

\section{Lie $n$-algebras from $n$-plectic manifolds} \label{main}
There are several clues that suggest that any $n$-plectic manifold
gives an $L_{\infty}$-algebra. 
%It was shown
%in our previous work \cite{Baez:2008bu} that a Lie 2-algebra can be explicitly
%constructed from the 2-plectic structure on any 2-plectic
%manifold. The underlying chain complex of this Lie 2-algebra is
%\[
%\cinf(M) \stackrel{d}{\to} \hamn{1},
%\]
%where $d$ is the de Rham differential. 
Comparing Eq.\ \ref{big_identity} to the generalized Jacobi identity 
(\ref{gen_jacobi}) suggests that, for an
$n$-plectic manifold, we should look for
Lie $n$-algebra structures on the chain complex 
\begin{equation} \label{complex}
\cinf(M) \stackrel{d}{\to} \Omega^{1}(M) \stackrel{d}{\to} \cdots \stackrel{d}{\to} 
\Omega^{n-2}(M) \stackrel{d}{\to} \hamn{n-1},
\end{equation}
with the $l_{1}$ map equal to $d$. We denote this complex as
$(L,d)$. Note that here we are using the de Rham differential as a degree
-1 operator. Hence $L_{0}=\hamn{n-1}$, while $L_{n-1}=\cinf(M)$.

Note that the bracket
$\brac{\cdot}{\cdot}$ given in Definition \ref{bracket_def} induces a
well-defined bracket $\blankbrac^{\prime}$ on the quotient 
\[
\g=\hamn{n-1}/d\Omega^{n-2}(M),
\]
where $d\Omega^{n-2}(M)$ is the space of exact $(n-1)$-forms. This is
because the Hamiltonian vector field of an exact $(n-1)$-form is the
zero vector field. It follows from Proposition \ref{no_jacobi} that
$\left (\g,\blankbrac^{\prime} \right)$ is, in fact, a Lie algebra. 

If $M$ is contractible, then the homology of $(L,d)$ is
\begin{align*}
H_{0}(L)&=\g, \\
H_{k}(L)&=0 \quad \text{for $0<k<n-1$},\\
H_{n-1}(L)&=\R.
\end{align*}

Therefore, the augmented complex 
\begin{equation} \label{aug_complex} 0 \to \R \hookrightarrow \cinf(M)
  \stackrel{d}{\to} \Omega^{1}(M) \stackrel{d}{\to} \cdots
  \stackrel{d}{\to} \Omega^{n-2}(M) \stackrel{d}{\to} \hamn{n-1}
\end{equation}
is a resolution of $\g$.

Barnich, Fulp, Lada, and Stasheff \cite{Barnich:1997ij} showed
that, in general, if $(C,\delta)$ is a resolution of a
vector space $V \cong H_{0}(C)$ and $C_{0}$ is equipped with a
skew-symmetric map $\tilde{l}_{2} \maps C_{0} \tensor C_{0} \to C_{0}$
that induces a Lie bracket on $V$, then
$\tilde{l}_{2}$ extends to an $L_{\infty}$-structure on $(C,\delta)$.
Hence we have the following proposition:

\begin{prop}\label{contract_thm}
Given a contractible $n$-plectic manifold $(M,\omega)$, there is an $L_{\infty}$-algebra
$(\tilde{L},\{l_{k} \})$ with underlying graded vector space
\[
\tilde{L}_{i} =
\begin{cases}
\hamn{n-1} & i=0,\\
\Omega^{n-1-i}(M) & 0 < i \leq n-1,\\
\R & i=n,
\end{cases}
\]
and $l_{1} \maps \tilde{L} \to \tilde{L}$ defined as 
\[
l_{1}(\alpha)=
\begin{cases}
\alpha, & \text{if $\deg{\alpha}=n$} \\
d\alpha & \text{if $\deg{\alpha} \neq n$,}
\end{cases}
\]
and all higher maps  $\left \{l_{k} \maps \tilde{L}^{\tensor k} \to \tilde{L}| 2
  \leq k < \infty \right\}$ are constructed inductively by using the bracket
\[
\brac{\cdot}{\cdot} \maps \tilde{L}_{0} \tensor \tilde{L}_{0} \to \tilde{L}_{0}, \quad
\brac{\alpha_{1}}{\alpha_{2}} = \ip{\alpha_{2}}\ip{\alpha_{1}} \omega,
\]
where $v_{\alpha_{1}},v_{\alpha_{2}}$ are the Hamiltonian vector fields corresponding
to the Hamiltonian forms $\alpha_{1}, \alpha_{2}$. Moreover the maps
$\{ l_{k} \}$ may be constructed so that
\[
l_{k}(\alphadk{k}) \neq 0 \quad \text{only if all $\alpha_{k}$ have
  degree 0},
\]
for $k \geq 2$. 
\end{prop}
\begin{proof}
  The proposition follows from Theorem 7 in the paper by Barnich,
  Fulp, Lada, and Stasheff \cite{Barnich:1997ij}. Since
  for any $n$-plectic manifold,
\[
\brac{\alpha}{d \beta}=0 \quad \forall \alpha \in \hamn{n-1} ~ \forall
\beta \in \Omega^{n-2}(M),
\]
the second remark following Theorem 7 in \cite{Barnich:1997ij} implies
that the maps $\{ l_{k}\}$ may be constructed so that they are trivial
when restricted to the positive-degree part of the $k$-th tensor power of $\tilde{L}$.
\end{proof}

For an arbitrary $n$-plectic manifold $(M,\omega)$, Proposition
\ref{contract_thm} guarantees the existence of $L_{\infty}$-algebras
locally. We want, of course, a global result in which the higher
$l_{k}$ maps are explicitly constructed using only the $n$-plectic
structure. Moreover, in our previous work on 2-plectic geometry \cite{Baez:2008bu}, we
were able to construct by hand a Lie 2-algebra on a 2-term complex
consisting of functions and Hamiltonian 1-forms. We did not need to
use a 3-term complex consisting of constants, functions, and
Hamiltonian 1-forms. Hence in the general case, we'd expect an
$n$-plectic manifold to give a Lie $n$-algebra whose underlying
complex is $(L,d)$, instead of a Lie $(n+1)$-algebra whose underlying
complex is the $(n+1)$-term complex used in the above proposition.

We can get an intuitive sense for what the maps $l_{k} \maps L^{\tensor k} \to L$
should be by unraveling the identity given in Definition \ref{Linfty}
for small values of $m$ and momentarily disregarding signs and
summations over unshuffles. For example, if $m=2$, then Eq.\
\ref{gen_jacobi} implies that the map $l_{2} \maps L \tensor L \to L$
must satisfy:
\begin{equation} \label{l2}
l_{1}l_{2} + l_{2} l_{1}=0.
\end{equation}
Obviously we want $l_{1}$ to be the de Rham differential and $l_{2}$
to be equal to the bracket $\brac{\cdot}{\cdot}$ when restricted to degree 0
elements:
\[
l_{2}(\alpha_{1},\alpha_{2}) = \pm \iota_{v_{\alpha_{2}}}
\iota_{v_{\alpha_{1}}} \omega=\brac{\alpha_{1}}{\alpha_{2}}\quad \forall \alpha_{i} \in
L_{0}=\hamn{n-1}.
\]
Now consider elements of degree 1. For example, if $\alpha \in
L_{0}$ and $\beta \in
L_{1}=\Omega^{n-2}(M)$, then
$l_{2}(\alpha,d\beta)=\brac{\alpha}{d\beta} =0$.
Therefore Eq.\ \ref{l2} implies
\[
dl_{2}(\alpha,\beta)=  l_{1} l_{2}(\alpha,\beta)= 0.
\]
Hence, when restricted to elements of degree 1, $l_{2}(\alpha,\beta)$
must be a closed $(n-2)$-form. We will choose this closed
form to be 0. In fact, we will choose $l_{2}$ to vanish on all elements with degree
$>0$, since, in general, we want the $L_{\infty}$ structure to 
only depend on the de Rham differential and the $n$-plectic structure.

Now suppose $l_{2}$ is defined as above and
let $m=3$.  Then Eq.\ \ref{gen_jacobi} implies:
\begin{equation}\label{l3}
l_{1} l_{3} + l_{2} l_{2} + l_{3} l_{1} =0.
\end{equation}
On degree 0 elements, $l_{1}=0$. Therefore it is clear from Proposition
\ref{no_jacobi} that the map $l_{3} \maps L^{\tensor 3} \to L$ when
restricted to degree 0 elements must be
\[
l_{3}(\alpha_{1},\alpha_{2},\alpha_{3}) = \pm \iota(v_{\alpha_{1}}
\wedge v_{\alpha_{2}} \wedge v_{\alpha_{3}}) \omega,\] where
$v_{\alpha_{i}}$ is the Hamiltonian vector field associated to
$\alpha_{i}$.  Now consider a degree 1 element of $L \tensor L \tensor
L$, for example: $\alpha_{1} \tensor \alpha_{2} \tensor \beta \in 
\hamn{n-1}\tensor \hamn{n-1} \tensor \Omega^{n-2}(M)$. Since
$l_{3}(\alpha_{1},\alpha_{2},d\beta)= \pm \iota(v_{\alpha_{1}}
\wedge v_{\alpha_{2}} \wedge v_{d\beta}) \omega=0$,
and $l_{2}$ vanishes on 
the positive-degree part of the $k$-th tensor power of $L$,
Eq.\ \ref{l3} holds if and only if
\[
dl_{3}(\alpha_{1},\alpha_{2},\beta)=0.
\]
Hence, when restricted to elements of degree 1, $l_{3}(\alpha_{1},\alpha_{2},\beta)$
must be a closed $(n-2)$-form. Again, we will choose this closed
form to be 0 by forcing $l_{3}$ to vanish on all elements with degree
$>0$. 

Observations like these bring us to our main theorem.  In general, we
will define the maps $l_{k} \maps L^{\tensor k} \to L$ on degree zero
elements to be completely specified (up to sign) by the $n$-plectic
structure $\omega$:
\[
l_{k}(\alphadk{k}) = \pm \iota(\vk{k}) \omega \quad \text{if $\deg{\alphak{k}}=0$},
\]
and trivial otherwise:
\[
l_{k}(\alphadk{k}) = 0 \quad \text{if $\deg{\alphak{k}} > 0$}.
\]

\begin{theorem} \label{main_thm}
Given an $n$-plectic manifold $(M,\omega)$, there is a Lie $n$-algebra
$\Lie(M,\omega)=(L,\{l_{k} \})$ with underlying graded vector space 
\[
L_{i} =
\begin{cases}
\hamn{n-1} & i=0,\\
\Omega^{n-1-i}(M) & 0 < i \leq n-1,
\end{cases}
\]
and maps  $\left \{l_{k} \maps L^{\tensor k} \to L| 1
  \leq k < \infty \right\}$ defined as
\[ 
l_{1}(\alpha)=d\alpha,
\]
if $\deg{\alpha}>0$ and
\begin{multline}
l_{k}(\alphadk{k}) =\\
\begin{cases}
0 & \text{if $\deg{\alphak{k}} > 0$}, \\
(-1)^{\frac{k}{2}+1} \iota(\vk{k}) \omega  & \text{if
  $\deg{\alphak{k}}=0$ and $k$ even},\\
(-1)^{\frac{k-1}{2}}\iota(\vk{k}) \omega  & \text{if
  $\deg{\alphak{k}}=0$ and $k$ odd},
\end{cases}
\end{multline}
for $k>1$, where $v_{\alpha_{i}}$ is the unique Hamiltonian vector field
associated to $\alpha_{i} \in \hamn{n-1}$.
\end{theorem}

\begin{proof}[Proof of Theorem \ref{main_thm}]
We begin by showing the maps $\{l_{k}\}$
are well-defined skew symmetric maps with $\deg{l_{k}}=k-2$.
If $\alphak{k} \in L^{\tensor \bullet}$ has
degree 0, then for all $\sigma \in \Sn_{k}$ the antisymmetry of
$\omega$ implies
\[
l_{k}(\alphasdk{k})=(-1)^{\sigma} l_{k}(\alphadk{k}).
\]
Since for each $i$, we have $\deg{\alpha_{i}}=0$, it follows that
$\epsilon(\sigma)=1$. Hence $l_{k}$ is skew symmetric and
well-defined. Since $\iota(\vk{k}) \omega \in
\Omega^{n+1-k}(M)=L_{k-2}$, we have $\deg{l_{k}}=k-2$.
We also have, by construction, $l_{k} = 0 $ for $k>n+1$.

Now we prove the maps satisfy Eq.\ \ref{gen_jacobi} in Definition
\ref{Linfty}. If $m=1$, then it is satisfied since $l_{1}$
is the de Rham differential. 
 If $m=2$, then a direct calculation shows
\[
l_{1}(l_{2}(\alpha_{1},\alpha_{2})) =
l_{2}(l_{1}(\alpha_{1}),\alpha_{2}) +
(-1)^{\deg{\alpha_{1}}}l_{2}(\alpha_{1},l_{1}(\alpha_{2})).
\]
Let $m > 2$. We will regroup the summands
in Eq.\ \ref{gen_jacobi} into two separate sums depending on the value
of the index $j$ and show that each of these is zero, thereby proving the
theorem.

We first consider the sum of the terms with $2 \leq j \leq m-2$:
\begin{equation} \label{term1}
\sum_{j=2}^{m-2}
   \sum_{\sigma \in \Sh(i,m-i)} \negthickspace \negthickspace \negthickspace
   (-1)^{\sigma}\epsilon(\sigma)(-1)^{i(j-1)} l_{j}
   (l_{i}(\alpha_{\sigma(1)}, \dots, \alpha_{\sigma(i)}), \alpha_{\sigma(i+1)},
   \ldots, \alpha_{\sigma(m)}).
\end{equation}
In this case we claim that for all $\sigma \in \Sh(i,m-i)$ we have 
\[
l_{j}(l_{i}(\alphasdk{i}),\alpha_{\sigma(i+1)},\hdots,\alpha_{\sigma(m)})=0.
\]
Indeed, if there exists an unshuffle such that the above
equality did not hold, then the definition of $l_{j} \maps L^{\tensor
  j} \to L$ implies
\[
\deg{l_{i}(\alphasdk{i}) \tensor \alpha_{\sigma(i+1)} \tensor \cdots \tensor \alpha_{\sigma(m)}}=0,
\] 
which further implies
\begin{equation} \label{step2}
\deg{l_{i}(\alphasdk{i})}=\deg{\alphask{i}} + i -2 =0.
\end{equation}
By assumption, $l_{i}(\alphasdk{i})$ must be non-zero and $j < m-1$
implies $i>1$. Hence we must have $\deg{\alphask{i}}=0$ and
therefore, by Eq.\ \ref{step2}, $i=2$. But this
implies $j=m-1$, which contradicts our bounds on $j$. So no such
unshuffle could exist, and therefore the sum (\ref{term1}) is zero.  

We next consider the sum of the terms $j=1$, $j=m-1$, and $j=m$:
\begin{equation} \label{the_sum}
\begin{split}
l_{1}( l_{m}(\alphadk{m})) + \sum_{\sigma \in \Sh(2,m-2)} 
(-1)^{\sigma}
\epsilon(\sigma)
l_{m-1}(l_{2}(\alpha_{\sigma(1)},\alpha_{\sigma(2)}),\alpha_{\sigma(3)},\hdots,\alpha_{\sigma(m)})
 \\
+ \sum_{\sigma \in \Sh(1,m-1)}  \negthickspace \negthickspace \negthickspace
(-1)^{\sigma}
\epsilon(\sigma)  (-1)^{m-1} 
l_{m}(l_{1}(\alpha_{\sigma(1)}),\alpha_{\sigma(2)},\hdots,\alpha_{\sigma(m)}).
\end{split}
\end{equation}
Note that if $\sigma \in \Sh(1,m-1)$ and $\deg{l_{1}(\alpha_{\sigma(1)})} >0$, then
\[
l_{m}(l_{1}(\alpha_{\sigma(1)}),\alpha_{\sigma(2)},\hdots,\alpha_{\sigma(m)})=0
\]
by definition of the map $l_{m}$. On the other hand, 
if $\deg{l_{1}(\alpha_{\sigma(1)})} =0$, then
$l_{1}(\alpha_{\sigma(1)})=d \alpha_{\sigma(1)}$ is Hamiltonian and
its Hamiltonian vector field is the zero vector field. Hence the third
term in (\ref{the_sum}) is zero. 

Since the map $l_{2}$ is degree 0, we only need to consider the
first two terms of (\ref{the_sum}) in the case when $\deg{\alphak{m}}
=0$. For the first term we have:
\[
l_{1}( l_{m}(\alphadk{m}))=
\begin{cases}
(-1)^{\frac{m}{2}+1} d\iota(\vk{m}) \omega & \text{if $m$ even},\\
(-1)^{\frac{m-1}{2}}d\iota(\vk{m}) \omega  & \text{if $m$ odd}.
\end{cases}
\]
Now consider the second term. If $\alpha_{i},\alpha_{j} \in \hamn{n-1}$ are Hamiltonian
$(n-1)$-forms then by Definition \ref{bracket_def},
$l_{2}(\alpha_{i},\alpha_{j})=\brac{\alpha_{i}}{\alpha_{j}}$.
By Proposition \ref{bracket_prop}, $l_{2}(\alpha_{i},\alpha_{j})$ is
Hamiltonian and its Hamiltonian vector field is $v_{\brac{\alpha_{i}}{\alpha_{j}}}=[v_{\alpha_{i}},v_{\alpha_{j}}]$.
Therefore for  $\sigma \in \Sh(2,m-2)$, we have
\begin{multline*}
l_{m-1}(l_{2}(\alpha_{\sigma(1)},\alpha_{\sigma(2)}),\alpha_{\sigma(3)},\hdots,\alpha_{\sigma(m)})=\\
\begin{cases}
(-1)^{\frac{m}{2}-1}\iota([v_{\alpha_{\sigma(1)}},v_{\alpha_{\sigma(2)}}] \wedge
  \cdots \wedge v_{\alpha_{\sigma(m)}}) \omega & \text{if $m$ even},\\
(-1)^{\frac{m+1}{2}}\iota([v_{\alpha_{\sigma(1)}},v_{\alpha_{\sigma(2)}}] \wedge
  \cdots \wedge v_{\alpha_{\sigma(m)}}) \omega  & \text{if $m$ odd}.
\end{cases}
 \end{multline*}
Since each $\alpha_{i}$ is degree 0, we can rewrite the sum over
$\sigma \in \Sh(2,m-2)$ as
\begin{multline*}
\sum_{\sigma \in \Sh(2,m-2)}
(-1)^{\sigma}
\epsilon(\sigma)
l_{m-1}(l_{2}(\alpha_{\sigma(1)},\alpha_{\sigma(2)}),\alpha_{\sigma(3)},\hdots,\alpha_{\sigma(m)}) =\\
\sum_{1\leq i<j \leq m} (-1)^{i+j-1} l_{m-1}(l_{2}(\alpha_{i},\alpha_{j}),\alpha_{1}, \alpha_{2},\hdots,
\hat{\alpha}_{i},\hdots,\hat{\alpha}_{j},\hdots,\alpha_{m}).
\end{multline*}
Therefore, if $m$ is even, the sum (\ref{the_sum}) becomes
\begin{multline*}
(-1)^{\frac{m}{2}+1} d\iota(\vk{m}) \omega  
+ (-1)^{\frac{m}{2}} \sum_{1\leq i<j \leq m} (-1)^{i+j} \iota([v_{\alpha_{i}},v_{\alpha_{j}}]
  \wedge v_{\alpha_{1}} \\ \wedge  \cdots 
\wedge \hat{v}_{\alpha_{i}} \wedge \cdots \wedge
  \hat{v}_{\alpha_{j}} \wedge \cdots \wedge v_{\alpha_{m}}) \omega
\end{multline*}
and, if $m$ is odd:
\begin{multline*}
(-1)^{\frac{m-1}{2}} d\iota(\vk{m}) \omega  
+ (-1)^{\frac{m-1}{2}} \sum_{1\leq i<j \leq m} (-1)^{i+j} \iota([v_{\alpha_{i}},v_{\alpha_{j}}]
  \wedge v_{\alpha_{1}} \\ \wedge
  \cdots \wedge \hat{v}_{\alpha_{i}} \wedge \cdots  \wedge
  \hat{v}_{\alpha_{j}} \wedge \cdots \wedge v_{\alpha_{m}}) \omega.
\end{multline*}
It then follows from Lemma \ref{tech_lemma} that, in either case, (\ref{the_sum}) is zero.
\end{proof}

It is clear that in the $n=1$ case, $\Lie(M,\omega)$ is the underlying
Lie algebra of the usual Poisson algebra of smooth functions on a
symplectic manifold. In the $n=2$ case, $\Lie(M,\omega)$ is the Lie
2-algebra obtained in our previous work with Baez and Hoffnung
\cite{Baez:2008bu}.  

For the $n=2$ case, it will be convenient for us in later chapters to
express the Lie 2-algebra $\Lie(M,\omega)$ in the language of Prop.\ \ref{L2A}:
\begin{prop}
\label{semistrict}
If $(M,\omega)$ is a $2$-plectic manifold, then there is a 
 Lie $2$-algebra $L_{\infty}(M,\omega)=(L,\blankbrac,J)$ where:
\begin{itemize}
\item $L_{0} =\hamn{1}$,
\item $L_{1}=\cinf(M)$,
\item the differential $L_{1} \stackrel{d}{\to} L_{0}$ is the de Rham differential,
\item the bracket $\blankbrac$ is $\brac{\cdot}{\cdot}$ in degree 0
  and trivial otherwise,
\item the Jacobiator is given by the linear map $J\maps \ham \tensor \ham \tensor 
\ham \to \cinf$, where $J(\alpha,\beta,\gamma) = \ip{\alpha}\ip{\beta}\ip{\gamma}\omega$.
\end{itemize}
\end{prop}
\begin{proof}
This follows from the fact that $d$, $\blankbrac$, and $J$ are the
structure maps $l_{1}$, $l_{2}$, and $l_{3}$, respectively, described in
Thm.\ \ref{main_thm}.
\end{proof}

Finally, we mention that the equality
\[
    d\brac{\alpha}{\beta} = -\iota_{[v_{\alpha},v_{\beta}]} \omega
\]
given in Proposition \ref{bracket_prop} implies the existence of a
bracket-preserving chain map
\[
\phi \maps \Lie(M,\omega) \to \VectH \left(M \right),
\]
which in degree 0 takes a Hamiltonian $(n-1)$-form 
$\alpha$ to its vector field $v_{\alpha}$. 
Here we consider the Lie algebra of Hamiltonian vector fields as a Lie
1-algebra whose underlying complex is concentrated in degree 0:
\[
\ldots \to 0 \to 0 \to \VectH \left(M \right).
\]
Hence $\phi$ is trivial in all higher degrees. In light of
Theorem \ref{main_thm}, $\phi$ becomes a strict morphism of
$L_{\infty}$-algebras. See the paper by Lada and Markl \cite{Lada-Markl} for the
definition of strict $L_{\infty}$-algebra morphisms.

\chapter{Lie 2-algebras from compact simple Lie groups} \label{lie_group_chapter}
Here we consider some Lie 2-algebras which arise on an important class of
2-plectic manifolds: compact simple Lie groups. Recall from Example
\ref{Lie_group_example} that such a group admits a 1 parameter family of
2-plectic structures given by a non-zero constant times the Cartan 3-form:
\begin{equation*} 
\nu_{k} =k \inp{\theta_{L}}{\bigl [ \theta_{L}, \theta_{L} \bigr ]}, \quad k \neq 0,
\end{equation*}
where $\theta_{L}$ is the left-invariant Maurer-Cartan form, and
$\inp{\cdot}{\cdot}$ is the inner product on the corresponding Lie
algebra, whose
bracket is $[\cdot,\cdot]$.
This 3-form plays an important role in the theory of affine Lie algebras, central
extensions of loop groups, and gerbes \cite{BCSS,Brylinski:1993,PressleySegal}.

Baez and Crans showed that the Lie algebra of a compact simple Lie
group $G$ can be used to build a Lie 2-algebra called the `string Lie
2-algebra' \cite{HDA6}. This Lie 2-algebra can be integrated to a special kind of category called
a Lie 2-group. For $G=\mathrm{Spin}(n)$, the geometric realization of this
Lie 2-group is homotopy equivalent to the topological group
$\mathrm{String}(n)$ \cite{BCSS,Henriques:2008}.  
The group $\mathrm{String}(n)$ naturally arises in the study of spin
structures on loop spaces \cite{Witten:1988}.

The structure of the string Lie 2-algebra associated to $G$ closely
resembles the structure of the Lie 2-algebra $\Lie(G,\nu_{k})$ of
Hamiltonian 1-forms on the 2-plectic manifold $(G,\nu_{k})$. In a
private communication, D.\ Stevenson asked if these Lie
2-algebras are quasi-isomorphic.  As we show in Section
\ref{string_Lie}, this turns out not to be true. However, we prove
that the string Lie 2-algebra is isomorphic to a particular sub Lie-2
algebra of $\Lie(G,\nu_{k})$, consisting of left-invariant Hamiltonian
1-forms. This gives a new geometric construction of the string Lie
2-algebra.  For another construction, based on central extensions of loop
groups, see the paper by Baez, Crans, Schreiber and Stevenson
\cite{BCSS}.  It will be interesting to see what can be learned from
comparing these approaches.

\section{Group actions on $n$-plectic manifolds}
\label{group_actions}
We begin by giving some basic results concerning group actions on $n$-plectic manifolds.
Suppose we have a Lie group acting on an $n$-plectic manifold $(M,\omega)$,
preserving the $n$-plectic structure.  In this situation the
Lie $n$-algebra $\Lie(M,\omega)$ constructed in Thm.\ \ref{main_thm}
has a sub-$n$-algebra consisting of invariant differential forms.

More precisely, let
$\mu \maps G \times M \to M$ be a left action of the Lie group $G$ on the
$n$-plectic manifold $\left(M,\omega \right)$, and assume this action
preserves the $n$-plectic structure:
\[\mu_{g}^{\ast} \omega = \omega,\]
for all $g \in G$.  Denote the subspace of invariant Hamiltonian $(n-1)$-forms 
by
\[\hamG{n-1} = \left \{ \alpha \in \hamn{n-1} ~ \vert ~ \forall g \in G ~
\mu_{g}^{\ast}\alpha = \alpha \right\}.\]
The Hamiltonian vector field of an
invariant Hamiltonian $(n-1)$-form is itself invariant under the
action of $G$:
\begin{proposition}
\label{invariant_vector_fields}
If $\alpha \in \hamG{n-1}$ and $v_{\alpha}$ is the Hamiltonian vector field
associated with $\alpha$, then $\pfor{g} v_{\alpha} = v_{\alpha}$ for
all $g \in G$.
\end{proposition}

\begin{proof}
The exterior derivative commutes with the pullback of the group
action. Therefore if $v_{1},\ldots, v_{n}$ are smooth vector fields, then
$d\alpha \left(\pfor{g} v_{1},\ldots,\pfor{g} v_{n}\right)=
d\alpha\left(v_{1},\ldots,v_{n} \right)$, since we are assuming $\alpha$
is $G$-invariant.  Since $\alpha \in \hamn{n-1}$, we have 
$d \alpha = -\ip{\alpha}\omega$, so 
\[
\omega\left(v_{\alpha},\pfor{g} v_{1},\ldots,\pfor{g} v_{n} \right)
= \omega \left(v_{\alpha},v_{1},\ldots,v_{n} \right)=
\omega\left(\pfor{g}v_{\alpha},\pfor{g} v_{1},\ldots,\pfor{g} v_{n} \right),
\]
where the last equality follows from $\pback{g}\omega=\omega$.
Therefore 
\[
\omega\left(v_{\alpha}-\pfor{g}v_{\alpha},\pfor{g} v_{1},\ldots,\pfor{g} v_{n} \right)=0.
\]
Since $\omega$ is nondegenerate, and $v_{1},\ldots, v_{n}$ are arbitrary, it
follows that $\pfor{g}v_{\alpha} = v_{\alpha}$.
\end{proof}

Let $\Omega^{k}(M)^{G}$ denote the subspace of invariant $k$-forms on $M$:
\[\Omega^{k}(M)^{G} = \left\{ \alpha \in \Omega^{k}(M) ~ \vert ~ \forall g \in G ~
\mu_{g}^{\ast}\alpha = \alpha \right\},\] 
and let $(L^{G},d)$ denote the $n$-term complex 
\[
\cinf(M)^{G} \stackrel{d}{\to} \Omega^{1}(M)^{G} \stackrel{d}{\to} \cdots \stackrel{d}{\to} 
\Omega^{n-2}(M)^{G} \stackrel{d}{\to} \hamG{n-1}.
\]
Clearly, this is a subcomplex of the underlying complex of the Lie
$n$-algebra $\Lie(M,\omega)$.
Moreover, the invariant differential forms on $M$ form a graded subalgebra that is
stable under exterior derivative and interior product with an invariant
%stable under exterior derivative and interior product by an invariant !!!
vector field \cite{GHV2}[Sec.\ III.4]. 
Since both the bracket 
introduced in Def.\ \ref{bracket_def} and the proof of Lemma \ref{tech_lemma}
depend only on compositions of these
operations, the Lie $n$-algebra structure described in Theorem
\ref{main_thm} restricts to a Lie $n$-algebra structure on the
subcomplex $L^{G}$.
Hence, we have the following theorem:
\begin{theorem}\label{invariant_LnA}
Given an $n$-plectic manifold $(M,\omega)$ equipped with group action 
$G \times M \to M$ preserving the $n$-plectic structure,
there is a Lie $n$-algebra
$\Lie(M,\omega)^{G}=(L^{G},\{l_{k} \})$ with underlying graded vector space 
\[
L^{G}_{i} =
\begin{cases}
\hamG{n-1} & i=0,\\
\Omega^{n-1-i}(M)^{G} & 0 < i \leq n-1,
\end{cases}
\]
and maps  $\left \{l_{k} \maps \bigl( L^{G} \bigr )^{\tensor k} \to
  L^{G} \vert 1
  \leq k < \infty \right\}$ defined as
\[ 
l_{1}(\alpha)=d\alpha,
\]
if $\deg{\alpha}>0$ and
\begin{multline}
l_{k}(\alphadk{k}) =\\
\begin{cases}
0 & \text{if $\deg{\alphak{k}} > 0$}, \\
(-1)^{\frac{k}{2}+1} \iota(\vk{k}) \omega  & \text{if
  $\deg{\alphak{k}}=0$ and $k$ even},\\
(-1)^{\frac{k-1}{2}}\iota(\vk{k}) \omega  & \text{if
  $\deg{\alphak{k}}=0$ and $k$ odd},
\end{cases}
\end{multline}
for $k>1$, where $v_{\alpha_{i}}$ is the unique invariant Hamiltonian vector field
associated to $\alpha_{i} \in \hamG{n-1}$.
\end{theorem}

\section{Compact simple Lie groups as 2-plectic manifolds}
Recall from Example \ref{Lie_group_example}  that for any compact simple Lie group $G$, the 2-plectic
structure $\nu_{k}=k \inp{\theta_{L}}{[\theta_{L},\theta_{L}]}$ is
left-invariant. Hence, Thm.\ \ref{invariant_LnA} implies there exists
a Lie 2-algebra whose underlying 2-term chain
complex is composed of left-invariant Hamiltonian 1-forms $\hamL$ on $G$ in
degree 0, and left-invariant functions $\cinfL$ in degree 1.

If $f \in \cinfL$, then by definition $f=f \circ L_{g}$ for all $g \in
G$. Hence $f$ must be a constant function, so $\cinfL$ may be identified 
with $\R$. Denote the space of all left invariant 1-forms as
$\Omega^{1}(G)^{L} \cong \g^{\ast}$, and left invariant vector fields
as $\X(G)^{L} \cong \g$. 
The following theorem characterizes the left invariant Hamiltonian 
1-forms.
\begin{theorem}
\label{left_invariant_1-forms}
Every left invariant 1-form on $\left(G,\nu_{k}\right)$ is
Hamiltonian. That is, \hfill \break $\hamL = \Omega^{1}(G)^{L}$.
\end{theorem} 
\begin{proof}
Recall that if $\alpha$ is a 1-form and $v_{0},v_{1}$
are vector fields, then
\[d\alpha\left(v_{0},v_{1}\right)= v_{0} \left(\alpha
\left(v_{1}\right)\right) - v_{1}\left(\alpha\left(v_{0}\right)\right)
- \alpha \left(\left[v_{0},v_{1}\right]\right).\]
Suppose now that $\alpha$ is a left invariant 1-form on $G$ and 
$v_{0},v_{1}$ are left invariant vector fields.  Then the smooth functions
$\alpha\left(v_{1}\right)$ and $\alpha\left(v_{0}\right)$ are also left 
invariant and therefore constant. Therefore the right hand side of
the above equality simplifies and we have
\[d\alpha\left(v_{0},v_{1}\right)= - \alpha \left(\left[v_{0},v_{1}\right]\right).\]

Let $\alpha \in \Omega^{1}(G)^{L}$ and let $\innerprod{\cdot}{\cdot}$ be the
inner product on $\g$ used in the construction of
$\nu_{k}$. Note we have two isomorphisms
\[
\g \xrightarrow{k\inp{\cdot}{\cdot}} \g^{\ast}, \quad \X(G)^{L} \xrightarrow{\theta_{L}} \g.
\]
Therefore, there exists a left invariant vector field $v_{\alpha} \in \X(G)^{L}$
such that $\alpha(v')=k\innerprod{\theta_{L}(v_{\alpha})}{\theta_{L}(v')}$ for
all left invariant vector fields $v' \in \X(G)^{L}$. 
Combining this with the above expression for $d \alpha$ gives
\[d \alpha\left(v_{0},v_{1}\right)=-k\innerprod{\theta_{L}(v_{\alpha})}{[\theta_{L}(v_{0}),\theta_{L}(v_{1})]},\]
which implies
\[d \alpha = -\ip{\alpha}\nu_{k}.\] 
Hence $\alpha \in \Gham$, and $\hamL=\Gham \cap \Omega^{1}(G)^{L} = \Omega^{1}(G)^{L}$.
\end{proof}

The most important application of Thm.\ \ref{left_invariant_1-forms}
is that it allows us to use Thm.\ \ref{invariant_LnA} and the
isomorphism $\hamL=\Omega^{1}(G)^{L} \cong \g^{\ast}$
to construct a Lie 2-algebra having $\g^{\ast}$ as its space of 0-chains, 
for any compact simple Lie group.  
Recalling the simpler definition of a Lie 2-algebra given in
Prop.\ \ref{L2A}, we summarize these facts in the following corollary.

\begin{corollary}
\label{left_inv_L2A}
If $G$ is a compact simple Lie group with Lie algebra $\g$ and 2-plectic structure $\nu_{k}$,
then there is a Lie 2-algebra $\Lie(G,\nu_{k})^{L}$
where:
\begin{itemize}
\item the space of 0-chains is $\g^{\ast}$,
\item the space of 1-chains is $\R$,
\item the differential is the exterior derivative $d \maps \R \to
  \g^{\ast}$ (i.e.\ $d=0$), 
\item{the bracket is
    $\brac{\alpha}{\beta}=\nu_{k}(v_{\alpha},v_{\beta},\cdot)$ in
    degree 0, and trivial otherwise,}
\item the Jacobiator is the linear map $J\maps \g^{\ast}\tensor \g^{\ast}\tensor
\g^{\ast}\to \R$ defined by $J(\alpha,\beta,\gamma) = -\nu_{k}(v_{\alpha},v_{\beta},v_{\gamma})$.
\end{itemize}
\end{corollary}
\noi In the statement of the above corollary, we are abusing notation slightly by viewing $\alpha \in
g^{\ast}$ as a left-invariant Hamiltonian 1-form. Note that the
corollary implies that we have a 1-parameter family of Lie 2-algebras:
\[
\left \{ \Lie(G,\nu_{k})^{L} \right \}_{k \neq 0}.
\]
Also, we see from the proof of Thm.\ \ref{left_invariant_1-forms}
that there is a simple correspondence between left invariant Hamiltonian
1-forms and left invariant Hamiltonian vector fields which relies on
the isomorphism between $\g$ and its dual space via the inner product
$\innerprod{\cdot}{\cdot}$. As a result, we have the following proposition
which will be useful in the next section.
\begin{proposition}
\label{Ham_vect}
If $G$ is a compact simple Lie group with 2-plectic structure
$\nu_{k}$ and $\innerprod{\cdot}{\cdot}$ is the
inner product on the Lie algebra $\g$ of $G$ used in the construction of
$\nu_{k}$, then there is an isomorphism of vector spaces
\[\varphi \maps \X \left(G\right)^L \stackrel{\sim}{\longrightarrow} \hamL \]
such that $\varphi(v)=k\innerprod{\theta_{L}(v)}{\theta_{L}(\cdot)}$ is the unique
left-invariant Hamiltonian 1-form whose Hamiltonian vector field is $v$.
\end{proposition}
% rewrote proof !!!
\begin{proof}
We show only uniqueness since the rest of the proposition follows
immediately from the arguments made in the proof of Thm.\ 
\ref{left_invariant_1-forms}. Let $\alpha$ and $\beta$ be
left invariant 1-forms. The arguments made in the aforementioned
proof imply $d\alpha = -\ip{\alpha}\nu_{k}$ and $d\beta =
-\ip{\beta}\nu_{k}$, where $v_{\alpha}$ and $v_{\beta}$ are the unique
left-invariant vector fields such that  $\alpha=
k\innerprod{\theta_{L}(v_{\alpha})}{\cdot}$ and $\beta= k\innerprod{\theta_{L}(v_{\beta})}{\cdot}$.
If $v_{\alpha}=v_{\beta}$ is the Hamiltonian vector field for both $\alpha$ and
$\beta$, then the nondegeneracy of the inner product
implies $\alpha =\beta$.
\end{proof}  

\begin{remark}
In general, if $\alpha$ and $\beta$ are Hamiltonian 1-forms sharing
the same Hamiltonian vector field, then $d(\alpha-\beta)=0$. 
Hence, Prop.\ \ref{Ham_vect} implies that there are no non-trivial left
invariant closed 1-forms. Since the left-invariant de Rham cohomology of $G$ is isomorphic to the Lie
algebra cohomology of $\g$, Prop.\ \ref{Ham_vect} is equivalent to the
well-known fact that $H^{1}_{\mathrm{CE}}(\g,\R)=0$ for any simple Lie algebra.
\end{remark}
   
\section{The string Lie 2-algebra}
\label{string_Lie} 
We have described how to construct a Lie
2-algebra of left-invariant forms, from any compact simple Lie group $G$, and any nonzero real
number $k$, using the 2-plectic structure $\nu_{k}$.  Now we show that
this Lie 2-algebra is isomorphic to the `string Lie 2-algebra' of $G$.

It was shown in previous work by Baez and Crans \cite{HDA6} that Lie
2-algebras can be classified up to equivalence by data consisting
of:
\begin{itemize}
\item{a Lie algebra $\g$,}
\item{a vector space $V$,}
\item{a representation $\rho \maps \g \to \End\left(V\right)$,}
\item{an element $[j]\in H^{3} \left(\g,V \right)$ of the Lie algebra
  cohomology of $\g$.}   
\end{itemize}
A Lie 2-algebra $L$ is constructed from this data by
setting the space of 0-chains $L_{0}$ equal to $\g$, the space
1-chains $L_{1}$ equal to $V$, and the differential to be the zero
map: $d=0$. The bracket $\left[\cdot,\cdot \right]
\maps L \otimes L \to L$ is defined to be the Lie bracket on $\g$ in degree
0, and defined in degrees 1 and 2 by:
\[ [x,a]=\rho_{x}(a), \qquad [a,x]=-\rho_{x}(a), \qquad [a,b]=0,\]
for all $x  \in L_{0}$ and $a,b \in L_{1}$. 
The Jacobiator is taken to be any 3-cocycle $j$ representing
the cohomology class $[j]$.

From this classification we can construct the \textbf{string Lie
2-algebra} $\g_{k}$ of a compact simple Lie group $G$ by taking
$\g$ to be the Lie algebra of $G$, $V$ to be $\R$, $\rho$ to
be the trivial representation, and 
\[  j(x,y,z)=k \innerprod{x}{\left[y,z\right]} \]
where $k \in \R$.  When $k \neq 0$,
the 3-cocycle $j$ represents a nontrivial cohomology class.  Note that
since $\rho$ is trivial, the bracket of $\g_{k}$ is trivial in all
degrees except 0.

It is natural to expect that the string Lie 2-algebra is closely
related to the Lie 2-algebra $\Lie(G,\nu_{k})^{L}$ described in Corollary
\ref{left_inv_L2A}, since both are built
using solely the trilinear form $k\inp{\cdot}{[\cdot,\cdot]}$ on $\g$.  Indeed, this
turns out to be the case:
\begin{theorem}
\label{string_Lie_Thm}
If $G$ is a compact simple Lie group with Lie algebra $\g$ 
and 2-plectic structure $\nu_{k}$, then the string Lie 2-algebra
$\g_{k}$ is isomorphic to the Lie 2-algebra $\Lie(G,\nu_{k})^{L}$
of left-invariant Hamiltonian 1-forms.
\end{theorem}
\begin{proof}
The underlying chain complex of $\g_{k}$ is $\R \xrightarrow{0} \g$,
while the underlying chain complex of $\Lie(G,\nu_{k})^{L}$ is
$\R \xrightarrow{0} \g^{\ast}$. The isomorphism given Prop.\
\ref{Ham_vect}:
\[
\varphi \maps \X(G)^{L} \iso \hamL, \quad \varphi(v)=k\inp{\theta_{L}(v)}{\theta_{L}(\cdot)}
\] 
induces an isomorphism of complexes
\[
\xymatrix{
\R \ar[d]_{\id} \ar[r]^{0} & \g \ar[d]^{\varphi} \\
\R  \ar[r]^{0} & \g^{\ast} 
}
\]
Note we implicitly used the identifications $\g \cong \X(G)$ and
$\g^{\ast} \cong \hamL$. Let 
$[\cdot,\cdot]$ and $\brac{\cdot}{\cdot}$ be the brackets of $\g_{k}$
and $\Lie(G,\nu_{k})^{L}$, respectively.
According to Def.\ \ref{homo}, we must show that the maps
$\brac{\cdot}{\cdot} \circ \left( \varphi \otimes \varphi \right)$ and
$\varphi \circ \left [\cdot,\cdot \right]$ are chain homotopic. They are,
in fact, equal. 

Indeed, if $v_1,v_2\in \g$, then it follows from Proposition \ref{Ham_vect} that
$\varphi(v_1)$, $\varphi(v_2)$, and $\varphi \left([v_{1},v_{2}]\right)$ are
the unique left invariant Hamiltonian 1-forms whose
Hamiltonian vector fields are $v_1$, $v_2$, and $[v_1,v_2]$,
respectively. But Proposition \ref{bracket_prop} implies 
\[d\brac{\varphi(v_1)}{\varphi(v_2)}= -\iota_{\left[v_{1},v_{2} \right]} \nu_{k}.\]
Hence $[v_{1},v_{2}]$ is also the Hamiltonian vector field of
$\brac{\varphi(v_{1})}{\varphi(v_{2})}$. It then follows from uniqueness that 
$\brac{\varphi(\cdot)}{\varphi(\cdot)}=\varphi \left([\cdot,\cdot] \right)$.
\end{proof}

We conclude this chapter by showing that $\Lie(G,\nu_k)$ and $\g_{k}$
are not equivalent. 
\begin{prop}
If $G$ is a compact simple Lie group with Lie algebra $\g$, then the
Lie 2-algebra of Hamiltonian 1-forms $\Lie(G,\nu_{k})$ and the string
Lie 2-algebra $\g_{k}$ are not quasi-isomorphic.
\end{prop}
\begin{proof}
By definition, any quasi-isomorphism of Lie 2-algebras must induce an
isomorphism on homology.
Hence, to prove the statement, it is sufficient to show that the homology of the complex 
\[
L_{\bullet}= \cinf(G) \xrightarrow{d} \Gham,
\]
is not isomorphic to the complex $ \R \xrightarrow{0} \g$.
We will prove this by showing that the degree 0 homology of $L_{\bullet}$
has dimension greater than $\dim \g = \dim \X(G)^{L}$.

Let $\theta_{R} \in \Omega^{1}(G,\g)$ be the right-invariant
Maurer-Cartan form. At any point $g \in G$, it can be written as
\[ 
\theta_{R} \vert_{g}(v) = R_{g^{-1} \ast}v, \quad v \in T_{g}G.
\]
Therefore, $\theta_{R} \vert_{g} = \Ad_{g} \theta_{L}
\vert_{g}$. Since the 2-plectic form $\nu_{k}$ is left and right
invariant, we have the equalities:
\begin{align*}
\nu_{k}&=k \inp{\theta_{L}}{[\theta_{L},\theta_{L}]} \\
&= k\inp{\Ad_{g}\theta_{L}}{\Ad_{g}[\theta_{L},\theta_{L}]} \\
&=k\inp{\Ad_{g}\theta_{L}}{[\Ad_{g}\theta_{L},\Ad_{g}\theta_{L}]} \\
&=k\inp{\theta_{R}}{[\theta_{R},\theta_{R}]}. 
\end{align*}
The last equality implies that we can use the proof of
Thm.\ \ref{left_invariant_1-forms} to show that every right invariant
form is Hamiltonian.

Since the Lie algebra $\g$ is simple, it is not abelian. Therefore,
there exists $x,y \in \g$ such that $[x,y] \neq 0$. Let $v^{x}$ be the
right invariant vector field equal to $x$ at the identity. That is,
\[
v^{x} \vert_{g} = R_{g \ast} x.
\]
Note that $v^{x}$ is the Hamiltonian vector field corresponding to the right
invariant Hamiltonian 1-form
$k\inp{\theta_{R}(v^{x})}{\theta_{R}}$. 
We claim $v^{x}$ is not left invariant. Indeed, if it was then the
equality
\[
L_{g \ast} x = v^{x} \vert_{g} = R_{g \ast} x
\]
would hold for all $g$. In particular, this implies
\[
\Ad_{\exp(ty)} x =x,
\]
and therefore
\[
[y,x]= \left. \frac{d}{dt} \Ad_{\exp(ty)} x \right \vert_{t=0} = 0,
\]
which contradicts our choice of $x$ and $y$. Hence 
\begin{equation} \label{intersect}
\X(G)^{L} \cap \mathrm{span}_{\R} v^{x} = 0
\end{equation}

The kernel of the surjection $\Gham \epi \XGham$ which sends a
Hamiltonian 1-form to its vector field is the space of closed
1-forms. Since $G$ is compact, its de Rham
cohomology is isomorphic to the Chevalley-Eilenberg cohomology
of $\g$. Since $\g$ is simple, its first cohomology group
vanishes. Hence every closed 1-form on $G$ is exact. Therefore,
\[
H_{0}(L_{\bullet}) = \Gham/d\cinf(G) \cong \XGham.
\]
The left invariant vector fields $\X(G)^{L} \cong \g$ are all
Hamiltonian by Prop.\ \ref{Ham_vect}.
Since $v^{x}$ is Hamiltonian, (\ref{intersect}) implies
\[
\dim \g < \dim \XGham =\dim H_{0}(L_{\bullet}).
\]
\end{proof}

\chapter{Stacks, gerbes, and Deligne cohomology} \label{stacks_chapter}
In this chapter, we begin the passage from the classical to the
quantum by introducing the technical machinery needed to geometrically quantize $n$-plectic manifolds.

A principal $\U(1)$-bundle $P$ over a manifold $M$ can be specified by
giving $\U(1)$-valued transition functions with respect to an open cover of $M$.
A connection on $P$ is given by specifying local 1-forms on $M$ that
satisfy a compatibility condition with the transition functions. The
exterior derivative of these 1-forms gives a global 2-form on
$M$ called the curvature of the connection. Conversely, if $M$ is
equipped with a closed 2-form $\omega$ satisfying a certain integrality
condition, then one can show that there exists a principal
$\U(1)$-bundle, with connection, on $M$ whose curvature is $\omega$.
When $\omega$ is also non-degenerate, 
the bundle or, equivalently, its associated Hermitian line bundle,
plays a major role in the geometric quantization of the symplectic manifold $(M,\omega)$.

Our goal is to generalize these facts to $n$-plectic geometry. 
We begin by observing that the word ``bundle''
can be replaced by the word ``sheaf''. From any fiber bundle $E \to M$,
one can construct a sheaf of sections, which assigns to an open set $U \ss M$ the set of
local sections $\sigma \maps U \to E$.
In particular, the sheaf of sections of a
principal $\U(1)$-bundle is what is known as a
`$\sh{\U(1)}$-torsor', where $\sh{\U(1)}$ denotes the sheaf of
sections of the trivial $\U(1)$-bundle.
These torsors can be equipped with extra
structure which gives a connection on the corresponding bundle. 

The higher analogue of a sheaf is what is known as a `stack'. In particular, the
higher analogue of a $\sh{\U(1)}$-torsor is a special kind of stack
called a $\U(1)$-gerbe. Just as the transition functions of a
$\sh{\U(1)}$-torsor give a 1-cocycle, the transition
functions of a $\U(1)$-gerbe give a 2-cocycle.  Stacks and
gerbes were originally developed within the context of algebraic
geometry by Grothendieck \cite{Grothendieck_stacks} and Giraud
\cite{Giraud:1971}, respectively.  More recent
work demonstrates that they naturally arise in differential geometry
as well.  Brylinski \cite{Brylinski:1993} showed that $\U(1)$-gerbes
on manifolds can be equipped with additional structures, which we call
`2-connections'.  These are the higher analogues of connections on
$\U(1)$-bundles. More precisely, a 2-connection on a $\U(1)$-gerbe
over $M$ is specified by local 1-forms and 2-forms on $M$ satisfying
various compatibility conditions. The exterior derivative of the
2-forms give a global closed 3-form called the `2-curvature'.
Conversely, if $M$ is equipped with a closed 3-form $\omega$
satisfying an integrality condition, then one can show that there
exists a $\U(1)$-gerbe with 2-connection on $M$ whose 2-curvature is
$\omega$. As we will see, in analogy with the symplectic case, 
$\U(1)$-gerbes with 2-connections play an important role in
the quantization of 2-plectic manifolds.

Brylinski's results rely heavily on a formalism called `Deligne
cohomology', which can be thought of as a
refinement of the usual \v{C}ech cohomology that classifies principal
bundles. In degree one, Deligne cohomology classifies principal $\U(1)$-bundles
equipped with a connection. Similarly, in degree two, it classifies
$\U(1)$-gerbes equipped with a 2-connection. It is easy to describe 
the higher degree groups as well.
% and we will use them to
%prequantize $n$-plectic manifolds in the next Chapter. 
However, geometric structures \cite{Gajer:1997} that are classified by
these groups are, in general, more difficult to work with.

Let us conclude this introduction by briefly outlining the main results found in the chapter.
We first review the basic theory of stacks and
gerbes. We then give a somewhat detailed description of 
Deligne cohomology, and we provide proofs of some statements not
easily found in the literature. After presenting Brylinski's construction
for equipping a gerbe with a 2-connection, we introduce what we call
a `2-line stack'. This stack categorifies the concept of
a Hermitian line bundle . We show that every $\U(1)$-gerbe with
2-connection has an associated 2-line stack with
2-connection. In the final section, we present Carey, Johnson, and Murray's
formalism \cite{Carey:2004}  for computing the holonomy of a
2-connection, which we will use in our quantization procedure for
2-plectic manifolds in Chapter \ref{quantization_chapter}.

\section{Stacks} \label{stacks_sec}
When introducing sheaf theory, one begins by first defining a
presheaf on a topological space $M$ as a contravariant functor $\Op(M)
\to \Set$. The objects of the category $\Op(M)$ are open sets of $M$ and the morphisms
are inclusion maps. Similarly, in the theory of stacks, we begin by 
defining fibered categories and prestacks. Just as a
presheaf assigns a set to each open set $U \ss M$, a fibered category
assigns a category to each such set.

\begin{definition}[\cite{Moerdijk:2002}] \label{fibered_cat_def}
A {\bf fibered category} $\F$ over $M$ consists of:
\begin{itemize}
\item{a category $\F(U)$ for each open set $U \subseteq M$,}
\item{a functor $i^\ast \maps \F(V) \to \F(U)$ for each inclusion $i
    \maps U \embed V$ of open sets,}
\item{a natural isomorphism $t_{i,j} \maps(ij)^{\ast}
    \stackrel{\sim}{\to} j^{\ast} i^{\ast}$ for each pair of
    composable inclusions 
\[
W \stackrel{j}{\embed} V \stackrel{i}{\embed} U,
\]
}
\end{itemize}
such that for any triple of composable inclusions
\[
Y \stackrel{k}{\embed} W \stackrel{j}{\embed} V \stackrel{i}{\embed} U
\]
the following diagram commutes:
\[
\xymatrix{
 (ijk)^{\ast} \ar[d]_{t_{i,jk}} \ar[r]^{t_{ij,k}}& k^*(ij)^{\ast} \ar[d]^{k^{\ast}t_{i,j}}\\
(jk)^{\ast}i^{\ast} \ar[r]^{t_{j,k} i^{\ast}}& k^{\ast} j^{\ast} i^{\ast}\\
}
\]
\end{definition}
The above definition implies that a fibered category is a
contravariant `pseudo-functor' $\F \maps
\Op(M) \to \Cat$. The following example of a fibered category is
perhaps the most important one for us.
\begin{example}[Sheaves on a manifold]\label{sh}
Let $M$ be a manifold. To each open set $U \ss M$, assign the category
$\Shf(U)$, whose objects are sheaves on $U$. To each inclusion of open
sets $V \stackrel{i}{\embed} U$ assign the functor
\[
\begin{array}{c}
\Shf(U) \stackrel{i^{\ast}}{\to} \Shf(V) \\
F \mapsto F \vert_V,
\end{array}
\]
where $F \vert_V$ is the restriction of the sheaf $F$ to the open set
$V$. For any open set $W \ss V$, we have $F \vert_{V}(W)=F(W)$. Hence, given
$W \stackrel{j}{\embed} V \stackrel{i}{\embed} U$, the functors
$(ij)^\ast$ and $j^{\ast}i^{\ast}$ are equal. Therefore, the natural
isomorphisms $t_{i,j}$ may be taken to be the identity.
\end{example}

\begin{definition}[\cite{Moerdijk:2002}]
A {\bf morphism} between fibered categories $\F$ and $\G$ over $M$ consists of
\begin{itemize}
\item{a functor $\phi_{U} \maps \F(U) \to \G(U)$ for every open set
    $U \ss M$,}
\item{a natural isomorphism $\alpha_{i} \maps \phi_{V} i^{\ast} \iso
    i^{\ast}\phi_{U}$ for every inclusion $V \stackrel{i}{\embed} U$, 
    such that for every pair of composable inclusions 
    $W \stackrel{j}{\embed} V \stackrel{i}{\embed} U$ the diagram
\[
\xymatrix{
\phi_{W}(ij)^{\ast} \ar[d]_{\phi_{W}\tau_{i,j}}  \ar[rr]^{\alpha_{ij}} & & (ij)^{\ast} \phi_{U}
\ar[d]^{\tau_{i,j}\phi_{U}} \\
 \phi_{W} j^{\ast}i^{\ast} \ar[r]^{\alpha_{j} i^{\ast}} & j^{\ast}
 \phi_{V}i^{\ast} \ar[r]^{j^{\ast}\alpha_{i}} & j^{\ast} i^{\ast}
   \phi_{U}
}
\]
commutes.
}
\end{itemize}
\end{definition}
\noi Recall that an isomorphism of presheaves is given by local isomorphisms
of sets. The corresponding notion for fibered categories is slightly weaker.
It incorporates equivalences of categories, rather than isomorphisms
of categories.

\begin{definition}\label{equivalence_def}
A morphism $(\phi,\alpha) \maps \F \to \G$ is an {\bf equivalence} iff every functor $\phi_{U}$ is
an equivalence of categories.\footnote{An equivalence in the sense of Def.\
  \ref{equivalence_def} is called a `strong equivalence' in \cite{Moerdijk:2002}.}
\end{definition}

If $\F$ is a fibered category over $M$, and $U \subseteq M$ is an open
set, then given any objects $x,y \in \F(U)$, one can construct a
presheaf on $U$ by assigning to an open set $V \stackrel{i}{\embed}U$ the set
$\Hom_{\F(V)}(i^{\ast}x,i^{\ast}y)$. We denote this presheaf $\sh{\Hom}_{\F}(x,y)$.
\begin{definition}[\cite{Moerdijk:2002}]
A fibered category $\F$ over $M$ is a {\bf prestack} iff for every open
set $U \subseteq M$ and objects $x,y \in \F(U)$, the presheaf
$\sh{\Hom}_{\F}(x,y)$ is a sheaf.
\end{definition}

Our definition of a stack will, again, come from Moerdijk \cite{Moerdijk:2002}.  
However, it is more convenient to give his definition using nerves of
open covers, which we will explain below. This makes our
notation appear more like Brylinski's \cite{Brylinski:1993} Def.\
5.2.1. However, we warn the reader that Brylinski's definition of a fibered category uses
a ``larger'' source category than $\Op(M)$. Its objects
are arbitrary local homeomorphisms into $M$. For what we need to do, it is not
necessary to use this larger category.
 
%(NOTE: THE FOLLOWING NERVE STUFF ALLOWS ME TO HIDE DESCENT CATEGORIES,
%MAKING THE STACK DEFINITION NICER TO READ. IT
%ALSO MAKES THE PROOF THAT TWISTED SHEAVES FORM A STACK WAY EASIER)
Given an open cover $\cU=\{U_{a}\}$ of an open set $V \subseteq M$,
we consider the disjoint union $\cU^{[0]}=\coprod_{a} U_{a}$, 
and the $n$-fold fiber product:
\begin{equation} \label{fiber_prod}
\cU^{[n]} = \underset{n+1}{\underbrace{\cU^{[0]}\times_{V} \cdots \times_{V} \cU^{[0]}}}
=\coprod_{a_{1}, a_{2}, \cdots, a_{n+1}} U_{a_{1}} \cap U_{a_{2}} \cap
\cdots \cap U_{a_{n+1}}.
\end{equation}
There is a map $p_{0} \maps \cU^{[0]} \to V$ given by the inclusion
maps $U_{a} \embed V$. Similarly, there exists $n+1$ maps
$p_{1,\ldots,\hat{k},\ldots,n+1} \maps \cU^{[n]}\to \cU^{[n-1]}$
determined by  inclusion maps of the form
\begin{equation}\label{inclu}
U_{a_1} \cap \cdots \cap U_{a_{n+1}} \embed
U_{a_1} \cap \cdots \cap U_{a_{k-1}} \cap \widehat{U_{a_k}} \cap
U_{a_{k+1}} \cap \cdots \cap U_{a_{n+1}}, 
\end{equation}
%These maps are the degeneracy maps for a simplicial manifold:
Putting these all together, we obtain the following diagram in the
category of manifolds:
\begin{equation} \label{nerve}
\xymatrix{ 
\cdots \ar@<+1.4ex>[r] \ar@<+.5ex>[r] \ar@<-.5ex>[r] \ar@<-1.4ex>[r]
&\cU^{[2]} \ar@<+.9ex>[r]^{p_{12}} \ar[r]\ar@<-.9ex>[r]_{p_{23}}& \cU^{[1]}
\ar@<+.5ex>[r]^{p_{1}} \ar@<-.5ex>[r]_{p_{2}}& \cU^{[0]} \ar[r]^{p_0}
& V.
}
\end{equation}
This is called the {\bf nerve} of the cover $\cU$.
In particular, the maps $p_{1},p_{2}$ are the projections from the
first and second factor, respectively, and
$p_{12},p_{13},p_{23}$ are the projections from the first and second, first and
third, and second and third factors, respectively. We sometimes will
slightly abuse notation by writing the compositions $p_{1}p_{ij}$ and
$p_{2}p_{ij}$ as $p_{i}$ and $p_{j}$, respectively.
The nerve of a cover is useful for expressing the various gluing properties of both sheaves and
stacks. 

Let us establish just a bit more notation. If $F$ is a presheaf on
$M$, then we define the product
\begin{equation} \label{presheaf_nerve}
F(\cU^{[n]}):= \prod_{a_{1},\ldots,a_{n+1}}F(U_{a_{1}} \cap \cdots
\cap U_{a_{n+1}}).
\end{equation}
Then applying $F$ to the diagram (\ref{nerve}) gives, for example,
\begin{equation}\label{sheaf_glue}
\xymatrix{ 
F(V) \ar[r]^{p^{\ast}_{0}} &F(\cU^{[0]}) \ar@<+.5ex>[r]^{p_{1}^{\ast}}
\ar@<-.5ex>[r]_{p_{2}^{\ast}}& F(\cU^{[1]}) \cdots,
}
\end{equation}
where $p_{i}^{\ast}$ are {\it maps between sets} corresponding to
restriction of sections. 
%The gluing axiom for sheaves can be expressed in the
%following way: Given any open set $V \ss M$ and any cover $\cU$ of
%$V$,  $F(V) \stackrel{p_{0}^{\ast}}{\to} $ is the equalizer of the
%diagram (\ref{sheaf_glue}) in the category of sets.
Now, if $\F$ is a fibered category on $M$, we define the category:
\[
\F(\cU^{[n]}):= \prod_{a_{1},\ldots,a_{n+1}}\F(U_{a_{1}} \cap \cdots
\cap U_{a_{n+1}}),
\]
where the product on the right-hand side is the product of
categories. We apply $\F$ to (\ref{nerve}) and obtain:
\[
\xymatrix{ 
\F(V) \ar[r]^{p_{0}^{\ast}}& \F(\cU^{[0]}) \ar@<+.5ex>[r]^{p_{1}^{\ast}}
\ar@<-.5ex>[r]_{p_{2}^{\ast}}& \F(\cU^{[1]})  
\ar@<+.9ex>[r]^{p_{12}^{\ast}} \ar[r]\ar@<-.9ex>[r]_{p_{23}^{\ast}}&
\F(\cU^{[2]}) \cdots
}
\]
Here, $p^{\ast}_{i},p^{\ast}_{ij}$ are {\it functors between
  categories}, which are determined by the
functors corresponding to the inclusions (\ref{inclu}).
Similarly, there are natural isomorphisms $t_{p_{i},p_{jk}}
\maps (p_{i}p_{jk})^{\ast}\to p_{jk}^{\ast}p_{i}^{\ast}$.

We can now give a relatively concise definition of a stack.
\begin{definition} \label{stack_def}
A prestack $\F$ over $M$ is a {\bf stack} if and only if given the
data: 
\begin{itemize}
\item{an open cover $\cU$ of an open set $V \subseteq M$,}
\item{an object $x \in \F(\cU^{[0]})$,}
\item{an isomorphism
\[
\phi \maps p_{2}^{\ast}x \stackrel{\sim}{\to} p_{1}^{\ast}x
\]
in $\F(\cU^{[1]})$ such that the following diagram in $\F(\cU^{[2]})$ commutes:
\[
\xymatrix{
p_{23}^{\ast}p_{2}^{\ast}x \ar[d]_{t^{-1}_{p_{2},p_{23}}} \ar[r]^{\phi} & p^{\ast}_{23}p^{\ast}_{1} x \ar[r]^{t^{-1}_{p_{1},p_{23}}} & p^{\ast}_{2} x \ar[d]^{t_{p_{2},p_{12}}}\\
p^{\ast}_{3} x \ar[d]_{t_{p_{2},p_{13}}} & & p^{\ast}_{12} p^{\ast}_{2} x \ar[d]^{\phi} \\
p^{\ast}_{13} p^{\ast}_{2} x \ar[d]_{\phi} & & p^{\ast}_{12} p^{\ast}_{1} x \ar[d]^{t^{-1}_{p_{1},p_{12}}} \\
p^{\ast}_{13} p^{\ast}_{1} x   \ar[rr]^{t^{-1}_{p_{1},p_{13}}}& & p^{\ast}_{1} p^{\ast}_{1} x 
}
\]
}
\end{itemize}
there exists an object $\tilde{x} \in \F(V)$, unique up to isomorphism,
together with an isomorphism
\[
\psi \maps p_{0}^{\ast} \tilde{x} \stackrel{\sim}{\to} x
\]
in $\F(\cU^{[0]})$ such that the following diagram in $\F(\cU^{[1]})$ commutes:
\[
\xymatrix{
p_{2}^{\ast}p^{\ast}_{0} \tilde{x}  \ar[d]_{\psi}
\ar[r]^{t_{p_{0},p_{2}}} &(p_{0}p_{2})^{\ast}\tilde{x} \ar @{=}[r]&(p_{0}p_{1})^{\ast}\tilde{x} \ar[r]^{t^{-1}_{p_{0},p_{1}}} & p^{\ast}_{1}p^{\ast}_{0}
\tilde{x} \ar[d]^{\psi}\\
p^{\ast}_{2} x \ar[rrr]^{\phi} &&& p^{\ast}_{1} x
}
\]
%where $\pi \maps \cU^{[1]} \to V$ is the projection $p_{0}p_{1}=p_{0}p_{2}$.
\end{definition}
\noi Hence, just as sections of a sheaf can glue together in a unique way,
objects in a stack can glue together uniquely up to isomorphism. In
addition, note that the prestack
condition implies that morphisms between objects can be glued together
as well. A morphism between stacks is simply a morphism between the
underlying fibered categories.

\begin{prop} \label{sh_stack}
Let $M$ be a manifold. The fibered category which assigns to an open
set $U \ss M$ the category $\Shf(U)$ of sheaves on $U$, as
defined in Example \ref{sh}, is a stack. 
\end{prop}
\begin{proof}
We refer the reader to Sec.\ 5.1 in \cite{Brylinski:1993} for the proof. 
\end{proof}
Finally, we mention that if $\F$ is a stack over $M$, then we will often refer to the objects of the
category $\F(M)$ as the {\bf global sections} of $\F$.

\section{Gerbes}
Roughly, gerbes are to stacks, as principal bundles are to fiber
bundles. To see this, let us first give the precise definition for a torsor.
\begin{definition}
Let $\sh{G}$ be the sheaf of smooth functions with values in the Lie
group $G$. A {\bf $\sh{G}$-torsor} over a manifold $M$ is a sheaf $F$
together with an action $\sh{G} \times F \to F$ such that for each $x
\in M$, there exists an open neighborhood $U$ of $x$ with the property
that for each open $V \ss U$, the set $F(V)$ is a principal
homogeneous $G(V)$-space.
\end{definition}
\noi The sheaf $\sh{G}$ itself is the {\bf trivial $\sh{G}$-torsor}.
Note the definition implies that if $F$ is a $\sh{G}$-torsor on
$M$, then $F$ is \textbf{locally isomorphic} to $\sh{G}$. That is, for
all $x\in M$ there exists an open neighborhood $U \ni x$, such that
restricted sheaves $F_{U}$ and $\sh{G}_{U}$ are isomorphic.  Morphisms
between $\sh{G}$-torsors are morphisms of the underlying sheaves which
respect the $\sh{G}$-action. As mentioned in the introduction to the
chapter, the sheaf of sections of a principal $G$-bundle is a
$\sh{G}$-torsor. Conversely, every $\sh{G}$-torsor is isomorphic to
such a sheaf of sections.

We can construct a fibered category on a manifold $M$ which assigns to
every open set $U$,  the category of $\sh{G}$-torsors over $U$.
Using the fact that $\Shf$ is a stack, it is not difficult to see that
this fibered category is also a stack, which we denote as  $\Tor_{G}$.
Just as $\sh{G}$-torsors are special kinds of sheaves, $\Tor_{G}$
is a special kind of stack. For example, for any open set $U$, the morphisms
in the category $\Tor_{G}(U)$ are all isomorphisms. Hence,
$\Tor_{G}(U)$ is a groupoid. In fact, it is a non-empty groupoid,
since we always have the trivial $\sh{G}$-torsor over every open set $U$.
Also, since every $\sh{G}$-torsor is locally isomorphic to $\sh{G}$,
any two $\sh{G}$-torsors in $\Tor_{G}(U)$ will become isomorphic when
pulled back to the category $\Tor_{G}(V)$, if $V$ is a ``small
enough'' open subset of $U$. By axiomatizing these facts, one arrives at
the definition of a $G$-gerbe. $\Tor_{G}$ itself is called the {\bf
  trivial $G$-gerbe}. In fact, as we will see, the definition implies
that a $G$-gerbe is a stack that is locally isomorphic to the stack $\Tor_{G}$. 
\begin{definition}[\cite{Brylinski:1993,Giraud:1971}] \label{G-gerbe}
Let $G$ be a Lie group. A stack $\G$ over $M$ is a {\bf $G$-gerbe} iff:
\begin{enumerate}
\item{for every open set $U \subseteq M$, the category $\G(U)$ is a groupoid,}
\item{there exists an open cover $\cU$ of $M$ such that the groupoid
    $\G(\cU^{[0]})$ is non-empty,}
\item{for every open set $V \subseteq M$ and every pair of objects $P,Q\in \G(V)$,
    there exists an open cover $\cU$ of $V$ such that $p_{0}^{\ast}P$
    and $p_{0}^{\ast}Q$ are isomorphic as objects in $\G(\cU^{[0]})$,}
\item{for every open set $U \subseteq M$ and every object $P
    \in \G(U)$, there exists a local isomorphism between the sheaf of
    groups $\sh{\Aut}_{\G}(P)=\sh{\Hom}_{\G}(P,P)$ 
    and the sheaf $\sh{G}_{U}$. This local isomorphism is unique up to
  inner automorphisms of $G$.}
\end{enumerate}
\end{definition}
\noi Roughly, a morphism between $G$-gerbes is a morphism between the underlying
stacks, which respects the local isomorphisms between the sheaves
$\sh{\Aut}_{\G}(P)$ and $\sh{G}_{U}$. See the definition following
Prop.\ 5.2.7 in \cite{Brylinski:1993} for the precise details.

% \begin{example}\label{trivial_gerbe}
% Let $M$ be a manifold and $G$ be a Lie group. Let $\Tor_{G}$ be the
% fibered category that assigns to each open set $U \ss
% M$ the category $\Tor_{G}(U)$ whose objects are $\sh{G}$-torsors on
% $U$ and whose morphisms are isomorphisms of $\sh{G}$-torsors. The proof of
% Prop.\ \ref{sh_stack} can be used to show $\Tor_{G}$ is a stack. It follows from
% the definition of $\sh{G}$-torsors that $\Tor_{G}$ is a $G$-gerbe. It is
% called the {\bf trivial $G$-gerbe} 
% \end{example}

\subsection*{The classification of $\U(1)$-gerbes} \label{gerbe_class}
From here on we shall only consider the case $G=\U(1)$. As we shall
see, $\U(1)$-gerbes are classified by the group $H^{3}(M,\Z)$, just as $H^{2}(M,\Z)$ classifies $\U(1)$-bundles. 

We first review the classification of principal $\U(1)$-bundles using sheaf
cohomology. We will always be working with paracompact
manifolds, therefore we canonically identify sheaf cohomology with its corresponding
\v{C}ech cohomology. 
Let us recall some basic facts concerning \v{C}ech cohomology. Let $F$ be a sheaf of abelian groups
on $M$, and let $\cU=\{U_{i}\}$ be an open cover. 
The space of {\bf \v{C}ech} {\boldmath $k$}{\bf -cochains} 
with values in $F$ is the abelian group
\begin{equation} \label{cech_cochain}
C^{k}(\cU,F) = \prod_{a_{1} < a_{2} < \ldots <a_{k+1}}F(U_{a_{1}} \cap
\cdots \cap U_{a_{k+1}}) \ss F(\cU^{[k]}).
\end{equation}
The {\bf \v{C}ech coboundary}: 
\[
C^{k}(\cU,F) \xto{\delta} C^{k+1}(\cU,F)
\]
is given, component-wise, by
\[
\delta (g)_{a_1,\ldots,a_{k+1}} = \sum_{j=1}^{k+1} (-1)^{j}
g_{a_1,\ldots, \widehat{a_{j}},\ldots, a_{k+1}} \vert_{U_{a_{1}} \cap \cdots \cap U_{a_{k+2}}}.
\]
The set of open covers of $M$ is a directed set, with the order given
by refinement. 
Therefore, the cohomology groups $H^{\bullet}(\cU,F)$ of
the complexes $(C^{\bullet}(\cU,F),\delta)$ form a direct system. The
{\bf \v{C}ech cohomology} of $M$ with values in $F$ is the direct
limit of these groups:
\[
H^{\bullet}(M,F) = \varinjlim_{\cU} H^{\bullet}(\cU,F).
\]
Recall that an open cover $\cU=\{U_{i}\}$ of $M$ is \textbf{good} iff every
non-empty intersection $U_{i_{1}} \cap \cdots \cap U_{i_{n}}$ is
contractible. Every manifold admits a good cover, and such covers are
cofinal in the aforementioned directed set. Hence, the direct
limit above can be computed by just considering good covers.

Let $P \to M$ be a principal $\U(1)$-bundle and $\cU=\{U_{i}\}$ an open
cover of $M$ admitting local trivializations of $P$. The corresponding
transition functions $g_{ij} \maps U_{i} \cap U_{j}\to
\U(1)$ satisfy the cocycle condition $g_{jk}g^{-1}_{ik}g_{ij}=1$ on
$U_{i} \cap U_{j} \cap U_{k}$, and hence give a class in
$H^{1}(M,\sh{\U(1)})$, the degree 1 cohomology group with
values in the sheaf of smooth $\U(1)$-valued functions. It is well-known
that $H^{1}(M,\sh{\U(1)})$ is in one-to-one correspondence with
isomorphism classes of principal $\U(1)$-bundles on the manifold $M$.

Let $\Zi$ denote the sheaf whose sections are
locally-constant functions with values in $\tpi \i \cdot \Z$, and
let $\icinf$ denote the sheaf of smooth imaginary-valued functions on $M$.
There is a short exact sequence
\begin{equation}\label{exp}
0 \to \Zi \embed \icinf \xrightarrow{\exp} \sh{\U(1)} \to 0,
\end{equation}
giving a long exact sequence in cohomology. Since $\icinf$ is a
 soft sheaf, the long exact sequence gives the isomorphisms:
\begin{equation}\label{chern_iso}
H^{k}(M,\sh{\U(1)}) \cong H^{k+1}(M,\Zi) \cong H^{k+1}(M,\Z).
\end{equation}
For $k=1$, the isomorphism (\ref{chern_iso}) associates to a
principal $\U(1)$-bundle its Chern class.

Now we consider the $k=2$ case, and explain how to obtain a $\U(1)$-valued 2-cocycle
from a $\U(1)$-gerbe $\G$. By the second axiom in
Def.\ \ref{G-gerbe}, there exists an open cover $\cU=\{U_{i}\}$ of the
manifold $M$,  
such that for all $i$, there exists an object $P_{i} \in \G(U_{i})$. By pulling back along refinements, we may
assume the following: $\cU$ is a good cover, there exists
isomorphisms of sheaves $\sh{\Aut}_{\G}(P_{i})\cong
\sh{\U(1)} \vert_{U_{i}}$ for all $P_{i}$ (by axiom 4), and there exists isomorphisms
\[
u_{ij} \maps P_{j} \vert_{U_{ij}}\iso  P_{i}\vert_{U_{ij}},
\]
where $P_{i}\vert_{U_{ij}}$ and $P_{j}\vert_{U_{ij}}$ are the
pullbacks of $P_{i}$ and $P_{j}$ to $\G(U_{i} \cap U_{j})$.
Therefore, by pulling back objects $P_{i}, P_{j}, P_{k}$ to $U_{i}
\cap U_{j} \cap U_{k}$, we have the commuting diagram
\[
\xymatrix{
P_{k} \vert_{U_{ijk}}  \ar[rr]^{u_{jk}} & &P_{j} \vert_{U_{ijk}} \ar[ld]^{u_{ij}} \\
&P_{i} \vert_{U_{ijk}} \ar[lu]^{u^{-1}_{ik}}&
}
\]
giving a morphism $u_{ik}^{-1} u_{ij} u_{jk} \in \sh{\Aut}_{\G}(P_{k})(U_{i} \cap U_{j} \cap U_{k})$.
Since $\sh{\Aut}_{\G}(P_{k}) (U_{i} \cap U_{j} \cap U_{k}) \cong
\sh{\U(1)}(U_{i} \cap U_{j} \cap U_{k})$, this automorphism
corresponds to a map $g_{ijk} \maps U_{i} \cap U_{j} \cap U_{k}  \to
\U(1)$. It is easy to see that $g_{ijk}$ satisfies the cocycle
condition on intersections $U_{i} \cap U_{j} \cap U_{k} \cap U_{l}$, and therefore gives a class $[g]
\in H^{2}(M,\sh{\U(1)})$.

Conversely, suppose $g_{ijk} \maps U_{i} \cap U_{j} \cap U_{k}  \to
\U(1)$ is a 2-cocycle on a good open cover $\cU=\{U_{i} \}$. 
Recall from the discussion preceding Def. \ref{G-gerbe} that $\Tor_{\U(1)}$ is a gerbe.
We construct a new gerbe $\G$ by ``twisting'' $\Tor_{\U(1)}$ by $h_{ijk}$.
Given an open set $V \ss M$, an object $(P_{i},u_{ij})$ in $\G(V)$ is
defined to be a collection of objects $P_{i} \in \Tor_{\U(1)}(V\cap U_{i})$, together with isomorphisms
\[
u_{ij} \maps P_{j} \vert_{V \cap U_{i} \cap U_{j}}\iso P_{i} \vert_{V \cap U_{i} \cap U_{j}}
\]
in $\Tor_{\U(1)}(V\cap U_{i} \cap U_{j})$, such that $u_{ik}^{-1}
u_{ij} u_{jk} = g_{ijk} \in \sh{\U(1)}(V\cap U_{i} \cap U_{j} \cap
U_{k})$. A morphism $(P_{i},u_{ij}) \to (P'_{i},u'_{ij})$ consists of
a family of morphisms of $\U(1)$-torsors $P_{i} \to P_{i}'$ whose
pullbacks in $\Tor_{\U(1)}(V\cap U_{i} \cap U_{j})$ commute with the
morphisms $u_{ij},u'_{ij}$. It is straightforward to show that by using the pullback
functors defined for $\Tor_{\U(1)}$, we obtain a stack $\G$ in this
way. To see that $\Aut_{\G}(P_{i},u_{ij})$ is locally isomorphic to
$\sh{\U(1)}$, note that such an automorphism must be given by a collection of
morphisms $ P_{i} \iso P_{i}$ corresponding to sections in
$\sh{\U(1)}(V \cap U_{i})$, which must agree when pulled back to $V
\cap U_{i} \cap U_{j}$. These glue to give a section in
$\sh{U(1)}(V)$, thereby establishing an isomorphism
$\Aut_{\G}(P_{i},u_{ij})(V) \cong \sh{\U(1)}(V)$. To show that the
other axioms in Def.\ \ref{G-gerbe} hold, one may show that the
categories $\G(U_{i})$ and $\Tor_{\U(1)}(U_{i})$ are equivalent for
all $U_{i}$. (This follows from the fact that $g_{ijk}$ restricted to
$U_{i}$ is a 2-coboundary since $H^{2}(U_{i},\sh{\U(1)})=0$. See Sec.\
5.2 in \cite{Brylinski:1993}.)  This construction, combined with the
isomorphism (\ref{chern_iso}) leads to the following theorem:
\begin{theorem}[\cite{Brylinski:1993,Giraud:1971}]\label{gerbe_class_theorem}
There is a one-to-one correspondence between
equivalence classes of $\U(1)$-gerbes on a manifold $M$ and classes in
$H^{3}(M,\Z)$.
\end{theorem}
 \noi In fact, one can go further and define the product of two
 $\U(1)$-gerbes, which is similar to the contracted product of
 principal $\U(1)$-bundles. The set of equivalence classes of
 $\U(1)$-gerbes therefore form an abelian group, and the bijection in the above
 theorem lifts to an isomorphism of groups.

$\U(1)$-gerbes can be equipped with structures that are the higher analogs of
connections and curvature. To classify these, we need to
introduce a more sophisticated cohomology theory.

\section{Deligne cohomology} \label{Deligne_sec}
To motivate this section, let us return to the familiar case of principal bundles.
If $P \to M$ is a principal $\U(1)$-bundle equipped with a connection,
then, in addition to the transition
functions $g_{ij}$, we have local 1-forms $\theta_{i} \in
\Omega^{1}(U_{i})$ satisfying a cocycle-like condition $\i \cdot (\theta_{i}-\theta_{j}) =g_{ij}^{-1}dg_{ij}$
on $U_{i} \cap U_{j}$. The curvature of the connection is the global
2-form $\omega$ on $M$ satisfying $\omega \vert_{U_{i}} =
d\theta_{i}$. 
%The short exact sequence of groups $\Z \to \R \to \U(1)$
%gives the long exact sequence
%\[
%\cdots \to H^{k-1}(M,\U(1)) \to  H^{k}(M,\Z) \to H^{k}(M,\R) \to
%\cdots,
%\]
%which maps the Chern class of $P$ in $H^{2}(M,\Z)$ to 
%$[\omega] \in H^{2}_{\mathrm{dR}}(M) \cong H^{2}(M,\R)$.

The classification of principal $\U(1)$-bundles equipped with
connection requires a refinement of the \v{C}ech
cohomology group $H^{1}(M,\sh{U(1)})$. 
The purpose of real Deligne cohomology 
is to make this notion precise. In fact, as we
will see, Deligne cohomology provides such a refinement for
any geometric objects classified by $H^{k}(M,\Z)$ for arbitrary
$k$. 

The primary reference for what follows is Sec.\ 1.5 of Brylinski
\cite{Brylinski:1993}. However, Brylinski works with the group $\Cx$
instead of $\U(1)$. What we call real Deligne cohomology is 
presented, without proofs, in Sec.\ 3 of Carey, Johnson, and Murray \cite{Carey:2004}.

Let $\Omega^{k}$ denote the sheaf of smooth differential $k$-forms on a
manifold $M$, and let $\dlog \maps \sh{\U(1)} \to \Omega^{1}$ be the differential operator
\[
\dlog := \frac{1}{\i}d\log.
\]
\begin{definition}[\cite{Carey:2004}]
The {\bf real Deligne cohomology} $H^{\bullet}(M,D_{n}^{\bullet})$
of $M$ is the \v{C}ech hyper-cohomology the exact sequence of sheaves: 
\[
D_{n}^{\bullet}:= ~\sh{\U(1)} \xto{\dlog}\Omega^{1} \stackrel{d}{\to}
\cdots \stackrel{d}{\to} \Omega^{n}, \quad n \geq 1.
\]
\end{definition}
\noi We compute $H^{\bullet}(M,D_{n}^{\bullet})$ in the
following way. Let $\cU=\{U_{i}\}$ be an open cover of $M$. We consider
the double complex of abelian groups:
\begin{equation}\label{cech_resolution}
\xymatrix{
\vdots & \vdots& \vdots && \vdots\\
C^{2}(\cU,\sh{\U(1)}) \ar[u]^{\delta} \ar[r]^-{\dlog}& C^{2}(\cU,\Omega^{1})
\ar[u]^{\delta} \ar[r]^{d} &
C^{2}(\cU,\Omega^{2}) \ar[u]^{\delta} \ar[r]^{d}&\cdots
\ar[r]^{d} & \ar[u]^{\delta} C^{2}(\cU,\Omega^{n}) \\
C^{1}(\cU,\sh{\U(1)})\ar[u]^{\delta} \ar[r]^-{\dlog}& C^{1}(\cU,\Omega^{1})
\ar[u]^{\delta} \ar[r]^{d} &
C^{1}(\cU,\Omega^{2}) \ar[u]^{\delta} \ar[r]^{d}&\cdots
\ar[r]^{d} & \ar[u]^{\delta} C^{1}(\cU,\Omega^{n}) \\
C^{0}(\cU,\sh{\U(1)}) \ar[u]^{\delta} \ar[r]^-{\dlog}& C^{0}(\cU,\Omega^{1})
\ar[u]^{\delta} \ar[r]^{d} &
C^{0}(\cU,\Omega^{2}) \ar[u]^{\delta} \ar[r]^{d}&\cdots
\ar[r]^{d} & \ar[u]^{\delta}C^{0}(\cU,\Omega^{n}) \\
}
\end{equation}
where $\delta$ is the usual
\v{C}ech co-boundary operator, and $C^{p}(\cU,\sh{\U(1)})$ and 
$C^{p}(\cU,\Omega^{k})$ denote the \v{C}ech $p$-cochains
(as defined in Eq.\ \ref{cech_cochain}).
The total complex of the double complex (\ref{cech_resolution}) is
\[
C^{0}(\cU,\sh{\U(1)}) \xto{\mathbf{d}}  C^{1}(\cU,\sh{\U(1)}) \oplus 
C^{0}(\cU,\Omega^{1}) \xto{\mathbf{d}}  
C^{2}(\cU,\sh{\U(1)}) \oplus 
C^{1}(\cU,\Omega^{1}) \oplus C^{0}(\cU,\Omega^{2}) \xto{\mathbf{d}}  \cdots,
\]
with total differential
\[
\begin{array}{c}
\d g = \delta g + (-1)^{p}\frac{1}{\i}d\log g, \quad g \in C^{p}(\cU,\sh{\U(1)})\\
\d\theta^{k} = \delta \theta^{k} +  (-1)^{p}d \theta^{k}, \quad \theta^{k}
\in C^{p}(\cU,\Omega^{k}).
\end{array}
\]
Let $H^{\bullet}(\cU,D_{n}^{\bullet})$ denote the cohomology of the
above total complex. The \v{C}ech hyper-cohomology of
$D_{n}^{\bullet}$ is, by definition, the direct limit
of the groups $H^{\bullet}(\cU,D_{n}^{\bullet})$ over all covers
\[
H^{\bullet}(M,D_{n}^{\bullet}) = \varinjlim_{\cU} H^{\bullet}(\cU,D_{n}^{\bullet}).
\]
If an open cover $\cU=\{U_{i}\}$ of $M$ is good, then it is well known that there is an isomorphism
\[
H^{\bullet}(M,D_{n}^{\bullet}) \cong H^{\bullet}(\cU,D_{n}^{\bullet}).
\]

We will be particularly interested in the groups $H^{n}(M,D_{n}^{\bullet})$,
which can be thought of as a refinement of the usual
\v{C}ech cohomology groups $H^{\bullet}(M,\sh{U(1)})$.  
\begin{definition} \label{n-cocycle_def}
A {\bf Deligne} {\boldmath $n$}{\bf -cocycle on $M$} is a representative of a class in $H^{n}(M,D_{n}^{\bullet})$ 
\end{definition}
\noi Hence, a Deligne $n$-cocycle is given by a cover $\cU$ of $M$ and
a collection $(g,\theta^{1},\theta^{2},\cdots,\theta^{n})$ with
\[
g \in C^{n}(\cU,\sh{\U(1)}), \quad \theta^{k} \in
C^{n-k}(\cU,\Omega^{k}),
\]
satisfying 
\begin{equation} \label{n-cocycle_conditions}
\begin{split}
\delta g &=1, \\
\delta \theta^{1} &= \frac{1}{\i}(-1)^{n-1} d \log g,\\
\delta \theta^{k} &= (-1)^{n-k}d \theta^{k-1}, ~ \text{for $2 \leq k \leq n$.}
\end{split}
\end{equation}
\noi We consider examples for $n=1$ and $n=2$ later on.  The projection
\[
D_{n}^{\bullet} \to D_{n}^{0}=\sh{\U(1)}
\]
gives a surjection in cohomology
\[
\begin{array}{c}
H^{n}(M,D_{n}^{\bullet}) \epi  H^{n}(M,\sh{U(1)})\\
\bigl [g,\theta^{1},\cdots,\theta^{n} \bigr] \mapsto [g]. 
\end{array}
\]
Hence, via the isomorphism $H^{p}(M,\sh{U(1)}) \cong H^{p+1}(M,\Zi)$, we have a surjection
\begin{equation} \label{chern}
c \maps H^{n}(M,D_{n}^{\bullet})\epi  H^{n+1}(M,\Zi). 
\end{equation}
We call $c([g,\theta^{1},\cdots,\theta^{n}])$ the \textbf{Chern class} of $[g,\theta^{1},\cdots,\theta^{n}]$. 

There is also a map of complexes
\[
\xymatrix{
\sh{U(1)} \ar[d] \ar[r]^{\dlog}& \Omega^{1}
\ar[d] \ar[r]^{d} &
\Omega^{2} \ar[d] \ar[r]^{d}&\cdots
\ar[r]^{d} & \ar[d]^{d}\Omega^{n} \\
0 \ar[r] & 0 \ar[r]&  0 \ar[r] & \cdots \ar[r] & \Omega^{n+1},
}
\]
given by the de Rham differential $d$.
The induced map on the corresponding \v{C}ech resolutions sends an $n$-cocycle
$(g,\theta^{1},\cdots,\theta^{n})$ to $d\theta^{n}\in C^{0}(\cU,\Omega^{n+1})$.
The equalities in (\ref{n-cocycle_conditions}) give $\delta
\theta^{n} = d\theta^{n-1}$. Hence, $\delta d\theta^{n}=0$, which
implies $d\theta^{n}$ is the restriction of a globally defined closed form.
This gives a map 
\begin{equation}\label{kappa_map}
\begin{split}
\kappa \maps H^{n}(M,D_{n}^{\bullet})\to  Z^{n+1}(M)\\
\kappa ([g,\theta^{1},\cdots,\theta^{n}]) = (-1)^{n} d\theta^{n},
\end{split}
\end{equation}
where $Z^{n+1}(M)$ are the closed $(n+1)$-forms on $M$. 
The forthcoming examples will make it clear why the sign $(-1)^{n}$ appears in
the definition of $\kappa$.
\begin{definition}\label{n-curvature}
The {\boldmath $n$}{\bf -curvature} of a Deligne $n$-cocycle
$(g,\theta^{1},\cdots,\theta^{n})$ on a manifold $M$ is the closed $(n+1)$-form
\[
\kappa([g,\theta^{1},\cdots,\theta^{n}]).
\]
\end{definition}
Let us consider some examples of Deligne $n$-cocycles and their $n$-curvatures.
% \begin{example}[$\U(1)$-valued functions]
% For $n=0$, a class $[g] \in H^{0}(M,D_{0}^{\bullet})$ is represented
% by a collection of maps $g_{i} \maps U_{i} \to \U(1)$ such that
% $g_{i}=g_{j}$ on $U_{i} \cap U_{j}$ and $g_{i}^{-1}dg_{i}=0$.
% Hence, $[g]$ is represented by a global $\U(1)$-valued function $\tilde{g} \maps M \to
% \U(1)$ with $\tilde{g}_{i}=g_{i}$. The 1-curvature is
% $\kappa([g])=\tilde{g}^{\ast}(d\theta)$. 
% \end{example}
\begin{example}[Principal $\U(1)$-bundles]\label{n=1}
For $n=1$, a class in $H^{1}(M,D_{1}^{\bullet})$ is represented
by  maps $g_{ij} \maps U_{i} \cap U_{j}  \to \U(1)$, and 1-forms
$\theta_{i} \in \Omega^{1}(U_{i})$, satisfying the cocycle conditions 
\begin{align*}
g_{jk}g^{-1}_{ik}g_{ij}=1, \quad \text{on $U_{i} \cap U_{j} \cap U_{k}$} \\
\i\cdot (\theta_{j}-\theta_{i}) =g_{ij}^{-1}dg_{ij},  \quad \text{on $U_{i} \cap U_{j}$}\\
\end{align*}
\noi The 1-curvature is the closed 2-form $\omega$ on $M$ satisfying
\[
\omega = -d\theta_{i} \quad \text{on $U_{i}$}.
\]

Let us consider two equivalent ways of realizing the above local data as a
geometric object. (Our convention follows Section 2.2 of \cite{Brylinski:1993}.)
First, it gives us a Hermitian line bundle $L \to M$, 
equipped with a connection $\conn$. The local
trivializations $s_{i} \maps U_{i} \to \U(1)$ of $L$ satisfy
\[
s_{i}=g_{ij} s_{j}, ~ \text{on $U_{i} \cap U_{j}$}.
\]
The connection $\conn$ is locally determined by the 1-forms
$-\theta_{i}$:
\[
\frac{\conn(s_{i})}{s_{i}} = -\i \cdot \theta_{i},
\]
which satisfy
\[
-\i \cdot(\theta_{i} -\theta_{j}) = g_{ij}^{-1}dg_{ij},
\]
because $(g,\theta)$ is a cocycle.
The curvature of the bundle is given by the global 2-form $-d\theta_{i}$.

Equivalently, the Deligne 1-cocycle gives a principal $\U(1)$-bundle
$P \to M$ equipped with a connection, i.e.\ a $\u(1)$-valued 1-form $\theta$ on $P$.
$L$ is the line bundle associated to $P$. Using a trivialization
$s \maps U_{i} \to P$, the connection 1-form on $P$ can be expressed
locally as
\[
s_{i}^{\ast} \theta = -\i \cdot \theta_{i}.
\]
Hence, the Deligne class $[g, \theta]$ corresponds to an isomorphism class of principal
$\U(1)$-bundles equipped with connection whose curvature is equal to
$\omega$, the 1-curvature of $[g, \theta]$.

This leads to the following theorem:
\begin{theorem}[\cite{Brylinski:1993}]
The group of isomorphism classes
of principal $\U(1)$-bundles with connection, on a manifold $M$, and
the degree one Deligne cohomology group $H^{1}(M,D_{1}^{\bullet})$ are
isomorphic.
\end{theorem}

\end{example}

\begin{example}[$\U(1)$-gerbes] \label{n=2}
The $n=2$ case will be particularly relevant for our work in the
subsequent chapters. 
A class $[g, A, B] \in H^{2}(M,D_{2}^{\bullet})$ is represented
by maps $g_{ijk} \maps U_{i} \cap U_{j} \cap U_{k} \to \U(1)$, 1-forms $A_{ij}
\in \Omega^{1}(U_{i} \cap U_{j})$, and 2-forms $B_{i} \in
\Omega^{2}(U_{i})$ satisfying the cocycle conditions:
\begin{equation} \label{2-cocycle}
\begin{split}
g_{jkl}g^{-1}_{ikl}g_{ijl}g^{-1}_{ijk}&=1 ~ \text{on $U_{i} \cap U_{j}
  \cap U_{k} \cap U_{l}$},\\
\i \cdot (A_{jk} - A_{ik} + A_{ij}) &= -g^{-1}_{ijk}dg_{ijk} ~ \text{on $U_{i} \cap U_{j} \cap U_{k}$},\\
B_{j} -B_{i} &= dA_{ij} ~ \text{on $U_{i} \cap U_{j}$}.
\end{split}
\end{equation}
The 2-curvature is the closed 3-form $\omega$ on $M$ satisfying
\[
\omega = dB_{i} \quad \text{on $U_{i}$}.
\]
We will see in Section \ref{2-connection_section} that
$\bigl[g, A, B \bigr]$ corresponds to an isomorphism class of a 
$\U(1)$-gerbe equipped with a 2-connection whose 2-curvature is $\omega$.
\end{example}

\subsection*{Integral differential forms}
In the remainder of this section, we determine which closed differential
forms can be realized as the $n$-curvature of a Deligne cocycle.  
Let $\Ri$ denote the sheaf whose sections are
locally-constant functions with values in $\i \cdot \R$. 
\begin{definition}\label{integral_def}
A closed differential form $\omega \in \Omega^{k}(M)$ is {\bf integral} iff
the class $\i \cdot [\omega]$ lies in the image of the composition
\begin{equation}\label{coeff}
H^{k}(M,\Zi) \to H^{k}(M,\Ri) \iso \i \cdot H^{k}_{\mathrm{dR}}(M).
\end{equation}
We denote by {\boldmath $Z^{k}(M)_{\mathrm{int}}$} the subspace of all closed
integral $k$-forms on $M$. 
\end{definition}
\noi Our goal is to show that the $n$-curvature of a Deligne $n$-cocycle is an
integral $(n+1)$-form, and conversely, every integral $(n+1)$-form is
the curvature of some Deligne $n$-cocycle. 

We begin by introducing
some necessary technical machinery.
Let $\Omega^{1\leq\bullet \leq k}$ denote the complex of sheaves
$\Omega^{1} \stackrel{d}{\to} \cdots \stackrel{d}{\to}
\Omega^{k}$ on a manifold $M$. Let $\R$ be the sheaf of locally constant
$\R$-valued functions. Let $\dim M=m$.
We consider the hyper-cohomology of the complex
$\cinf \stackrel{d}{\to} \Omega^{1 \leq \bullet  \leq m}$
via the double complex:
\[
\xymatrix{
\vdots& \vdots && \vdots\\
C^{2}(\cU,\cinf) \ar[u]^{\delta} \ar[r]^{d} &
C^{2}(\cU,\Omega^{1}) \ar[u]^{\delta} \ar[r]^-{d}&\cdots
\ar[r]^-{d} & \ar[u]^{\delta} C^{m}(\cU,\Omega^{m}) \\
C^{1}(\cU,\cinf) \ar[u]^{\delta} \ar[r]^{d} &
C^{1}(\cU,\Omega^{1}) \ar[u]^{\delta} \ar[r]^-{d}&\cdots
\ar[r]^-{d} & \ar[u]^{\delta} C^{1}(\cU,\Omega^{m}) \\
C^{0}(\cU,\cinf) \ar[u]^{\delta} \ar[r]^{d} &
C^{0}(\cU,\Omega^{1}) \ar[u]^{\delta} \ar[r]^-{d}&\cdots
\ar[r]^-{d} & \ar[u]^{\delta}C^{0}(\cU,\Omega^{m}),
}
\]
where $\cU=\{U_{i}\}$ is a good cover.
The total differential is:
\[
\begin{array}{c}
\d\theta^{k} = \delta \theta^{k} +  (-1)^{p}d \theta^{k}, \quad \theta^{k}
\in C^{p}(\cU,\Omega^{k}), ~ 0 \leq k \leq m.
\end{array}
\]
\noi Suppose $\omega$ is a closed $(n+1)$-form on $M$, with $n < m$. Let $p_{0} \maps
\coprod U_{i} \to M$ be the usual inclusion map. Then $p^{\ast}_{0}
\omega$ is in the group $C^{0}(\cU,\Omega^{n+1})$, and gives
a class
\[
[0,\ldots,p^{\ast}_{0} \omega,\ldots 0]
\in H^{n+1}\bigl(M,\cinf \stackrel{d}{\to}
\Omega^{1 \leq \bullet  \leq m} \bigr).
\]
We also consider the augmented complex
\[
\R \stackrel{\iota}{\to} \cinf \stackrel{d}{\to} \Omega^{1 \leq \bullet  \leq m}.
\]
If $r \in C^{n+1}(\cU,\R)$ represents a class
$[r] \in H^{n+1}(M,\R)$, then it also gives a class in the total cohomology
\[
[r,\ldots,0] 
\in H^{n+1}\bigl(M,\cinf \stackrel{d}{\to}
\Omega^{1 \leq \bullet  \leq m} \bigr).
\]
The following proposition essentially gives the well-known isomorphism:
$H^{\bullet}(M,\R) \cong H^{\bullet}_{\mathrm{dR}}(M)$, which was
implicitly used in Def.\ \ref{integral_def}.
\begin{prop}\label{Cech_deRham}
The $(n+1)$-cocycles $(r,\ldots,0)$ and $(0,\ldots,p^{\ast}_{0}
\omega,\ldots,0)$ are cohomologous if and only if there exists
differential forms $\theta^{k} \in C^{n-k}(\cU,\Omega^{k})$ for
$k=0,\ldots,n$ such that
\begin{align*}
d\theta^{n} &=p^{\ast}_{0} \omega \\
\delta \theta^{k} &= (-1)^{n-k} d\theta^{k-1} ~ \text{for $1 \leq k \leq n$},\\
\delta \theta^{0} & = (-1)^{n} r.
\end{align*}
\end{prop}
\begin{proof}
The conditions given for the differential forms $\theta^{k}$ are
equivalent to the statement
\[
(r,\ldots,0) + \d (\theta^{0},\ldots,\theta^{n})= 
(0,\ldots,p^{\ast}_{0} \omega,\ldots,0),
\]
where $\d$ is the total differential of the above double complex.
\end{proof}
\noi One can always find a unique class $[r] \in H^{n+1}(M,\R)$ such that
$[r,\ldots,0] =[0,\ldots,p^{\ast}_{0}\omega,\ldots,0]$. 
Moreover, $\omega$ is integral if and only if $\i \cdot [r] \in H^{n+1}(M,\Zi)$.

Let $Z^{k}$ denote the sheaf of closed 
$k$-forms. We will need the following lemma:
\begin{lemma} \label{tech_lemma_quasi_iso}
For $n \geq 1$, the complex $\sh{\U(1)} \xto{\dlog} \Omega^{1 \leq \bullet \leq
  n-1} \stackrel{d}{\to} Z^{n}$ is quasi-isomorphic to the constant
sheaf $\U(1)$.
\end{lemma}
\begin{proof}
We proceed via induction, starting with $\sh{\U(1)} \xto{\dlog} Z^{1}$.
Consider the short exact sequence of complexes of sheaves:
\[
\xymatrix{
\U(1) \ar[d] \ar[r]^{\mathrm{incl}} & \sh{\U(1)} \ar[d]^{\dlog} \ar[rr]^{\dlog}&&
Z^{1} \ar[d]^{\id}\\
0 \ar[r] & Z^{1} \ar[rr]^{\id}&&Z^{1}
}
\]
Since $H^{\bullet}(M,Z^{1} \stackrel{\id}{\to}Z^{1}) =0$, and
$H^{\bullet}(M,\U(1) \to 0) =H^{\bullet}(M,\U(1))$, the long exact
sequence in cohomology gives:
\[
H^{\bullet}(M,\sh{\U(1)} \xto{\dlog} Z^{1}) \cong
H^{\bullet}(M,\U(1)).
\]
Now assume $n >1$ and 
\[
H^{\bullet}(M,\sh{\U(1)} \xto{\dlog} \Omega^{1 \leq \bullet \leq
  n-1} \stackrel{d}{\to} Z^{n}) \cong H^{\bullet}(M,\U(1)).
\]
Again, we have a short exact sequence of complexes:
\[
\xymatrix{
\sh{\U(1)} \ar[d]^{\dlog} \ar[r]^{\id} &\sh{\U(1)}\ar[d]^{\dlog} \ar[r] & 0 \ar[d]\\
\Omega^{1} \ar[d]^{d}\ar[r]^{\id} &\Omega^{1} \ar[d]^{d}\ar[r] &  0 \ar[d]\\
\vdots \ar[d]^{d} & \vdots \ar[d]^{d} & \vdots \ar[d] & \\ 
\Omega^{n-1} \ar[d]^{d}\ar[r]^{\id} &\Omega^{n-1} \ar[d]^{d}\ar[r] &  0 \ar[d]\\
Z^{n} \ar[d]\ar[r]^{\mathrm{incl}} &\Omega^{n} \ar[d]^{d}\ar[r]^{d} &  Z^{n+1} \ar[d]^{\id}\\
0 \ar[r] & Z^{n+1} \ar[r]^{\id}&Z^{n+1}
}
\]
The long exact sequence in cohomology combined with the induction
hypothesis gives the desired result.
\end{proof}

We now prove:
\begin{prop} \label{curvature=integral}
The curvature $(n+1)$-form of a Deligne $n$-cocycle is integral.
\end{prop}
\begin{proof}
We consider the short exact sequence of complexes of sheaves
\begin{equation}\label{ses_1}
\xymatrix{
\sh{\U(1)} \ar[d]^{\dlog} \ar[r]^{\id} &\sh{\U(1)}\ar[d]^{\dlog} \ar[r] & 0 \ar[d]\\
\Omega^{1} \ar[d]^{d}\ar[r]^{\id} &\Omega^{1} \ar[d]^{d}\ar[r] &  0 \ar[d]\\
\vdots \ar[d]^{d} & \vdots \ar[d]^{d} & \vdots \ar[d] & \\ 
\Omega^{n-1} \ar[d]^{d}\ar[r]^{\id} &\Omega^{n-1} \ar[d]^{d}\ar[r] &  0 \ar[d]\\
Z^{n} \ar[r]^{\mathrm{incl}} &\Omega^{n} \ar[r]^{(-1)^{n}d} &  Z^{n+1}
}
\end{equation}
The complex on the left is $\sh{\U(1)} \stackrel{\dlog}{\to} \Omega^{1 \leq \bullet \leq
  n-1} \stackrel{d}{\to} Z^{n}$, while the middle complex is
$D^{\bullet}_{n}$. The complex on the right is the shifted complex
$Z^{n+1}[-n]$. Note that:
\begin{align*}
H^{n-1} \bigl(M,Z^{n+1}[-n] \bigr)&=  0\\
H^{n} \bigl(M,Z^{n+1}[-n] \bigr)&=H^{0}(M,Z^{n+1})=Z^{n+1}(M).
\end{align*}
This, in combination  with Lemma \ref{tech_lemma_quasi_iso}, implies we have a long exact sequence
\begin{equation} \label{long_exact_sequence}
0 \to H^{n}(M,\U(1)) \to H^{n}(M,D^{\bullet}_{n})
\stackrel{\kappa}{\to} Z^{n+1}(M) \stackrel{f}{\to} H^{n+1}(M,\U(1)),
\end{equation}
where $\kappa=(-1)^{n}d$ is the curvature map given in (\ref{kappa_map}), and
$f$ is the composition of the
connecting homomorphism
\[
Z^{n+1}(M) \stackrel{\partial}{\to} H^{n+1}(M,\sh{\U(1)} \stackrel{\dlog}{\to} \Omega^{1 \leq \bullet \leq
  n-1} \stackrel{d}{\to} Z^{n}),
\]
with the isomorphism given by Lemma
\ref{tech_lemma_quasi_iso}. 
The proposition is proven if we can show that $f(\omega) =0$ implies $\omega$ is integral.

We proceed by working through the definition of $\del$. 
Let $\cU=\{U_{i}\}$ be a good cover of $M$, and
take the \v{C}ech resolution of the complexes corresponding to the 3
columns in (\ref{ses_1}). Let $A^\b$, $B^\b$, and $K^\b$ be the total
complexes associated to the resolutions of the left, middle, and right columns,
respectively, of (\ref{ses_1}). In particular, we have
\begin{align*}
A^{n}&= C^{n}(\cU,\sh{\U(1)}) \oplus C^{n-1}(\cU,\Omega^1) \oplus
\cdots \oplus C^{1}(\cU,\Omega^{n-1}) \oplus C^{0}(\cU,Z^{n})\\
A^{n+1}&= C^{n+1}(\cU,\sh{\U(1)}) \oplus C^{n}(\cU,\Omega^1) \oplus
\cdots \oplus C^{2}(\cU,\Omega^{n-1}) \oplus C^{1}(\cU,Z^{n}),
\end{align*}

\begin{align*}
B^{n}&= C^{n}(\cU,\sh{\U(1)}) \oplus C^{n-1}(\cU,\Omega^1) \oplus
\cdots \oplus C^{1}(\cU,\Omega^{n-1}) \oplus C^{0}(\cU,\Omega^{n})\\
B^{n+1}&= C^{n+1}(\cU,\sh{\U(1)}) \oplus C^{n}(\cU,\Omega^1) \oplus
\cdots \oplus C^{2}(\cU,\Omega^{n-1}) \oplus C^{1}(\cU,\Omega^{n}),
\end{align*}
and
\[
K^{n} = C^{0}(\cU,Z^{n+1}), \quad K^{n+1} = C^{1}(\cU,Z^{n+1}).
\]

The connecting homomorphism is defined using the diagram
\[
\xymatrix{
\ar[d]_{\mathbf{d}} A^n \ar[r] & \ar[d]^{\mathbf{d}} B^n \ar[r]^{\kappa}& K^{n}
\ar[d]^{\delta}\\
A^{n+1} \ar[r] &  B^{n+1} \ar[r]^{\kappa}& K^{n+1}
}
\]
Given $\omega \in Z^{n+1}(M)$, 
we have $p^{\ast}_{0}( \omega)$ in the
group $K^{n}$, where $p_0 \maps \coprod U_{i} \to M$ is the inclusion.
We next find an $n$-chain in $B^n$ which maps to
$p^{\ast}_{0}( \omega)$, via the map $\kappa=(-1)^{n}d$.
Proposition \ref{Cech_deRham}, in combination with the isomorphism between
\v{C}ech and de Rham cohomology, implies there exists
$r \in C^{n+1}(\cU,\R)$ representing a class $[r] \in
H^{n+1}(M,\R)$ and differential forms $\theta^{0},\ldots,\theta^{n}$
with $\theta^{k} \in C^{n-k}(\cU,\Omega^{k})$ such that
\begin{align*}
d\theta^{n} &=(-1)^{n}p^{\ast}_{0}(\omega) \\
\delta \theta^{k} &= (-1)^{n-k} d\theta^{k-1} ~ \text{for $1 \leq k \leq n$},\\
\delta \theta^{0} & = -r.
\end{align*}
Setting $g=\exp(\i \cdot \theta^{0})$ gives the 
$n$-chain $(g,\theta^{1},\ldots,\theta^{n}) \in B^{n}$, which, by construction, 
is mapped to $p^{\ast}_{0}( \omega) \in K^n$ by $\kappa$.
We then apply the total differential $\d$ to $(g,\theta^{1},\ldots,\theta^{n})$.
The conditions on the forms $\theta^{k}$ imply
\[
\d (g,\theta^{1},\ldots,\theta^{n}) = (\delta g, 0 \ldots,0,\delta \theta^{n}).
\]
The quasi-isomorphism in Lemma \ref{tech_lemma_quasi_iso} sends the $(n+1)$-cocycle 
$(\delta g, 0 \ldots,0,\delta \theta^{n}) \in B^{n+1}$
to   
\[
\delta g = \exp \bigl( \i \cdot \delta \theta^{0} \bigr) = \exp \bigl( -\i\cdot r\bigr)\in C^{n+1}(\cU,U(1)). 
\]
Hence, we have determined $f$:
\[
f(\omega) = \bigl [\exp\bigl( -\i \cdot r\bigr) \bigr] \in H^{n+1}(M,\U(1)).
\]
Finally, recall that the short exact sequence  $0 \to \Zi \to \Ri
\stackrel{\exp}{\to} \U(1) \to 0$ gives the long exact sequence
\[
\cdots \to H^{n+1}(M,\Zi) \to H^{n+1}(M,\Ri) \to H^{n+1}(M,\U(1)) \to \cdots.
\]
Therefore, if $f(\omega) = 0$, then we have
$\i \cdot [r] \in H^{n+1}(M,\Zi)$, which implies $\omega$ is integral.
\end{proof}
\noi The converse statement is:
\begin{prop}\label{integral=curvature}
If $\omega$ is a closed integral $(n+1)$-form, then there exists a Deligne
$n$-cocycle whose $n$-curvature is $\omega$.
\end{prop}
\begin{proof}
The statement follows from the exactness of
(\ref{long_exact_sequence}) and the definition of the map 
$Z^{n+1}(M) \stackrel{f}{\to} H^{n+1}(M,\U(1))$ given in the proof
of the previous proposition.
\end{proof}

\section{2-Connections on $\U(1)$-gerbes}\label{2-connection_section}
Here we present Brylinski's formalism \cite[Sec.\ 5.3]{Brylinski:1993}
which describes how to equip a $\U(1)$-gerbe with a `2-connection',
and how such a structure is related to a Deligne 2-cocycle.
Recall that the set of connections on a
$\U(1)$-principal bundle over $M$ forms an affine space modeled on the
vector space $\i \cdot \Omega^{1}(M)$. We can think of
connections on $P$ as global sections of a sheaf, which we denote
as $\Co(P)$. Given an open set $U \ss M$, $\Co(P)(U)$ is the set of
connections on the restriction of the bundle $P$ to $U$. Since each
set  $\Co(P)(U)$ is equipped with a principal homogeneous
$\Omega^{1}(U)$-space, the sheaf $\Co(P)$ is a $\Omega^{1}$-torsor.

The above discussion implies that given an object $P \in
\Tor_{\U(1)}$, we can assign to it a $\Omega^{1}\vert_{U}$-torsor
$\Co(P)$. This sheaf satisfies some compatibility conditions that
correspond to familiar facts about connections on bundles:
\begin{itemize}
\item{Given an inclusion $V \stackrel{i}{\embed} U$, we have an
    equality of sheaves on $V$: $i^{\ast} \Co(P)=\Co(i^{\ast} P)$.
}
\item{Given an isomorphism of $\U(1)$-torsors $\phi \maps P_{1} \iso
    P_{2}$ on $U$, we have an obvious isomorphism of $\Omega^{1}
    \vert_{U}$-torsors $\phi_{\ast} \maps \Co(P_{1}) \iso
   \Co(P_{2})$.}

\item{If the isomorphism in (2) is an automorphism $g \maps P \iso P$
    corresponding to a section $g \in \sh{\U(1)}(U)$, then we have the
    ``gauge transformation'' 
\[
g_{\ast} (\conn) = \conn - g^{-1}dg, \quad  \forall \conn \in \Co(P).
\]
}
\end{itemize}
Any $\U(1)$-gerbe is locally isomorphic to $\Tor_{\U(1)}$, therefore
it makes sense to axiomatize the above construction for arbitrary
gerbes.
\begin{definition}[\cite{Brylinski:1993}]\label{conn_structure}
Let $\G$ be a $\U(1)$-gerbe over $M$. A {\bf connective structure} on
$\G$ is an assignment to every object $P \in \G(U)$ for every open set
$U \ss M$, a $\Omega^{1}\vert_{U}$-torsor $\Co(P)$ equipped with the
following data:
\begin{enumerate}
\item{For every inclusion $V \stackrel{i}{\embed} U$, an isomorphism
of $\Omega^{1} \vert_{V}$-torsors
\[
\alpha_{i} \maps i^{\ast} \Co(P) \iso \Co(i^{\ast}P),
\]
where $i^{\ast} \Co(P)$ is the pullback of $\Co(P)$ as an object in $\Shf(U)$,
such that for any composable pair $W \stackrel{j}{\embed} V \stackrel{i}{\embed} U$
the diagram
\[
\xymatrix{
 j^{\ast} i^{\ast}\Co(P) \ar @{=}[d] \ar[r]^{j^{\ast}\alpha_{i}} &
 j^{\ast}\Co(i^{\ast}P)\ar[r]^{\alpha_{j}} & \Co(j^{\ast}i^{\ast}P)
 \ar[d]^{t_{i,j \ast}} \\
 (ij)^{\ast}\Co(P)  \ar[rr]^{\alpha_{ij}} &&\Co((ij)^{\ast}P)
 }
 \]
commutes.
}
\item{For any isomorphisms $\phi \maps P_{1} \iso P_{2}$ and
$\psi \maps P_{2} \iso P_{3}$ in $\G(U)$, isomorphisms of
$\Omega^{1}\vert_{U}$-torsors
\[
\phi_{\ast} \maps \Co(P_{1}) \iso \Co(P_{2}),\quad \psi_{\ast} \maps \Co(P_{2}) \iso \Co(P_{3}),
\]
such that $(\psi \circ \phi)_{\ast}=\psi_{\ast} \circ \phi_{\ast}$ and
the diagram
\[
\xymatrix{
i^{\ast}\Co(P_{1}) \ar[d]^{\alpha_{1,i}}  \ar[r]^{i^{\ast} \phi_{\ast}} &i^{\ast}\Co(P_{2})
\ar[d]^{\alpha_{2,i}}\\
\Co(i^{\ast} P_{1}) \ar[r]^{(i^{\ast} \phi)_{\ast}} &\Co(i^{\ast} P_{2}).
}
\]
commutes. Moreover, if $\sh{\Aut}_{\G}(P)(U) \cong \sh{\U(1)}(U)$ and
$g \in \sh{\U(1)}(U)$, then $g_{\ast} \maps \Co(P) \iso \Co(P)$ is the
map
\[
\conn \mapsto \conn - g^{-1}dg.
\]
}
\end{enumerate}
\end{definition}

If $\Co(P)$ is the sheaf of connections on a principal $\U(1)$-bundle
$P \to M$,  then to each section $\conn \in \Co(P)$, we can assign a 2-form
$K(\conn)$ on $M$ corresponding to its curvature. This fact motivates
the next definition.

\begin{definition}[\cite{Brylinski:1993}] \label{curving}
Let $\G$ be a $\U(1)$-gerbe over $M$ equipped with a connective
structure $P \mapsto \Co(P)$. A {\bf curving} of the connective
structure is an assignment to every object $P \in \G(U)$, and every
section $\conn \in \Co(P)(U)$, for every open set $U \ss M$, a 2-form
$K(\conn) \in \Omega^{2}(U)$ with the following properties:
\begin{enumerate}
\item{Given an inclusion $V \stackrel{i}{\embed} U$ of open sets, and
    the associated isomorphism $\alpha_{i} \maps i^{\ast} \Co(P) \iso
    \Co(i^{\ast}P)$, the equality
    \[
    K(\alpha_{i}(i^{\ast} \conn)) = i^{\ast}K(\conn)
    \]
    holds, where $i^{\ast}K(\conn)$ is the usual pullback of differential
    forms.
}
\item{Given an isomorphism $\phi \maps P \iso P'$ in $\G(U)$ and the
    associated isomorphism $\phi_{\ast} \maps \Co(P) \to \Co(P')$, the
    equality
    \[
    K(\conn)=K(\phi_{\ast} (\conn))
    \]
    holds.
}
\item{If $\theta$ is a 1-form on $U$, then $K(\conn + \i \cdot \theta)=K(\conn)
    + d\theta$.
}
\end{enumerate}
\noi We say $\G$ is $\U(1)$-gerbe equipped with a {\bf 2-connection} iff it
is equipped with a connective structure and a curving.
\end{definition}

Finally, let us describe how 2-connections are related to Deligne
2-cocycles. Let $\G$ be a $\U(1)$-gerbe on $M$ equipped with a 2-connection.
As we described in Section \ref{gerbe_class}, we may choose a cover
$\{U_{i}\}$ such that there exists objects $P_{i} \in \G(U_{i})$,
isomorphisms $u_{ij} \maps P_{j} \vert_{U_{ij}}\iso P_{i} \vert_{U_{ij}}$ in $\G(U_{i} \cap
U_{j})$, and a 2-cocycle $g_{ijk} = u^{-1}_{ik}u_{ij}u_{jk} \in
\sh{\Aut}_{\G}(P_{k}) (U_{i} \cap U_{j} \cap U_{k}) \cong
\sh{\U(1)}(U_{i} \cap U_{j} \cap U_{k})$. We choose a
section $\conn_{i} \in \Co(P_{i})(U_{i})$ for each $i$. The
restriction of $\conn_{i}$ to $U_{i} \cap U_{j}$ gives a section of
$\Co(P_{i} \vert_{U_{ij}})$ by axiom 1 of Def.\ \ref{conn_structure},
which we will also denote as $\conn_{i}$. The isomorphisms $u_{ij}$
induce isomorphisms $u_{ij \ast} \maps \Co(P_{j} \vert_{U_{ij}}) \iso \Co(P_{i} \vert_{U_{ij}})$
of $\Omega^{1} \vert_{U_{ij}}$-torsors. Hence, $\conn_{i}$ and $u_{ij
  \ast}\conn_{j}$ are both sections of $\Co(P_{i}
\vert_{U_{ij}})$. This implies that there exists 1-forms $A_{ij}$ on
$U_{i} \cap U_{j}$ such that
\begin{equation} \label{Aij_conn}
\i \cdot A_{ij} = \conn_{i} - u_{ij \ast} \conn_{j}.
\end{equation}
Restricting the above equalities to $U_{i} \cap
U_{j} \cap U_{k}$ gives
\[
\i \cdot (A_{jk} - A_{ik} + A_{ij}) = \conn_{i} - (u^{-1}_{ik} u_{ij}u_{jk})_{\ast}
\conn_{i}.
\]
Axiom 2 of Def.\ \ref{conn_structure} implies that the right-hand side
of this equation is $g_{ijk}dg_{ijk}$. Hence,
\[
\i \cdot (A_{jk} - A_{ik} + A_{ij}) = g_{ijk}dg_{ijk}.
\]

The curving on $\G$ assigns a 2-form $B_{i}=K(\conn_{i})$ on each
$U_{i}$. On the intersections $U_{i} \cap U_{j}$, axiom 1 of Def.\
\ref{curving} implies that $K(\conn_{i})$ is just the restriction of
$B_{i}$. It follows from axiom 2 of the same definition that
\[
B_{j}=K(\conn_{j})=K(u_{ij \ast} \conn_{j}),
\]
and, by applying $K$ to Eq.\ \ref{Aij_conn}, we obtain
\[
B_{i}-B_{j} = dA_{ij}.
\]
By comparing these calculations with Eqs.\ \ref{2-cocycle} in  Example
\ref{n=2}, we see that we've obtained from $\G$ a Deligne 2-cocycle
$(g,-A, B)$ whose 2-curvature is given by the 3-form $\omega=dB_{i}$.
This leads to the following theorem.

\begin{theorem}[\cite{Brylinski:1993}]
There is a one-to-one correspondence between the set of equivalence classes
of $\U(1)$-gerbes with 2-connection on a manifold $M$ and
the degree two Deligne cohomology group
$H^{2}(M,D_{2}^{\bullet})$.
\end{theorem}

\section{2-Line stacks} \label{2-line_stack_section}
The category of principal $\U(1)$-bundles with connection over a manifold
is equivalent to the category of Hermitian line bundles with connection.
This equivalence sends a principal bundle to its associated line
bundle. The goal of this section is to construct an analogous
associated object to a $\U(1)$-gerbe with 2-connection. We call this
the `associated 2-line stack'. In the subsequent chapters on quantization, it will be
convenient to consider both the principal bundle/gerbe perspective and
the line bundle/2-line stack perspective.

\subsection*{Twisted vector bundles}
We begin by introducing the concept of twisting a Hermitian vector bundle by a
$\sh{\U(1)}$-valued \v{C}ech  2-cocycle. Hermitian vector bundles on a
manifold $M$ are equivalent to certain locally free sheaves with extra structure.
It is well-known that these vector bundles form a stack $\HVB$
over $M$, which inherits its structure as a fibered category from the
stack of sheaves $\Shf$. Let $\cU=\{U_{i}\}$ be a cover of $M$. Assume we have a
vector bundle $E_{i} \in \HVB(U_{i})$ for each $U_{i}$ and
isomorphisms of vectors bundles preserving the Hermitian structure
$\phi_{ij} \maps E_{j} \vert_{U_{i} \cap U_{j}}\iso E_{i} \vert_{U_{i}
  \cap U_{j}}$ such that the composition $\phi^{-1}_{ik}\circ
\phi_{ij}\circ \phi_{jk}$ 
is the identity automorphism of
the vector bundle $E_{k} \vert_{U_{i} \cap U_{j} \cap U_{k}} \in \HVB(U_{i} \cap U_{j} \cap
U_{k})$. Comparing this data with Def.\ \ref{stack_def}, we see
that we are giving an object $(E_{i}) \in \HVB(\cU^{[0]})$, and an isomorphism 
$(\phi_{i}) \maps p^{\ast}_{2} (E_{i})\iso p^{\ast}_{1} (E_{i})$, which
satisfies the necessary gluing conditions to give a global vector
bundle $E \to M$ in $\HVB(M)$. The restriction of $E$ to each $U_{i}$ is
isomorphic to the bundle $E_{i}$.

Now let $g \in C^{2}(\cU,\sh{\U(1)})$ be a 2-cocycle given by the functions $g_{ijk}
\maps U_{i} \cap U_{j} \cap U_{k} \to \U(1)$. If $E \in \HVB(U_{i} \cap U_{j} \cap
U_{k})$ is a Hermitian vector bundle, then $g$ induces an
automorphism of $E$ (preserving the Hermitian structure), which
corresponds to multiplying sections of $E$ by $g_{ijk}$.
We consider, as above, an object $(E_{i}) \in \HVB(\cU^{[0]})$, and an isomorphism 
$(\phi_{i}) \maps p^{\ast}_{2} (E_{i})\iso p^{\ast}_{1}
(E_{i})$. However, this time we require $\phi^{-1}_{ik}\circ
\phi_{ij}\circ \phi_{jk} =g_{ijk}$, instead of the identity. 
Unless $g_{ijk}$
is a co-boundary, this twisting prevents us from gluing the $E_{i}$'s
together to form a global Hermitian vector bundle. Hence, we have the
following definition:
\begin{definition}\label{twistedvb}
Let $\cU=\{U_{i}\}_{i \in \mathcal{I}}$ be an open cover of $M$ and $g \in C^{2}(\cU,\sh{\U(1)})$ a 2-cocycle.
A {\boldmath $g$}{\bf-twisted Hermitian vector bundle} over $M$ consists of the following data: 
\begin{itemize}
\item{on each $U_{i}$, a Hermitian vector bundle
\[
(E_{i}, \inp{\cdot}{\cdot}_{i}),
\]
}
\item{on each $U_{ij}= U_{i} \cap U_{j}$, an isomorphism of Hermitian vector bundles
\[
\phi_{i j} \maps E_{j}\vert_{U_{i j}} \iso E_{i} \vert_{U_{i j}},
\]
such that 
for all $i, j,k$  in $\mathcal{I}$:
\[
\phi^{-1}_{ik}\circ \phi_{ij}\circ \phi_{jk}= g_{i j k} \cdot
\]
where $g_{i j k} \cdot$ is the automorphism of $E_{k} \vert_{U_{ijk}}$ 
corresponding to multiplication by
\[
g_{ijk} \maps U_{i} \cap U_{j} \cap U_{k} \to \U(1).
\]
}
\end{itemize}
A {\bf morphism} $\psi \maps (E_{i},\phi_{i j}) \to
(E'_{i},\phi'_{i j})$ of $g$-twisted Hermitian vector bundles over $U$
consists of a collection of morphisms of Hermitian vector bundles
\[
\psi_{i} \maps E_{i}  \to E'_{i}, \quad 
\]
for each $i \in \mathcal{I}$ such that
\[
\psi_{i} \circ \phi_{i j} = \phi'_{i j} \circ \psi_{j}.
\]
\end{definition}
\noi Notice that the definition of a twisted vector bundle mimics the
construction we described in Sec.\ \ref{gerbe_class} for obtaining a gerbe from a 2-cocycle.

Let $\HVB^{g}(M)$ denote the category of $g$-twisted Hermitian vector
bundles over $M$. We first consider the case when $g$ is the trivial cocycle.
\begin{prop} \label{trivial_twist_prop}
If $g=1 \in C^{2}(\cU,\sh{\U(1)})$ is the trivial 2-cocycle, then
$\HVB^{g}(M)$ is equivalent to the category $\HVB(M)$.
\end{prop}
\begin{proof}
If $g$ is trivial, then the data which describes a twisted bundle is
the same data needed to glue local objects of a stack into a global object (Def.\
\ref{stack_def}). Hence, given a trivially twisted bundle
$(E_{i},\phi_{ij})$, there exists a global vector bundle $E$ whose
restriction to each $U_{i}$ is isomorphic to $E_{i}$. Indeed, the
category $\HVB^{g=1}(M)$ is a category of `descent data' for
the stack $\HVB$. (See
Appendix \ref{stack_appendix}.)
The fact that $\HVB$ is a stack implies $\HVB(M)$ is equivalent to this
category of descent data \cite{Moerdijk:2002}.
\end{proof}
\noi The next proposition implies that, up to
equivalence, $\HVB^{g}(M)$ depends only on the class $[g] \in H^{2}(\cU,\sh{\U(1)})$.   
\begin{prop} \label{twistedvb_equiv}
If $g$, $g' \in C^{2}(\cU,\sh{\U(1)})$ are cohomologous 2-cocycles, then the categories $\HVB^{g}(M)$ and
$\HVB^{g'}(M)$ are equivalent.
\end{prop}
\begin{proof}
Let $h \in C^{2}(\cU,\sh{\U(1)})$ be a 2-cochain such that $g = g' +
\delta h$. If $(E_{i},\phi_{ij})$ is an object of $\HVB^{g}(M)$, then
we can define Hermitian vector bundle automorphisms
\[
h_{ij} \maps E_i \vert_{U_{ij}}\iso E_i \vert_{U_{ij}}
\] 
over each open set $U_{ij}=U_{i} \cap U_{j}$ corresponding to
multiplying the sections of $E_i \vert_{U_{ij}}$
by $h_{ij} \maps U_{ij} \to \U(1)$. This gives new
isomorphisms
\[
\psi_{ij} = h_{ij} \circ \phi_{ij} \maps E_{j} \xto{\sim} E_{i}.
\]
Since the $\phi_{ij}$'s are $\C$-linear, the morphisms $\psi_{ij}$  satisfy on $U_{ijk}$:
\begin{align*}
\psi^{-1}_{ik}\circ \psi_{ij} \circ \psi_{jk} &= (h^{-1}_{ik} h_{ij} h_{jk}) g_{ijk}\\
&=g_{ijk} + \delta h\\
&= g^{\prime}_{ijk}.
\end{align*}
Hence, there is a functor from $\HVB^{g}(M)$ to $\HVB^{g'}(M)$, 
determined by the map $(E_{i},\phi_{ij}) \mapsto (E_{i},\psi_{ij})$
on objects, and the identity map on morphisms.
This functor gives the desired equivalence of categories.
\end{proof}

If $g \in C^{2}(\cU,\sh{\U(1)})$ and $g' \in C^{2}(\cU',\sh{\U(1)})$
are 2-cocycles related by a refinement, then one can show that the
categories $\HVB^{g}(M)$ and $\HVB^{g'}(M)$ are equivalent. (See, for
example, Lemma 1.2.3 in \cite{Caldararu:2000}.) Hence, up to
equivalence, we can uniquely associate the category
$\HVB^{g}(M)$ to the class $[g] \in H^{2}(M,\sh{\U(1)})$.

The next proposition implies that $g$-twisted Hermitian vector bundles are
the global sections of certain a stack which we think of as being associated to the
$\U(1)$-gerbe whose equivalence class is determined by $[g]$.
\begin{prop} \label{2-line_stack_prop}
Given a 2-cocycle $g \in C^{2}(\cU,\sh{\U(1)})$ on a manifold $M$,
there exists a stack over $M$ whose category of global
sections is equivalent to the category $\HVB^{g}(M)$
of $g$-twisted Hermitian vector bundles over $M$.
\end{prop}
\begin{proof}
The fact that twisted vector bundles or, more generally, twisted
coherent sheaves, form a stack is a known result in complex algebraic
geometry \cite{Antieau:2009}[Sec.\ 2.2], \cite{Schapira:2006}[Cor.\ 5.4.8].
 The idea of the proof is simple. We
construct the stack by gluing together the local stacks $\HVB
\vert_{U_{i}}$ of Hermitian vector bundles over $U_{i}$, using the
2-cocycle $g$. However, the proof requires us to introduce additional technology for stacks,
so we give the details in Appendix \ref{stack_appendix}.
 \end{proof}
\noi The stack described in Prop.\ \ref{2-line_stack_prop} is unique
up to equivalence of stacks. We slightly abuse notation and denote it
$\HVB^{g}$, so that we may identify the global sections with twisted
bundles in  $\HVB^{g}(M)$.

% \begin{prop}\label{twistedvb_stack_equiv}
% If $g$, $g' \in \sh{\U(1)}(\cU^{[2]})$ are cohomologous 2-cocycles, then the stacks $\HVB^{g}$ and
% $\HVB^{g'}$ are equivalent.
% \end{prop}
% \begin{proof}
% The equivalence of categories described in Prop.\
% \ref{twistedvb_equiv}  gives local equivalences
% $\phi_{V} \maps  \HVB^{g}(V) \to  \HVB^{g'}(V)$.
% \textit{double check this:}\\
% The natural isomorphisms $\alpha_{i}$ described in
% Def.\ \ref{equivalence_def} are just the the identity.
% \end{proof}

\subsection*{Twisted bundles as sections of a 2-bundle}
The sheaf of sections of a complex line bundle is constructed by
using the transition functions to glue together local smooth functions $U \to \C$.
There is a formalism known as `2-bundle theory' which categorifies
this idea \cite{BaezSchreiber:2005,Bartels:2004}. The total
space of a smooth 2-bundle over a manifold is, roughly, a category whose objects
and morphisms are themselves manifolds.\footnote{This is an example of
  what is called  a smooth `2-space' which is a slight generalization of the more familiar concept of a Lie groupoid.}
In this context, the complex line is replaced by $\mathsf{Vect}_{\C}$, the category of
finite-dimensional complex vector spaces. This category was
interpreted by Kapranov and Voevodsky \cite{Kapranov:1995} 
as a rank 1 `2-vector space'. A complex 2-line bundle is therefore a
2-bundle whose fibers are categories equivalent to
$\mathsf{Vect}_{\C}$.  A section of the 2-bundle is determined locally by a
particular kind of functor $U \to \mathsf{Vect}_{\C}$, where the open
set $U \ss M$ is given the structure of a trivial category. Roughly speaking,
such a functor assigns a vector space to each point in $U$ in a smooth
way, and hence determines a vector bundle over $U$. These local
sections can be glued together using 2-cocycles (cf.\ Def. \ref{twistedvb}), in analogy
with the line bundle case. Bartels' work \cite{Bartels:2004} implies that the
``sheaf of sections'' of a 2-bundle over $M$ is indeed a stack over $M$.
We will not use 2-bundle theory in this work. However, this rough
sketch provides the motivation for interpreting the stack $\HVB^{g}$
as the higher analog of a Hermitian line bundle.
\begin{definition}
Let $g \in C^{2}(\cU,\sh{\U(1)})$ be a 2-cocycle on $M$, and let
$\G$ be the corresponding $\U(1)$-gerbe whose equivalence class
is $[g] \in H^{2}(M,\sh{\U(1)})$. The
{\bf 2-line stack associated to} {\boldmath $\G$}
is the stack $\HVB^{g}$.
\end{definition}
\noi Note that Thm.\ \ref{gerbe_class_theorem}, Prop.\ \ref{twistedvb_equiv}, 
and Lemma 1.2.3 in \cite{Caldararu:2000}
imply that the 
2-line stack associated to a gerbe is unique up to equivalence.

\subsection*{2-Connections on 2-line stacks}
If we equip a $\U(1)$-gerbe with a 2-connection, 
then it is reasonable to expect that this extra structure can be
transferred to its associated 2-line stack. 
Hence, we next consider twisting a
Hermitian vector bundle, equipped with connection, by a Deligne
2-cocycle.
\begin{definition}\label{twistedvb_conn}
Let  $\cU=\{U_{i}\}_{i \in \mathcal{I}}$ be an open cover of $M$
and  $\xi=(g,A,B) \in C^{2}(\cU,\sh{\U(1)}) \oplus
C^{1}(\cU,\Omega^{1}) \oplus C^{0}(\cU,\Omega^{2})$
a Deligne 2-cocycle. A {\boldmath $\xi$}{\bf-twisted Hermitian vector bundle with connection} over
$M$ consists of the following data: 
\begin{itemize}
\item{on each $U_{i}$, a Hermitian vector bundle equipped
    with a Hermitian connection
\[
(E_{i}, \inp{\cdot}{\cdot}_{i},\conn_{i}),
\]
}
\item{on each $U_{ij}= U_{i} \cap U_{j}$, an isomorphism of Hermitian vector bundles
\[
\phi_{i j} \maps E_{j}\vert_{U_{i j}} \iso
E_{i} \vert_{U_{i j}},
\]
such that 
\[
\phi_{i j} \conn_{j} - \conn_{i} \phi_{ij}= \i \cdot A_{ij} \tensor \phi_{ij},
\]
and for all $i,j,k$  in $\mathcal{I}$:
\[
\phi^{-1}_{ik}\circ \phi_{ij}\circ \phi_{jk}= g_{i j k} \cdot
\]
where $g_{i j k} \cdot$ is the automorphism of $E_{k} \vert_{U_{ijk}}$ 
corresponding to multiplication by
\[
g_{ijk} \maps U_{i} \cap U_{j} \cap U_{k} \to \U(1).
\]
}
\end{itemize}
A {\bf morphism} $\psi \maps (E_{i}, \conn_{i},\phi_{i j})\to
(E'_{i},\conn^{\prime}_{i},\phi'_{i j})$ of $\xi$-twisted Hermitian vector bundles with
connection consists of a collection of connection-preserving morphisms of Hermitian vector bundles
\[
\psi_{i} \maps (E_{i},\conn_{i})  \to (E'_{i},\conn'_{i}) \quad 
\]
for each $i \in \mathcal{I}$ such that
\[
\psi_{i} \circ \phi_{i j} = \phi'_{i j} \circ \psi_{j}.
\]
\end{definition}
The above definition and the cocycle conditions on
$(g,A,B)$ force a compatibility between the curvatures $\conn_{i}^{2}$ of
the vector bundles $(E_{i}, \conn_{i})$. More precisely, for all
  $i,j \in \mathcal{I}$, we have
\begin{equation} \label{B_equation}
\phi_{ij} \circ (\conn^2_{j} - \i \cdot B_{j} \tensor \id)
=(\conn^2_{i} - \i \cdot B_{i} \tensor \id) \circ \phi_{ij}.
\end{equation}
\begin{definition}\label{twisted_flat_def}
We say a $\xi$-twisted Hermitian vector bundle with connection $(E_{i}, \conn_{i},\phi_{ij})$
is {\bf twisted-flat} iff for all $i \in \mathcal{I}$
\[
\conn^2_{i} - \i \cdot B_{i} \tensor \id=0.
\]
\end{definition}
\noi We interpret a twisted-flat section of a 2-line stack to be the 2-plectic
analogue of a covariant constant section of a Hermitian line
bundle. We will use this analogy in our quantization procedure for
2-plectic manifolds in Chapter \ref{quantization_chapter}.
Not surprisingly, bundles twisted by a Deligne
2-cocycle are global sections of a stack.
\begin{prop} \label{2-line_stack_conn_prop}
Given a Deligne 2-cocycle $\xi$
on a manifold $M$, there exists a stack $\HVB^{\xi}$ over $M$ whose 
category of global sections is equivalent to the category of $\xi$-twisted Hermitian vector
bundles with connections over $M$.
\end{prop}
\begin{proof}
The proof is essentially identical to the one given in Appendix
\ref{stack_appendix} for Proposition \ref{2-line_stack_prop}.
\end{proof}

The last definition of this section completes the analogy between
2-line stacks and line bundles.
\begin{definition} \label{2-line_stack_assoc}
Let $\xi$ be a Deligne 2-cocycle on $M$ and let
$\G$ be the corresponding $\U(1)$-gerbe with 2-connection whose equivalence class
is $[\xi] \in H^{2}(M,D^{\bullet}_{2})$. The
{\bf 2-line stack equipped with 2-connection} associated to $\G$
is the stack $\HVB^{\xi}$.
\end{definition}

\section{Holonomy of Deligne classes} \label{holonomy_sec}
Suppose we have a trivial principal $\U(1)$-bundle $P \to M$ equipped with
connection. The connection in this case is given by a 1-form $\theta$
on $M$. The holonomy of this connection is the function
\begin{equation} \label{triv}
S^{1} \stackrel{\gamma}{\to} M \mapsto \exp \bigl(  i \oint_{S^{1}}
\gamma^{\ast} \theta \bigr),
\end{equation}
from loops in $M$ to $\U(1)$.
If $\cU = \{U_{i}\}$ is a cover of $M$, and $p_{0} \maps
\cU^{[0]} \to M$ is the inclusion (\ref{fiber_prod}) then
$p_{0}^{\ast} \theta$ is a 1-form on $\cU^{[0]}= \coprod_{i} U_{i}$.
Therefore, the Deligne 1-cocycle corresponding to the bundle $P$ with connection
$\theta$ is $(1,p_{0}^{\ast}\theta)$. It is reasonable to define the
holonomy of this Deligne 1-cocycle to be the function given in
(\ref{triv}). Locally, every bundle with connection is isomorphic
to the trivial bundle equipped with a 1-form. Therefore, by
gluing the local functions (\ref{triv}) together, we can compute
the holonomy of any Deligne 1-cocycle, which would correspond
to the usual notion of the holonomy of a  bundle with
connection. Carey, Johnson, and Murray \cite{Carey:2004} give a
construction that does precisely this for both Deligne 1-cocycles and 
Deligne 2-cocycles. This allows one to define the `2-holonomy' of a
$\U(1)$-gerbe equipped with a 2-connection. We will use their construction
in our quantization of 2-plectic manifolds in Chapter \ref{quantization_chapter}.

The construction begins by first observing that if $\alpha$ is an $n$-form on $M$ and
$\cU=\{U_{i}\}$ is a good cover, then we can construct a Deligne $n$-cocycle
$(1,0,\ldots,0,p^{\ast}_{0}\alpha)$ by generalizing the $n=1$ case
described in the previous paragraph. We therefore have an
inclusion of groups
\[
\begin{array}{c}
\Omega^{n}(M) \stackrel{\iota}{\to} H^{n}(M,D^{\bullet}_{n}) \cong H^{n}(\cU,D^{\bullet}_{n})\\
\alpha \mapsto [1,0,\ldots,0,p^{\ast}_{0} ( \alpha)].
\end{array}
\]
We also have the sequence
\[
\Omega^{n}(M) \stackrel{\iota}{\to} H^{n}(M,D^{\bullet}_{n})
\stackrel{c}{\to} H^{n+1}(M,\Zi).
\]
\noi Here, $c$ is the map (\ref{chern}) which sends a Deligne class to its Chern class.
Clearly, the image of $\iota$ projects to the trivial class in
$H^{n}(M,\sh{\U(1)}) \cong H^{n+1}(M,\Zi)$. Therefore, $c \circ \iota=0$. 
Moreover, we have the following proposition, which is given
without proof in \cite{Carey:2004}. 
%we will need 
%in order to show the holonomy of a Deligne cocycle is well-defined.
\begin{prop}
Let $Z^{n}(M)_{\mathrm{int}}$ be the subspace of all closed
integral $n$-forms on a manifold $M$. 
The sequence of groups:
\begin{equation}\label{long_exact_seq}
0 \to Z^{n}(M)_{\mathrm{int}} \embed \Omega^{n}(M)  \stackrel{\iota}{\to} H^{n}(M,D^{\bullet}_{n})
\stackrel{c}{\to} H^{n+1}(M,\Zi) \to 0
\end{equation}
is exact.
\end{prop}
\begin{proof}
We have already discussed the surjectivity of the map $c$ in 
Sec.\ \ref{Deligne_sec}. To show $\ker c \ss \im \iota$,
suppose 
\[
(g,\theta^{1},\ldots,\theta^{n}) \in 
C^{n}(\cU,\sh{\U(1)}) \oplus C^{n-1}(\cU,\Omega^1) \oplus
\cdots \oplus C^{1}(\cU,\Omega^{n-1}) \oplus C^{0}(\cU,\Omega^{n})
\]
is a Deligne $n$-cocycle relative to a good
open cover $\cU$ such that $c([g,\theta^{1},\ldots,\theta^{n}])=0$.
Then, the isomorphism $H^{n}(M,\sh{\U(1)}) \cong H^{n+1}(M,\Zi)$ implies
there exists a cochain $h \in C^{n-1}(\cU,\sh{\U(1)})$ such that
$\delta h =g$. Since $\cU$ is good, a staircase construction in the double complex
(\ref{cech_resolution}) shows there exists $k$-forms 
\[
\eta^{k} \in C^{n-k-1}(\cU,\Omega^{k}), \quad 1 \leq k \leq n-1
\]
such that
\begin{equation} \label{proofeq1}
\begin{split}
\theta^{1} &= \delta \eta^{1} + (-1)^{n-1}\frac{1}{\i} d \log h, \\ 
\theta^{k} & = \delta \eta^{k} + (-1)^{n-k} d \eta^{k-1}, \quad 2 \leq
k \leq n-1.
\end{split}
\end{equation}
In particular, for $k=n-1$, we have $\theta^{n-1} = \delta \eta^{n-1} - d \eta^{n-2}$,
and hence
\[
d \theta^{n-1} = \delta d\eta^{n-1}.
\]
The fact that $(g,\theta^{1},\ldots,\theta^{n})$ is a cocycle implies
\[
\delta \theta^{n} - d \theta^{n-1} =0.
\]
Combining the two equalities gives $\delta( \theta^{n} -
d\eta^{n-1})=0$. Hence $\theta^{n} -d\eta^{n-1}$ is a cocycle in
$C^{0}(\cU,\Omega^{n})$. Therefore there exists a global $n$-form
$\alpha \in \Omega^{n}(M)$ such that $p^{\ast}_{0}(\alpha) = \theta^{n} -d\eta^{n-1}$.
This result, combined with the Eqs.\ \ref{proofeq1} imply
\[
(g,\theta^{1},\ldots,\theta^{n}) -
\mathbf{d}(h,\eta^{1},\ldots,\eta^{n-1}) =
(1,0,\ldots,0,p^{\ast}_{0}\alpha),
\]
where $\mathbf{d}$ is the total differential of the double complex
(\ref{cech_resolution}). Hence $\ker c = \im \iota$.

Next we show $Z^{n}(M)_{\mathrm{int}} \ss \ker \iota$. Suppose
$\alpha$ is a closed integral $n$-form. Then Prop.\
\ref{integral=curvature} implies there exists a Deligne
$(n-1)$-cocycle $(h,\eta^{1},\ldots,\eta^{n-1})$ representing a class
in $H^{n-1}(M,D^{\bullet}_{n-1})$ whose $(n-1)$-curvature is
$\alpha$. By definition of the curvature, this means
\[
p^{\ast}_{0} \alpha = (-1)^{n-1}d\eta^{n-1}.
\]
Embeding this cocycle in the complex $D^{\bullet}_{n}$ and applying
the total differential gives:
\[
\mathbf{d}(h,\eta^{1},\ldots,\eta^{n-1}) = (1,0,\ldots,0,d \eta^{n-1})
= (1,0,\ldots,0,(-1)^{n-1}p^{\ast}_{0} \alpha).
\]
Hence
\[
(1,0,\ldots,0,p^{\ast}_{0} \alpha)- (-1)^{n-1} \mathbf{d}(h,\eta^{1},\ldots,\eta^{n-1}) = (1,0,\ldots,0),
\]
which implies $\iota(\alpha) =0$.

Finally, we show $\ker \iota \ss Z^{n}(M)_{\mathrm{int}}$. Let $\alpha$
be a $n$-form on $M$  such that
\[
\iota(\alpha)=[1,0,\ldots,p^{\ast}_{0} \alpha] = [1,0,\ldots,0].
\]
Hence the curvature of the cocycle
$(1,0,\ldots,p^{\ast}_{0}\alpha)$ 
is zero. By definition of the curvature, this implies $p^{\ast}_{0}
d \alpha=0$. Therefore $\alpha$ is closed. Furthmore, by assumption,
there exists a cochain
\[
(h,\eta^{1},\ldots,\eta^{n-1}) \in 
C^{n-1}(\cU,\sh{\U(1)}) \oplus C^{n-2}(\cU,\Omega^1) \oplus
\cdots \oplus C^{0}(\cU,\Omega^{n-1}) 
\]
such that
\[
(1,0,\ldots,p^{\ast}_{0}\alpha)-\mathbf{d}(h,\eta^{1},\ldots,\eta^{n-1}) = (1,0,\ldots,0).
\]
By definition of the differential $\mathbf{d}$, this implies
\[
p^{\ast}_{0} \alpha = d\eta^{n-1},
\]
and:
\begin{align*}
\delta h &=1, \\
\delta \eta^{1} &= (-1)^{n} \frac{1}{\i}d \log h,\\
\delta \eta^{k} &= (-1)^{n-k-1}d \eta^{k-1}, ~ \text{for $2 \leq k \leq n-1$.}
\end{align*}
Hence, $(h,\eta^{1},\ldots,\eta^{n-1})$ is a Deligne $(n-1)$-cocycle
representing a class in $H^{n-1}(M,D^{\bullet}_{n-1})$ whose
$(n-1)$-curvature is $(-1)^{n-1}\alpha$. Therefore Prop.\ \ref{curvature=integral} implies
that $\alpha$ is integral.
\end{proof}

Let $[g,\theta^{1},\ldots,\theta^{n}] \in H^{n}(M,D^{\bullet}_{n})$ be a degree
$n$ class relative to an open cover $\cU=\{U_{i}\}$ of $M$.
Let $\sigma \maps \Sigma^{n} \to M$ be a map from a compact,
oriented $n$-dimensional manifold into $M$. It is easy to see that the pullback
$[\sigma^{\ast} g,\sigma^{\ast}\theta^{1},\ldots,\sigma^{\ast}\theta^{n}]$ is a degree $n$
class in $H^{n}(\sigma^{-1}\cU,D^{\bullet}_{n})$ relative to the open
cover $\sigma^{-1}\cU=\{\sigma^{-1}(U_{i})\}$ of $\Sigma^{n}$.
Since $H^{n+1}(\Sigma^{n},\Zi) \cong H^{n+1}(\Sigma^{n},\Z)=0$, the
sequence (\ref{long_exact_seq}) implies there exists an $n$-form
$\alpha$ on $\Sigma^{n}$ such that
\begin{equation}\label{alpha_def}
\iota(\alpha)=[\sigma^{\ast} g,\sigma^{\ast}\theta^{1},\ldots,\sigma^{\ast}\theta^{n}].
\end{equation}
Hence, we can integrate $\alpha$ and take the exponential
\begin{equation} \label{element}
\exp \bigl (i \int_{\Sigma^{n}} \alpha \bigr)
\end{equation}
to obtain an element of $\U(1)$. Note that if $\alpha'$ is any other
$n$-form satisfying $\iota(\alpha')=[\sigma^{\ast}
g, \sigma^{\ast}\theta^{1},\ldots,\sigma^{\ast}\theta^{n}]$, then the sequence
(\ref{long_exact_seq}) implies $\alpha-\alpha'$ is integral, which
further implies
\[
\int_{\Sigma^{n}} (\alpha-\alpha') \in \tpi \Z.
\]
Therefore the element (\ref{element}) only depends on the class
$[\sigma^{\ast} g,
\sigma^{\ast}\theta^{1},\ldots,\sigma^{\ast}\theta^{n}]$, which allows
us to give the following definition:
\begin{definition}[\cite{Carey:2004}] \label{holonomy_def}
Let $[g,\theta^{1},\ldots,\theta^{n}] \in H^{n}(M,D^{\bullet}_{n})$ be a
degree $n$ Deligne class. The {\bf $n$-holonomy} of a
map $\sigma \maps \Sigma^{n} \to M$ is the element
\[
\hol([g,\theta^{1},\ldots,\theta^{n}],\sigma):= \exp \bigl(i \int_{\Sigma^{n}} \alpha\bigr)
\]
of $\U(1)$, where $\alpha \in \Omega^{n}(\Sigma^{n})$ is the $n$-form defined  in
Eq.\ \ref{alpha_def}.
\end{definition}
\noi It is straightforward to verify for $n=1$, that
$\hol([(g,\theta],\sigma)$ is the usual holonomy of a principal
$\U(1)$-bundle with transition functions and connection 1-forms 
representing the class $[g,\theta]$. Similarly, for
gerbes we have:
\begin{definition} \label{2-hol_2-conn_def}
The {\bf 2-holonomy of a 2-connection} on a $\U(1)$-gerbe corresponding
to the Deligne 2-cocycle $(g,A,B)$ is the assignment to every map
$\sigma: \Sigma^{n} \to M$, the element 
\[
\hol([g,A,B],\sigma) \in\U(1).
\]
\end{definition}
\noi Since the 2-holonomy of the gerbe depends only on the Deligne class, we
can just as easily define the 2-holonomy for the associated 2-line
stack with 2-connection.

\chapter{Prequantization of $2$-plectic manifolds} \label{prequant_chapter}
In Chapter \ref{algebra_chapter}, we showed that any $n$-plectic
manifold gives rise to a Lie $n$-algebra. This generalizes the
well-known fact that the functions on a symplectic manifold $(M,\omega)$
form a Poisson algebra. In the symplectic case,
the geometric quantization procedure of
Kirillov \cite{Kirillov:2004}, Kostant \cite{Kostant:1970}, 
and Souriau \cite{Souriau:1967} (KKS) involves 
constructing faithful representations of this algebra using 
structures that naturally arise on $M$. 
The first step of this procedure is called
prequantization. Our goal in this chapter is to generalize this to
$2$-plectic manifolds, and prequantize the Lie 2-algebra of Hamiltonian 1-forms. 

\section{Overview of prequantization}
In symplectic geometry, prequantization itself begins by assigning 
to a symplectic manifold either a principal $\U(1)$-bundle, or a
Hermitian line bundle, with
connection whose curvature corresponds to the symplectic 2-form.
In this chapter, we will use principal bundles. 

\begin{definition}[\cite{Souriau:1967}] \label{KKS_prequant_def}
A {\bf prequantized symplectic manifold} is a symplectic manifold $(M,\omega)$
equipped with a principal $\U(1)$-bundle $P \to M$ with
connection, such that the curvature of the connection is
$\omega$.
\end{definition}
\noi Definition \ref{n-cocycle_def} and Example \ref{n=1} in the
previous chapter imply that a prequantized symplectic manifold is a
symplectic manifold equipped with a Deligne 1-cocycle whose 1-curvature
is $\omega$. This observation allows us to generalize Def.\
\ref{KKS_prequant_def} to the $n$-plectic case.
\begin{definition}\label{n-plectic_prequant_def}
A {\bf prequantized} {\boldmath$n$}{\bf -plectic manifold} is an $n$-plectic manifold
$(M,\omega)$ equipped with a Deligne $n$-cocycle $\xi$ whose $n$-curvature
is $\omega$.
\end{definition}
\noi Not every $n$-plectic manifold can be prequantized.
Indeed, Propositions \ref{curvature=integral} and \ref{integral=curvature} imply:
\begin{proposition}
An $n$-plectic manifold $(M,\omega)$ is prequantizable if and only if
$\omega$ is integral.
\end{proposition}
\noi In Prop. \ref{curvature=integral}, we considered the
long exact sequence
\[
0 \to H^{n}(M,\U(1)) \to H^{n}(M,D^{\bullet}_{n})
\stackrel{\kappa}{\to} Z^{n+1}(M) \stackrel{f}{\to} H^{n+1}(M,\U(1)),
\]
where $Z^{n+1}(M)$ is the space of closed  $(n+1)$-forms, and $\kappa$
is the curvature map. The sequence shows that the manifold may have several 
non-equivalent prequantizations. Indeed, the prequantizations of
$(M,\omega)$ are classified by the Deligne cohomology group $H^{n}(M,D^{\bullet}_{n})$.

Let $(M,\omega,\xi)$ be a prequantized symplectic manifold and let $P \xrightarrow{\pi} M$
be the $\U(1)$-bundle with connection corresponding to
the Deligne 1-cocycle $\xi$. From this geometric data, the
KKS procedure for prequantization gives a faithful representation of the Poisson algebra 
$(\cinf(M),\brac{\cdot}{\cdot})$ as unitary operators on a Hilbert
space. This representation can be constructed by using the `Atiyah
algebroid' associated to $P$. The Atiyah algebroid is an example
of a Lie algebroid: roughly, a vector bundle $A \to M$ equipped with a
bundle map to the tangent bundle of $M$, and a Lie algebra structure
on its space of global sections.  The total space of the Atiyah
algebroid is the quotient $A=TP/U(1)$. Sections of $A$ are
$U(1)$-invariant vector fields on $P$. A connection on $P$ is
equivalent to a splitting $s \maps TM \to A$ of the short exact sequence
\[
0 \to \R \times M \to A\stackrel{\pi_{\ast}}{\to} TM \to 0
\]
where the map $\R \times M \to A$ corresponds to identifying the
vertical subspace of $T_pP$ with the Lie algebra $\u(1) \cong
\R$. As we will see, those sections of $A$ which act as infinitesimal symmetries preserving the connection (or
splitting) form a Lie subalgebra that is isomorphic to the Poisson
algebra.  This implies that the Poisson algebra acts as linear differential operators
on the $\C$-valued functions on $P$. In particular,
the algebra acts on functions $f \maps P \to \C$ with the property
\[
f(pg)=g^{-1}f(p), \quad g \in \U(1).
\]
A simple calculation shows that such functions correspond to global
sections of the Hermitian line bundle associated to $P$.
Compactly supported global sections of this line bundle form a
vector space equipped with a Hermitian inner product.
The $L^{2}$-completion of this space is called the `prequantum
Hilbert space.' We shall consider this Hilbert space in more detail in
the next chapter.

If the symplectic manifold is connected, then the Poisson algebra gives what is known as
the `Kostant-Souriau central extension' of the Lie algebra of
Hamiltonian vector fields \cite{Kostant:1970}. The symplectic form,
evaluated at a point, gives a representative of the degree 2 class 
corresponding to this extension in the Lie algebra cohomology of the
Hamiltonian vector fields. 
The fact that this central extension is quantized, rather
than the Hamiltonian vector fields themselves, is the reason why the
concept of `phase' is introduced in quantum mechanics.

The goal of this chapter is to generalize the above prequantization procedure to 2-plectic
manifolds. We already know from Chapter \ref{stacks_chapter} that, for
a prequantized 2-plectic manifold, a $\U(1)$-gerbe with 2-connection
plays the role of the $\U(1)$-principal bundle. But what is the
2-plectic analogue of the Atiyah algebroid? We answer this question in
this chapter by considering a more general problem: understanding the relationship
between 2-plectic geometry and the theory of `Courant
algebroids.' Roughly, a Courant algebroid is a vector bundle that
generalizes the structure of a Lie algebroid equipped with a symmetric
nondegenerate bilinear form on the fibers. They were first used by
Courant \cite{Courant} to study generalizations of pre-symplectic and
Poisson structures in the theory of constrained mechanical
systems. Curiously, many of the ingredients found in 2-plectic
geometry are also found in the theory of `exact' Courant
algebroids. An exact Courant algebroid is a Courant algebroid whose
underlying vector bundle $C \to M$ is an extension of the tangent
bundle by the cotangent bundle:
\[
0 \to T^{\ast}M \to C \to TM \to 0.
\] 
In a letter to Weinstein, \v{S}evera \cite{Severa1} described how
exact Courant algebroids arise in 2-dimensional variational problems
(e.g.\ bosonic string theory), and showed that they are classified up
to isomorphism by the degree 3 de Rham cohomology of $M$. From any
closed 3-form on $M$, one can explicitly construct an exact Courant
algebroid equipped with an `isotropic' splitting of the above short
exact sequence, using local 1-forms and 2-forms that satisfy cocycle
conditions \cite{Bressler-Chervov,Hitchin:2004ut,Gualtieri:2007}.

\v{S}evera's classification implies that every 2-plectic
manifold $(M,\omega)$ gives a unique exact Courant algebroid (up to
isomorphism) whose class is represented by the 2-plectic
structure. However, there are more interesting similarities between
2-plectic structures and exact Courant algebroids.  Roytenberg and
Weinstein \cite{Roytenberg-Weinstein} showed that the bracket on the
space of global sections of a Courant algebroid induces an
$L_{\infty}$ structure. If we are considering an exact Courant
algebroid, then the global sections can be identified with vector
fields and 1-forms on the base space. Roytenberg and Weinstein's
results imply that these sections, when combined with the smooth
functions on the base space, form a Lie 2-algebra
\cite{Roytenberg_L2A}. Moreover, the Jacobiator of the Lie 2-algebra
encode a closed 3-form representing the \v{S}evera class
\cite{Severa-Weinstein}.

The first new result we present in this chapter is that there exists a Lie
2-algebra morphism which embeds the Lie 2-algebra of Hamiltonian
1-forms on a 2-plectic manifold $(M,\omega)$ into the Lie 2-algebra of
global sections of the corresponding exact Courant algebroid
$C$ equipped with an isotropic splitting.  Moreover, this
morphism gives an isomorphism between the Lie 2-algebra of Hamiltonian
1-forms and the sub Lie 2-algebra consisting of those sections of
$C$ which preserve the splitting via a particular kind of
adjoint action. This result holds without any integrality condition on
the 2-plectic structure. However, its meaning becomes clear in the
context of prequantization: It is the higher analogue of the isomorphism
between the underlying Lie algebra of the Poisson algebra on a
prequantized symplectic manifold, and the Lie sub-algebra of sections
of the Atiyah algebroid that preserve the connection on the associated
principal bundle. Hence, we see that the 2-plectic analogue of the
Atiyah algebroid associated to a principal $U(1)$-bundle is an exact
Courant algebroid associated to a $U(1)$-gerbe. This idea that exact
Courant algebroids are higher Atiyah algebroids has been discussed
previously in the literature
\cite{Bressler-Chervov,Gualtieri:2007}. However, this is the first
time the analogy has been understood using Lie $n$-algebras within the
context of prequantization.

The second result presented here involves identifying the 2-plectic
analogue of the Kostant-Souriau central extension. On a 2-plectic
manifold, associated to every Hamiltonian 1-form is a Hamiltonian
vector field. These vector fields form a Lie algebra, which we can view
as a trivial Lie 2-algebra, whose underlying chain complex is
concentrated in degree 0, and whose bracket satisfies the Jacobi
identity on the nose. For any 1-connected (i.e.\ connected and simply
connected) 2-plectic manifold, we show that the Lie 2-algebra of
Hamiltonian 1-forms is quasi-isomorphic to a `strict central
extension' of the trivial Lie 2-algebra of Hamiltonian vector fields
by the abelian Lie 2-algebra $\R \to 0$. 
%This abelian Lie 2-algebra is
%known as $b\u(1)$. 
Furthermore, we show that this extension corresponds to
a degree 3 class in the Lie algebra cohomology of the Hamiltonian
vector fields with values in the trivial representation. In analogy
with the symplectic case, a 3-cocycle representing this class can be
constructed by using the 2-plectic form. It follows from the
aforementioned results relating a 2-plectic manifold $(M,\omega)$ to
the Courant algebroid $C$, that the sub Lie 2-algebra of sections of
$C$ that preserve the splitting is also quasi-isomorphic to this
central extension, and can be interpreted as the prequantization of the
Lie 2-algebra of Hamiltonian 1-forms. 
%Phases originate from the
%presence of $b\u(1)$, which integrates to an important Lie 2-group
%called $BU(1)$.

\section{Prequantization of symplectic manifolds} \label{symplectic_sec} 
%\paragraph{Lie algebroids from closed 2-forms} 
%\label{symp_algebroid}
In this section, we briefly review the construction of 
Lie algebroids on symplectic manifolds and describe an embedding of
the Poisson algebra into the Lie algebra of sections of the
algebroid. We emphasize the role played by the Atiyah algebroid in prequantization and the
construction of the Kostant-Souriau central extension. 

We begin by reviewing the construction of a Lie algebroid, which
ultimately will describe how phases arise in the prequantization of
symplectic manifolds. A section of this Lie algebroid is a vector
field on the base manifold together with a `phase', or more precisely,
a real-valued function.
\begin{definition}[\cite{Mackenzie:1987}] 
A {\bf Lie algebroid} over a manifold $M$ is a real
vector bundle $A \to M$ equipped with a bundle map (called the anchor) $\rho \maps A \to
TM$, and a Lie algebra bracket $\blankbrac_{A} \maps \Gamma(A) \tensor
\Gamma(A) \to \Gamma(A)$ such that the induced map
\[
\Gamma(\rho) \maps \Gamma(A) \to \X(M)
\]
is a morphism of Lie algebras, and for all $f \in \cinf(M)$ and
$e_{1},e_{2} \in \Gamma(A)$ we have
the Leibniz rule
\[
\abrac{e_{1}}{f e_{2}} = f \abrac{e_{1}}{e_{2}} + \rho(e_{1})(f)e_{2}. 
\]
A Lie algebroid with surjective anchor map is called a
{\bf transitive Lie algebroid}. 
\end{definition}
The main ideas of the following construction are presented in Sec.\ 17
of Cannas da Silva and Weinstein \cite{CdS-Weinstein:1999}. We provide
the details here in order to compare to the 2-plectic case in Sec.\
\ref{geometric}. Let $(M,\omega)$ be a manifold equipped with a closed
2-form, e.g.\ a pre-symplectic manifold. By a \textbf{trivialization}
of $\omega$, we mean a cover $\{U_{i}\}$ of $M$, equipped with 1-forms
$\theta_{i} \in \Omega^{1}(U_{i})$, and smooth functions $g_{ij} \in
\cinf(U_{i} \cap U_{j})$, such that
\begin{align} 
 \omega \vert_{U_{i}} = d\theta_{i} \\
(\theta_{j}-\theta_{i}) \vert_{U_{ij}} = dg_{ij}, \label{cocycle1}
\end{align}
where $U_{ij}=U_{i} \cap U_{j}$.
Every manifold admits a good cover, hence every closed 2-form admits a trivialization.   
Given such a trivialization of $\omega$, we can construct a transitive Lie
algebroid over $M$. Over each $U_{i}$ we consider the Lie algebroid
\[
A_{i}=TU_{i} \oplus \R \to U_{i},
\]
with bracket 
\[
\aibrac{v_{1} + f_{1}}{v_{2} + f_{2}} = [v_{1},v_{2}] + v_{1}(f_{2}) -
v_{2}(f_{1}), 
\]
for all  $v_{i} + f_{i} \in \X(U_{i}) \oplus \cinf(U_{i})$, and anchor $\rho$ given
by the projection onto $TU_{i}$.
From the 1-forms $dg_{ij} \in \Omega^{1}(U_{ij})$, we can construct
transition functions
\begin{align*}
G_{ij} \maps U_{ij} \to GL(n+1),\\
G_{ij}(x)= 
\begin{pmatrix}
1 & 0\\
dg_{ij} \vert_{x} & 1
\end{pmatrix},
\end{align*}
which act on a point $v_{x} + r \in A_{i} \vert_{U_{ij}}$ by 
\[
G_{ij}(x)(v_{x} + r) = v_{x} + r + dg_{ij}(v_{x}).
\]
Clearly, each $G_{ij}$ satisfies the cocycle conditions on $U_{ijk}$ by virtue of Eq.\ \ref{cocycle1}.
Therefore, we have over $M$ the vector bundle
\[
A= \coprod_{x \in M} T_{x}U_{i} \oplus \R / \sim,
\]
where the equivalence is defined via the functions $G_{ij}$ in
the usual way. For any sections $v_{i} + f_{i}$ of $A_{i}
\vert_{U_{ij}}$, a direct calculation shows that
\[
\aibrac{G_{ij}(v_{1} + f_{1})}{G_{ij}(v_{2} + f_{2})} = G_{ij}([v_{1},v_{2}] + v_{1}(f_{2}) -
v_{2}(f_{1})).
\]
Hence, the local bracket descends to a well-defined bracket
$\tabrac{\cdot}{\cdot}$ on the quotient. Henceforth, 
$(A,\tabrac{\cdot}{\cdot},\rho)$ will denote this transitive Lie algebroid
associated to the closed 2-form $\omega$.

It is easy to see that the above Lie algebroid is an extension of the
tangent bundle
\[
0 \to M \times \R \to A \stackrel{\rho}{\to} TM \to 0.
\]
Moreover, the 1-forms $\theta_{i} \in \Omega^{1}(U_{i})$ induce a splitting
\[
s \maps TM \to A
\]
of the above sequence defined as
\begin{equation} \label{A_splitting}
s(v_{x})= v_{x} - \theta_{i}(v_{x}), \quad \forall ~ v_{x} \in TU_{i}.
\end{equation}
By a slight abuse of notation, we denote the horizontal lift
$\Gamma(s) \maps \X(M) \to \Gamma(A)$ also by $s$. Hence
every section $e \in \Gamma(A)$ is of the form $e=s(v) +f$,
for some $v\in \X(M)$ and $f \in \cinf(M)$. Using the local definition
of the splitting and the fact that $\omega \vert_{U_{i}}=d\theta_{i}$,
a direct calculation shows that
\begin{equation} \label{tabrac_def}
\tabrac{s(v_{1}) + f_{1}}{s(v_{2}) +
  f_{2}} = s \bigl ([v_{1},v_{2}] \bigr) +
v_{1}(f_{2})- v_{2}(f_{1}) 
-\ip{2}\ip{1}\omega,
\end{equation}
for all sections $s(v_{i}) + f_{i}$. The failure of
the splitting $s \maps TM \to A$ to preserve the Lie bracket on
sections is measured by the 2-form $\omega$:
\[
[s(v_{1}),s(v_{2})]_{A} = s([v_{1},v_{2}]) - \omega(v_{1},v_{2}),
\quad \forall v_{1},v_{2} \in \X(M).
\]

It is a simple exercise to show that a different choice of
trivialization gives a Lie algebroid equipped with a splitting that
is isomorphic to $A$ equipped with the splitting given in Eq.\ \ref{A_splitting}.

\subsection*{The Poisson algebra}
Let $(M,\omega)$ be a symplectic
manifold. Here $\brac{f}{g}=\omega(v_{f},v_{g})$ denotes the Poisson
bracket on smooth functions. The vector field $v_{f}$, satisfying 
the equality $df=-\ip{f}\omega$, is the unique Hamiltonian vector
field corresponding to the function $f$.  We denote the Lie algebra of
Hamiltonian vector fields by $\Xham$.  Let
$(A,\tabrac{\cdot}{\cdot},\rho)$ be the Lie algebroid
associated to $\omega$ and $s \maps TM \to A$ be the
splitting defined in Eq.\ \ref{A_splitting}. We are interested in a particular Lie
sub-algebra of $\Gamma(A)$  acting on the subspace
$s(\X(M)) \subseteq \Gamma(A) $ via the adjoint action.
\begin{definition} \label{preserve_split_def}
A section $a=s(v) + f \in \Gamma(A)$ {\bf preserves the splitting} $s \maps TM
\to A$ iff $\forall v' \in \X(M)$ 
\[
\bigl [a,s(v') \bigr]_{A} = s([v,v']).
\]
The subspace of sections that preserve the splitting is 
denoted as $\Gamma(A)^{s}$.
\end{definition}

\begin{prop} \label{preserve_lie_alg}
$\Gamma(A)^{s}$ is a Lie subalgebra of $\Gamma(A)$.
\end{prop}
\begin{proof}
Follows directly from the fact that the bracket on $\Gamma(A)$ and the
bracket on $\X(M)$ both satisfy the Jacobi identity.
\end{proof}
It is easy to show that a section $s(v)+f$ preserves the splitting if
and only if $v=v_{f}$. In fact: 
\begin{prop} \label{lie_alg_iso}
The underlying Lie algebra of the Poisson algebra
$\bigl(\cinf(M),\brac{\cdot}{\cdot} \bigr)$ is isomorphic to the Lie algebra
$\bigl(\Gamma(A)^{s},\tabrac{\cdot}{\cdot} \bigr)$.
\end{prop}
\begin{proof}
For any vector field $v' \in \X(M)$, it follows from Eq.\
\ref{tabrac_def} that we have $\tabrac{s(v)+f}{s(v')}=s([v,v'])$ if
and only if
\[
v'(f)+\omega(v',v)=0,
\]
and hence $df=-\iota_{v}\omega$.
Therefore the injective map
\[
\phi \maps \cinf(M) \to \Gamma(A)^{s}, \quad \phi(f)=s(v_{f})
+ f
\]
is also surjective. If  $v_{f}$ and $v_{g}$ are Hamiltonian vector fields corresponding to the functions
$f$ and $g$, respectively, then
\begin{align*}
\left [ \phi(f), \phi(g) \right ]_{A} &=\left [s(v_{f}) + f, s(v_{g}) + g
\right ]_{A} \\
&=  s([v_{f},v_{g}]) + \bigl (v_{f}(g) - v_{g}(f) \bigr) 
- \ip{g}\ip{f}\omega \\
&=s([v_{f},v_{g}]) + \omega(v_{f},v_{g})\\
&= \phi([v_{f},v_{g}]).
\end{align*}
\end{proof}

\subsection*{Prequantization and Atiyah algebroids}
Definition \ref{n-plectic_prequant_def} implies that a prequantized symplectic manifold
is an integral symplectic manifold equipped with Deligne 1-cocycle. 
By definition, this 1-cocycle corresponds to 
a collection of 1-forms $\theta_{i} \in
\Omega^{1}(U_{i})$, and $U(1)$-valued functions $g_{ij} \maps
U_{ij} \to U(1)$ defined on a good cover $\{U_{i}\}$ such that
\begin{align*}
\omega&=d\theta_{i} \quad \text{on $U_{i}$},\\
\i(\theta_{j} -\theta_{i}) &= g^{-1}_{ij}dg_{ij} \quad \text{on $U_{ij}$},\\
g_{jk}g^{-1}_{ik}g_{ij}&=1 \quad \text{on $U_{ijk}$}.
\end{align*}

The Deligne 1-cocycle also gives, of course, a trivialization of the
2-form $\omega$, and therefore the transitive Lie algebroid
$(A,\tabrac{\cdot}{\cdot},\rho)$ over $M$ equipped with the
splitting $s\maps TM \to A$. However in this case, the
functions $g_{ij} \maps U_{ij} \to U(1)$ are the transition
functions of a principal $\U(1)$-bundle $P$ with connection. 
Therefore, by identifying $\u(1)$ with
$\i \cdot \R$, we see that $A$ is isomorphic to the \textbf{Atiyah
  algebroid} $TP/U(1)$. A point in $A$ corresponds to a
vector field along the fiber $\pi^{-1}(x)$ that is invariant under the
right $U(1)$ action. Hence a global section of $A$
corresponds to a $U(1)$-invariant vector field on $P$.

Splittings of $0 \to M \times \R \to A \to TM \to 0$
correspond to connection 1-forms on $P$. 
The connection 1-form $\theta \in \Omega^{1}(P)$ corresponding to the
local forms $\theta_{i}$ induces a `left-splitting'
$\hat{\theta} \maps A \to M \times \R$ such that $\hat{\theta} \circ s =0$.
It is straightforward to show that $a \in \Gamma(A)^{s}$ if
and only if 
\[
\L_{a} \theta =0.
\]
That is, a section of the Atiyah algebroid preserves the splitting if
and only if it preserves the corresponding connection on $P$. For a
prequantized symplectic manifold, the Lie algebra
$\Gamma(A)^{s}$ is a Lie sub-algebra of derivations on
$\cinf(P)_{\C}$ and therefore on the global sections of the associated
line bundle of $P$. Proposition \ref{lie_alg_iso} then implies that we
have a faithful representation of the Poisson
algebra $\left(\cinf(M), \brac{\cdot}{\cdot} \right)$.

\subsection*{The Kostant-Souriau central extension} 
If $(M,\omega)$ is a connected symplectic manifold, then we
have a short exact sequence of Lie algebras
\begin{equation} \label{KS1}
0 \to \u(1) \to \cinf(M) \to \Xham \to 0
\end{equation}
The underlying Lie algebra of the Poisson algebra is known as the
Kostant-Souriau central extension of the Lie algebra of Hamiltonian
vector fields \cite{Kostant:1970}. If $\sigma \maps \Xham \to \cinf(M)$ is
a splitting of the underlying sequence of vector spaces, then the
failure of $\sigma$ to be a strict (i.e.\ bracket-preserving) Lie algebra
morphism is measured by the difference
\[
\{\sigma(v_{1}),\sigma(v_{2})\} - \sigma([v_{1},v_{2}])
\]
which represents a degree 2 class in the
Chevalley-Eilenberg cohomology $H^{2}_{\mathrm{CE}}(\Xham,\R)$. This class
can be represented by using the symplectic form. More specifically, pick a point
$x\in M$ and let $c \in \Hom(\Lambda^{2}\Xham,\R)$ be the cochain
given by:
\[
c(v,v')= - \omega(v,v') \vert_{x}, \quad \forall v,v' \in \Xham.
\] 
The fact that $c$ is a cocycle follows from the bracket
$\brac{\cdot}{\cdot}$ satisfying the Jacobi identity.
One can show that the class $[c]$ does not depend on the choice of $x\in
M$.

If $(M,\omega)$ is a prequantized connected symplectic manifold, then
Prop.\ \ref{lie_alg_iso} implies that the `quantized Poisson algebra'
gives an isomorphic central extension
\[
0 \to \u(1) \to \Gamma(A)^{s} \to \Xham \to 0.
\]
This central extension is
responsible for introducing phases into the quantized system. Two
functions $f$ and $f'$ differing by a constant $r \in \u(1)$ will have
the same Hamiltonian vector fields and therefore give the same flows
on $M$. However, their quantizations will give unitary transformations
which differ by a phase $\exp (2\pi \i r )$.

\section{Courant algebroids} \label{courant_sec}
Now, we begin our investigation of the 2-plectic case.
First, we recall some basic facts and examples of Courant
algebroids and then we proceed to describe \v{S}evera's classification of exact Courant
algebroids. The Courant algebroid will act as the 2-plectic analogue
of the Atiyah algebroid.

There are several equivalent definitions of a Courant
algebroid found in the literature. The following
definition, due to Roytenberg \cite{Roytenberg_thesis}, is equivalent
to the original definition given by Liu, Weinstein, and Xu \cite{Liu:1997}.
\begin{definition} \label{courant_algebroid} A {\bf Courant algebroid}
  is a vector bundle $C \to M$ equipped with a nondegenerate symmetric
  bilinear form $\innerprod{\cdot}{\cdot}$ on the bundle, a
  skew-symmetric bracket $\cbrac{\cdot}{\cdot}$ on $\Gamma(C)$, and a
  bundle map (called the {\bf anchor}) $\rho \maps C \to TM$ such that
  for all $ e_{1},e_{2},e_{3}\in \Gamma (C)$ and for all $f,g \in
  \cinf(M)$ the following properties hold:
\begin{enumerate}
\item{$\cbrac{e_1}{\cbrac{e_2}{e_3}} - \cbrac{\cbrac{e_1}{e_2}}{e_3} 
-\cbrac{e_2}{\cbrac{e_1}{e_3}}=-DT(e_{1},e_{2},e_{3}),$  
}

\item{$\rho(\cbrac{e_{1}}{e_{2}})=[\rho(e_{1}),\rho(e_{2})]$, \label{axiom2}
} 

\item{$\cbrac{e_{1}}{fe_{2}}=f\cbrac{e_{1}}{e_{2}}+\rho (e_{1})(f)e_{2}-\half
\langle e_{1},e_{2}\rangle {D}f$,
}

\item{$\langle {D}f,{D}g\rangle =0$,}

\item{$\rho(e_{1})\left(\langle
e_{2},e_{3}\rangle\right) =\langle \cbrac{e_{1}}{e_{2}}+\half {D}\langle
e_{1},e_{2}\rangle ,e_{3}\rangle +\langle e_{2},\cbrac{e_{1}}{e_{3}}+\half
{D}\langle e_{1},e_{3}\rangle \rangle$,}
\end{enumerate}
\noi where $[\cdot,\cdot]$ is the Lie bracket of vector fields,
$D \maps \cinf(M) \to \Gamma(C)$ is the map defined by $\innerprod{D f}{e}=\rho(e)f$, and 
\[
T(e_{1},e_{2},e_{3})= \frac{1}{6}
\left(\innerprod{\cbrac{e_1}{e_2}}{e_3} +
\innerprod{\cbrac{e_3}{e_1}}{e_2} + \innerprod{\cbrac{e_2}{e_3}}{e_1} \right). 
\]
\end{definition}
The bracket in Definition \ref{courant_algebroid} is skew-symmetric,
but the first property implies that it needs only to satisfy the
Jacobi identity ``up to $D T$''. Note that the vector bundle $C \to M$ may
be identified with $C^{\ast} \to M$ via the bilinear form
$\innerprod{\cdot}{\cdot}$ and therefore we have the dual map
\[ \rho^{\ast} \maps T^{\ast}M \to C.\] 
Hence the map $D$ is simply the pullback of the de Rham differential by $\rho^{\ast}$. 

There is a commonly used alternate definition given by \v{S}evera \cite{Severa1}
for Courant algebroids which involves a bracket operation on sections
that satisfies a Jacobi identity but is not skew-symmetric.
\begin{definition}\label{alt_courant_algebroid}
  A {\bf Courant algebroid} is a vector bundle $ C\rightarrow M $
  together with a nondegenerate symmetric bilinear form $ \innerprod
  {\cdot }{\cdot } $ on the bundle, a bilinear operation $\dbrac{\cdot}{\cdot}$ on
  $ \Gamma (C) $, and a bundle map $ \rho \maps C\rightarrow TM $ such
  that for all $ e_{1},e_{2},e_{3}\in \Gamma (C)$ and for all $f \in
  \cinf(M)$ the following properties hold:
\begin{enumerate}
\item{$
    \dbrac{e_{1}}{\dbrac{e_{2}}{e_{3}}}=\dbrac{\dbrac{e_{1}}{e_{2}}}{e_{3}}+\dbrac{e_{2}}{\dbrac{e_{1}}{e_{3}}}$,}
\item {$\rho (\dbrac{e_{1}}{e_{2}})=[\rho (e_{1}),\rho (e_{2})]$,} 
\item {$ \dbrac{e_{1}}{fe_{2}}=f\dbrac{e_{1}}{e_{2}}+\rho (e_{1})(f)e_{2}$,}
\item {$ \dbrac{e_{1}}{e_{1}}=\half D \innerprod {e_{1}}{e_{1}}$,}
\item {$ \rho (e_{1})\left(\innerprod {e_{2}}{e_{3}}\right)=\innerprod
    {\dbrac{e_{1}}{e_{2}}}{e_{3}}+\innerprod {e_{2}}{\dbrac{e_{1}}{e_{3}}}$,}
\end{enumerate}
where $[\cdot,\cdot]$ is the Lie bracket of vector fields, and
$D \maps \cinf(M) \to \Gamma(C)$ is the map defined by $\innerprod{D f}{e}=\rho(e)f$.
\end{definition}
Roytenberg \cite{Roytenberg_thesis} showed that $C \to M$ is a Courant
algebroid in the sense of Definition \ref{courant_algebroid} with
bracket $\cbrac{\cdot}{\cdot}$, bilinear form
$\innerprod{\cdot}{\cdot}$ and anchor $\rho$ if and only if $C \to M$
is a Courant algebroid in the sense of Definition
\ref{alt_courant_algebroid} with the same anchor and bilinear form but
with bracket $\dbrac{\cdot}{\cdot}$ given by
\begin{equation}
 \dbrac{e_{1}}{e_{2}} = \cbrac{e_{1}}{e_{2}} + \half D
 \innerprod{e_{1}}{e_{2}}. \label{dorfman}
\end{equation}
All Courant algebroids in this chapter are considered to be Courant algebroids in
the sense of Definition \ref{courant_algebroid}. We introduced
Definition \ref{alt_courant_algebroid} mainly to connect our
discussion here with previous results in the literature.

\begin{example} \label{standard}
An important example of a Courant algebroid is the \textbf{standard
  Courant algebroid} $C=TM \oplus T^{\ast}M$ over any manifold
  $M$ equipped with the \textbf{standard Courant
  bracket}:
\begin{equation} 
\cbrac{v_{1} + \alpha_{1}}{v_{2} + \alpha_{2}}=
[v_{1},v_{2}] +\L_{v_{1}}\alpha_{2}- \L_{v_{2}}\alpha_{1} -
\half d\innerprodm{v_{1} + \alpha_{1}}{v_{2} + \alpha_{2}}, \label{standard_bracket}
 \end{equation}
where
\begin{equation} \label{skew}
\innerprodm{v_{1} + \alpha_{1}}{v_{2} + \alpha_{2}} = \ip{1}\alpha_{2}
- \ip{2}\alpha_{1}
\end{equation}
is the \textbf{standard skew-symmetric pairing}.
The bilinear form is given by the \textbf{standard symmetric pairing}:
\begin{equation}
\innerprodp{v_{1} + \alpha_{1}}{v_{2} + \alpha_{2}}=
\ip{1}\alpha_{2} + \ip{2}\alpha_{1}. \label{standard_innerprod}
\end{equation}
The anchor $\rho \maps C \to TM$ is the projection
map, and $D=d$ is the de Rham differential. The bracket
$\cbrac{\cdot}{\cdot}$ is the skew-symmetrization of the \textbf{standard
  Dorfman bracket}:
\begin{equation} \label{std_dorf}
\dbrac{v_{1} + \alpha_{1}}{v_{2} + \alpha_{2}}=
[v_{1},v_{2}] + \L_{v_{1}}\alpha_{2} - \ip{2}d\alpha_{1},
\end{equation}
which plays the role of the bracket given in
Definition \ref{alt_courant_algebroid}. 
%In this example, the identity
%(\ref{dorfman}) is
%\begin{equation} \label{dorf_courant}
%\cbrac{v_{1} + \alpha_{1}}{v_{2} + \alpha_{2}}=
%\dbrac{v_{1} + \alpha_{1}}{v_{2} + \alpha_{2}} - \half
%d\innerprod{v_{1} + \alpha_{1}}{v_{2} + \alpha_{2}}.
%\end{equation}
\end{example}

The standard Courant algebroid is the prototypical example of
an \textbf{exact Courant algebroid} \cite{Bressler-Chervov}.
\begin{definition} \label{exact} A Courant algebroid $C \to M$ with
  anchor map $\rho \maps C \to TM$ is {\bf exact} iff
\[
0 \to T^{\ast}M \stackrel{\rho^{\ast}}{\to} C \stackrel{\rho}{\to} TM
\to 0\]
is an exact sequence of vector bundles.
\end{definition}

\subsection*{The \v{S}evera class of an exact Courant algebroid}
\v{S}evera's classification \cite{Severa1} originates in the idea that a particular
kind of splitting of the above short exact sequence corresponds to defining a
connection.
\begin{definition}\label{connection}
A {\bf splitting} of an exact Courant algebroid $C$ over a manifold $M$ is a map of
vector bundles $s \maps TM \to C $ such that
\begin{enumerate}
\item{ $\rho \circ s = \id_{TM}$,}
\item{$\innerprod{s(v_{1})}{s(v_{2})}=0$ for all $v_{1},v_{2} \in TM$,}
\end{enumerate}
where $\rho \maps C \to TM$ and $\innerprod{\cdot}{\cdot}$ are the
anchor and bilinear form, respectively.
\end{definition}
In other words, a splitting of an exact Courant algebroid is an isotropic splitting of
the sequence of vector bundles. Bressler and Chervov call splittings `connections'
\cite{Bressler-Chervov}.
If $s$ is a splitting and $B \in
\Omega^2(M)$ is a 2-form then one can construct a new splitting:
\begin{equation}
\left(s+B\right)(v)=s(v)+ \rho^{\ast}B(v,\cdot). \label{2-form_action}
\end{equation}
Furthermore, one can show that any two splittings on an exact Courant
algebroid must differ by a 2-form on
$M$ in this way. Hence the space of splittings on an exact Courant algebroid
is an affine space modeled on the vector space of 2-forms
$\Omega^{2}(M)$ \cite{Bressler-Chervov}.

The failure of a splitting to preserve the bracket gives a suitable
notion of `curvature'. Given vector fields $v_{1}, v_{2}, v_{3}$ on $M$, it can be shown that
the function
\[
\omega(v_{1},v_{2},v_{3})= \innerprod{\cbrac{s\left(v_{1}\right)}{s\left(v_{2}\right)}}{s(v_{3})}
\]
defines a closed 3-form on $M$ \cite{Bressler-Chervov}.
This is the curvature 3-form of an exact
Courant algebroid over $M$. It gives a well-defined cohomology class in
$H^{3}_{\mathrm{DR}}(M)$, independent of the choice of
splitting. 
\begin{definition}[\cite{Gualtieri:2007}] \label{Severa_class} The {\bf \v{S}evera class}
  of an exact Courant algebroid with bracket $\cbrac{\cdot}{\cdot}$ and bilinear
  form $\innerprod{\cdot}{\cdot}$ is the cohomology
  class $[-\omega] \in H^{3}_{\mathrm{DR}}(M)$, where
\[
\omega(v_{1},v_{2},v_{3})=
\innerprod{\cbrac{s\left(v_{1}\right)}{s\left(v_{2}\right)}}{s(v_{3})}.
\]
\end{definition}

\section{Courant algebroids and $2$-plectic
  geometry} \label{geometric} 
In this section, we describe a relationship between Courant algebroids
and 2-plectic manifolds which can be understood as the higher analogue 
of the relationship between Atiyah algebroids and symplectic manifolds. 

We begin by recalling how to
explicitly construct an exact Courant algebroid with \v{S}evera class
$[\omega]$. This is the 3-form version of the construction that gives
a transitive Lie algebroid over a pre-symplectic manifold, which was
previously discussed in Sec. \ref{symplectic_sec}.  The approach
given here is essentially identical to those given by Gualtieri
\cite{Gualtieri:2007}, Hitchin \cite{Hitchin:2004ut}, and \v{S}evera
\cite{Severa1} .

Let $(M,\omega)$ be a manifold equipped with a closed 3-form.
A trivialization of $\omega$
is an open cover$\{U_{i}\}$ of $M$ equipped with 2-forms  
$B_{i} \in \Omega^{2}(U_{i})$, and 1-forms $A_{ij} \in \Omega^{1}(U_{ij})$ 
on intersections such that
\begin{equation} \label{cocycle2}
\begin{split}
 \omega \vert_{U_{i}} &= dB_{i} \\
(B_{j}-B_{i}) \vert_{U_{ij}} &= dA_{ij}.
\end{split}
\end{equation}
Given such a trivialization, over each open set $U_{i}$ consider the
bundle $C_{i}=TU_{i} \oplus
T^{\ast}U_{i}\to U_{i}$ equipped with the standard pairing
\begin{equation} \label{twist_innerprod}
\innerprodp{v_{1} + \alpha_{1}}{v_{2} + \alpha_{2}}_{i}= \ip{1}\alpha_{2}
+ \ip{2}\alpha_{1},
\end{equation}
$v_{1},v_{2} \in \X(U_{i}),$  $\alpha_{1},\alpha_{2} \in\Omega^{1}(U_{i}),$
which has signature $(n,n)$. On double intersections, it is
easy to see that 
\[
\innerprodp{v_{1} + \ip{1}dA_{ij} +\alpha_{1}}{v_{2} + \ip{2}dA_{ij} +
  \alpha_{2}}_{i} = \innerprodp{v_{1} + \alpha_{1}}{v_{2} + \alpha_{2}}_{i}.
\]
Hence the 2-forms $\{dA_{ij}\}$ generate transition functions
\begin{align*}
G_{ij} \maps U_{ij} \to SO(n,n),\\
G_{ij}(x)= 
\begin{pmatrix}
1 & 0\\
dA_{ij} \vert_{x} & 1
\end{pmatrix},
\end{align*}
which satisfy the cocycle conditions on $U_{ijk}$ by virtue of Eq.\
\ref{cocycle2}. Therefore, we have over  $M$ the vector bundle
\[
C= \coprod_{x \in M} T_{x}U_{i} \oplus T^{\ast}_{x}U_{i} / \sim,
\]
equipped with a bilinear form denoted as $\innerprodp{\cdot}{\cdot}$.
$C$ sits in the exact sequence
\[
0 \to T^{\ast}M \stackrel{\jmath}{\rightarrow} C \stackrel{\rho}\to TM \to 0,
\]
where the anchor $\rho$ is induced by the projection $T^{\ast}U_{i} \oplus TU_{i}
\to TU_{i}$, and $\jmath$ is the inclusion.

The 2-forms $B_{i}$ induce a bundle map  $s \maps TM \to C$
\begin{equation} \label{canonical_split}
s(v_{x}) = v_{x} - B_{i}(v_{x}) \quad \text{if $x\in U_{i}$},
\end{equation}
It follows from Eq.\ \ref{cocycle2} that $s$ is well-defined when $x \in
U_{ij}$. It is easy to see that this map is an isotropic splitting
(Def.\ \ref{connection}).
Hence every section $e\in \Gamma(C)$ can be uniquely expressed
as
\[
e= s(v) + \alpha,
\]
for some $v\ \in \X(M)$ and $\alpha \in \Omega^{1}(M)$.  As
before, we use $s$ to also denote the map $\Gamma(s) \maps \X(M) \to
\Gamma(C)$. The anchor map is just
\begin{equation} \label{twist_anchor}
\rho \bigl(s(v)+\alpha \bigr)=v.
\end{equation}

Given sections $s(v_{1})+\alpha_{1}$, $s(v_{2}) + \alpha_{2} \in
\Gamma(C)$, a local calculation using Eq.\
\ref{canonical_split}  gives
\begin{equation} \label{split_innerprod}
\begin{split}
\innerprodp{s(v_{1})+\alpha_{1}}{s(v_{2})+\alpha_{2}}
&= \ip{1}\alpha_{2} - \ip{1}\ip{2}B_{i} + \ip{2}\alpha_{1} -
\ip{2}\ip{1}B_{i}\\
&= \innerprodp{v_{1}+\alpha_{1}}{v_{2}+\alpha_{2}}.
\end{split}
\end{equation}
The above equality holds, in fact, for any splitting $s' \maps TM \to
C$, since $s-s'$ is a 2-form on $M$ and therefore skew-symmetric.
The bracket on $\Gamma(C)$ is defined over the open set
$U_{i}$ by:
\[
\tcbrac{s(v_{1}) + \alpha_{1}}{s(v_{2}) +
  \alpha_{2}}\vert_{U_{i}} =\sbrac{s(v_{1}) + \alpha_{1}}{s(v_{2}) +
  \alpha_{2}}_{i}
\]
where $\sbrac{\cdot}{\cdot}_{i}$ is the standard Courant bracket
(\ref{standard_bracket}) on $C_{i}$. Since the 2-forms $\{dA_{ij}\}$
are closed, it follows by direct computation that on double
intersections $U_{ij}$:
\[
\sbrac{G_{ij} (v_{1} + \alpha_{1})}{G_{ij}(v_{2} +
  \alpha_{2})}_{i} = G_{ij} \bigl(\sbrac{v_{1} + \alpha_{1}}{v_{2} +
  \alpha_{2}}_{i} \bigr).
\]
Hence the bracket $\tcbrac{\cdot}{\cdot}$ is indeed globally well-defined. 
Using the local definition of the bracket and the splitting, as well
as the fact that $dB_{i}=\omega$, it is easy to show that
\begin{equation} \label{twist_bracket}
\begin{split}
\tcbrac{s(v_{1}) + \alpha_{1}}{s(v_{2}) +
  \alpha_{2}} &= s \bigl ([v_{1},v_{2}] \bigr) +
\L_{v_{1}}\alpha_{2}- \L_{v_{2}}\alpha_{1} \\
 & \quad -\half d\innerprodm{v_{1} + \alpha_{1}}{v_{2} + \alpha_{2}}
-\ip{2}\ip{1}\omega.
\end{split}
\end{equation}
The bracket $\tcbrac{\cdot}{\cdot}$ is called the \textbf{twisted Courant bracket}. 
A analogous construction using the standard Dorfman bracket (\ref{std_dorf}) 
on $C_{i}$ gives the \textbf{twisted Dorfman
  bracket}:
\begin{equation} \label{twist_bracket_dorf}
\tdbrac{s(v_{1}) + \alpha_{1}}{s(v_{2}) + \alpha_{2}}=
s\bigl([v_{1},v_{2}]\bigr) + \L_{v_{1}}\alpha_{2} - \ip{2}d\alpha_{1} -\ip{2}\ip{1}\omega.
\end{equation}
These brackets were studied in detail by \v{S}evera and Weinstein
\cite{Severa-Weinstein,Severa1}.

It is straightforward to check that $C \to M$ equipped with the
aforementioned bilinear form, anchor, and bracket
$\tcbrac{\cdot}{\cdot}$ is an exact Courant algebroid (Definition
\ref{courant_algebroid}). Just as in Lie algebroid case, the construction of $C$
is independent of the choice of trivialization up to a splitting-preserving
isomorphism.

A direct calculation shows that
\[
-\omega(v_{1},v_{2},v_{3}) =
\innerprodp{\tcbrac{s\left(v_{1}\right)}{s\left(v_{2}\right)}}{s(v_{3})}.
\]
Hence, the Courant algebroid $C$ has
\v{S}evera class $[\omega]$. Of course, we are interested in the
special case when $\omega$ is a 2-plectic structure.
We summarize the above discussion with the following proposition:
\begin{prop} \label{2plectic_courant}
Let $(M,\omega)$ be a $2$-plectic manifold. Up to isomorphism, there
exists a unique exact Courant
algebroid $C$ over $M$, with
bilinear form $\innerprodp{\cdot}{\cdot}$,  
anchor map $\rho$, and 
bracket $\tcbrac{\cdot}{\cdot}$ 
given in Eqs.\ \ref{twist_innerprod}, \ref{twist_anchor}, and
\ref{twist_bracket}, respectively,
and  equipped with a splitting whose curvature is $-\omega$.
\end{prop}

\subsection*{Lie $2$-algebras from Courant algebroids}
Next, we describe how a Courant algebroid gives a Lie 2-algebra. From
here on, we shall describe a Lie 2-algebra using the terminology given
in Prop.\ \ref{L2A}, i.e.\ as a 2-term chain complex, equipped with a
bracket and a Jacobiator.

Recall that the space of global sections of a transitive Lie
algebroid associated to a closed 2-form gives  a Lie algebra.
As we shall see, the global sections of a  Courant algebroid
form a Lie 2-algebra. Given any Courant algebroid $C \rightarrow M$ with bilinear form
$\innerprod{\cdot}{\cdot}$, bracket $\cbrac{\cdot}{\cdot}$, and anchor
$\rho \maps C \to TM$, one can construct a
$2$-term chain complex 
\[
L \quad = \quad \cinf(M) \stackrel{D}{\rightarrow} \Gamma(C),
\]
with differential $D=\rho^{\ast}d$ where $d$ is the de Rham differential. The bracket
$\cbrac{\cdot}{\cdot}$ on global sections can be extended to a chain
map $\sbrac{\cdot}{\cdot} \maps L \tensor L \to L$. If $e_1,e_2$ are
degree 0 chains then $\sbrac{e_{1}}{e_{2}}$ is the original bracket.
If $e$ is a degree 0 chain and $f,g$ are degree 1 chains, then we
define:
\begin{align*}
\sbrac{e}{f} &= -\sbrac{f}{e} = \half \innerprod{e}{D f}  \\
\sbrac{f}{g}&=0.
\end{align*}   
It was shown by Roytenberg and Weinstein \cite{Roytenberg-Weinstein}
that this extended bracket gives a $L_{\infty}$-algebra. Roytenberg's
later work \cite{Roytenberg_graded,Roytenberg_L2A} implies that a
brutal truncation of this $L_{\infty}$-algebra is a Lie 2-algebra
whose underlying complex is $L$.  
For the Courant algebroid
$C$ associated to a 2-plectic manifold, their result implies:
\begin{theorem}\label{courant_L2A}
If $C$ is the exact Courant algebroid given in Proposition \ref{2plectic_courant}
then there is a  Lie $2$-algebra  
$L_{\infty}(C)=(L,\blankbrac, J)$ where:
\begin{itemize}
\item $L_{0}=\Gamma(C)$,
\item $L_{1}=\cinf(M)$,
\item the differential $L_{1} \stackrel{d}{\to} L_{0}$ is the de Rham differential
\item{the bracket $\sbrac{\cdot}{\cdot}$ is 
\[
[e_{1},e_{2}]= \tcbrac{e_{1}}{e_{2}} \quad  \text{in degree 0}
\]
and
\[
  [e,f]=-[f,e]=\half \innerprodp{e}{df} \quad \text{in degree 1},
\]
}
\item the Jacobiator is the linear map $J\maps \Gamma(C) \tensor \Gamma(C) \tensor \Gamma(C) \to \cinf(M)$ defined by
\begin{align*} 
J(e_{1},e_{2},e_{3}) &= -T(e_{1},e_{2},e_{3})\\
&=-\frac{1}{6}
\left(\innerprodp{\tcbrac{e_1}{e_2}}{e_3} +
\innerprodp{\tcbrac{e_3}{e_1}}{e_2} \right. \\
 & \left.\quad + \innerprodp{\tcbrac{e_2}{e_3}}{e_1} \right). 
\end{align*}
\end{itemize}
\end{theorem}
More precisely, the theorem follows from Example 5.4 of
\cite{Roytenberg_L2A} and Section 4 of \cite{Roytenberg_graded}. On
the other hand, the original construction of Roytenberg and Weinstein
gives a $L_{\infty}$-algebra on the complex:
\[
0 \rightarrow \ker D \stackrel{\iota}{\rightarrow} \cinf(M)
\stackrel{D}{\rightarrow}\Gamma(C),
\]
with trivial structure maps $l_{n}$ for $n > 3$. Moreover, the map
$l_{2}$ (corresponding to the bracket $\sbrac{\cdot}{\cdot}$ given above) is trivial in
degree $>1$ and the map $l_{3}$ (corresponding to the Jacobiator $J$) is
trivial in degree $>0$. Hence these maps induce the above Lie 2-algebra
structure on $\cinf(M)\stackrel{D}{\rightarrow}\Gamma(C)$.

\subsection*{The algebraic relationship between 2-plectic and
  Courant} 
Associated to any 2-plectic manifold $(M,\omega)$, is a Lie 2-algebra $\Lie(M,\omega)$
(Thm.\ \ref{main_thm}). In Prop.\ \ref{semistrict}, we described this
Lie 2-algebra as a 2-term chain complex $L =(L_{1} \xto{d} L_{0})$
equipped with a bracket $\blankbrac$ and Jacobiator $J$ where:
\begin{itemize}
\item $L_{0} =\hamn{1}$ is the space of Hamiltonian 1-forms,
\item $L_{1}=\cinf(M)$,
\item the differential $L_{1} \stackrel{d}{\to} L_{0}$ is the de Rham differential,
\item the bracket $\blankbrac$ is $\brac{\alpha}{\beta}=\omega(v_{\alpha},v_{\beta},\cdot)$ in degree 0
  and trivial otherwise,
\item the Jacobiator is given by the linear map $J\maps \ham \tensor \ham \tensor 
\ham \to \cinf$, where $J(\alpha,\beta,\gamma) = \omega(v_{\gamma},v_{\beta},v_{\alpha})$.
\end{itemize}
\noi Also associated to $(M,\omega)$, is the exact
Courant algebroid
$(C,\tcbrac{\cdot}{\cdot},\innerprodp{\cdot}{\cdot},\rho)$ described
in Prop.\ \ref{2plectic_courant},
equipped with a splitting $s \maps TM \to C$ whose curvature is $-\omega$. 
From this Courant algebroid, we obtain the Lie 2-algebra $\Lie(C)$
described in Thm.\ \ref{courant_L2A}.

We now describe the relationship between $\Lie(M,\omega)$ and $\Lie(C)$.
We understand this as the 2-plectic analogue 
of the relationship described in Sec.\ \ref{symplectic_sec}
between the Poisson algebra of a symplectic
manifold and the Lie algebra associated to the transitive Lie algebroid over the manifold.
\begin{theorem} \label{main_courant_thm} Let $(M,\omega)$ be a $2$-plectic
  manifold and let $C$ be its corresponding Courant
  algebroid. Let $L_{\infty}(M,\omega)$ and $L_{\infty}(C)$ be the 
  Lie 2-algebras corresponding to $(M,\omega)$ and $C$,
  respectively.  There exists a morphism of Lie 2-algebras embedding
  $L_{\infty}(M,\omega)$ into $L_{\infty}(C)$.
\end{theorem} 

Before we prove the theorem, we introduce some technical
lemmas to ease the calculations. Recall from Eq.\ \ref{skew} that the
formula for the standard skew-symmetric pairing
on $\X(M) \oplus \Omega^{1}(M)$:
\[
\innerprodm{v_{1} + \alpha_{1}}{v_{2} + \alpha_{2}} =
\ip{1}\alpha_{2} - \ip{2}\alpha_{1}.
\]
In what follows, by the symbol ``$\cp$'' we mean cyclic permutations of the symbols $\alpha,\beta,\gamma$.
\begin{lemma}\label{calc1}
If $\alpha, \beta \in \ham$ with corresponding Hamiltonian vector fields
$v_{\alpha},v_{\beta}$, then
$\L_{v_{\alpha}}\beta=\brac{\alpha}{\beta} + d \ip{\alpha}\beta$.
\end{lemma}
\begin{proof} Since $\L_{v} = \iota_v d + d \iota_v$,
\[  \L_{v_{\alpha}}{\beta} = 
\ip{\alpha} d \beta + d \ip{\alpha} \beta =
-\ip{\alpha}\ip{\beta} \omega + d \ip{\alpha} \beta =
\brac{\alpha}{\beta} + d \ip{\alpha} \beta .\]
\end{proof}

\begin{lemma} \label{calc2}
If $\alpha,\beta,\gamma \in \ham$ with corresponding Hamiltonian vector fields
$v_{\alpha},v_{\beta},v_{\gamma}$, then
\begin{align*}
\iota_{[v_{\alpha},v_{\beta}]} \gamma + \cp &= 
-3\ip{\alpha}\ip{\beta}\ip{\gamma}\omega +
\ip{\alpha}d\innerprodm{v_{\beta}+\beta}{v_{\gamma} + \gamma}\\ & \quad  +
\ip{\gamma}d\innerprodm{v_{\alpha}+\alpha}{v_{\beta} + \beta} +
\ip{\beta}d\innerprodm{v_{\gamma}+\gamma}{v_{\alpha} + \alpha}.
\end{align*}
\end{lemma}
\begin{proof}
The identity $\iota_{[v_{\alpha},v_{\beta}]}=\L_{v_{\alpha}}\ip{\beta}
-\ip{\beta}\L_{v_{\alpha}}$ and Lemma \ref{calc1} imply:
\begin{align*}
\iota_{[v_{\alpha},v_{\beta}]} \gamma &=
\lie{\alpha}{\ip{\beta}}\gamma - {\ip{\beta}}\lie{\alpha}\gamma\\
&=\lie{\alpha}{\ip{\beta}}\gamma -\ip{\beta}
\left(\brac{\alpha}{\gamma} + d \ip{\alpha}\gamma\right)\\
&=\ip{\alpha}d\ip{\beta}\gamma -\ip{\beta}\ip{\gamma}\ip{\alpha}\omega 
  - \ip{\beta}d\ip{\alpha}\gamma,
\end{align*}
where the last equality follows from the definition of the bracket.

Therefore we have:
\begin{align*}
\iota_{[v_{\gamma},v_{\alpha}]} \beta  
&=\ip{\gamma}d\ip{\alpha}\beta -\ip{\alpha}\ip{\beta}\ip{\gamma}\omega 
  - \ip{\alpha}d\ip{\gamma}\beta,\\
\iota_{[v_{\beta },v_{\gamma}]}\alpha  
&=\ip{\beta}d\ip{\gamma}\alpha -\ip{\gamma}\ip{\alpha}\ip{\beta}\omega 
  - \ip{\gamma}d\ip{\beta}\alpha,
\end{align*}
and Eq.\ \ref{skew} implies
\[
\ip{\alpha}d\ip{\beta}\gamma - \ip{\alpha}d\ip{\gamma}\beta=
\ip{\alpha}d\innerprodm{v_{\beta}+\beta}{v_{\gamma} + \gamma}.
\]
The statement then follows.
\end{proof}

\begin{lemma}\label{calc3}
If $\alpha, \beta \in \ham$ with corresponding Hamiltonian vector fields
$v_{\alpha},v_{\beta}$, then
\[
\lie{\alpha}{\beta}-\lie{\beta}{\alpha}=2\brac{\alpha}{\beta} + d\innerprodm{v_{\alpha}
  + \alpha}{v_{\beta} + \beta}.
\]
\end{lemma}
\begin{proof}
Follows immediately from Lemma \ref{calc1} and Eq.\ \ref{skew}.
\end{proof}

We have all we need to give a proof of Thm.\ \ref{main_courant_thm}.
\begin{proof}[Proof of Theorem \ref{main_courant_thm}]
Let
\begin{gather*}
L= \cinf(M) \stackrel{d}{\to} \ham,\\
\lbrac{\cdot}{\cdot} \maps L \tensor L\to L,\\
 J_{L} \maps L \tensor L \tensor L  \to L
\end{gather*}
denote the underlying chain complex, bracket, and Jacobiator of the Lie
2-algebra $L_{\infty}(M,\omega)$. Similarly,
\begin{gather*}
L' = \cinf(M) \stackrel{d}{\to} \Gamma(C),\\
\lpbrac{\cdot}{\cdot} \maps L' \tensor L' \to L',\\
J_{L'} \maps L' \tensor L' \tensor L' \to L'
\end{gather*}
denotes the underlying chain
complex, bracket, and Jacobiator of the Lie 2-algebra
$L_{\infty}(C)$.

We construct a Lie 2-algebra morphism from $L_{\infty}(M,\omega)$ to $L_{\infty}(C)$.
Recall from Def. \ref{homo}, that such a morphism consists of
\begin{itemize}
\item{a chain map $\phi \maps L \to L'$, and}
\item{a chain homotopy $\Phi \maps L \tensor L \to L'$ from the chain
  map
\[     \begin{array}{ccl}  
     L \tensor L & \to & L'   \\
     x \tensor y & \longmapsto & \phi \left( [x,y] \right)
  \end{array}
\]
to the chain map
\[     \begin{array}{ccl}  
     L \tensor L & \to & L'   \\
     x \tensor y & \longmapsto & \left [ \phi(x),\phi(y) \right]^{\prime},
  \end{array}
\]
}
\end{itemize}
such that the following equation holds:
\begin{equation} \label{coherence}
\begin{array}{l}
\phi_1(J(x,y,z))- J^{\prime}(\phi_0(x),\phi_0(y), \phi_0(z)) = \\
\Phi(x,[y,z]) -\Phi([x,y],z) - \Phi(y,[x,z]) - [\Phi(x,y),\phi_0(z)]^{\prime}\\
+ [\phi_0(x), \Phi(y,z)]^{\prime}- [\phi_0(y),\Phi(x,z)]^{\prime}.
\end{array}
\end{equation}

Let $s \maps TM \to C$ be the splitting. Let $\phi_{0} \maps
\ham \to \Gamma(C)$ be given by
\[
\phi_{0}(\alpha)=s(v_\alpha) +\alpha,
\]
where $v_{\alpha}$ is the Hamiltonian vector field corresponding to
$\alpha$.
Let $\phi_{1}\maps \cinf(M) \to \cinf(M)$ be the identity.
Then $\phi \maps L \to L'$ is a chain
map, since the Hamiltonian vector field of an exact 1-form is zero.
Let $\Phi \maps \ham \tensor \ham \to \cinf(M)$ be given
by
\[
\Phi(\alpha,\beta)=-\half \innerprodm{v_{\alpha} + \alpha}{v_{\beta} + \beta}.
\]

Now we show $\Phi$ is a well-defined chain homotopy in the sense of
Def.\ \ref{homo}. We have
\begin{equation}\label{lpbrac_eq}
\begin{split}
\lpbrac{\phi_{0}(\alpha)}{\phi_{0}(\beta)}&=  \tcbrac{s(v_{\alpha}) +
  \alpha}{s(v_{\beta}) + \beta}\\
&=  s([v_{\alpha},v_{\beta}]) + \L_{v_{\alpha}}\beta - \L_{v_{\beta}}\alpha
 -\ip{\beta}\ip{\alpha} \omega\\
  &\quad - \half d \innerprodm{v_{\alpha} + \alpha }{v_{\beta} +
    \beta}\\
  &=s([v_{\alpha},v_{\beta}]) + \brac{\alpha}{\beta} + \half d
  \innerprodm{v_{\alpha} + \alpha }{v_{\beta} + \beta}\\
&=s([v_{\alpha},v_{\beta}]) + \lbrac{\alpha}{\beta} - d\Phi(\alpha,\beta).
\end{split}
\end{equation}
The second line above is just the definition of the twisted Courant
bracket (Eq.\ \ref{twist_bracket}), while the 
second to last line follows from Lemma \ref{calc3} and 
Def.\ \ref{bracket_def} of the bracket $\{\cdot,\cdot\}$.
By Prop.\ \ref{bracket_prop}, the Hamiltonian vector
field of $\brac{\alpha}{\beta}$ is $[v_{\alpha},v_{\beta}]$. Hence we
have:
\[
\phi_{0}(\lbrac{\alpha}{\beta}) -\lpbrac{\phi_{0}(\alpha)}{\phi_{0}(\beta)}
=d\Phi(\alpha,\beta).
\]

In degree 1, the bracket $\lbrac{\cdot}{\cdot}$ is trivial. It
follows from the definition of $\lpbrac{\cdot}{\cdot}$ that
\[
\phi_{1}(\lbrac{\alpha}{f})- 
\lpbrac{\phi_{0}(\alpha)}{\phi_{1}(f)} = -\half \innerprodp{s(v_{\alpha}) + \alpha}{df}.
\]
From Eq.\ \ref{split_innerprod}, we have
\[
\innerprodp{s(v_{\alpha}) + \alpha}{df}=\innerprodp{s(v_{\alpha}) +
  \alpha}{s(0)+ df}= \ip{\alpha}df.
\]
Therefore 
\[
\phi_{1}(\lbrac{\alpha}{f})- 
\lpbrac{\phi_{0}(\alpha)}{\phi_{1}(f)} = \Phi(\alpha,df),
\]
and similarly
\[
\phi_{1}(\lbrac{f}{\alpha})- 
\lpbrac{\phi_{1}(f)}{\phi_{0}(\alpha)} = \Phi(df,\alpha).
\]
Therefore $\Phi$ is a chain homotopy.

It remains to show the coherence condition (Eq.\ \ref{coherence} in
Definition \ref{homo}) is satisfied.  First we rewrite the Jacobiator
$J_{L'}$ using the second to last line of (\ref{lpbrac_eq}):
\begin{align*}
  J_{L'}(\phi_0(\alpha),\phi_0(\beta), \phi_0(\gamma))&=-\frac{1}{6}
  \innerprodp{\lpbrac{\phi_{0}(\alpha)}{\phi_0(\beta)}}{\phi_0(\gamma)}
  +
  \cp \\
  &=-\frac{1}{6}\innerprodp{
    s([v_{\alpha},v_{\beta}])
    +\brac{\alpha}{\beta}-d\Phi(\alpha,\beta)}{s(v_{\gamma})+\gamma}
  \\ & \quad +  \cp.
\end{align*}
From the definition of the bracket $\brac{\cdot}{\cdot}$ and the symmetric
pairing, we have
\begin{equation} \label{sub1}
  J_{L'}(\phi_0(\alpha),\phi_0(\beta), \phi_0(\gamma))=
 -\frac{1}{2}\ip{\gamma}\ip{\beta}\ip{\alpha}\omega 
-\frac{1}{6} \bigl( \iota_{[v_{\alpha},v_{\beta}]} \gamma-
\ip{\gamma}d\Phi(\alpha,\beta) + \cp \bigr).
\end{equation}
Lemma \ref{calc2} implies
\begin{equation} \label{lemma_implies}
\iota_{[v_{\alpha},v_{\beta}]} \gamma + \cp = -3 \ip{\alpha}\ip{\beta}\ip{\gamma}\omega 
- \bigl (2 \ip{\gamma}d\Phi(\alpha,\beta) + \cp \bigr),
\end{equation}
so  Eq.\ \ref{sub1} becomes
\[
 J_{L'}(\phi_0(\alpha),\phi_0(\beta), \phi_0(\gamma))=
\ip{\alpha}\ip{\beta}\ip{\gamma}\omega + \bigl(\frac{1}{2}
\ip{\gamma}d\Phi(\alpha,\beta) + \cp \bigr).
\]
By definition, $J_{L}(\alpha,\beta,\gamma)=\ip{\alpha}\ip{\beta}\ip{\gamma}\omega$.
Therefore, in this case, the left-hand side of Eq.\ \ref{coherence} is
\begin{equation}\label{LHS}
\phi_1(J_{L}(\alpha,\beta,\gamma)) - J_{L'}(\phi_0(\alpha),\phi_0(\beta), \phi_0(\gamma)) =
-\frac{1}{2}\ip{\gamma}d\Phi(\alpha,\beta) + \cp.
\end{equation}

Since the brackets and homotopy $\Phi$ are skew-symmetric, the right-hand side of Eq.\
\ref{coherence} can be rewritten as:
\begin{equation}\label{RHS}
\bigl (\Phi(\alpha,\lbrac{\beta}{\gamma}) + \cp \bigr)-
\bigl ( \lpbrac{\Phi(\alpha,\beta)}{\phi_{0}(\gamma)}  + \cp \bigr).
\end{equation}
Consider the first term in Eq.\ \ref{RHS}. The Hamiltonian vector field corresponding to
$\lbrac{\beta}{\gamma}=\brac{\beta}{\gamma}$ is
$[v_{\beta},v_{\gamma}]$. Therefore the definition of $\Phi$ implies
\[
\Phi(\alpha,\lbrac{\beta}{\gamma}) + \cp =-\frac{3}{2}\ip{\gamma}\ip{\beta}\ip{\alpha}\omega
+\frac{1}{2}\bigl( \iota_{[v_{\beta},v_{\gamma}]}\alpha + \cp\bigr).
\]
It then follows from Lemma \ref{calc2} (see Eq.\ \ref{lemma_implies}) that
\[
\Phi(\alpha,\lbrac{\beta}{\gamma}) + \cp = -  \ip{\gamma}d\Phi(\alpha,\beta) + \cp.
\]
By definition of the bracket $\lpbrac{\cdot}{\cdot}$, the second term
in Eq.\ \ref{RHS} can be written as
\[
\lpbrac{\Phi(\alpha,\beta)}{\phi_{0}(\gamma)}  + \cp = 
-\frac{1}{2} \ip{\gamma}d\Phi(\alpha,\beta) + \cp.
\]
Hence the coherence condition:
%\begin{multline*}
\[
\phi_1(J_{L}(\alpha,\beta,\gamma)) -
J_{L'}(\phi_0(\alpha),\phi_0(\beta), \phi_0(\gamma)) =\\
\Phi(\alpha,\lbrac{\beta}{\gamma}) - 
\lpbrac{\Phi(\alpha,\beta)}{\phi_{0}(\gamma)}  + \cp
\]
%\end{multline*}
is satisfied, and $(\phi,\Phi) \maps L_{\infty}(M,\omega)
\to L_{\infty}(C)$ is a morphism of Lie 2-algebras.
\end{proof}

We now focus on a particular sub-Lie 2-algebra of
$L_{\infty}(C)$. The following definition is due to
\v{S}evera \cite{Severa1} and is a generalization of Def.\ \ref{preserve_split_def}:
\begin{definition} \label{courant_conn_preserve_def}
Let $C$ be the exact Courant algebroid given in Prop.\ \ref{2plectic_courant}
equipped with a splitting $s \maps TM \to C$. We say 
a section $e=s(v)+\alpha$ {\bf preserves the splitting} iff 
$\forall v' \in \X(M)$
\[
\tdbrac{e}{s(v')}=s([v,v']).
\]
The subspace of sections that preserve the splitting is 
denoted as $\Gamma(C)^{s}$.
\end{definition}
Note that the twisted Dorfman bracket is used in the above definition
rather than the twisted Courant bracket. Since it satisfies the Jacobi
identity, it gives a `strict' adjoint action on sections of $C$. 
The 2-plectic analogue of Proposition \ref{preserve_lie_alg} is:
\begin{prop}
If $C$ is the exact Courant algebroid given in Proposition \ref{2plectic_courant}
equipped with the splitting $s \maps TM \to C$, then there is
a  Lie $2$-algebra
$L_{\infty}(C)^{s}=(L,\blankbrac, J)$ where:
\begin{itemize}
\item $L_{0}=\Gamma(C)^{s}$,
\item $L_{1}=\cinf(M)$,
\item the differential $L_{1} \stackrel{d}{\to} L_{0}$ is the de Rham differential
\item{the bracket $\sbrac{\cdot}{\cdot}$ is 
\[
[e_{1},e_{2}]= \tcbrac{e_{1}}{e_{2}} \quad  \text{in degree 0}
\]
and
\[
  [e,f]=-[f,e]=\half \innerprodp{e}{df} \quad \text{in degree 1},
\]
}
\item the Jacobiator is the linear map $J\maps \Gamma(C) ^{s} \tensor \Gamma(C) ^{s} \tensor \Gamma(C) ^{s} \to \cinf(M)$ defined by
\begin{align*} 
J(e_{1},e_{2},e_{3}) &= -T(e_{1},e_{2},e_{3})\\
&=-\frac{1}{6}
\left(\innerprodp{\tcbrac{e_1}{e_2}}{e_3} +
\innerprodp{\tcbrac{e_3}{e_1}}{e_2} \right. \\
 & \left.\quad + \innerprodp{\tcbrac{e_2}{e_3}}{e_1} \right). 
\end{align*}
\end{itemize}
\end{prop}
\begin{proof}
  Let $v'$ be a vector field on $M$. By the definition
  of the twisted Dorfman bracket (Eq.\ \ref{twist_bracket_dorf}), it follows
  that $\tdbrac{df}{s(v')}=0$ $\forall f\in \cinf(M)$. Hence the
  complex $L$ is well-defined. 
 We now show that
  $\Gamma^{s}(C)$ is closed under the twisted Courant
  bracket.
Suppose $e_{1}$ and $e_{2}$ are sections preserving the
splitting. Let $e_{i}=s(v_{i})+\alpha_{i}$. Since the twisted Dorfman
bracket and the Lie bracket of vector fields satisfy the Jacobi
identity, we have:
\[
\tdbrac{\tdbrac{e_{1}}{e_{2}}}{s(v')}=  s([[v_{1},v_{2}],v']).
\]
From  Eq.\ \ref{dorfman}, we have the identity:
\[
\tcbrac{e_{1}}{e_{2}}=\tdbrac{e_{1}}{e_{2}} - \half d\innerprodp{e_{1}}{e_{2}}.
\]
Therefore:
\begin{align*}
\tdbrac{\tcbrac{e_{1}}{e_{2}}}{s(v')} &= 
\tdbrac{\tdbrac{e_{1}}{e_{2}}}{s(v')} - \half \tdbrac{d\innerprodp{e_{1}}{e_{2}}}{s(v')}\\
&=s([[v_{1},v_{2}],v']).
\end{align*}
It follows from Theorem \ref{courant_L2A} that the Lie 2-algebra
axioms are satisfied.
\end{proof}

This next result is essentially a corollary of Thm.\
\ref{main_courant_thm}. However, it is important since
it is the 2-plectic analogue of Prop.\
\ref{lie_alg_iso}. 
\begin{theorem} \label{lie_2_alg_iso}
$L_{\infty}(M,\omega)$ and $L_{\infty}(C)^{s}$ are isomorphic
as Lie 2-algebras.
\end{theorem}
\begin{proof}
Recall that in Theorem \ref{main_courant_thm} we constructed a morphism of
Lie 2-algebras given by a chain map $\phi \maps L_{\infty}(M,\omega) \to
L_{\infty}(C)$:
\[
\phi_{0}(\alpha)=s(v_{\alpha}) + \alpha, \quad \phi_{1}=\id,
\]
and a homotopy $\Phi \maps \ham \tensor \ham \to \cinf(M)$:
\[
\Phi(\alpha,\beta)=-\half \innerprodm{v_{\alpha} + \alpha}{v_{\beta} + \beta}.
\]
Let $v' \in \X(M)$ and $e = s(v) + \alpha$. 
By definition of the twisted Dorfman bracket, $\tdbrac{e}{s(v')}=s[v,v']$ if and only if 
$\iota_{v'} \bigl(d\alpha +\iota_{v}\omega \bigr)=0$.
Hence a section  of $C$
preserves the splitting if and only if it lies in the image of the
chain map $\phi$. Since this map is also injective, the statement follows.
\end{proof}

Theorem \ref{lie_2_alg_iso} suggests that we interpret the Lie 2-algebra $L_{\infty}(C)^{s}$ as
the prequantization of the Lie 2-algebra of ``observables''
$L_{\infty}(M,\omega)$.  Clearly, these results further support the
idea that exact Courant algebroids play the role of higher Atiyah
algebroids \cite{Bressler-Chervov,Gualtieri:2007}.  However,
interpreting $L_{\infty}(C)^{s}$ as `operators' or as infinitesimal
symmetries of a $\U(1)$-gerbe with 2-connection is still a work in progress. It is likely
that significant progress would be made by solving the larger problem
of how to integrate an exact Courant algebroid to a Lie 2-groupoid.

\section{Central extensions of Lie 2-algebras} \label{2plectic_extend}
In this section, we constructing the 2-plectic
version of the Kostant-Souriau central extension, which we discussed
in Sec.\ \ref{symplectic_sec}.
First some preliminary definitions:
\begin{definition}
A Lie 2-algebra $(L,\blankbrac, J)$ is {\bf trivial} iff $L_{1}=0$.
\end{definition}
\noi Any Lie algebra $\g$ gives a trivial Lie
2-algebra whose underlying  complex is
\[
0 \to \g.
\]
In particular, the Lie algebra of Hamiltonian vector fields $\Xham$ is a trivial Lie 2-algebra.
\begin{definition}
A Lie 2-algebra $(L,\blankbrac, J)$
is {\bf abelian} iff $\blankbrac=0$ and $J=0$.
\end{definition}
\noi Hence an abelian Lie 2-algebra is just a 2-term chain complex. 

\begin{definition} \label{extend_def}
If $L$, $L'$, and $L''$ are Lie
2-algebras, then
$L'$ is a {\bf strict extension} of $L''$ by $L$
iff there exists Lie 2-algebra morphisms
\[
(\phi,\Phi) \maps L \to L', \quad
(\phi',\Phi') \maps L' \to L''
\]
such that the chain maps $\phi$, $\phi'$
give a short exact sequence of the underlying chain complexes
\[
L \stackrel{\phi}{\to} L' \stackrel{\phi'}{\to} L''.
\]
We say $L'$ is a {\bf strict central
  extension} of $L''$ iff
$L'$ is an extension of
$L''$ by $L$ and
\[
\left [ \im \phi, L' \right]' =0.
\]
\end{definition}
\begin{remark}
These definitions will be sufficient for our work here. However,
they are, in general, too strict. For example, one can have homotopies
between morphisms between Lie 2-algebras, and therefore we should
consider sequences that are only exact up to homotopy as
``exact''. Fully weak extensions for degree-wise finite-dimensional
Lie $n$-algebras have recently been described as particular homotopy
pushouts in the closed model category of differential graded (dg)
algebras \cite{Urs_nlab}.  The opposite of this model structure is
taken to be, by definition, a presentation of the
$(\infty,1)$-category of degree-wise finite-dimensional
$L_{\infty}$-algebras. For infinite-dimensional Lie $n$-algebras, such
as the ones we consider here, it is likely that one can find a
suitable definition in a similar manner by using a closed model
category structure on the category of dg co-algebras.
\end{remark}

We would like to understand how $L_{\infty}(M,\omega)$ is a central extension
of $\Xham$ as a Lie 2-algebra. Our first two results are quite general
and hold for any 2-plectic manifold $(M,\omega)$.
\begin{prop} \label{gen_extend}
If $(M,\omega)$ is a 2-plectic manifold, then
the Lie 2-algebra $L_{\infty}(M,\omega)$ is a central extension of the
trivial Lie 2-algebra $\Xham$ by the abelian Lie 2-algebra
\[
\cinf(M) \stackrel{d}{\to} \cOmega^{1}(M),
\]
consisting of smooth functions and closed 1-forms.
\end{prop}
\begin{proof}
Consider the following short exact sequence of complexes:
\begin{equation}\label{ses}
    \xymatrix{
        \cOmega^1(M) \ar[r]^{\jmath} & \ham \ar[r]^{p}  & \Xham\\
        \cinf(M) \ar[u]^{d}\ar[r]^{\id}  & \cinf(M) \ar[u]^{d}
        \ar[r]& 0 \ar[u]}
\end{equation}
The map $\jmath \maps \cOmega^{1}(M) \to \ham$ is the inclusion, and
\[
p \maps \ham \to \Xham, \quad p(\alpha)=v_{\alpha}
\]
takes a Hamiltonian 1-form to its corresponding vector field. It
follows from Prop. \ref{bracket_prop} that $p$ preserves the bracket. In fact,
all of the horizontal chain maps give strict Lie 2-algebra morphisms
(i.e.\ all homotopies are trivial). The Hamiltonian vector field
corresponding to a closed 1-form is zero. Thus,
if $\alpha$ is closed, then for all $\beta \in \ham$
we have $\sbrac{\alpha}{\beta}_{L_{\infty}(M,\omega)}=\brac{\alpha}{\beta}=0$.
Hence $L_{\infty}(M,\omega)$ is a central extension of $\Xham$. 
\end{proof}

\begin{prop} \label{L2A_extend_1} 
Let $(M,\omega)$ be a 2-plectic manifold. Given $x\in M$, there
is a Lie 2-algebra  $L_{\infty}(\Xham,x)=(L,\blankbrac,J_{x})$ where
\begin{itemize}
\item{$L_{0}=\Xham$,}
\item{$L_{1}=\R$,}
\item{the differential $L_{1} \stackrel{d}{\to} L_{0}$ is trivial ($d=0$),}
\item{the bracket $\blankbrac$ is the Lie bracket on $\Xham$
in degree 0 and trivial in all other degrees}
\item{the Jacobiator is the linear map 
\[
J_{x} \maps \Xham \tensor
    \Xham \tensor \Xham \to \R
\] 
defined by
\[
J_{x}(v_{1},v_{2},v_{3})= \ip{1}\ip{2}\ip{3}\omega\vert_{x}.
\]
}
\end{itemize}
Moreover, $J_{x}$ is a 3-cocycle in the Chevalley-Eilenberg cochain
complex \linebreak $\Hom(\Lambda^{\bullet}\Xham, \R)$.
\end{prop} 
\begin{proof}
  We have a bracket defined on a complex with trivial differential
  that satisfies the Jacobi identity ``on the nose''. Hence to show
  $L_{\infty}(\Xham,x)$ is a Lie 2-algebra, it sufficient to show that the
  Jacobiator $J_{x}(v_{1},v_{2},v_{3})$ satisfies Eq.\ \ref{big_J} in
  Def.\ \ref{L2A} for $x \in M$. This follows immediately from Thm.\
  \ref{semistrict}. The classification theorem of Baez and Crans
(Thm. 55 in \cite{HDA6}) implies that $J_{x}$ satisfying Eq.\
\ref{big_J} in the definition of a Lie 2-algebra is equivalent to
$J_{x}$  being a 3-cocycle with values in the trivial representation.
\end{proof}

Recall that in the symplectic case, if the manifold is connected, then 
the Poisson algebra is a central extension of the Hamiltonian vector
fields by the Lie algebra $\u(1) \cong \R$. The categorified analog of
the Lie algebra $\u(1)$ is the abelian Lie 2-algebra $b\u(1)$ whose
underlying chain complex is simply
\[
\R \to 0.
\]
It is natural to suspect that, under
suitable topological conditions, the abelian Lie algebra $\cinf(M)
\stackrel{d}{\to} \cOmega^{1}(M)$ introduced in Prop. \ref{gen_extend}
is related to $b\u(1)$. 

Let us first assume that the 2-plectic manifold is connected. Note
that the Jacobiator $J_{x}$ of the Lie 2-algebra $L_{\infty}(\Xham,x)$ 
introduced in Prop.\ \ref{L2A_extend_1} depends
explicitly on the choice of $x \in M$. However, if $M$ is connected,
then the cohomology class $J_{x}$ represents as a 3-cocycle does not
depend on $x$. This fact has important implications for $L_{\infty}(\Xham,x)$:

\begin{prop} \label{L2A_extend_2}
If $(M,\omega)$ is a connected 2-plectic manifold and 
$J_{x}$ is the 3-cocycle given in Prop.\ \ref{L2A_extend_1}, then
the cohomology class $[J_{x}]\in H^{3}_{\mathrm{CE}}(\Xham,\R)$
is independent of the choice of $x \in M$. Moreover, given any other point
$y \in M$, the Lie 2-algebras $L_{\infty}(\Xham,x)$ and
$L_{\infty}(\Xham,y)$ are quasi-isomorphic.
\end{prop}
\begin{proof}
To prove that $[J_{x}]$ is independent of $x$, we use a construction
similar to the proof of Prop.\ 4.1 in \cite{Brylinski:1990}. 
The Chevalley-Eilenberg differential
\[
\delta \maps \Hom(\Lambda^{n}
\Xham, \R)\to \Hom(\Lambda^{n+1} \Xham, \R)
\]
 is defined by
\[
(\delta c)(v_{1},\ldots,v_{n+1}) = \sum_{1 \leq i < j \leq n}
(-1)^{i+j}c([v_{i},v_{j}],v_{1},\cdots,
\hat{v}_{i},\cdots,\hat{v}_{j},\ldots,v_{n+1}).
\]
Note that if $c$ is an arbitrary 2-cochain then 
\[
(\delta c) (v_{\alpha},v_{\beta},v_{\gamma})   =-c([v_{\alpha},v_{\beta}],v_{\gamma}) +
c([v_{\alpha},v_{\gamma}],v_{\beta}) - c([v_{\beta},v_{\gamma}],v_{\alpha}).
\]
Now let $y \in M$. Let $\Gamma \maps [0,1] \to M$ be a path from $x$ to
$y$. Given $v_{\alpha},v_{\beta} \in \Xham$, define
\[
c(v_{\alpha},v_{\beta}) = \int_{\Gamma} \omega(v_{\alpha},v_{\beta},\cdot).
\]
Clearly, $c$ is a 2-cochain. We claim
\[
J_{y}(v_{\alpha},v_{\beta},v_{\gamma})
-J_{x}(v_{\alpha},v_{\beta},v_{\gamma}) = (\delta c) (v_{\alpha},v_{\beta},v_{\gamma}) 
\]
The failure of $\brac{\cdot}{\cdot}$ to satisfy the Jacobi identity implies
\[
d\ip{\alpha}\ip{\beta}\ip{\gamma}\omega = 
\brac{\alpha}{\brac{\beta}{\gamma}} -
    \brac{\brac{\alpha}{\beta}}{\gamma} 
    -\brac{\beta}{\brac{\alpha}{\gamma}},
\]
and, from the definition of $\brac{\cdot}{\cdot}$, we have
\[
d\ip{\alpha}\ip{\beta}\ip{\gamma}\omega = 
-\omega([v_{\alpha},v_{\beta}],v_{\gamma},\cdot)  
+\omega([v_{\alpha},v_{\gamma}],v_{\beta},\cdot)  
- \omega([v_{\beta},v_{\gamma}],v_{\alpha},\cdot).  
\]
Integrating both sides of the above equation gives
\begin{align*}
\int_{\Gamma} d\ip{\alpha}\ip{\beta}\ip{\gamma}\omega  &=
J_{y}(v_{\alpha},v_{\beta},v_{\gamma})
-J_{x}(v_{\alpha},v_{\beta},v_{\gamma})\\ 
&= -\int_{\Gamma}\omega([v_{\alpha},v_{\beta}],v_{\gamma},\cdot)  
+ \int_{\Gamma} \omega([v_{\alpha},v_{\gamma}],v_{\beta},\cdot)  
- \int_{\Gamma}\omega([v_{\beta},v_{\gamma}],v_{\alpha},\cdot)\\
%&=-c([v_{\alpha},v_{\beta}],v_{\gamma}) +
%c([v_{\alpha},v_{\gamma}],v_{\beta}) -
%c([v_{\beta},v_{\gamma}],v_{\alpha}) \\
&= (\delta c) (v_{\alpha},v_{\beta},v_{\gamma}).
\end{align*}

It follows from Thm.\ 57 in Baez and Crans \cite{HDA6} that
$[J_{x}]=[J_{y}]$ implies $L_{\infty}(\Xham,x)$ and $L_{\infty}(\Xham,y)$ 
are quasi-isomorphic (or `equivalent' in their terminology). 
\end{proof}

Now we impose further conditions on our 2-plectic manifold.  From here
on, we assume $(M,\omega)$ is 1-connected (i.e.\ connected and simply
connected). This is the 2-plectic analogue of the requirement that the
symplectic manifold in Sec.\ \ref{symplectic_sec} be connected. It will
allow us to construct several elementary, yet interesting,
quasi-isomorphisms of Lie 2-algebras.

\begin{prop} \label{L2A_extend_3}
If $M$ is a 1-connected manifold, then the abelian
Lie 2-algebra $\cinf(M) \stackrel{d}{\to} \cOmega^{1}(M)$ is
quasi-isomorphic to $b\u(1)$.
\end{prop}
\begin{proof}
Let $x \in M$. The chain map
\[
    \xymatrix{
        \cinf(M) \ar[d]_{\mathrm{ev}_{x}} \ar[r]^{d} & \cOmega^{1}(M) \ar[d]\\
        \R \ar[r]  & 0
}
\]
is a quasi-isomorphism.
\end{proof}

\begin{prop} \label{L2A_extend_4} 
If $(M,\omega)$ is a 1-connected
  2-plectic manifold and $x \in M$, then the Lie 2-algebras
  $L_{\infty}(M,\omega)$ and $L_{\infty}(\Xham,x)$ are quasi-isomorphic.
\end{prop} 

\begin{proof}
  We construct a quasi-isomorphism from $L_{\infty}(M,\omega)$ to
  $L_{\infty}(\Xham,x)$.  There is a chain map
\[
    \xymatrix{
        \cinf(M) \ar[d]_{\mathrm{ev}_{x}} \ar[r]^{d} & \ham \ar[d]^{p}\\
        \R \ar[r]^{0}  & \Xham
}
\]
with $\mathrm{ev}_{x}(f)=f(x)$ and $p(\alpha)=v_{\alpha}$. Since
$p$ preserves the bracket, we take $\Phi$ in Def. \ref{homo} to be
the trivial homotopy. Eq.\ \ref{coherence} holds
since:
\[
\mathrm{ev}_{x}(\omega(v_{\gamma},v_{\beta},v_{\alpha}))=J_{x}(v_{\alpha},v_{\beta},v_{\gamma}),
\]
and therefore we have constructed a Lie 2-algebra morphism. Since $M$
is connected, the homology of the complex $\cinf(M) \stackrel{d}{\to}
\ham$ is just $\R$ in degree 1 and $\ham/d\cinf(M)$ in degree 0.  The
kernel of the surjective map $p$ is the space of closed 1-forms, which
is $d\cinf(M)$ since $M$ is simply connected.
\end{proof} 

We can summarize the results given in Props.\ \ref{gen_extend}
\ref{L2A_extend_1} \ref{L2A_extend_3}, and \ref{L2A_extend_4} with
the following commutative diagram:
\[
\xymatrix{ \cOmega^{1}(M) \ar @{~>}[rd] \ar[rr]^{\jmath} && \ham
  \ar @{~>}[rd]_{p} \ar[rr]^{p} && \Xham \ar[rd]^{\id}\\
  &0 \ar[rr]  && \Xham \ar[rr] && \Xham \\
  \cinf(M) \ar'[r][rr] \ar[uu]^{d} \ar @{~>}[rd]_{\mathrm{ev}_{x}} && \cinf(M) \ar'[u][uu]
  \ar @{~>}[rd]_{\mathrm{ev}_{x}} \ar'[r][rr] && 
  0 \ar'[u][uu]\ar[rd] \\
  & \R \ar[uu] \ar[rr] && \R \ar[rr] \ar[uu] && 0 \ar[uu] }
\]
The back of the diagram shows $L_{\infty}(M,\omega)$ as the central extension of
the trivial Lie 2-algebra $\Xham$. The front shows 
$L_{\infty}(\Xham,x)$ as a central extension of $\Xham$ by
$b\u(1)$. The morphisms going from back to front are all
quasi-isomorphisms. Thus we have the 2-plectic analogue of the
Kostant-Souriau central extension:
\begin{theorem}
If $(M,\omega)$ is a 1-connected 2-plectic manifold, then
$L_{\infty}(M,\omega)$ is quasi-isomorphic to a central extension of
the trivial Lie 2-algebra $\Xham$ by $b\u(1)$.
\end{theorem}
Also, from Prop.\ \ref{lie_2_alg_iso} we know that
$L_{\infty}(M,\omega)$ is isomorphic to the Lie 2-algebra
$L_{\infty}(C)^{s}$ consisting of sections of the
Courant algebroid $C$ which preserve a chosen splitting $s
\maps TM \to C$. Therefore:
\begin{corollary}
If $(M,\omega)$ is a 1-connected 2-plectic manifold, then
$L_{\infty}(C)^{s}$ is quasi-isomorphic to a central extension of
the trivial Lie 2-algebra $\Xham$ by $b\u(1)$.
\end{corollary}
A comparison of the above corollary to
the results discussed in Sec.\ \ref{symplectic_sec} suggests that $L_{\infty}(C)^{s}$ be
interpreted as the quantization of $L_{\infty}(M,\omega)$ with
$b\u(1)$ giving rise to the quantum phase.

Finally, note that a splitting of the short exact sequence of complexes 
\[
    \xymatrix{
        0 \ar[r] & \Xham \ar[r]^{\id}  & \Xham\\
        \R \ar[u]\ar[r]^{\id}  & \R \ar[u]^{0} 
        \ar[r]& 0 \ar[u]}
\]
is the identity map in degree 0 and the trivial map in degree
1. Obviously the splitting preserves the bracket but does not preserve
the Jacobiator. Indeed, the failure of the splitting to be a strict
Lie 2-algebra morphism between $\Xham$ and $L_{\infty}(\Xham,x)$ is
due to the presence of the 3-cocycle $J_{x}$.

\section*{Summary}
Many new results have been given in this chapter, so we conclude with
a brief summary. We defined a
prequantized $n$-plectic manifold to be an integral $n$-plectic
manifold equipped with Deligne $n$-cocycle.
If $(M,\omega)$
is a 0-connected, prequantized symplectic manifold, then 
there exists a principal $U(1)$-bundle over $M$ 
equipped with a connection whose curvature is $\omega$, and a corresponding Atiyah
algebroid $A \to M$ equipped with a splitting such that
the Lie algebra of sections of $A$ which preserve the splitting
is isomorphic to a central extension of the Lie algebra of Hamiltonian
vector fields:
\[
\u(1) \to \cinf(M) \to \Xham.
\]
This central extension gives a cohomology class in
$H^{2}_{\mathrm{CE}}(\Xham,\R)$ which can be represented by the
symplectic form evaluated at a point in $M$.

Analogously, if $(M,\omega)$
is a 1-connected, prequantized 2-plectic manifold, then 
there exists a $U(1)$-gerbe over $M$ 
equipped with a connection and curving whose 3-curvature is $\omega$, 
and a corresponding exact Courant algebroid $C \to M$
equipped with a splitting such that
the Lie 2-algebra of sections of $C$ which preserve the splitting
is quasi-isomorphic to a central extension of the (trivial) Lie 2-algebra of Hamiltonian
vector fields:
\[
b\u(1) \to L_{\infty}(\Xham) \to \Xham.
\]
This central extension gives a cohomology class in
$H^{3}_{\mathrm{CE}}(\Xham,\R)$ which can be represented by the
2-plectic form evaluated at a point in $M$.

\chapter{Geometric quantization of 2-plectic manifolds} 
\label{quantization_chapter}
In the previous chapter, we first considered prequantization for symplectic
manifolds, and then generalized the procedure to 2-plectic manifolds.
We were primarily concerned with prequantizing the algebra of
observables, i.e.\ the Poisson algebra in the symplectic case, and the
Lie 2-algebra of Hamiltonian 1-forms in the 2-plectic case.
In this chapter, we switch our focus from quantizing observables to 
quantizing states.

Prequantization is a simple and elegant construction. However,
numerous examples in symplectic geometry show that it
is only the first step of a two-part process. Full quantization
involves using additional structures in order to construct the correct
space of quantum states. This process was developed over time by
considering particular examples. We
suspect that the development of a complete geometric quantization procedure  for
2-plectic manifolds will follow a similar
path. In this chapter, we generalize aspects of the quantization process
for symplectic manifolds to the 2-plectic case by
using the higher geometric structures introduced
in earlier chapters. The result is a simple procedure 
for quantizing 2-plectic manifolds, which we apply to
a particular example of interest. To the best of our knowledge, this
is the first geometric quantization procedure ever developed for such manifolds.

Let us provide some motivation for why additional work
beyond prequantization is needed in order to obtain the correct
quantum states.
In the last chapter, we described a prequantized symplectic manifold
as a symplectic manifold equipped with principal $\U(1)$-bundle with connection.
A natural choice for the  quantum state space is
the space of square-integrable global sections of the Hermitian line
bundle associated to the principal bundle. This is often called the
`prequantum Hilbert space'. It comes equipped with an inner product given by
integrating the fiber-wise Hermitian inner product of sections with
respect to the symplectic volume form. However, from the physicist's point of
view, this space is too large to be the space of quantum states of a
physical system.
% However, global sections of $L$ can have arbitrarily
% small support, and hence, when interpreted as wavefunctions on a
% classical phase space, give probability densities which violate the
% Heisenberg uncertainty condition. Hence, 

For example, recall that the cotangent bundle of a manifold is a
symplectic manifold, equipped with its canonical symplectic
structure $\omega = \sum_{i}dp_{i} \wedge dq^{i}$. It is, in fact, an
integral symplectic manifold since $\omega$ is exact. The
sections in the prequantized Hilbert space locally look like functions $f(q^{i},p_{i})$ of $2n$
variables corresponding to the ``position'' coordinates $q^{i}$ of the base manifold 
and the ``momentum'' coordinates $p_{i}$ of the fibers. 
These functions can have arbitrarily
small support, and hence, when interpreted as wavefunctions on a
classical phase space, give probability densities which violate the
Heisenberg uncertainty condition. To get around this problem, one
reduces the size of the Hilbert space by taking
the subspace consisting of those sections satisfying
$\partial f/\partial p_{i}=0$. Hence,  the number of
``variables'' is reduced from $2n$ to $n$, by only considering those
sections constant along the fibers.

Consider another example that is perhaps more mathematically interesting.
The coadjoint orbits of the Lie group $\SU(2)$ correspond to 2-spheres
centered about the origin in
$\su(2)^{\ast} \cong \R^3$. Each orbit is a symplectic manifold
equipped with what is known as the `KKS symplectic
form'. This 2-form is integral if the radius of the sphere is one-half of a non-negative integer.
On each integral orbit, we have the prequantized Hilbert space,
consisting of global square-integrable sections of a
Hermitian line bundle. This Hilbert space is infinite
dimensional. However, we can equip the orbit with a
complex structure and consider only holomorphic sections i.e.\ those
sections which locally are functions $f(z^{i},\bar{z}^{i})$ satisfying
$\partial f/ \partial \bar{z}^i=0$. This smaller space of holomorphic sections is much more interesting.
First, it is finite-dimensional. Moreover, it is an irreducible
representation of $\SU(2)$. This way of obtaining representations from
coadjoint orbits by geometric quantization is quite general, and is known as Kirillov's
orbit method \cite{Kirillov:2004}. Note that, again, the size of the
prequantum space is reduced by decreasing the number of variables.

Hence, it is important to consider prequantized symplectic manifolds
equipped with additional structure in order to cut down the number of
admissible sections in the prequantum Hilbert space. In both of the
above examples, the extra structure corresponds to a special
integrable distribution called a `polarization'. We introduced
real $k$-polarizations for $n$-plectic manifolds in Chapter
\ref{nplectic_geometry} precisely for this reason, and we see that
real 1-polarizations appeared in our first example. The second
example employed the use of a `complex polarization'. These 
structures certainly play an important role in symplectic geometry \cite[Chap.\
5]{Woodhouse:1991}. Unfortunately, it is not yet clear how to generalize them to the
$n$-plectic case. Hence, we only will consider real $k$-polarizations
for $n>1$.

For symplectic manifolds, the output from quantization is a Hilbert space of quantum
states. As we will see, the output from quantizing a 2-plectic manifold
is a category of quantum states. In the last section of
this chapter, we consider in detail an example
in which the states correspond to objects in a representation
category. This suggests that 2-plectic quantization can categorify Kirillov's
orbit method.

\section{Geometric quantization of symplectic manifolds} \label{geo_quant_sym_sec}
As usual, it is instructive to consider the symplectic case first.
Consider a prequantized symplectic manifold $(M,\omega,\xi)$, where
$\xi$ is a Deligne 1-cocycle. Recall from Example \ref{n=1} in Chap.\ \ref{stacks_chapter}
that $\xi=(g,\theta)$ is specified by an open cover $\{U_{i}
\}$ of $M$, local 1-forms $\theta_{i} \in \Omega^{1}(U_{i})$, and $\U(1)$-valued functions
$g_{ij} \maps U_{i} \cap U_{j} \to \U(1)$ satisfying certain cocycle conditions.
In this chapter, we realize this 1-cocycle as the transition functions and local
connection forms of a Hermitian line bundle $(L,\inp{\cdot}{\cdot})$ equipped with a 
connection $\conn$. We let $\Gamma(L)_c$ denote the smooth sections of $L$ with compact support.
The prequantum Hilbert space is defined to be the completion of $\Gamma(L)_{c}$ with
respect to the inner product $(\sigma_1,\sigma_2)=\int_{M}\inp{\sigma_1}{\sigma_2}\omega^{n}$.

Recall from Def.\ \ref{k-polarization} that a real polarization on $M$
is a foliation $F$ of $M$ whose leaves are immersed Lagrangian submanifolds.
\begin{definition}
A {\bf quantized symplectic manifold} is a prequantized symplectic
manifold $(M,\omega,\xi)$ equipped with a real polarization $F$.
\end{definition}

\subsection{The Bohr-Sommerfeld variety} 
We use Deligne cocycles in some parts of this section,
rather than the more traditional language of bundles, in order to make
the analogy with the 2-plectic case as clear as possible.
In the 2-plectic case, we use Deligne cocycles, rather than
stacks directly, since the cocycles behave better under pullbacks and restrictions.

Given a quantized symplectic manifold $(M,\omega,\xi,F)$, let $D_{F} \ss TM$
denote the corresponding involutive distribution. A good candidate for
the quantum Hilbert space is the space constructed from those
sections of $\Gamma(L)_{c}$ which are covariantly constant along each
leaf of $F$:
\[
H= \bigl \{ \sigma \in \Gamma(L)_{c} ~ \vert ~ \conn_{v}\sigma =0 ~ \forall v \in \Gamma(D_{F}) \bigr \}.
\]

Unfortunately, the topology of the leaves will often force this space to be trivial.
For example, if the leaves of the foliation
are not compact, then we must have
$\sigma=0$ for all $\sigma \in H$. Otherwise, the integral of $\inp{\sigma}{\sigma}\omega^{n}$ will diverge.

There are additional topological obstructions which are more interesting.
Let $\Lambda \subseteq M$ be a leaf of the foliation $F$.
Since the restriction $(L \vert_{\Lambda},\conn \vert_{\Lambda})$
is a flat Hermitian line bundle, it is completely determined by its
holonomy representation 
\[
\oint\conn \vert_{\Lambda}\maps \pi_{1}(M) \to \U(1).
\] 
If $\sigma$ is a section of $L$ which is covariantly
constant along $F$, then $\sigma \vert_{\Lambda}$ is a covariantly
constant global section of $(L \vert_{\Lambda},\conn
\vert_{\Lambda})$. Hence $\sigma \vert_{\Lambda}$ is either
zero, or $(L \vert_{\Lambda},\conn \vert_{\Lambda})$ is the
trivial bundle with trivial connection, i.e.\  $\oint
\conn \vert_{\Lambda}=1$. 

So, we should consider only the leaves on which the restricted bundle has trivial
holonomy. In the language of Section \ref{holonomy_sec}, 
these are the leaves $\Lambda \stackrel{i}{\to}M$ with the property
that given a map $\sigma \maps S^{1} \to \Lambda$, the corresponding
holonomy (Def.\ \ref{holonomy_def}) of the  Deligne 1-cocycle $\xi \vert_{\Lambda}=i^{\ast} \xi$ is
trivial: $\hol(\xi \vert_{\Lambda},\sigma)=1$.
This leads us to the following definition.
\begin{definition}
Let $(M,\omega,\xi,F)$ be a quantized symplectic manifold. The
{\bf Bohr-Sommerfeld variety} {\boldmath $V_{\mathrm{BS}}$} associated to $F$ is the
union of all leaves $\Lambda$ of $F$ which satisfy 
\[
\hol(\xi \vert_{\Lambda},\sigma)=1
\]
for all maps $\sigma \maps S^{1} \to \Lambda$.
\end{definition}

The relation with the Bohr-Sommerfeld conditions from physics comes from the fact that
$\Lambda$ is contained in the Bohr-Sommerfeld variety if and only if for every loop
$\gamma$ in $\Lambda \cap U_{i}$:
\[
\exp \left (\i \oint_{\gamma} \theta_{i} \right)=1 \Leftrightarrow
\oint_{\gamma} \theta_{i}= 2\pi n_{\gamma}, \quad n_{\gamma} \in \Z,
\]
where $\theta_{i}$ is the local connection 1-form on $U_{i}$.

The use of Bohr-Sommerfeld varieties in geometric quantization was
developed considerably by \'{S}niatycki \cite{Sniatycki:1975}. He showed
that the correct quantum Hilbert space is the completion of the space of
sections of $L \vert_{V_{\mathrm{BS}}}$ which are covariantly constant
along each leaf contained in the variety. In general, 
such a section will not be the pullback of a global smooth section of $L \to
M$. Instead, it corresponds to a `distributional section' of $L$ 
\cite{Sniatycki:1975}[Sec.\ 5]. \'{S}niatycki's work motivates the
next definition.
\begin{definition}
Let $(M,\omega,\xi,F)$ be a quantized symplectic manifold,
$V_{\mathrm{BS}}$ be the corresponding Bohr-Sommerfeld variety, and
$L \vert_{V_{\mathrm{BS}}}$ be Hermitian line bundle associated to the
Deligne 1-cocycle $\xi \vert_{V_{\mathrm{BS}}}$.
The {\bf quantum state space} $Q(V_{\mathrm{BS}})$is the space of sections of 
$L \vert_{V_{\mathrm{BS}}}$ which are covariantly constant
along each leaf contained in $V_{\mathrm{BS}}$.
\end{definition}

\subsection{Example: $\R^2\setminus\{0\}$} \label{harm_osc_sec}
In this example, we will construct the quantum state space associated to
the punctured plane $M=\R^2\setminus\{0\}$ equipped with the 2-form 
\[
\omega = r dr \wedge dt, 
\]
with $0 < r < \infty, ~ 0 \leq t < 2\pi.$
Since $\omega =d\theta$, where 
\[
\theta=H dt, \quad H=\frac{1}{2}r^2,
\]
we see $(M,\omega)$ is an integral symplectic manifold. Hence $\theta$
is a connection 1-form on the trivial Hermitian line bundle $L=M
\times \C$. 

There is an obvious foliation $F$ of $M$ whose leaves are
concentric circles of radius $R >0$ about the origin. 
Since $\omega$ is a volume form on $M$, our discussion in
Example \ref{Rn_example} implies $F$ is a polarization.
The corresponding distribution $D_{F}$ is the vector field $\partial/\partial t$.

Let us first consider global sections of $L$ covariantly constant
along the leaves of $F$ in order to see why the Bohr-Sommerfeld
variety enters the picture.
Such a section $\psi$ must satisfy:
\[
\conn_{\del/\del t} \psi =0.
\]
Since $\conn = d + \i \cdot \theta$, this is equivalent to
$\psi$ satisfying the differential equation
\[
\frac{\del \psi}{\del t} = -\frac{\i}{2}r^2\psi,
\]
which has solutions of the form 
\[
\psi(r,t) = \exp(-\frac{\i}{2}r^2 t) g(r).
\]
However, such a solution must also satisfy:
\[
\psi(r,t)=\psi(r,t+2\pi).
\]
Hence, $\psi(r,t)$ must vanish if $\frac{r^2}{2}$ is not an
integer, and therefore no non-trivial smooth solution exists.

Now let us consider the Bohr-Sommerfeld variety associated to $F$.
Let the leaf $S^{1}_{R}$ correspond to a circle of radius $R$.
The Bohr-Sommerfeld condition implies
\[
\oint_{S^{1}_{R}} \theta = \frac{1}{2}R^2\int^{2 \pi}_{0} d t \in 2\pi\Z. 
\]
Hence, the variety corresponds to the integer level sets of $H$:
\[
V_{\mathrm{BS}}= \bigcup_{n \in \N ^{+}} H^{-1}(\{n\}),
\]
and the quantum state space $Q(V_{\mathrm{BS}})$ consists of linear combinations of
functions of the form
\[
\psi_{n}(t)=\exp(-nt \i) g(\sqrt{2n}).
\]

This quantized symplectic manifold is closely related to the
quantization of the simple harmonic oscillator. We can interpret $M$
as the classical phase space of the oscillator, and $H$ as a
Hamiltonian function which measures the energy of the oscillator.
It takes the familiar form $H=\frac{1}{2}(p^2 +q^2)$ in cartesian coordinates. 
The level sets of $H$ are the leaves
of the foliation and correspond to the classically allowed states in
phase-space with constant energy $\frac{1}{2}R^2$. 
The Bohr-Sommerfeld condition restricts the allowed 
states of the oscillator to those in $V_{\mathrm{BS}}$ thereby
quantizing the energy of the oscillator. The quantum values for
the energy are the non-negative integers. The sections $\psi_{n}(t)$ represent the
quantum states which satisfy the
Schr\"{o}dinger equation
\[
\i \cdot  \frac{\del \psi}{\del t} = \hat{H} \psi.
\]

Strictly speaking, this is not the correct quantization of the simple
harmonic oscillator, since its quantum energy states are actually $n + 1/2$.
Obtaining these shifted values for the energy requires using a more sophisticated
approach involving the `meta-plectic correction'
\cite{Sniatycki:1975}, \cite{Woodhouse:1991}[Ch.\ 10]. 

\section{Categorified geometric quantization} \label{cat_quant_section}
Now we present the 2-plectic analogue of the previously discussed
quantization process.
We start with a prequantized 2-plectic manifold $(M,\omega,\xi)$, where
$\xi$ is a Deligne 2-cocycle. 
From Example \ref{n=2} in Chap.\ \ref{stacks_chapter}, we know
that $\xi=(g, A, B)$ is specified by an open cover
$\{U_{i}\}$ of $M$, local 2-forms $B_{i} \in \Omega^{2}(U_{i})$,
local 1-forms $A_{ij} \in \Omega^{1}(U_{i} \cap U_{j})$,  
and $\U(1)$-valued functions
$g_{ijk} \maps U_{i} \cap U_{j} \cap U_{k} \to \U(1)$ satisfying certain cocycle conditions.
Recalling Definition \ref{2-line_stack_assoc},
we realize this cocycle as the 2-line stack $\HVB^{\xi}$ equipped with a
2-connection. 

We defined real $k$-polarizations for
$n$-plectic manifolds in Def.\ \ref{k-polarization}. Recall that,
unlike the symplectic case, there are several ways to define
orthogonal complements for $n$-plectic manifolds. Hence, there are
different ways to generalize the notion of Lagrangian submanifold, and therefore
real polarization, to the $n$-plectic case. For the 2-plectic case,
we can consider either 1-polarizations or 2-polarizations.
Regardless, the definitions in the previous section for symplectic
manifolds naturally generalize:
\begin{definition} \label{quant_2-plectic_def}
A {\bf quantized 2-plectic manifold} is a prequantized 2-plectic
manifold $(M,\omega,\xi)$ equipped with a real $k$-polarization $F$.
\end{definition}
\noi The next definition uses the notion of 2-holonomy for a Deligne 2-cocycle (Def.\
\ref{2-hol_2-conn_def}). 
\begin{definition}\label{2-plectic_BS_def}
Let $(M,\omega,\xi,F)$ be a quantized 2-plectic manifold. The
{\bf Bohr-Sommerfeld variety} {\boldmath $V_{\mathrm{BS}}$} associated to $F$ is the
union of all leaves $\Lambda$ of $F$ which satisfy 
\[
\hol(\xi \vert_{\Lambda},\sigma)=1
\]
for all maps $\sigma \maps \Sigma^{2} \to \Lambda$,
where $\Sigma^{2}$ is a compact, oriented 2-manifold.
\end{definition}

The Bohr-Sommerfeld variety is, by construction, a disjoint union of
immersed submanifolds in $M$. The inclusion map
$V_{\mathrm{BS}} \xrightarrow{i} M$ is smooth, and we can pull-back
the Deligne 2-cocycle $\xi$ to $V_{\mathrm{BS}}$. If $\xi$ is defined
with respect to an open cover $\{U_{i}\}$ of $M$, then $\xi
\vert_{V_{\mathrm{BS}}}$ is a 2-cocycle with respect to the cover
$\{U_{i} \cap V_{ \mathrm{BS}} \}$.
In analogy with the
symplectic case, we consider global sections of the 2-line stack
$\HVB^{\xi}$ over $V_{\mathrm{BS}}$, where by $\xi$ we mean $\xi \vert_{V_{\mathrm{BS}}}$.
Proposition
\ref{2-line_stack_conn_prop} implies that the category of such global
sections is equivalent to the category of $\xi
\vert_{V_{\mathrm{BS}}}$-twisted Hermitian vector bundles over
$V_{\mathrm{BS}}$. In Definition \ref{twisted_flat_def}, we described what it means for a twisted
bundle to be twisted-flat. We interpret twisted-flatness to be the 2-plectic
analogue of covariantly constant. 

Let $(E_{i}, \conn_{i},\phi_{ij})$ be a $\xi
\vert_{V_{\mathrm{BS}}}$-twisted Hermitian vector
bundle over the Bohr-Sommerfeld variety. Recall from Def.\ \ref{twistedvb_conn},
that such a bundle is given by the following data: Over
each open set $V_{i}=U_{i} \cap V_{\mathrm{BS}}$,
a Hermitian vector bundle with connection $(E_{i},\conn_{i})$, 
and, over each intersection $V_{i} \cap V_{j}$,
an isomorphism $\phi_{ij}$ between the pullbacks of bundles $E_{j}$
and $E_{i}$. The isomorphisms $\phi_{ij}$ are required to satisfy
compatibility relations with the 1-forms $A_{ij} \vert_{V_{\mathrm{BS}}}$ on $V_{i} \cap
V_{j}$, and with the $\U(1)$-valued functions $g_{ijk}
\vert_{V_{\mathrm{BS}}}$ on $V_{i} \cap V_{j} \cap V_{k}$.

We can pull this twisted bundle
back to any leaf $\Lambda \ss V_{\mathrm{BS}}$ in the obvious
way, resulting in a bundle twisted by
$\xi \vert_{\Lambda}=(g \vert_{\Lambda}, A \vert_{\Lambda}, B \vert_{\Lambda})$. It is twisted-flat iff
the equality
\[
\conn^2_{i} \vert_{\Lambda} - \i \cdot B_{i} \vert_{\Lambda} \tensor \id=0.
\]
holds for all $i$. 
Twisted bundles satisfying the above for all leaves
$\Lambda \ss V_{\mathrm{BS}}$ form a full subcategory of $\HVB^{\xi}(V_{\mathrm{BS}})$.
Hence, we have a categorified analogue of the quantum state space:
\begin{definition} \label{quant_state_cat_def}
Let $(M,\omega,\xi,F)$ be a quantized 2-plectic manifold and
$V_{\mathrm{BS}}$ be the corresponding Bohr-Sommerfeld variety.
The {\bf quantum state category} $\Quant(V_{\mathrm{BS}})$ is the 
subcategory of $\HVB^{\xi}(V_{\mathrm{BS}})$
consisting of twisted Hermitian vector bundles that are twisted-flat
along each leaf contained in ${V_{\mathrm{BS}}}$.
\end{definition}

\subsection{Example: $\R^3\setminus\{0\}$} \label{cat_quant_example_sec}
In this section, we consider an example in detail which will reveal
several interesting aspects of our quantization procedure for
2-plectic manifolds. We construct the quantum state category
associated to the manifold $M=\R^3\setminus\{0\}$ equipped with the
2-plectic form
\[
\omega = \frac{1}{r^2} dx^{1} \wedge dx^{2} \wedge dx^{3},
\] 
where $r$ is given by the usual Euclidean norm.
In analogy with the example involving the symplectic manifold
$\R^{2}\setminus\{0\}$, we will see how the Bohr-Sommerfeld variety is
used to overcome certain topological obstructions.

One reason for considering the 3-form $\omega$ is because 
$\omega = d B$, where
\[
B = \frac{1}{r^{2}} (x dy\wedge dz + y dz \wedge dx + z dx \wedge
dy).
\]
Restricting the 2-form $B$ to a sphere centered about the origin gives
the famous KKS symplectic form. We mentioned this symplectic structure
and the role it plays in representation theory in the introduction to
the chapter. We shall make use of this fact later on in Sec.\ \ref{rep_theory}.

We prequantize $(M,\omega)$ with the Deligne 2-cocycle
$\xi=(1,0, B)$. More precisely, we choose a good open cover $\{U_{i} \}$ of $M$ and
consider $\xi$ as the restriction of $\xi$ to this cover.

\subsubsection*{Global sections of $\HVB^{\xi}$}
Before we proceed further, let us characterize the global sections of
$\HVB^{\xi}$ i.e.\ $\xi$-twisted Hermitian vector bundles over $M$.
Let $(E_{i},\conn_{i},\phi_{ij})$ be such a bundle. Since $\xi=(1,0,
B)$ projects to the trivial class in $H^{2}(M,\sh{U(1)})$, we are dealing
with trivially twisted vector bundles with connection.
Let $\HVB_{\conn}(M)$ denote the category whose objects are Hermitian
vector bundles over $M$ equipped with connection.
The following proposition says we can identify trivially twisted
bundles with ordinary bundles.
\begin{prop} \label{just_conn_prop}
The categories $\HVB^{\xi}(M)$ and $\HVB_{\conn}(M)$ are equivalent.
\end{prop} 
\begin{proof}
Since $\xi=(1,0,\i \cdot B)$, Def.\ \ref{twistedvb_conn} implies that an object of
$\HVB^{\xi}(M)$ is given by a Hermitian vector bundle with connection
$(E_{i},\conn_{i})$ on each $U_{i}$, an isomorphism $\phi_{ij} \maps
E_{j} \vert_{U_{ij}} \iso E_{i} \vert_{U_{ij}}$, which preserves the
connection $\phi_{ij} \conn_{j}  =
\conn_{i} \phi_{ij}$ on $U_{ij}$, such that
$\phi_{ik}^{-1} \phi_{ij} \phi_{jk}=1$ on $U_{ijk}$ 
A morphism $(E_{i},\conn_{i},\phi_{ij}) \to (E'_{i},\conn'_{i},\phi'_{ij})$ is given by a
collection of bundle morphisms $E_{i}
\xrightarrow{f_{i}} E'_{i}$ which preserve the connection $\conn'_{i}
f_{i} = f_{i}\conn_{i}$, satisfying $f_{i} \phi_{ij} = \phi'_{ij}
f_{j}$ on $U_{ij}$. 

Now, consider the functor $F \maps \HVB_{\conn}(M) \to \HVB^{\xi}(M)$ which
sends a vector bundle $(E,\conn)$ to the trivially twisted bundle $(E
\vert_{U_{i}},\conn \vert_{U_{i}}, \phi_{ij}=\id)$, and a morphism $f$
to its restriction on each $U_{i}$. We shall show $F$ is full, faithful,
and essentially surjective, and hence gives an equivalence of categories.
For essential surjectivity, we must show that given $(E_{i},\conn_{i},\phi_{ij})$
there exists an object $(E,\conn)$ such that $F(E,\conn)$ is isomorphic
to $(E_{i},\conn_{i},\phi_{ij})$. By unraveling Def.\ \ref{stack_def} for a
stack, we see that the above data for a trivially twisted bundle implies there exists a Hermitian
vector bundle with connection $(E,\conn)$ on $M$ and connection
preserving isomorphisms $E \vert_{U_{i}} \xto{\psi_{i}} E_{i}$ on $U_{i}$ such
that $\phi_{ij} \psi_{j} = \psi_{i}$ on $U_{ij}$. Hence, the $\psi_{i}$
give an isomorphism in $\HVB^{\xi}(M)$ between $F(E,\conn)$ and $(E_{i},\conn_{i},\phi_{ij})$.

It's clear that $F$ is faithful (i.e.\ injective on morphisms). For
fullness, we must show $F \maps \Hom(E,E') \to \Hom(F(E),F(E'))$ is
surjective. Let $E \vert_{U_{i}} \xto{f_{i}} E' \vert_{U_{i}}$ denote
a morphism between $F(E)$ and $F(E')$. Since
$\phi_{ij}=\phi'_{ij}=\id$, it follows from the definition of morphism
that $f_{i}=f_{j}$ on each $U_{ij}$. Since morphisms between bundles
form a sheaf, there exists a unique global morphism $E \xto{f} E'$
such that $f \vert_{U_{i}} =f_{i}$. Hence, the proposition is proven.
\end{proof}

\subsubsection*{Topological considerations}
There is an obvious foliation $F$ of $M$ whose leaves $S^{2}_{R}$ are
concentric spheres of radius $R >0$ about the origin. 
Since $\omega$ is a volume form on $M$, our discussion in
Example \ref{Rn_example} implies that $F$ is a 2-polarization.
Hence, $(M,\omega,\xi,F)$ is a quantized 2-plectic manifold.

% The restriction
% of $B$ to any leaf is obviously closed, and is a volume form:
% \[
% \int_{S^{2}_{R}} B \vert_{R} = 4 \pi R
% \]

To see why the Bohr-Sommerfeld variety is needed,
let us consider global sections of $\HVB^{\xi}$ which are twisted-flat
along the leaves of $F$. By Prop.\ \ref{just_conn_prop}, any global
section can be thought of as a Hermitian vector bundle $E\to M$ with
connection $\conn$. Let $E \vert_{R}$ denote the restriction of this bundle
to a leaf $S^{2}_{R}$. By definition, $E \vert_{R}$ is twisted-flat if its
curvature satisfies
$\conn^{2} \vert_{R} = \i \cdot B \vert_{R} \tensor \id$. 
%If $\Omega_{ij}$ is a local curvature 2-form for $\conn^{2}$, then
% the Chern classes of $E_{R}$ can be determined from the characteristic
% equation
% \[
% \det( I + t \Omega_{ij}) = \det( I + tB \cdot I).
% \]
% In particular, the first Chern class is 
% \[
% c_{1}(E_{R})= \frac{1}{2 \pi} \mathrm{rank}(E) \cdot B \vert_{R}.
% \]
% Since representatives of Chern classes are integral forms, we 
% have the condition
% \[
% \int_{S^{2}_{R}} c_{1}(E_{R}) = 2 R \cdot \mathrm{rank}(E) \in \Z.
% \]
% Hence, it is clear that we can only consider leaves of $F$ whose radii $R$
% satisfy the above equality. 

% \subsubsection*{Reduction of bundles}
% There are, in fact, further obstructions to obtaining twisted flat
% sections along the leaves of $F$ which are even stronger than those
% given above. 
The next proposition implies $B \vert_{R}$ 
must be an integral 2-form.
\begin{prop} \label{just_lines_prop}
If $E$ is a rank $n$ Hermitian vector bundle with connection $\conn$
on $S^{2}\vert_{R}$ with curvature $\conn^{2}= \i\cdot B \vert_{R} \tensor \id$, then there is an
isomorphism of bundles
\[
E \iso  L_{1} \oplus L_{2} \oplus \cdots \oplus L_{n}
\]
where $L_{i}$ is a Hermitian line bundle with connection whose curvature
2-form is $B \vert_{R}$. Moreover, the $L_{i}$'s are all isomorphic as line
bundles with connection.
\end{prop}
\noi The proposition can be proven using classical results from
differential geometry. Let $P$ be a principal $G$-bundle over a connected manifold
$M$ equipped with a  $\g$-valued connection 1-form $\theta$. 
Such a bundle is said to be {\bf reducible} to a principal $G'$-bundle
$P' \xto{\iota} P$ iff $G'$ is a subgroup of $G$,
and the inclusion map $\iota$ commutes with the group action of
$G'$. The connection $\theta$ reduces to a connection on $P'$
iff its pullback along the inclusion takes values in the Lie algebra
of $G'$.

Given  $p \in P$, let $H(p)$ denote the set of points in $P$ which are
joined by a piece-wise smooth horizontal path in $P$.
Let $\mathrm{Hol}_{p}(\theta)$ be the holonomy group based at $p \in P$
i.e.\ the subgroup of $G$ consisting of elements $g$ such that $p$ and
$pg$ are joined by a piece-wise smooth horizontal loop in
$P$. Similarly, let $\mathrm{Hol}^{0}_{p}(\theta)$ be the subgroup
consisting of those $g$ such that $p$ and $pg$ are connected by a
contractible horizontal loop. Both of these subgroups are, in fact,
Lie subgroups.
The following is Theorem 7.1 in Kobayashi-Nomizu \cite{Kobay-Nomizu}.
\begin{theorem}[Reduction Theorem]\label{reduction_thm}
A principal $G$-bundle $P$ with connection $\theta$ is reducible to a
principal bundle with total space $H(p)$ and structure group 
$\mathrm{Hol}_{p}(\theta)$. Furthermore, $\theta$ reduces to a connection
on $H(p)$.
\end{theorem}
\noi Next, we recall the Ambrose-Singer Theorem.
\begin{theorem}[\cite{Ambrose-Singer}]
If $\Omega$ is the curvature 2-form of a principal $G$-bundle $P$
with connection $\theta$, then the Lie algebra of $\mathrm{Hol}_{p}(\theta)$
is the subspace of $\g$ spanned by all elements of the form
$\Omega_{q}(v_1,v_2)$, where $q \in H(p)$ and $v_{1},v_{2}$ are
horizontal tangent vectors at $q$.
\end{theorem}
Now we give the proof of our proposition.

\begin{proof}[Proof of Proposition \ref{just_lines_prop}]
Let $(P,\theta)$ be the principal
$\U(n)$-bundle with connection whose associated bundle is $E$. 
Let $p \in P$. Since the curvature of $E$ is $\i \cdot B \vert_{R} \tensor \id$, the Ambrose-Singer Theorem
implies the Lie algebra of $\mathrm{Hol}_{p}(\theta)$ is
\begin{equation} \label{lie_alg_eq}
\underset{n}{\underbrace{\u(1) \times \cdots \times \u(1)}},
\end{equation}
where $n=\mathrm{rank}(E)$.
The reduced holonomy group $\mathrm{Hol}^{0}_{p}(\theta)$ is the
connected component of $\mathrm{Hol}_{p}(\theta)$ containing the
identity. Therefore its Lie algebra is also (\ref{lie_alg_eq}) and
hence 
\[
\mathrm{Hol}^{0}_{p}(\theta) = \underset{n}{\underbrace{\U(1) \times \cdots \times \U(1)}}.
\]
Since $S^{2}$ is simply connected, $\mathrm{Hol}_{p}(\theta)=\mathrm{Hol}^{0}_{p}(\theta)$.
Therefore, by the Reduction Theorem, $P$ reduces to a $\U(1) \times
\cdots \times \U(1)$ bundle, which implies that $E$ is isomorphic to a
direct sum of line bundles 
\[
E'=L_{1} \oplus \cdots \oplus L_{n}.
\]

Indeed, if $g_{ab} \maps U_{ab} \to \U(n)$ are local transition
functions for $E$, the above isomorphism implies that there exists local functions $f
\maps U_{a} \to \U(n)$ and transition functions
\[
\begin{array}{c}
h_{ab} \maps U_{ab} \to \U(1) \times \cdots \times \U(1)\\
 x \mapsto (h^{1}_{ab}(x),\ldots,h^{n}_{ab}(x)),
\end{array}
\]
such that $h_{ab} = f_{a} g_{ab} f_{b}^{-1}$. If $\Omega'_{a}$ and $\Omega_{a}$ are the
local curvature 2-forms for $E'$ and $E$,
respectively, then
\[
\Omega'_{a} = f_{a} \Omega_{a} f^{-1}_{a} = \i f_{a} B \vert_{R} \cdot I
f^{-1}_{a} =\i B \vert_{R} \cdot I,
\]
where $I$ is the identity matrix.
Hence, the connection $\conn_{i}$ on the line bundle $L_{i}$ induced by the reduction
has curvature $B \vert_{R}$.

Finally, we show that all the line bundles $(L_{i},\conn_{i})$ are
isomorphic. We do so by showing that their local data of transition
functions and 1-forms all represent the same class in the degree 1
Deligne cohomology of $S^{2}$.
In the proof of Prop.\ \ref{curvature=integral}, 
we showed that the sequence (\ref{long_exact_sequence}) is exact.
Hence, the following sequence is exact:
\[
0 \to H^{1}(S^{2},\U(1)) \to H^{1}(S^{2},D^{\bullet}_{1})
\stackrel{\kappa}{\to} Z^{2}(S^{2}) \stackrel{f}{\to} H^{2}(S^{2},\U(1)),
\]
which relates cohomology with $\U(1)$-coefficients to the Deligne
cohomology group $H^{1}(S^{2},D^{\bullet}_{1})$. The map $\kappa$
sends a Deligne class to the closed 2-form corresponding to
its curvature. The Universal Coefficient Theorem implies:
\[
H^{1}(S^{2},\U(1)) \cong \Hom(H_{1}(S^2,\Z),\U(1)) =0.
\]
Hence the curvature map $\kappa$ is injective. Therefore line
bundles with the same curvature are isomorphic. This completes the proof.
\end{proof}

\subsubsection*{Constructing the Bohr-Sommerfeld variety}
The quantum state category is the subcategory of sections of
the stack $\HVB^{\xi}$ which are twisted-flat over the leaves 
contained in the Bohr-Sommerfeld variety. Proposition
\ref{just_lines_prop} implies that we have no hope of finding such
sections if the 2-form $B$ is not integral, since it must be the
curvature of a line bundle.  Remarkably, 
the 2-plectic Bohr-Sommerfeld variety, obtained by categorifying the
symplectic definition, resolves this issue.  
\begin{prop}
The 2-form $B$ restricts to an integral 2-form on a leaf
of the foliation $F$ if and only if the leaf is contained in the
Bohr-Sommerfeld variety.
\end{prop}
\begin{proof}
Let $S^{2}_{R}$ be a leaf and assume $B \vert_{R}$ is integral.
Recall from Def.\ \ref{integral_def} this means that the class $[B\vert_{R}]$ is in the image of the map
\[
H^{2}(S^{2}_{R},2\pi\Z) \to H^{2}(S^{2}_{R},\R) \iso H^{2}_{\mathrm{dR}}(S^{2}_{R}).
\]
There are canonical isomorphisms which identify singular cohomology
with smooth singular cohomology for arbitrary coefficients, and $\R$-valued 
smooth singular cohomology with de Rham cohomology \cite{Warner}[Sec.\ 5.34].
Using these isomorphisms, $B \vert_{R}$ is integral if and only if
\[
\int_{\Delta^{2}} s^{\ast} B \vert_{R} \in 2\pi \Z
\]
for all smooth simplicies $s \maps \Delta^{2} \to S^{2}_{R}$.
By Def.\ \ref{2-plectic_BS_def}, $S^{2}_{R}$ is contained in the
Bohr-Sommerfeld variety if and only if 
\[
\hol(\xi,\sigma)=1
\]
for all maps $\sigma \maps \Sigma^{2} \to S^{2}_{R}$, where
$\Sigma^{2}$ is a compact oriented 2-manifold. Definitions
\ref{holonomy_def} and \ref{2-hol_2-conn_def} imply $\hol(\xi,\sigma)=1$ if and only if 
\[
\int_{\Sigma^{2}} \sigma^{\ast} B \vert_{R} \in  2\pi \Z.
\]
Since $B \vert_{R}$ is integral, $\sigma^{\ast} B \vert_{R}$ is
integral for any such map $\sigma$. Let $\sum_{i} n_{i} s_{i}$ represent the fundamental class
in $H_{2}(\Sigma^{2}) \cong \Z$, where $s_{i} \maps \Delta^{2} \to
\Sigma^{2}$ are smooth simplicies. Then
\[
\int_{\Sigma^{2}} \sigma^{\ast} B \vert_{R} = \sum_{i} n_{i}
\int_{\Delta^{2}} s_{i}^{\ast} B \vert_{R} \in 2 \pi \Z.
\]
Hence, $S^{2}_{R}$ is contained in the variety.

Conversely, assume $S^{2}_{R}$ is a leaf in the Bohr-Sommerfeld
variety. Then, by taking $\sigma=\id$, we have
\begin{equation} \label{fund_class}
\int_{S^{2}_{R}} B \vert_{R} \in 2\pi \Z.
\end{equation}
We claim that this implies $B \vert_{R}$ is integral. Indeed, since
$S^{2}$ is simply connected, the Universal Coefficient Theorem implies
we have a commuting diagram
\[
\xymatrix{
H^{2}(S^{2}, 2 \pi \Z) \ar[d] \ar[r]^-{\sim} & \Hom(H_{2}(S^{2}), 2 \pi\Z) \ar[d]\\
H^{2}(S^{2},\R) \ar[r]^-{\sim} & \Hom(H_{2}(S^{2}), \R). 
}
\] 
Hence, $B \vert_{R}$ is integral if and only if for all classes $[s] \in
H_{2}(S^{2}_{R})$, we have 
\[
\int_{\Delta^{2}} s^{\ast}B \vert_{R} \in 2\pi \Z.
\]
Any such class is an integer multiple of the fundamental class
representing $S^{2}_{R}$. Therefore the integral (\ref{fund_class}) gives the
desired result.
\end{proof}
\begin{corollary}\label{integer_radius_cor}
A sphere with radius $R$ is contained in the
Bohr-Sommerfeld variety if and only if
\[
R \in \frac{1}{2} \Z.
\]
\end{corollary}
\begin{proof}
Such a sphere is contained in the variety if and only if $B
\vert_{R}$ is integral, i.e.\ if and only if
\[
\int_{S^{2}_{R}} B \vert_{R} =4 \pi R \in 2\pi \Z.
\]
\end{proof}
\noi Hence the variety is, precisely, the subspace
\[
V_{\mathrm{BS}} = \coprod^{\infty}_{n=1} S^{2}_{n/2}.
\]

\subsubsection*{The quantum state category}
Now we can characterize the quantum state category $\Quant(V_{\mathrm{BS}})$, i.e.\ the
subcategory of $\HVB^{\xi}(V_{\mathrm{BS}})$ whose objects are
those $\xi$-twisted Hermitian bundles 
 over $V_{\mathrm{BS}}$ which are twisted-flat along each leaf.
The results obtained in the previous sections imply:
\begin{theorem} \label{quant_thm}
There is a one-to-one correspondence between
isomorphism classes of objects in $\Quant(V_{\mathrm{BS}})$ and isomorphism classes of 
Hermitian vector bundles (with connection) over the Bohr-Sommerfeld
variety whose restriction to any leaf $S^{2}_{n/2}$ is of the form
\[
L \oplus L \oplus \cdots \oplus L
\]
where $L$ is a line bundle with curvature $B \vert_{n/2}$.
\end{theorem}

\begin{proof}
Let $(E_{i},\conn_{i},\phi_{ij})$ be an object in
$\Quant(V_{\mathrm{BS}})$. 
It is a bundle twisted by the trivial cocycle $(1,0,\i \cdot B)$
on the cover $\{V_{i}\}$, where 
\[
V_{i}=U_{i} \cap  V_{\mathrm{BS}} = \coprod^{\infty}_{n=1} U_{i} \cap S^{2}_{n/2}
\]
Hence, the proof of Prop.\ \ref{just_conn_prop} can be used to show there exists a
Hermitian vector bundle $E$ over $V_{\mathrm{BS}}$, unique up to
isomorphism, with connection $\conn$ such that 
\[
(E \vert_{V_{i}},\conn \vert_{V_{i}},\id)  \cong
(E_{i},\conn_{i},\phi_{ij}).
\]
Since $(E_{i},\conn_{i},\phi_{ij})$ is twisted flat, the restriction of
the curvature of the bundle $(E,\conn)$ to a leaf $S^{2}_{n/2}$ 
satisfies
\[
\conn^{2} \vert_{n/2} = \i \cdot B \vert_{n/2} \tensor \id. 
\]
Hence, Prop. \ref{just_lines_prop} implies that the
restriction of $E$ to $S^{2}_{n/2}$ is isomorphic to direct sum of
line bundles 
\[
\bigoplus^{k}_{i=1} L_{i},
\]
Here, $k \geq 0$, and each $L_{i}$ is the line bundle,
unique up to isomorphism, with curvature $B \vert_{n/2}$.

By reversing this argument, any Hermitian vector bundle over $V_{\mathrm{BS}}$ 
whose restriction to a leaf is isomorphic to the direct sum above
represents a unique isomorphism class of trivially twisted bundles
in $\Quant(V_{\mathrm{BS}})$
\end{proof}

\section{Applications to representation theory} \label{rep_theory}
As previously mentioned, some of the most important applications of geometric
quantization lie in representation theory. Here we present evidence that the categorified
geometric quantization of 2-plectic manifolds has similar 
uses. Roughly, the idea is the following:
In ordinary geometric quantization, sections in the quantum space $Q(V_{\mathrm{BS}})$
correspond to vectors in a representation of a Lie group. In
categorified geometric quantization, sections in the quantum category
$\Quant(V_{\mathrm{BS}})$ correspond to representations i.e.\ objects
in a representation category of a Lie group.

In particular, we describe this correspondence in detail for the
example $M = \R^{3}\setminus\{0\}$ considered in the previous section. The
2-spheres in the associated Bohr-Sommerfeld variety are
special coadjoint orbits of the Lie group $\SU(2)$, via the
identification $\su(2)^{\ast} \cong \R^3$. These 2-spheres equipped
with the restriction of the 2-form $B$ are sympletic manifolds, and,
through ordinary geometric quantization, they give irreducible
representations of $\SU(2)$. As we will see, these facts imply that
the quantum state category, obtained via the categorified quantization
of $M$, is closely related to the category of finite dimensional
representations of $\SU(2)$.

Let $S^{2}_{n/2}$ be the 2-sphere of radius $n/2$,
i.e.\ a leaf in the Bohr-Sommerfeld variety. By identifying $S^3$ 
with the unit sphere in $\C^{2}$, we use the Hopf fibration $S^{3} \to S^{2}_{n/2}$:
\[
\bigl (Z_{0},Z_{1} \bigr) \mapsto \frac{n}{2} \bigl(Z_{1} \bar{Z}_{0} + Z_{0}
\bar{Z}_{1}, ~ i Z_{1} \bar{Z}_{0} - iZ_{0}\bar{Z}_{1}, ~ Z_{0} \bar{Z}_{0} - Z_{1}
\bar{Z}_{1} \bigr)
\]
to identify $S^{2}_{n/2}$ with $\CP^{1}$. Choosing the affine
coordinate $w =Z_{1}/Z_{0}$, 
the 2-form 
\[
B \vert_{n/2} = \frac{4}{n^{2}} (x dy\wedge dz + y dz \wedge dx + z dx \wedge
dy)
\]
becomes
\[
B \vert_{n/2} = n \i \frac{dw \wedge d\bar{w}}{(1+ w\bar{w})^{2}}.
\]
(See Woodhouse \cite{Woodhouse:1991} Sections 3.5, 8.4, and 9.2 for
details.)

Recall that the hyperplane bundle $H\to \CP^{1}$ is the holomorphic line bundle
whose fiber over each point $[Z_{0},Z_{1}] \in \CP^{1}$ is the dual
space of the corresponding line in $\C^{2}$. The curvature of this bundle is $B \vert_{1/2}$.
Hence, $B \vert_{n/2}$ is the curvature of the tensor product
\[
H^{\tensor n} \to \CP^{1}.
\]
In fact, up to isomorphism, $H^{\tensor n}$ is the unique holomorphic line bundle with
this curvature. Let $(\zeta^{0},\zeta^{1})$ be the coordinates on the dual space $\C^{2^{\ast}}$.
It can be shown using standard complex analysis that the
global holomorphic sections $\Gamma(H^{\tensor n}) _{\mathrm{h}}$ are the degree $n$
homogeneous polynomials in the variables $(\zeta^{0},\zeta^{1})$
\cite{Griffiths-Harris}[Sec.\ 1.3].

There is an action of the group $\SU(2) \ss \mathrm{SL}(2,\C)$ on the polynomials
$\Gamma(H^{\tensor n})_{\mathrm{h}}$ which is induced by its obvious action on
$\C^{2}$. In fact, for each $n$, $\Gamma(H^{\tensor n})_{\mathrm{h}}$ is an
irreducible representation, which represents the unique isomorphism
class of irreducible representations of dimension $n+1$
\cite{Kirillov:2004}[Sec.\ A3.2]. Moreover, any finite dimensional
representation of $\SU(2)$ is isomorphic to a finite direct sum of
irreducibles. Hence, any such representation is isomorphic to
the holomorphic global sections of a direct sum
\[
H^{\tensor n_{1}} \oplus H^{\tensor n_{2}} \oplus \cdots \oplus H^{\tensor n_{k}}
\]
of line bundles over $\CP^{1}$. Note that the trivial bundle over $\C
P^{1}$ is the line bundle $H^{\tensor n}$ with $n=0$. Its global sections
are the holomorphic functions on $\CP^{1}$, i.e.\ the constants $\C$.

Now we show how all of this is related to the quantization of the
2-plectic manifold $\R^{3} \setminus \{0\}$.
Theorem \ref{quant_thm} implies that an isomorphism class of objects in the
quantum state category $\Quant(V_{\mathrm{BS}})$ can be identified
with a collection of bundles:
\begin{align*}
k_{1} \cdot L_{1} & \to S^{2}_{1/2} \\
k_{2} \cdot L_{2} & \to S^{2}_{1} \\
k_{3} \cdot L_{3} & \to S^{2}_{3/2} \\
\vdots& \quad \vdots \\
k_{n} \cdot L_{n} & \to S^{2}_{n/2} \\
\vdots& \quad \vdots 
\end{align*}
which are unique up to isomorphism.
Here, $k_{n}$ is a non-negative integer, and $k_{n} \cdot L_{n}$ is the direct sum of line bundles 
\[
k_{n} \cdot L_{n} = \underset{k_{n}}{\underbrace{L_{n} \oplus L_{n} \oplus \cdots \oplus L_{n}}},
\]
where $L_{n}$ is the line bundle with curvature $B \vert_{n/2}$.
By identifying each sphere with $\CP^{1}$, the above discussion
implies we can identify each line bundle with a tensor power of the
hyperplane bundle $H^{\tensor n}$. By taking global holomorphic sections, 
each copy of $H^{\tensor n}$ is then identified with
$\mathrm{Sym}^{n}(\C^{2^{\ast}})$, the space of degree $n$
homogeneous polynomials in 2 variables, which is a $(n+1)$-dimensional
irreducible representation of $\SU(2)$. (See Figure \ref{fig}.)
\begin{figure}[h!] \label{fig}
 \centering
\[
\xymatrix{
k_{1} \cdot L_{1} \ar[r]& S^{2}_{1/2} & \ar @{|->}[r]  && k_{1} \cdot
H\ar[r]& \CP^{1} & \ar @{|->}[r]^-{\Gamma_{\mathrm{h}}}&& k_{1} \cdot \mathrm{Sym^{1}}(\C^{2^{\ast}})\\
 k_{2} \cdot L_{2} \ar[r]& S^{2}_{1} & \ar @{|->}[r]  && k_{2} \cdot
 H^{\tensor 2}\ar[r]& \CP^{1}& \ar @{|->}[r]^-{\Gamma_{\mathrm{h}}}&& k_{2} \cdot \mathrm{Sym^{2}}(\C^{2^{\ast}})\\
  k_{3} \cdot L_{3} \ar[r]& S^{2}_{3/2} & \ar @{|->}[r]  && k_{3}
  \cdot H^{\tensor 3}\ar[r]& \CP^{1}  &\ar @{|->}[r]^-{\Gamma_{\mathrm{h}}}&& k_{3} \cdot \mathrm{Sym^{3}}(\C^{2^{\ast}})\\
  \vdots & \vdots &  && \vdots & \vdots & && \vdots\\
  k_{n} \cdot L_{n} \ar[r]& S^{2}_{n/2} & \ar @{|->}[r]  && k_{n} \cdot H^{\tensor n}\ar[r]& \CP^{1} &\ar @{|->}[r]^-{\Gamma_{\mathrm{h}}}&& k_{n} \cdot \mathrm{Sym^{n}}(\C^{2^{\ast}})\\
  \vdots & \vdots & && \vdots & \vdots & && \vdots
}
\]
  \caption{The quantum state given by the
    collection of vector bundles $\{k_{1} \cdot L_{1} \to S^{2}_{1/2}$,
    $k_{2} \cdot L_{2} \to S^{2}_{1}, \ldots\}$ is identified 
    with the representation $k_{1} \cdot \mathrm{Sym^{1}}(\C^{2^{\ast}})
    \oplus k_{2} \cdot \mathrm{Sym^{2}}(\C^{2^{\ast}}) \oplus \cdots$
    of $\SU(2)$.}
\end{figure}

\noi Note this procedure gives all finite-dimensional representations of
$\SU(2)$ except for those built using the 1-dimensional trivial representation. This is
because the sphere of radius 0 (the origin) is not contained in the
Bohr-Sommerfeld variety. Hence, we have proven:

\begin{theorem}\label{rep_theory_thm}
There is a one-to-one correspondence between
isomorphism classes of objects in the quantum state category $\Quant(V_{\mathrm{BS}})$ and
isomorphism classes of finite-dimensional representations
of $\SU(2)$ whose decomposition into irreducibles does not contain the
trivial representation.
\end{theorem}

\chapter{Summary and future work} \label{conclusions}
We conclude this thesis by providing a summary which ties together the
main results of the previous chapters. Along the way, we make some brief remarks regarding
open problems and possible directions for future research.

\subsection*{Polarizations on $n$-plectic manifolds}
In Chapter \ref{nplectic_geometry}, we presented the basic
geometric facts needed for our study of $n$-plectic
manifolds. In particular, we considered the
$n$-plectic analogues for Lagrangian submanifolds and real polarizations.
There are at least $n$ different ways to generalize the definition of a Lagrangian
submanifold to $n$-plectic geometry. This is due to the fact that
there are $n$ different ways to define the notion of orthogonal complement
on an $n$-plectic vector space (Def. \ref{ortho_comp}). 
Since real polarizations in symplectic
geometry are foliations whose leaves are Lagrangian submanifolds, we
have at least $n$ different kinds of real polarizations on an $n$-plectic
manifold (Def.\ \ref{k-polarization}). Polarizations play an important role in geometric
quantization, but it is not clear which
definition of polarization for $n$-plectic manifolds is ``best'' in this context. 

Moreover, it is unknown if an $n$-plectic analogue of a complex
polarization exists. It is possible that one could use ideas from generalized complex
geometry \cite{Gualtieri:2007} and the theory of `higher Dirac
structures' \cite{Zambon:2010ka} to help develop such polarizations.

\subsection*{Lie $n$-algebras from $n$-plectic manifolds}
In Chapter \ref{algebra_chapter}, we showed that an $n$-plectic
structure on $M$ induces a bracket on the space of Hamiltonian
$(n-1)$-forms. The bracket is skew-symmetric, but only satisfies the Jacobi
identity up to homotopy.
We proved that this bracket gives a Lie $n$-algebra $\Lie(M,\omega)$, whose underlying $n$-term chain complex
consists of Hamiltonian $(n-1)$-forms and all
differential forms of lower degrees (Thm. \ref{main_thm}).
When $n=1$, the Hamiltonian forms are the
smooth functions, and the Lie $1$-algebra
is just the underlying Lie algebra of the usual Poisson algebra of a
symplectic manifold. For certain 2-plectic manifolds,
our previous work with Baez and Hoffnung implies that
the associated Lie 2-algebra can be used to describe the ``observable
algebra'' of the classical bosonic string \cite{Baez:2008bu}.

In Appendix \ref{leibniz_appendix}, we showed that an $n$-plectic
manifold also gives a dg Leibniz algebra $\Leib(M,\omega)$ on the same
complex (Prop.\ \ref{n-plectic_Leibniz}). Its bracket satisfies
Jacobi, but is skew-symmetric only up to homotopy.  For the 2-plectic
case, we showed $\Lie(M,\omega)$ and $\Leib(M,\omega)$ are isomorphic
in Roytenberg's category of weak Lie 2-algebras (Thm.\
\ref{isomorphism_L2A_thm}). The objects of this category are 2-term
$L_{\infty}$-algebras whose structure maps are skew-symmetric only up
to homotopy.  In general, we would like to conjecture that some sort
of equivalence such as this holds for $n>2$. Unfortunately, it is not
clear in what category this should occur. Indeed, developing a theory
of weak Lie $n$-algebras is an open problem.  Perhaps by studying the
relationships between the structures specifically on $\Lie(M,\omega)$
and $\Leib(M,\omega)$ for arbitrary $n$ one could get a sense of what
explicit coherence conditions would be needed to give a good
definition.

On the other hand, there are structures known as
`Loday-$\mathbf{\infty}$ algebras' (or sh Leibniz algebras)
\cite{Ammar:2008} that generalize the definition of an
$L_{\infty}$-algebra by, again, relaxing the skew symmetry condition
on the structure maps. However, this time the skew symmetry is not
required to hold up to homotopy.  Hence any dg Leibniz algebra is a
Loday-$\infty$ algebra. Any $L_{\infty}$-algebra is as well. Therefore
there may be an isomorphism between $\Lie(M,\omega)$ and
$\Leib(M,\omega)$ in this category for $n \geq 2$.

\subsection*{Lie 2-algebras from compact simple Lie groups} 
In Chapter \ref{lie_group_chapter}, we considered compact simple Lie groups
as 2-plectic manifolds. Any compact simple Lie group $G$ admits a 1 parameter family of canonical 2-plectic
structures $\{\nu_{k}\}$, given by non-zero multiples of
Cartan 3-form (Ex.\ \ref{Lie_group_example}). We proved that the
associated Lie 2-algebra $\Lie(G,\nu_{k})$ contains a sub Lie 2-algebra consisting of left invariant Hamiltonian
1-forms (Cor.\ \ref{left_inv_L2A}). We showed that this sub-algebra is not equivalent to
$\Lie(G,\nu_{k})$, however, it is isomorphic to the so-called string Lie 2-algebra
associated to $(G,\nu_{k})$ (Thm. \ref{string_Lie_Thm}). The string Lie 2-algebra plays an
important role in string theory and in the theory of loop groups.

Our results suggest a close link between these areas and 2-plectic
geometry. There are many possible directions for future work here. In
particular, it would be interesting to understand the relationship
between $\Lie(G,\nu_{k})$ and the algebra of observables for certain
string theory models called `WZW models'. 

\subsection*{Gerbes, 2-line stacks, and 2-bundles}
In Chapter \ref{stacks_chapter}, we presented the technical tools needed to
develop a geometric quantization theory for 2-plectic manifolds.
The work of Brylinski \cite{Brylinski:1993} implies that if $(M,\omega)$ is a 2-plectic manifold 
and $\omega$ is an integral 3-form, then $\omega$ can be realized as the 2-curvature of 
a $\U(1)$-gerbe equipped with a 2-connection. 
If $\{U_{i}\}$ is an open cover of $M$, then
locally, a $\U(1)$-gerbe with 2-connection
is determined by a Deligne 2-cocycle i.e\ a collection of $\U(1)$-valued transition functions $g_{ijk}$ on
$U_{i} \cap U_{j} \cap U_{k}$, 1-forms $A_{ij}$ on $U_{i} \cap U_{j}$, and
2-forms $B_{i}$ on $U_{i}$, with $dB_{i}=\omega$, satisfying certain
compatibility conditions (Ex.\ \ref{n=2}).

% This is in complete
% analogy with the symplectic case, where an integral symplectic form 
% is the curvature of a principal $\U(1)$-bundle equipped with
% connection, which is locally determined by a collection of transition
% functions $g_{ij}$ and 1-forms $\theta_{i}$.

Every principal $\U(1)$-bundle with connection has an associated Hermitian
line bundle with connection. Similarly, we showed that every $\U(1)$-gerbe with
2-connection over a 2-plectic manifold has an associated 2-line stack
with 2-connection (Prop.\ \ref{2-line_stack_conn_prop}). The
category of global sections of the 2-line stack is equivalent to the
category of Hermitian vector bundles twisted by the gerbe's Deligne
2-cocycle $\xi=(g_{ijk},A_{ij},B_{i})$.
Such a twisted bundle is given locally by a collection of
Hermitian vector bundles $E_{i}$ with connection $\conn_{i}$ (Def.\ \ref{twistedvb_conn}). The
twisting by $\xi$ characterizes the
obstruction to gluing these bundles together into a global 
bundle over $M$. A $\xi$-twisted Hermitian
vector bundle is twisted-flat if, for each $E_{i}$, the curvature $\conn^{2}_{i}$ is equal
to $\i \cdot B_{i} \tensor \id$ (Def.\ \ref{twisted_flat_def}). This is the 2-plectic analogue of a
flat section of a Hermitian line-bundle.
We also showed that there is a good notion of  holonomy (Def.\ \ref{2-hol_2-conn_def}) for
2-line stacks equipped with 2-connection, given by Carey, Johnson, and Murray's
formula for the 2-holonomy of the Deligne class $[\xi]$ \cite{Carey:2004}.
The 2-holonomy plays an important role in our quantization procedure
for 2-plectic manifolds. In particular, it is used in our definition for the 2-plectic version of the
Bohr-Sommerfeld variety (Def.\ \ref{2-plectic_BS_def}). 

It is not obvious, at first glance, why twisted Hermitian vector bundles
are the 2-plectic analogues of sections of a Hermitian line bundle.
In the same chapter, we sketch an argument supporting this point of
view using Bartels' work in 2-bundle theory \cite{Bartels:2004}. 
It becomes clear that twisted Hermitian vector bundles should be understood as sections of a
2-vector bundle of rank 1. It would be very interesting to make this argument more precise, and perhaps recast
our results within the context of 2-vector bundles. For example, for
line bundles, one can consider different kinds of sections  e.g.\
smooth, square-integrable, etc.
Similarly, the 2-bundle approach might suggest that we consider 2-lines stacks
whose sections are more general than twisted bundles, e.g.\ twisted
coherent sheaves. This could have important consequences for the output
of our geometric quantization procedure for 2-plectic manifolds.

\subsection*{2-Plectic prequantization and Courant algebroids}
We defined a prequantized 2-plectic manifold to be a 2-plectic manifold
equipped with a Deligne 2-cocycle (Def.\ \ref{n-plectic_prequant_def}). 
This 2-cocycle can be realized
geometrically as a $\U(1)$-gerbe with 2-connection, or as its associated
2-line stack. This is in complete analogy with the symplectic case, where
we prequantize using either a principal $\U(1)$-bundle or a Hermitian
line bundle. In Section \ref{symplectic_sec} , we first recall how to
prequantize the Poisson algebra on a symplectic manifold
equipped with a principal $\U(1)$-bundle $P$ with connection. By
prequantizing, we mean faithfully representing the Poisson
algebra as linear differential operators.
This is done by considering the Atiyah algebroid $A$ associated to
$P$. There is an injective Lie algebra morphism from the Poisson
algebra to the Lie algebra of global sections of $A$, which identifies the Poisson
algebra with those invariant vector fields on $P$ whose flows preserve
the connection (Prop.\ \ref{lie_alg_iso}). 

% This algebroid
% is a vector bundle over the manifold whose space of global
% sections is the Lie algebra of $\U(1)$-invariant vector
% fields on $P$. Hence, sections of $A$ are infinitesimal symmetries of
% $P$ i.e.\ their flows give principal bundle automorphisms of
% $P$. There is short exact sequence of vector bundles over $M$
% \[
% M \times \u(1) \to A \to TM,
% \]
% whose splittings $TM \to A$ correspond to connections on $P$.
% There is an injective Lie algebra morphism from the Poisson
% algebra to global sections of $A$, which identifies the Poisson
% algebra with those invariant vector fields on $P$ whose flows preserve
% the connection. 
For the 2-plectic case, we described a known construction which gives a Courant algebroid $C$
over a prequantized 2-plectic manifold $(M,\omega)$ equipped with a
$\U(1)$-gerbe with 2-connection (Sec.\ \ref{geometric}). 
In this case, $C$ is a vector bundle over $M$ whose sections are locally
given by vector fields and 1-forms on $M$. Its space of global sections
form a Lie 2-algebra. There is a short exact sequence
of vector bundles over $M$
\[
T^{\ast}M  \to C \to TM,
\]
whose splittings $TM \to C$ correspond to 2-connections on the
$\U(1)$-gerbe over $M$. We prove the existence of an injective Lie
2-algebra morphism from the Lie 2-algebra $\Lie(M,\omega)$ of
observables on $M$ to the Lie 2-algebra $\Lie(C)$ of global sections
of $C$ (Thm.\ \ref{main_courant_thm}).
This morphism identifies $\Lie(M,\omega)$ with a sub-Lie 2-algebra of
$\Lie(C)$, which, in a certain sense, preserves the 2-connection of the
gerbe (Thm.\ \ref{lie_2_alg_iso}). We interpret this sub-algebra as the prequantization of
$\Lie(M,\omega)$. Also, we show that this construction
gives the higher analogue of the well known Kostant-Souriau central
extension in symplectic geometry (Sec.\ \ref{2plectic_extend}).

This prequantization process gives an interesting
relationship between Courant algebroids and prequantized 2-plectic manifolds.
Let us give two possible directions for future work based on
these results.

\begin{enumerate}
\item{
Sections of the Atiyah algebroid $A$ over a prequantized symplectic
manifold equipped with the principal $\U(1)$-bundle are differential operators on a Hilbert
space. This Hilbert space is constructed from global sections of
the associated Hermitian line bundle.
The higher analogue of this Hilbert space is the category of   
global sections of the 2-line stack associated to a $\U(1)$-gerbe.
In what way, if at all, do sections of the Courant algebroid over a prequantized
2-plectic manifold act as ``operators'' on this higher analogue of a
Hilbert space?
}

\item{
Recall that sections of the Atiyah algebroid are 
infinitesimal $U(1)$-equivariant symmetries of the corresponding
principal $U(1)$-bundle. Integration gives the `gauge groupoid' over $M$,
whose elements correspond to the equivariant automorphisms
of the principal bundle \cite{CdS-Weinstein:1999}[Sec.\ 17.1]. 
Our results suggest that the Courant
algebroid is the higher analogue of the Atiyah algebroid.
So, how can we understand sections of the Courant algebroid on a
prequantized 2-plectic manifold as infinitesimal automorphisms of
the corresponding $U(1)$-gerbe? In other words, what is the Lie
2-groupoid that integrates this Courant algebroid, and how does it
act as the `gauge 2-groupoid' of the $U(1)$-gerbe?
}
\end{enumerate}

\subsection*{2-plectic quantization and representation theory}
In the last chapter, we categorifed \'{S}niatycki's \cite{Sniatycki:1975} 
quantization procedure for symplectic manifolds, which
employs Bohr-Sommerfeld varieties to overcome topological obstructions
that arise when using real polarizations (Sec.\ \ref{cat_quant_section}). This categorification gives 
a simple procedure for quantizing a 2-plectic manifold, and the
resulting output is a category of quantum states (Def.\ \ref{quant_state_cat_def}). An object of
this category is a twisted Hermitian vector bundle over the
Bohr-Sommerfeld variety (Def.\ \ref{2-plectic_BS_def})  whose restriction to each leaf contained in the variety
is twisted-flat.

In Section \ref{cat_quant_example_sec}, we considered an interesting
example: $M=\R^3 \setminus \{0\}$ equipped with a volume form
$\omega=dB$. We quantized $M$ by equipping it with the trivial Deligne
2-cocycle $\xi=(1,0,B)$, and a 2-polarization whose leaves are spheres
centered about the origin.  The restriction of the 2-form $B$ to such
a sphere is the KKS symplectic form, which arises in Kirillov's orbit
method for constructing representations of Lie groups. This is not
surprising, since $\R^3$ is isomorphic to the dual of the Lie algebra
$\su(2)$, and each sphere is isomorphic to a coadjoint orbit.  We then
showed that in this example, a leaf of the polarization is contained
in the Bohr-Sommerfeld variety if and only if it is a sphere of radius
$n/2$, where $n$ is an integer (Cor.\ \ref{integer_radius_cor}). The orbit method identifies such a
sphere with the irreducible representation of $\SU(2)$ whose dimension
is $n+1$.

Next, we proved that any twisted bundle in the associated category of quantum
states is isomorphic to a direct sum of line bundles over spheres
contained in the variety (Thm.\ \ref{quant_thm}). The fact that the twisted bundle is
twisted-flat 
on each sphere implies that each of these line bundles must be
isomorphic to a tensor power of the hyperplane bundle over $\CP^{1}$.
This allowed us to identify a quantum state with a representation of
$\SU(2)$. More precisely, we proved that isomorphism
classes of objects in the quantum state category are in one-to-one
correspondence with isomorphism classes of finite-dimensional representations
of $\SU(2)$ whose decomposition into irreducibles does not contain the
trivial representation (Thm.\ \ref{rep_theory_thm}).

It is unfortunate that we are unable to obtain the
trivial representation via our quantization procedure. However, it is
not surprising. Our procedure identifies spheres of radius
$n/2$ with irreducibles of dimension  $n+1$. Hence, the trivial
representation corresponds to the origin in $\su(2)^{\ast}$, which is not in $M$. In some sense,
this identification needs to be shifted so that the sphere of radius
$n/2$ is identified with the irreducible representation of dimension $n$. This is very
similar to the $1/2$ shift which arises in the usual geometric
quantization of the simple harmonic oscillator (Sec.\ \ref{harm_osc_sec}).

We believe we have just scratched the surface of a deeper 
relationship between representation theory and the geometric
quantization of 2-plectic manifolds. Indeed, our example suggests that
2-plectic quantization can give a categorifed analogue of the orbit method.
We conclude by mentioning two related directions for future work along these
lines.
\begin{enumerate}
\item{
It is well known that closed integral forms on a
manifold $M$ can be mapped to closed integral forms on $LM$, the space
of free loops of $M$, by a process called `transgression'. Moreover,
this process sends a $\U(1)$-gerbe equipped with 2-connection on $M$
to a principal $\U(1)$-bundle with connection on $LM$
\cite{Brylinski:1993}[Ch.\ 6]. This suggests
that the categorifed geometric quantization of a 2-plectic manifold
may, in some way, correspond to ordinary geometric quantization on
$LM$. (We are overlooking subtleties here, such as the fact that
transgression need not preserve non-degeneracy.)
For example, perhaps there is some 2-plectic structure on $\su(2)$ whose
quantization gives a category of quantum states, with objects
corresponding to certain representations of the loop group $L\SU(2)$
obtained by applying the orbit method to the loop algebra $L\su(2)$.
}

\item{
Much work has been done on quantizing the conjugacy classes of compact
simple Lie groups via a variety of methods, all of which rely on the
Cartan 3-form in some way \cite{Meinrenken:2009,Mohrdieck:2004}.
The output of these quantization procedures gives information about the representation
theory of the corresponding loop group. Every compact simple Lie group,
equipped with the Cartan 3-form, is a 2-plectic manifold.
Hence, it is natural to suspect that 2-plectic quantization of Lie
\textit{groups} is also related to the representation theory of  loop groups.
We have preliminary results which suggest that such a relationship exists, although,
even in the simple case of $\SU(2)$, many issues remain unresolved.
}
\end{enumerate}

%\nocite{*}
%\singlespacing
%\bibliographystyle{alpha}
%\bibliographystyle{plain}
\bibliographystyle{abbrv}
\bibliography{bibfile}

\def\polhk#1{\setbox0=\hbox{#1}{\ooalign{\hidewidth
  \lower1.5ex\hbox{`}\hidewidth\crcr\unhbox0}}} \def\cprime{$'$}
\begin{thebibliography}{10}

\bibitem{Ambrose-Singer}
W.~Ambrose and I.~M. Singer.
\newblock A theorem on holonomy.
\newblock {\em Trans. Amer. Math. Soc.}, 75:428--443, 1953.

\bibitem{Ammar:2008}
M.~Ammar and N.~Poncin.
\newblock Coalgebraic approach to the {L}oday infinity category, stem
  differential for {$2n$}-ary graded and homotopy algebras.
\newblock {\em Ann. Inst. Fourier (Grenoble)}, 60(1):355--387, 2010.
\newblock Also available as
  \href{http://arxiv.org/abs/0809.4328}{arXiv:0809.4328}.

\bibitem{Antieau:2009}
B.~Antieau.
\newblock Cohomological obstruction theory for {B}rauer classes and the
  period-index problem.
\newblock Available as \href{http://arxiv.org/abs/0909.2352}{arXiv:0909.2352}.

\bibitem{HDA6}
J.~C. Baez and A.~S. Crans.
\newblock Higher-dimensional algebra. {VI}. {L}ie 2-algebras.
\newblock {\em Theory Appl. Categ.}, 12:492--538, 2004.
\newblock Also available as
  \href{http://arxiv.org/abs/math/0307263}{arXiv:math/0307263}.

\bibitem{Baez:2008bu}
J.~C. Baez, A.~E. Hoffnung, and C.~L. Rogers.
\newblock Categorified symplectic geometry and the classical string.
\newblock {\em Comm. Math. Phys.}, 293(3):701--725, 2010.
\newblock Also available as
  \href{http://arxiv.org/abs/0808.0246}{arXiv:0808.0246}.

\bibitem{Baez:2009uu}
J.~C. Baez and C.~L. Rogers.
\newblock Categorified symplectic geometry and the string {L}ie 2-algebra.
\newblock {\em Homology, Homotopy Appl.}, 12(1):221--236, 2010.

\bibitem{BaezSchreiber:2005}
J.~C. Baez and U.~Schreiber.
\newblock Higher gauge theory.
\newblock In {\em Categories in {A}lgebra, {G}eometry and {M}athematical
  {P}hysics}, volume 431 of {\em Contemp. Math.}, pages 7--30. Amer. Math.
  Soc., Providence, RI, 2007.
\newblock Also available as
  \href{http://arxiv.org/abs/math/0511710}{arXiv:math/0511710}.

\bibitem{BCSS}
J.~C. Baez, D.~Stevenson, A.~S. Crans, and U.~Schreiber.
\newblock From loop groups to 2-groups.
\newblock {\em Homology, Homotopy Appl.}, 9(2):101--135, 2007.
\newblock Also available as
  \href{http://arxiv.org/abs/math/0504123}{arXiv:math/0504123}.

\bibitem{Barnich:1997ij}
G.~Barnich, R.~Fulp, T.~Lada, and J.~Stasheff.
\newblock The sh {L}ie structure of {P}oisson brackets in field theory.
\newblock {\em Comm. Math. Phys.}, 191(3):585--601, 1998.
\newblock Also available as
  \href{http://arxiv.org/abs/hep-th/9702176}{arXiv:hep-th/9702176}.

\bibitem{Bartels:2004}
T.~Bartels.
\newblock {\em Higher Gauge Theory I: 2-Bundles}.
\newblock PhD thesis, University of California, Riverside, 2004.
\newblock Available as
  \href{http://arxiv.org/abs/math/0410328}{arXiv:math/0410328}.

\bibitem{Bressler-Chervov}
P.~Bressler and A.~Chervov.
\newblock Courant algebroids.
\newblock {\em J. Math. Sci. (N. Y.)}, 128(4):3030--3053, 2005.
\newblock Also available as
  \href{http://arxiv.org/abs/hep-th/0212195}{arXiv:hep-th/0212195}.

\bibitem{Brylinski:1990}
J.-L. Brylinski.
\newblock Noncommutative {R}uelle-{S}ullivan type currents.
\newblock In {\em The {G}rothendieck {F}estschrift, {V}ol.\ {I}}, volume~86 of
  {\em Progr. Math.}, pages 477--498. Birkh\"auser Boston, Boston, MA, 1990.

\bibitem{Brylinski:1993}
J.-L. Brylinski.
\newblock {\em Loop {S}paces, {C}haracteristic {C}lasses and {G}eometric
  {Q}uantization}.
\newblock Modern Birkh\"auser Classics. Birkh\"auser Boston Inc., Boston, MA,
  2008.

\bibitem{Caldararu:2000}
A.~Caldararu.
\newblock {\em Derived {C}ategories of {T}wisted {S}heaves on {C}alabi-{Y}au
  {M}anifolds}.
\newblock PhD thesis, Cornell University, 2000.

\bibitem{CdS-Weinstein:1999}
A.~Cannas~da Silva and A.~Weinstein.
\newblock {\em Geometric {M}odels for {N}oncommutative {A}lgebras}, volume~10
  of {\em Berkeley Mathematics Lecture Notes}.
\newblock American Mathematical Society, Providence, RI, 1999.

\bibitem{Cantrijn:1999}
F.~Cantrijn, A.~Ibort, and M.~de~Le{\'o}n.
\newblock On the geometry of multisymplectic manifolds.
\newblock {\em J. Austral. Math. Soc. Ser. A}, 66(3):303--330, 1999.

\bibitem{Carey:2004}
A.~L. Carey, S.~Johnson, and M.~K. Murray.
\newblock Holonomy on {D}-branes.
\newblock {\em J. Geom. Phys.}, 52(2):186--216, 2004.
\newblock Also available as
  \href{http://arxiv.org/abs/hep-th/0204199}{arXiv:hep-th/0204199}.

\bibitem{Courant}
T.~J. Courant.
\newblock Dirac manifolds.
\newblock {\em Trans. Amer. Math. Soc.}, 319(2):631--661, 1990.

\bibitem{Forger:2002ak}
M.~Forger, C.~Paufler, and H.~R{\"o}mer.
\newblock The {P}oisson bracket for {P}oisson forms in multisymplectic field
  theory.
\newblock {\em Rev. Math. Phys.}, 15(7):705--743, 2003.
\newblock Also available as
  \href{http://arxiv.org/abs/math-ph/0202043}{arXiv:math-ph/0202043}.

\bibitem{Freed:1994}
D.~S. Freed.
\newblock Higher algebraic structures and quantization.
\newblock {\em Comm. Math. Phys.}, 159(2):343--398, 1994.
\newblock Also available as
  \href{http://arxiv.org/abs/hep-th/9212115}{arXiv:hep-th/9212115}.

\bibitem{Gajer:1997}
P.~Gajer.
\newblock Geometry of {D}eligne cohomology.
\newblock {\em Invent. Math.}, 127(1):155--207, 1997.
\newblock Also available as
  \href{http://arxiv.org/abs/alg-geom/9601025}{arXiv:alg-geom/9601025}.

\bibitem{Giraud:1971}
J.~Giraud.
\newblock {\em Cohomologie non ab\'elienne}.
\newblock Springer-Verlag, Berlin, 1971.
\newblock Die Grundlehren der mathematischen Wissenschaften, Band 179.

\bibitem{GHV2}
W.~Greub, S.~Halperin, and R.~Vanstone.
\newblock {\em Connections, {C}urvature, and {C}ohomology. {V}ol. {II}: {L}ie
  {G}roups, {P}rincipal {B}undles, and {C}haracteristic {C}lasses}.
\newblock Academic Press, New York-London, 1973.

\bibitem{Griffiths-Harris}
P.~Griffiths and J.~Harris.
\newblock {\em Principles of {A}lgebraic {G}eometry}.
\newblock Wiley Classics Library. John Wiley \& Sons Inc., New York, 1994.

\bibitem{Grothendieck_stacks}
A.~Grothendieck.
\newblock Technique de descente et th\'eor\`emes d'existence en g\'eometrie
  alg\'ebrique. {I}. {G}\'en\'eralit\'es. {D}escente par morphismes
  fid\`element plats.
\newblock In {\em S\'eminaire {B}ourbaki, {V}ol.\ 5}, pages Exp.\ No.\ 190,
  299--327. Soc. Math. France, Paris, 1995.

\bibitem{Gualtieri:2007}
M.~Gualtieri.
\newblock Generalized complex geometry.
\newblock Available as
  \href{http://arxiv.org/abs/math/0703298}{arXiv:math/0703298}.

\bibitem{Helein:2002wf}
F.~H{\'e}lein.
\newblock Hamiltonian formalisms for multidimensional calculus of variations
  and perturbation theory.
\newblock In {\em Noncompact {P}roblems at the {I}ntersection of {G}eometry,
  {A}nalysis, and {T}opology}, volume 350 of {\em Contemp. Math.}, pages
  127--147. Amer. Math. Soc., Providence, RI, 2004.
\newblock Also available as
  \href{http://arxiv.org/abs/math-ph/0212036}{arXiv:math-ph/0212036}.

\bibitem{Henriques:2008}
A.~Henriques.
\newblock Integrating {$L_\infty$}-algebras.
\newblock {\em Compos. Math.}, 144(4):1017--1045, 2008.
\newblock Also available as
  \href{http://arxiv.org/abs/math/0603563}{arXiv:math/0603563}.

\bibitem{Hitchin:2004ut}
N.~Hitchin.
\newblock Generalized {C}alabi-{Y}au manifolds.
\newblock {\em Q. J. Math.}, 54(3):281--308, 2003.
\newblock Also available as
  \href{http://arxiv.org/abs/math/0209099v1}{arXiv:math/0209099}.

\bibitem{Ibort:2000}
A.~Ibort.
\newblock Multisymplectic geometry: generic and exceptional.
\newblock In {\em Proceedings of the {IX} {F}all {W}orkshop on {G}eometry and
  {P}hysics ({V}ilanova i la {G}eltr\'u, 2000)}, volume~3 of {\em Publ. R. Soc.
  Mat. Esp.}, pages 79--88. R. Soc. Mat. Esp., Madrid, 2001.

\bibitem{Kapranov:1995}
M.~M. Kapranov and V.~A. Voevodsky.
\newblock {$2$}-categories and {Z}amolodchikov tetrahedra equations.
\newblock In {\em Algebraic {G}roups and {T}heir {G}eneralizations: {Q}uantum
  and {I}nfinite-dimensional {M}ethods ({U}niversity {P}ark, {PA}, 1991)},
  volume~56 of {\em Proc. Sympos. Pure Math.}, pages 177--259. Amer. Math.
  Soc., Providence, RI, 1994.

\bibitem{Kijowski:1973gi}
J.~Kijowski.
\newblock A finite-dimensional canonical formalism in the classical field
  theory.
\newblock {\em Comm. Math. Phys.}, 30:99--128, 1973.

\bibitem{KT:1979}
J.~Kijowski and W.~M. Tulczyjew.
\newblock {\em A {S}ymplectic {F}ramework for {F}ield {T}heories}, volume 107
  of {\em Lecture Notes in Physics}.
\newblock Springer-Verlag, Berlin, 1979.

\bibitem{Kirillov:2004}
A.~A. Kirillov.
\newblock {\em Lectures on the {O}rbit {M}ethod}, volume~64 of {\em Graduate
  Studies in Mathematics}.
\newblock American Mathematical Society, Providence, RI, 2004.

\bibitem{Kobay-Nomizu}
S.~Kobayashi and K.~Nomizu.
\newblock {\em Foundations of {D}ifferential {G}eometry. {V}ol {I}}.
\newblock John Wiley \& Sons, New York-London, 1963.

\bibitem{KosSchwarz:2004}
Y.~Kosmann-Schwarzbach.
\newblock Derived brackets.
\newblock {\em Lett. Math. Phys.}, 69:61--87, 2004.
\newblock Also available as
  \href{http://arxiv.org/abs/math/0312524}{arXiv:math/0312524}.

\bibitem{Kostant:1970}
B.~Kostant.
\newblock Quantization and unitary representations. {I}. {P}requantization.
\newblock In {\em Lectures in {M}odern {A}nalysis and {A}pplications, {III}},
  pages 87--208. Lecture Notes in Math., Vol. 170. Springer, Berlin, 1970.

\bibitem{Lada-Markl}
T.~Lada and M.~Markl.
\newblock Strongly homotopy {L}ie algebras.
\newblock {\em Comm. Algebra}, 23(6):2147--2161, 1995.
\newblock Also available as
  \href{http://arxiv.org/abs/hep-th/9406095v1}{arXiv:hep-th/9406095}.

\bibitem{LS}
T.~Lada and J.~Stasheff.
\newblock Introduction to {SH} {L}ie algebras for physicists.
\newblock {\em Internat. J. Theoret. Phys.}, 32(7):1087--1103, 1993.
\newblock Also available as
  \href{http://arxiv.org/abs/hep-th/9209099}{hep-th/9209099}.

\bibitem{Liu:1997}
Z.-J. Liu, A.~Weinstein, and P.~Xu.
\newblock Manin triples for {L}ie bialgebroids.
\newblock {\em J. Differential Geom.}, 45(3):547--574, 1997.
\newblock Also available as
  \href{http://arxiv.org/abs/dg-ga/9508013}{arXiv:dg-ga/9508013}.

\bibitem{Mackenzie:1987}
K.~C.~H. Mackenzie.
\newblock {\em General {T}heory of {L}ie {G}roupoids and {L}ie {A}lgebroids},
  volume 213 of {\em London Mathematical Society Lecture Note Series}.
\newblock Cambridge University Press, Cambridge, 2005.

\bibitem{Martinet:1970}
J.~Martinet.
\newblock Sur les singularit\'es des formes diff\'erentielles.
\newblock {\em Ann. Inst. Fourier (Grenoble)}, 20(fasc. 1):95--178, 1970.

\bibitem{Meinrenken:2003}
E.~Meinrenken.
\newblock The basic gerbe over a compact simple {L}ie group.
\newblock {\em Enseign. Math. (2)}, 49(3-4):307--333, 2003.
\newblock Also available as
  \href{http://arxiv.org/abs/math/0209194}{arXiv:math/0209194}.

\bibitem{Meinrenken:2009}
E.~Meinrenken.
\newblock On the quantization of conjugacy classes.
\newblock {\em Enseign. Math. (2)}, 55(1-2):33--75, 2009.
\newblock Also available as
  \href{http://arxiv.org/abs/0707.3963}{arXiv:0707.3963}.

\bibitem{Moerdijk:2002}
I.~Moerdijk.
\newblock Introduction to the language of stacks and gerbes.
\newblock Available as
  \href{http://arxiv.org/abs/math/0212266}{arXiv:math/0212266}.

\bibitem{Mohrdieck:2004}
S.~Mohrdieck and R.~Wendt.
\newblock Integral conjugacy classes of compact {L}ie groups.
\newblock {\em Manuscripta Math.}, 113(4):531--547, 2004.
\newblock Also available as
  \href{http://arxiv.org/abs/math/0303118}{arXiv:math/0303118}.

\bibitem{PressleySegal}
A.~Pressley and G.~Segal.
\newblock {\em Loop {G}roups}.
\newblock Oxford Mathematical Monographs. The Clarendon Press Oxford University
  Press, New York, 1986.

\bibitem{Quillen:1969}
D.~Quillen.
\newblock Rational homotopy theory.
\newblock {\em Ann. of Math. (2)}, 90:205--295, 1969.

\bibitem{Rogers:2011}
C.~L. Rogers.
\newblock 2-plectic geometry, {C}ourant algebroids, and categorified
  prequantization.
\newblock Submitted. Available as
  \href{http://arxiv.org/abs/1009.2975}{arXiv:1009.2975}.

\bibitem{Rogers:2010}
C.~L. Rogers.
\newblock {$L_\infty$}-algebras from multisymplectic geometry.
\newblock To appear in \textsl{Lett.\ Math.\ Phys}. Available as
  \href{http://arxiv.org/abs/1005.2230}{arXiv:1005.2230}.

\bibitem{Roman-Roy:2009}
N.~Rom{\'a}n-Roy.
\newblock Multisymplectic {L}agrangian and {H}amiltonian formalisms of
  classical field theories.
\newblock {\em SIGMA Symmetry Integrability Geom. Methods Appl.}, 5:Paper 100,
  25, 2009.
\newblock Also available as
  \href{http://arxiv.org/abs/math-ph/0506022}{arXiv:math-ph/0506022}.

\bibitem{Roytenberg_thesis}
D.~Roytenberg.
\newblock {\em Courant Algebroids, Derived Brackets and Even Symplectic
  Supermanifolds}.
\newblock PhD thesis, University of California, Berkeley, 1999.
\newblock Available as
  \href{http://arxiv.org/abs/math/9910078}{arXiv:math/9910078}.

\bibitem{Roytenberg_graded}
D.~Roytenberg.
\newblock On the structure of graded symplectic supermanifolds and {C}ourant
  algebroids.
\newblock In {\em Quantization, {P}oisson brackets and beyond ({M}anchester,
  2001)}, volume 315 of {\em Contemp. Math.}, pages 169--185. Amer. Math. Soc.,
  Providence, RI, 2002.
\newblock Also available as
  \href{http://arxiv.org/abs/math/0203110}{arXiv:math/0203110}.

\bibitem{Roytenberg_L2A}
D.~Roytenberg.
\newblock On weak {L}ie 2-algebras.
\newblock In {\em X{XVI} {W}orkshop on {G}eometrical {M}ethods in {P}hysics},
  volume 956 of {\em AIP Conf. Proc.}, pages 180--198. Amer. Inst. Phys.,
  Melville, NY, 2007.
\newblock Also available as
  \href{http://arxiv.org/abs/0712.3461}{arXiv:0712.3461}.

\bibitem{Roytenberg-Weinstein}
D.~Roytenberg and A.~Weinstein.
\newblock Courant algebroids and strongly homotopy {L}ie algebras.
\newblock {\em Lett. Math. Phys.}, 46(1):81--93, 1998.
\newblock Also available as
  \href{http://arxiv.org/abs/math/9802118}{arXiv:math/9802118v1}.

\bibitem{SSS:2009}
H.~Sati, U.~Schreiber, and J.~Stasheff.
\newblock {$L_\infty$}-algebra connections and applications to {S}tring- and
  {C}hern-{S}imons {$n$}-transport.
\newblock In {\em Quantum {F}ield {T}heory}, pages 303--424. Birkh\"auser,
  Basel, 2009.

\bibitem{Schapira:2006}
P.~Schapira.
\newblock Categories, sites, sheaves and stacks: {C}ourse at {U}niversity of
  {P}aris {VI}.
\newblock Available as
  \href{http://people.math.jussieu.fr/~schapira/lectnotes/Sta.pdf}{\url{http:/%
/people.math.jussieu.fr/~schapira/lectnotes/Sta.pdf}}.

\bibitem{Schlessinger:1979}
M.~Schlessinger and J.~Stasheff.
\newblock Deformation theory and rational homotopy type.
\newblock 1979.
\newblock Preprint.

\bibitem{Urs_nlab}
U.~Schreiber.
\newblock $\infty$-{L}ie algebra cohomology.
\newblock Available as
  \href{http://ncatlab.org/nlab/show/infinity-Lie+algebra+cohomology}{\url{htt%
p://ncatlab.org/nlab/show/infinity-Lie+algebra+cohomology}}.

\bibitem{Schreiber:2005}
U.~Schreiber.
\newblock {\em From Loop Space Mechanics to Nonabelian Strings}.
\newblock PhD thesis, Universit{\"a}t Duisburg-Essen, 2005.
\newblock Available as
  \href{http://arxiv.org/abs/hep-th/0509163}{arXiv:hep-th/0509163}.

\bibitem{Severa-Weinstein}
P.~{\v{S}}evera and A.~Weinstein.
\newblock Poisson geometry with a 3-form background.
\newblock {\em Progr. Theoret. Phys. Suppl.}, (144):145--154, 2001.
\newblock Also available as
  \href{http://arxiv.org/abs/math/0107133}{arXiv:math/0107133}.

\bibitem{Sniatycki:1975}
J.~{\'S}niatycki.
\newblock Wave functions relative to a real polarization.
\newblock {\em Internat. J. Theoret. Phys.}, 14(4):277--288, 1975.

\bibitem{Sniatycki:1980}
J.~{\'S}niatycki.
\newblock {\em Geometric {Q}uantization and {Q}uantum {M}echanics}, volume~30
  of {\em Applied Mathematical Sciences}.
\newblock Springer-Verlag, New York, 1980.

\bibitem{Souriau:1967}
J.-M. Souriau.
\newblock Quantification g\'eom\'etrique. {A}pplications.
\newblock {\em Ann. Inst. H. Poincar\'e Sect. A (N.S.)}, 6:311--341, 1967.

\bibitem{Swann:1991}
A.~Swann.
\newblock Hyper-{K}\"ahler and quaternionic {K}\"ahler geometry.
\newblock {\em Math. Ann.}, 289(3):421--450, 1991.

\bibitem{Severa1}
P.~\v{S}evera.
\newblock Letter to {A}lan {W}einstein.
\newblock Available at \href{http://sophia.dtp.fmph.uniba.sk/~severa/letters/}
  {\url{http://sophia.dtp.fmph.uniba.sk/~severa/letters/}}.

\bibitem{Warner}
F.~W. Warner.
\newblock {\em Foundations of {D}ifferentiable {M}anifolds and {L}ie {G}roups}.
\newblock Scott, Foresman and Co., Glenview, Ill.-London, 1971.

\bibitem{Weyl:1935}
H.~Weyl.
\newblock Geodesic fields in the calculus of variation for multiple integrals.
\newblock {\em Ann. of Math. (2)}, 36(3):607--629, 1935.

\bibitem{Witten:1988}
E.~Witten.
\newblock The index of the {D}irac operator in loop space.
\newblock In {\em Elliptic {C}urves and {M}odular {F}orms in {A}lgebraic
  {T}opology ({P}rinceton, {NJ}, 1986)}, volume 1326 of {\em Lecture Notes in
  Math.}, pages 161--181. Springer, Berlin, 1988.

\bibitem{Woodhouse:1991}
N.~Woodhouse.
\newblock {\em Geometric {Q}uantization}.
\newblock The Clarendon Press Oxford University Press, New York, 1980.

\bibitem{Zambon:2010ka}
M.~Zambon.
\newblock {$L_{\infty}$}-algebras and higher analogues of {D}irac structures.
\newblock Available as \href{http://arxiv.org/abs/1003.1004}{arXiv:1003.1004}.

\end{thebibliography}

\appendix

\chapter{Other algebraic structures on $n$-plectic manifolds} \label{leibniz_appendix}

There are other structures besides Lie $n$-algebras which can generalize the
Poisson bracket to $n$-plectic manifolds.  Here we show that any
$n$-plectic manifold gives rise to another kind of algebraic structure
known as a differential graded (dg) Leibniz algebra. A dg Leibniz
algebra is a graded vector space equipped with a degree $-1$ differential
and a bilinear bracket that satisfies a Jacobi-like identity, but does
not need to be skew-symmetric.  There is an interesting relationship between
the bilinear bracket on the Lie $n$-algebra and the bracket on the
corresponding dg Leibniz algebra. When $n=2$, these algebras can be
compared directly as objects in Roytenberg's category of `weak Lie
2-algebras' \cite{Roytenberg_L2A}. A weak Lie 2-algebra is a Lie 2-algebra whose $k=2$ bracket
satisfies skew-symmetry only up to a chain homotopy. This homotopy
must satisfy compatibility relations with the homotopy controlling the
failure of the Jacobi identity. We show that the Lie 2-algebra and the
2-term dg Leibniz algebra arising from a 2-plectic manifold are
isomorphic as weak Lie 2-algebras. We are unable to extend this result to
the $n>2$ case, since there is currently no definition available for weak $L_{\infty}$-algebras.

\section{dg Leibniz algebras}
In symplectic geometry, every
function $f \in \cinf(M)$ is Hamiltonian.  We also have the equality:
\begin{equation} \label{poisson}
\brac{f}{g}= \ip{f}dg= \L_{v_{f}}g
\end{equation}
for all $f,g \in \hamn{0}=\cinf(M)$. Hence $\brac{f}{\cdot}$ is a
degree zero derivation on $\hamn{0}$, which makes
$(\hamn{0},\blankbrac)$ a Poisson algebra. In general, for $n >1$, 
an equality such as Eq.\ \ref{poisson} does not hold, and
Hamiltonian forms are obviously not closed under wedge product.
Therefore, we shouldn't expect the Lie $n$-algebra $\Lie(M,\omega)$ to
behave like a Poisson algebra. But we do have the following simple lemma:
\begin{lemma} \label{der_lemma}
Let $(M,\omega)$ be an $n$-plectic manifold. If $\alpha,\beta \in
\hamn{n-1}$ are Hamiltonian forms, then
\[
\L_{v_{\alpha}} \beta= \brac{\alpha}{\beta} + d \ip{\alpha} \beta.
\]
\end{lemma}
\begin{proof}
Definitions \ref{hamiltonian} and \ref{bracket_def} imply:
\begin{align*}
\L_{v_{\alpha}} \beta &= \ip{\alpha} d\beta + d\ip{\alpha} \beta \\
&= -\ip{\alpha} \ip{\beta}\omega + d\ip{\alpha} \beta \\
&=\brac{\alpha}{\beta} + d \ip{\alpha} \beta.
\end{align*}
\end{proof}

Lemma \ref{der_lemma} suggests that we interpret the $(n-1)$-form
$\L_{v_{\alpha}}\beta$ as a type of bracket on $\hamn{n-1}$, equal to
the bracket $\brac{\cdot}{\cdot}$ modulo boundary
terms. To this end, we consider an algebraic structure known as a
differential graded (dg) Leibniz algebra. 

\begin{definition} \label{dg_leibniz}
A {\bf differential graded Leibniz algebra}
$(L,\delta,\bbrac{\cdot}{\cdot})$ is a graded vector
space $L$ equipped with a degree -1 linear map $\delta \maps L \to
L$ and a degree 0 bilinear map $\bbrac{\cdot}{\cdot} \maps L \tensor L
\to L$ such that the following identities hold:
\begin{gather}
\delta \circ \delta =0 \\
\delta \bbrac{x}{y} = \bbrac{\delta x}{y} +
(-1)^{\deg{x}}\bbrac{x}{\delta y} \label{bracket_der}\\
\bbrac{x}{\bbrac{y}{z}}= \bbrac{\bbrac{x}{y}}{z} + (-1)^{\deg{x}\deg{y}}
\bbrac{y}{\bbrac{x}{z}}\label{leib_jacobi},
\end{gather}
for all $x,y,z \in L$.
\end{definition}
In the literature, dg Leibniz algebras are also called dg Loday algebras.
This definition presented here is equivalent to the one given by Ammar and Poncin
\cite{Ammar:2008}. Note that the second condition given in the
definition above can be interpreted as the Jacobi identity. Hence 
if the bilinear map $\bbrac{\cdot}{\cdot}$ is skew-symmetric,
then a dg Leibniz algebra is a DGLA.

We now show that every $n$-plectic manifold gives a dg Leibniz
algebra.
\begin{prop} \label{n-plectic_Leibniz}
Given an $n$-plectic manifold $(M,\omega)$, there is a differential
graded Leibniz algebra $\Leib(M,\omega)=(L,\delta,\bbrac{\cdot}{\cdot})$ with underlying
graded vector space 
\[
L_{i} =
\begin{cases}
\hamn{n-1} & i=0,\\
\Omega^{n-1-i}(M) & 0 < i \leq n-1,
\end{cases}
\]
and maps $ \delta \maps L \to L$, $\bbrac{\cdot}{\cdot} \maps L
\tensor L \to L$ defined as
\[
\delta(\alpha)=d \alpha,
\]
if $\deg{\alpha} > 0$ and
\[
\bbrac{\alpha}{\beta}=
\begin{cases}
\L_{v_{\alpha}}\beta & \text{if $\deg{\alpha}=0$,}\\
0 & \text{if $\deg{\alpha} >0$,}
\end{cases}
\]
where $v_{\alpha}$ is the Hamiltonian vector field associated to $\alpha$.
\end{prop}
\begin{proof}
If $\alpha,\beta \in L_{0}=\hamn{n-1}$ are Hamiltonian, then Lemma \ref{der_lemma}
implies
$d\bbrac{\alpha}{\beta}=d\brac{\alpha}{\beta}=-\iota_{[v_{\alpha},v_{\beta}]} \omega$.
Hence $\bbrac{\alpha}{\beta}$ is Hamiltonian. For $\deg{\beta} > 0$,
we have $\deg{\L_{v_{\alpha}}\beta}=\deg{\beta}$, since the Lie
derivative is a degree zero derivation. Hence $\bbrac{\cdot}{\cdot}$
is a bilinear degree 0 map. 

We next show that Eq.\ \ref{bracket_der} of Definition \ref{dg_leibniz}
holds. If $ \deg{\alpha} > 1$, then it holds
trivially. If $\deg{\alpha} =1$, then
$\bbrac{\alpha}{\beta}=\bbrac{\alpha}{\delta\beta}=0$ for all $\beta \in L$
by definition, and $\bbrac{\delta \alpha}{\beta}=0$ since the
Hamiltonian vector field associated to $d\alpha$ is zero. If
$\deg{\alpha}=0$ and $\deg{\beta}=0$, then
$\deg{\bbrac{\alpha}{\beta}}=0$. Hence all terms in
\eqref{bracket_der} vanish by definition. The last case to
consider is
$\deg{\alpha}=0$ and $\deg{\beta} > 0$. We have 
\[
\delta \bbrac{\alpha}{\beta}= d \L_{v_{\alpha}}\beta =
\L_{v_{\alpha}} d\beta =\bbrac{\alpha}{\delta \beta}.
\]

Finally, we show the Jacobi identity \eqref{leib_jacobi}
holds. Let $\alpha, \beta, \gamma \in L$.  Then the left hand side of
\eqref{leib_jacobi} is $\bbrac{\alpha}{\bbrac{\beta}{\gamma}}$, while
the right hand side is $\bbrac{\bbrac{\alpha}{\beta}}{\gamma} +
(-1)^{\deg{\alpha}\deg{\beta}}\bbrac{\beta}{\bbrac{\alpha}{\gamma}}$. Note equality holds trivially
if $\deg{\alpha} >0$ or $\deg{\beta} >0$. Otherwise, we use the
identity
\[
\L_{[v_{1},v_{2}]} = \L_{v_{1}} \L_{v_{2}} - \L_{v_{2}}\L_{v_{1}},
\]
and the fact that $d\bbrac{\alpha}{\beta}=-\iota_{[v_{\alpha},v_{\beta}]} \omega$
to obtain the following equalities:
\begin{align*}
\bbrac{\alpha}{\bbrac{\beta}{\gamma}} &= \L_{v_{\alpha}}
\L_{v_{\beta}} \gamma\\
&=\L_{[v_{\alpha},v_{\beta}]}\gamma + \L_{v_{\beta}}
\L_{v_{\alpha}} \gamma\\
&=\bbrac{\bbrac{\alpha}{\beta}}{\gamma} + 
\bbrac{\beta}{\bbrac{\alpha}{\gamma}}.
\end{align*}
\end{proof}
One interesting aspect of the dg Leibniz structure is that it
interprets the bracket of Hamiltonian $(n-1)$-forms geometrically as
the change of an observable along the flow of a Hamiltonian vector
field. Leibniz algebras, in fact, naturally arise in a
variety of geometric settings e.g.\  in Courant algebroid theory
and, more generally, in the derived bracket
formalism \cite{KosSchwarz:2004}.  It would be interesting to compare
$\Leib(M,\omega)$ to the Leibniz algebras that appear in these other
formalisms.

\section{Weak Lie 2-algebras}
When $(M,\omega)$ is a symplectic manifold, $L_{\infty}(M,\omega)$ and
$\Leib(M,\omega)$ give the same Lie algebra: the Poisson algebra of functions.
It would be nice if we could show that for any $n$-plectic manifold, 
$\Lie(M,\omega)$ and $\Leib(M,\omega)$ are also ``the same'', i.e.\
equivalent as objects in some category containing both
$L_{\infty}$-algebras and dg Leibniz algebras. This may seem unlikely
at first since the brackets which induce these structures have
different properties. For example, $\brac{\cdot}{\cdot}$ is
skew-symmetric, while $\bbrac{\cdot}{\cdot}$, in general, is not.
However, we have the following proposition.
\begin{prop} \label{brackets_same_prop}
Let $(M,\omega)$ be an $n$-plectic manifold, and 
$\brac{\cdot}{\cdot}$ and $\bbrac{\cdot}{\cdot}$ be the brackets given
in Def.\ \ref{bracket_def} and Prop.\ \ref{n-plectic_Leibniz},
respectively. If $\alpha$ and $\beta$ are Hamiltonian $(n-1)$-forms,
then
\[
\bbrac{\alpha}{\beta} + \bbrac{\beta}{\alpha} = d \left( \ip{\alpha}
  \beta + \ip{\beta} \alpha \right).
\]
\end{prop}
\begin{proof}
The statement follows from the formula $\L_{v} = \iota_{v}d + d\iota_{v}$.
\end{proof}
So, we seek a category whose objects originate from
weakening, up to homotopy, both the skew-symmetric axiom and the Jacobi identity. 
Unfortunately, no such category exists,
unless $n=2$. In this case, by extending the work of Baez and Crans \cite{HDA6}, Roytenberg
\cite{Roytenberg_L2A} developed what are known as 2-term weak
$L_{\infty}$-algebras, or `weak Lie 2-algebras'. In a weak Lie
2-algebra, the skew symmetry condition on the maps given in Definition
\ref{Linfty} is relaxed. In particular, the bilinear map $l_{2} \maps
L \tensor L \to L$ is skew-symmetric only up to homotopy. 
This homotopy must satisfy a coherence condition, as well as compatibility
conditions with the homotopy that controls the failure of the Jacobi identity.
The goal of this section is to show that if $(M,\omega)$ is a
2-plectic manifold, then $\Lie(M,\omega)$ and $\Leib(M,\omega)$ are
isomorphic as weak Lie 2-algebras.

\begin{definition}[\cite{Roytenberg_L2A}] \label{WL2A_def}
A {\bf weak Lie 2-algebra} is a 2-term chain complex of vector spaces
$L = (L_1\stackrel{d}\rightarrow L_0)$ equipped with the following structure:
\begin{itemize}
\item a chain map $\blankbrac\maps L \tensor L\to L$ called the {\bf
bracket};
\item a chain homotopy $S \maps L\tensor L \to L$
%not antisymmetric!!!
from the chain map
\[     \begin{array}{ccl}  
     L \tensor L &\to& L   \\
     x \tensor y &\longmapsto& [x,y]  
  \end{array}
\]
to the chain map
\[     \begin{array}{ccl}  
     L \tensor L & \to & L   \\
     x \tensor y & \longmapsto & -[y,x]  
  \end{array}
\]
called the {\bf alternator};
\item a chain homotopy $J \maps L \tensor L \tensor L
 \to L$ 
from the chain map
\[     \begin{array}{ccl}  
     L \tensor L \tensor L & \to & L   \\
     x \tensor y \tensor z & \longmapsto & [x,[y,z]]  
  \end{array}
\]
to the chain map
\[     \begin{array}{ccl}  
     L \tensor L \tensor L& \to & L   \\
     x \tensor y \tensor z & \longmapsto & [[x,y],z] + [y,[x,z]]  
  \end{array}
\]
called the {\bf Jacobiator}.
\end{itemize}
In addition, the following equations are required to hold:
\begin{equation}
\begin{array}{c}
  [x,J(y,z,w)] + J(x,[y,z],w) +
  J(x,z,[y,w]) + [J(x,y,z),w] \\ + [z,J(x,y,w)] 
  = J(x,y,[z,w]) + J([x,y],z,w) \\ + [y,J(x,z,w)] + J(y,[x,z],w) + J(y,z,[x,w]),
\end{array}
\end{equation}
\begin{equation}
    J(x,y,z)+J(y,x,z)=-[S(x,y),z],
\end{equation}
\begin{equation} \label{WL2A_eq3}
    J(x,y,z)+J(x,z,y)=[x,S(y,z)]-S([x,y],z)-S(y,[x,z]),
\end{equation}
\begin{equation}
    {S(x,[y,z])} = S([y,z],x).
\end{equation}
\end{definition}

A weak Lie 2-algebra homomorphism is a chain map between the underlying
chain complexes that preserves the bracket up to coherent chain homotopy.
More precisely:
\begin{definition}[\cite{Roytenberg_L2A}]
\label{weak_homo}
Given Lie 2-algebras $L$ and $L'$ with bracket, alternator and 
Jacobiator $\blankbrac$, $S$, $J$ and $\blankbrac^\prime$, $S^\prime$, 
$J^\prime$ respectively, a {\bf homomorphism} from $L$ to $L'$
consists of:
\begin{itemize}
\item{a chain map $\phi=\left(\phi_{0},\phi_{1}\right) \maps L \to L'$, and}
%\item{a chain homotopy $\Phi \maps L \tensor L \to L$ from the chain !!!
\item{a chain homotopy $\Phi \maps L \tensor L \to L^{\prime}$ from the chain
  map
\[     \begin{array}{ccl}  
     L \tensor L & \to & L^{\prime}   \\
     x \tensor y & \longmapsto & \left [ \phi(x),\phi(y) \right]^{\prime},
  \end{array}
\]
to the chain map
\[     \begin{array}{ccl}  
     L \tensor L & \to & L^{\prime}   \\
     x \tensor y & \longmapsto & \phi \left( [x,y] \right)
  \end{array}
\]
}
\end{itemize}
such that the following equations hold:
\begin{equation}\label{homo_eq_1}
    S^{\prime}\left(\phi_{0}(x),\phi_{0}(y)\right)-\phi_{1}(S\left(x,y
    \right))=\Phi(x,y)+\Phi(y,x),
\end{equation}
\begin{equation}\label{homo_eq_2}
     \begin{array}{l}
       J^{\prime}\left(\phi_{0}(x),\phi_{0}(y),\phi_{0}(z)\right)-\phi_{1}\left(J\left(
       x,y,z \right)\right)\\
       =[\phi_{0}(x),\Phi(y,z)]^{\prime}-[\phi_{0}(y),\Phi(x,z)]^{\prime}-[\Phi(x,y),\phi_{0}(z)]^{\prime}\\
       -\Phi([x,y],z)-\Phi(y,[x,z])+\Phi(x,[y,z]).
     \end{array}
\end{equation}
\end{definition}
\noi The details involved in composing Lie 2-algebra homomorphisms are given
by Roytenberg \cite{Roytenberg_L2A}.  We say a Lie 2-algebra homomorphism
with an inverse is an {\bf isomorphism}.

Lie 2-algebras in the sense of Prop.\ \ref{L2A} are
weak Lie 2-algebras that satisfy skew-symmetry on the
nose. They are called semi-strict Lie 2-algebras in this context,
since the Jacobi identity may still fail to hold. More precisely:
\begin{definition}[\cite{Roytenberg_L2A}]
A weak Lie 2-algebra $(L,\blankbrac,S,J)$ is {\bf semi-strict}
iff $S=0$, and {\bf hemi-strict} iff $J=0$.
\end{definition}
Note that the bracket of a hemi-strict Lie 2-algebra satisfies a Jacobi identity of the form
\[ 
[x,[y,z]]-[[x,y],z] - [y,[x,z]] =0,
\]
but it is not necessarily skew-symmetric.
In fact, any hemi-strict Lie 2-algebra is a
2-term dg Leibniz algebra. For 2-plectic manifolds, we have
a converse:

\begin{prop}
If $(M,\omega)$ is a 2-plectic manifold, then $\Leib(M,\omega)$ is a
hemi-strict Lie 2-algebra with: 
\begin{itemize}
\item{underlying complex 
\[
L=\cinf(M)\stackrel{d}{\to} \hamn{1},
\]
}
\item{bracket given by
\[
[\alpha,\beta]=\bbrac{\alpha}{\beta}= \L_{v_{\alpha}} \beta,
\]
in degree 0, and
\begin{align*}
[\alpha,f]&=\bbrac{\alpha}{f}= \L_{v_{\alpha}} f, \\
[f,\alpha]&=\bbrac{f}{\alpha}=0
\end{align*}
in degree 1, 
}
\item{alternator given by:
\[
S(\alpha,\beta)=\ip{\alpha}\beta + \ip{\beta}\alpha,
\]
}
\item{Jacobiator given by:
\[
J(\alpha,\beta,\gamma)=0.
\]
}
\end{itemize}
\end{prop}

\begin{proof}
The axioms for a weak Lie 2-algebra given in Def.\ \ref{WL2A_def}  
are verified by straightforward calculations using the Cartan
calculus. In particular, the fact that Eq.\ \ref{WL2A_eq3} is
satisfied follows from the identity:
\[
\L_{v} \iota_{w} \alpha = \iota_{[v,w]} \alpha + \iota_{w} \L_{v} \alpha.
\]
\end{proof}

For a 2-plectic manifold, we view $\Lie(M,\omega)$ as a weak Lie 2-algebra
with trivial alternator. 
\begin{prop} \label{semi-strict_prop}
If $(M,\omega)$ is a 2-plectic manifold, then $\Lie(M,\omega)$ is a
semi-strict Lie 2-algebra with: 
\begin{itemize}
\item{underlying complex 
\[
L=\cinf(M)\stackrel{d}{\to} \hamn{1},
\]
}
\item{bracket given by
\[
[\alpha,\beta]=\brac{\alpha}{\beta}= \omega(v_{\alpha},v_{\beta},\cdot),
\]
in degree 0, and
\begin{align*}
[\alpha,f]&=0\\
[f,\alpha]&=0\\
\end{align*}
in degree 1, 
}
\item{alternator given by:
\[
S(\alpha,\beta)=0,
\]
}
\item{Jacobiator given by:
\[
J(\alpha,\beta,\gamma)=\omega(v_{\gamma},v_{\beta},v_{\alpha}).
\]
}
\end{itemize}
\end{prop}
\begin{proof}
By setting $S=0$, in Def.\ \ref{WL2A_def}, we recover 
the usual notion of a Lie 2-algebra Def.\ \ref{LnA}.
Hence, the statement follows from Prop.\ \ref{semistrict}.
\end{proof}

The main result of this section is the
following theorem:
\begin{theorem} \label{isomorphism_L2A_thm}
If $(M,\omega)$ is a 2-plectic manifold, then $\Lie(M,\omega)$ and
$\Leib(M,\omega)$ are isomorphic as weak Lie 2-algebras.
\end{theorem}

\begin{proof}
Since the underlying chain complexes of $\Lie(M,\omega)$ and
$\Leib(M,\omega)$ are the same, we build a weak Lie 2-algebra
isomorphism (Def.\ \ref{weak_homo}) using the identity chain map
\[
\phi_{0}=\id, \quad \phi_{1}=\id.
\]  
Let $\Phi \maps \hamn{1} \tensor
\hamn{1} \to \cinf(M)$ be the map:
\[
\Phi(\alpha,\beta)=\ip{\alpha}\beta.
\]
Proposition \ref{brackets_same_prop} and
a straightforward calculation show that $\Phi$ gives a chain homotopy:
\[
\xymatrix{
L_0 \tensor L_1 \oplus L_1 \tensor L_0 \ar[d]_{\blankbrac  -\blankbrac'}
\ar[rr] &&  \ar @{-->}[dll]_{\Phi} L_0 \tensor L_0  \ar[d]^{\blankbrac  -\blankbrac'}\\
L'_1 \ar[rr]^{d} && L'_0
}
\]
where $L_{1}=L'_{1}=\cinf(M)$, $L_0=L_0'=\hamn{1}$, $\blankbrac$ is
the bracket on $\Leib(M,\omega)$, and $\blankbrac'$ is the bracket on $\Lie(M,\omega)$.

The alternator for $\Leib(M,\omega)$ is 
\[
S(\alpha,\beta)=\ip{\alpha}\beta + \ip{\beta}\alpha
=\Phi(\alpha,\beta)  + \Phi(\beta,\alpha).
\]
Since the alternator for $\Lie(M,\omega)$ is trivial, the above equality
implies Eq.\ \ref{homo_eq_1} in Def.\ \ref{weak_homo} holds.

Since the Jacobiator of $\Leib(M,\omega)$ is trivial, 
the left hand side of Eq.\ \ref{homo_eq_2} only involves the Jacobiator
$J'(\alpha,\beta,\gamma)=\omega(v_{\gamma},v_{\beta},v_{\alpha})$
of $\Lie(M,\omega)$. Using the definition of the brackets $\blankbrac$
and $\blankbrac'$, the right hand side of Eq.\ \ref{homo_eq_2}
becomes:
\[
\bbrac{\alpha}{\Phi(\beta,\gamma)} -
\bbrac{\beta}{\Phi(\alpha,\gamma)} - \bbrac{\Phi(\alpha,\beta)}{\gamma}\\
-\Phi(\brac{\alpha}{\beta},\gamma) - \Phi(\beta,\brac{\alpha}{\gamma})
+ \Phi(\alpha,\brac{\beta}{\gamma}).
\]
By expanding the above using $\Phi(\alpha,\beta)=\ip{\alpha}\beta$, 
$\bbrac{\Phi(\alpha,\beta)}{\gamma}=0$,
$\bbrac{\alpha}{\Phi(\beta,\gamma)}=\L_{v_{\alpha}}\Phi(\beta,\gamma)$, 
$\bbrac{\alpha}{\beta}=\omega(v_{\alpha},v_{\beta},\cdot)$, and the identity
\[
\L_{v} \iota_{w} \alpha = \iota_{[v,w]} \alpha + \iota_{w} \L_{v} \alpha,
\]
the
right hand side of Eq.\ \ref{homo_eq_2}
becomes:
\[
\ip{\beta} d \ip{\alpha} \gamma + \ip{\beta}\ip{\alpha} d \gamma -
\ip{\beta} d \ip{\alpha} \gamma +
2\omega(v_{\alpha},v_{\beta},v_{\gamma}).
\]
Since $\omega(v_{\alpha},v_{\beta},v_{\gamma})=-\ip{\beta}\ip{\alpha}d\gamma$, the
above expression simplifies to:
\[
\omega(v_{\alpha},v_{\beta},v_{\gamma})=-\omega(v_{\gamma},v_{\beta},v_{\alpha})=
-J'(\alpha,\beta,\gamma),
\]
which is the left hand side of Eq.\ \ref{homo_eq_2}.
Hence, $(\phi_{0},\phi_{1},\Phi)$ satisfies the axioms for an isomorphism
of weak Lie 2-algebras.
\end{proof}

\chapter{Twisted bundles and the proof of Proposition
  \ref{2-line_stack_prop}} \label{stack_appendix}
Recall Def.\ \ref{twistedvb} of a Hermitian vector bundle twisted by a
2-cocycle $g \in C^{2}(\cU,\sh{\U(1)})$ on an open cover
$\cU=\{U_{i}\}$ of a manifold $M$. Such an object is given by the
following data:
\begin{itemize}
\item{on each $U_{i}$, a Hermitian vector bundle
\[
(E_{i}, \inp{\cdot}{\cdot}_{i}),
\]
}
\item{on each $U_{ij}= U_{i} \cap U_{j}$, an isomorphism of Hermitian vector bundles
\[
\phi_{i j} \maps E_{j}\vert_{U_{i j}} \iso
E_{i} \vert_{U_{i j}},
\]
such that $\forall i, j, k \in \mathcal{I}$:
\[
\phi^{-1}_{ik}\circ \phi_{ij}\circ \phi_{jk}= g_{i j k} \cdot
%\begin{array}{c}
%\phi_{i i} =\id_{E_{i}},\\
%\phi_{i j} =\phi_{j i}^{-1},\\
%\end{array}
\]
where $g_{i j k} \cdot$ is the automorphism of $E_{k} \vert_{U_{ijk}}$ 
corresponding to multiplication by
\[
g_{ijk} \maps U_{i} \cap U_{j} \cap U_{k} \to \U(1).
\]
}
\end{itemize}
Also, recall that a morphism $f \maps (E_{i},\phi_{i j}) \to
(E'_{i},\phi'_{i j})$ of $g$-twisted Hermitian vector bundles over $M$
consists of a collection of morphisms of Hermitian vector bundles
\[
f_{i} \maps E_{i}  \to E'_{i}, \quad 
\]
for each $i \in \mathcal{I}$ such that
\[
f_{i} \circ \phi_{i j} = \phi'_{i j} \circ f_{j}.
\]
In this section, we will prove Prop. \ref{2-line_stack_prop} from Chapter \ref{stacks_chapter}:
\begin{proposition*}
Given a 2-cocycle $g \in C^{2}(\cU,\sh{\U(1)})$ on a manifold $M$,
there exists a stack over $M$ whose category of global
sections is equivalent to the category $\HVB^{g}(M)$
of $g$-twisted Hermitian vector bundles over $M$.
\end{proposition*}
\noi As we will see, the proof follows from the fact that locally defined stacks
can be glued together to form a stack over $M$, in analogy with the
well-known result for sheaves.

We need to introduce some more machinery for stacks. First, just as we
have natural transformations between functors, we can define fibered
transformations between morphisms of fibered categories:
\begin{definition*}[\cite{Moerdijk:2002}] 
Let $(\phi,\alpha), (\psi,\beta) \maps \F \to \G$ be morphisms
between fibered categories. A {\bf fibered transformation} $\mu \maps
(\phi,\alpha) \to  (\psi,\beta)$ consists of natural transformations
\[
\mu_{U} \maps \phi_{U} \to \psi_{U},
\]
for each $U \ss M$, such that given an inclusion $i \maps V \to U$ of
open sets, the diagram of natural transformations 
\[
\xymatrix{
\phi_{V} i^{\ast} \ar[d]_{\mu_{V} i^{\ast}} \ar[r]^{\alpha_{i}}& i^{\ast} \phi_{U} \ar[d]^{i^{\ast}\mu_{U}} \\
\psi_{V} i^{\ast} \ar[r]^{\beta_{i}}& i^{\ast} \psi_{U}
}
\]
commutes. We say $\mu$ is a {\bf fibered isomorphism} if each
$\mu_{U}$ is a natural isomorphism.
\end{definition*}

Next, we describe the category of descent data associated to a fibered
category over $M$ and an open cover of $M$. One can think of this as
the data needed to glue together locally defined sections into a
global section.
\begin{definition}[\cite{Moerdijk:2002}]\label{descent_def}
Let $\F$ be a fibered category over $M$ and let $\cU=\{U_{i} \}$ be an
open cover of $M$. The category $\Des(M,\cU)$ of {\bf descent data}
has:
\begin{itemize}
\item{As objects, collections $(x_{i},\psi_{ij})$ where each $x_{i}$
    is an object of $\F(U_{i})$, and each
\[
\psi_{ij} \maps x_{j} \vert_{U_{ij}} \iso x_{i} \vert_{U_{ij}}
\]
is an isomorphism in $\F(U_{ij})$ required to satisfy the 
conditions
\begin{equation} \label{descent_def_eq}
\psi_{ik}^{-1} \circ \psi_{ij} \circ \psi_{jk} =\id
\end{equation}
in $\F(U_{ijk})$.
}

\item{As morphisms, $(x_{i},\psi_{ij}) \xto{f} (x'_{i},\psi'_{ij})$, a
    collection of morphisms
\[
x_{i} \xto{f_{i}} x'_{i}
\]
in $\F(U_{i})$ such that the diagram
\[
\xymatrix{
x_{j} \vert_{U_{ij}} \ar[d]_{\psi_{ij}}\ar[r]^{f_{j}} & x'_{j} 
\vert_{U_{ij}} \ar[d]^{\psi'_{ij}}\\
x_{i} \vert_{U_{ij}} \ar[r]^{f_{i}} &  x'_{i}\vert_{U_{ij}} 
}
\]
commutes in $\F(U_{ij})$.
}
\end{itemize}
\end{definition}
\noi Categories of descent data are sometimes used  directly in the definitions
for pre-stack and stack. We observe that if $\F$ is a fibered category, for any
open cover $\cU$, there is a functor $ D \maps \F(M) \to \Des(\F,\cU)$
which sends an object $x \in \F(M)$ to $(x \vert_{U_{i}}, \psi_{ij}=\id)$ in the descent category. If
$\F$ is a prestack, then this functor is fully faithful i.e.\ a bijection on morphisms. 
We have used variations of the next proposition in Chapters 5 and 7. 
\begin{prop}\label{descent_prop}
If $\F$ is a stack over $M$ and $\cU$ is an open cover of $M$, then
the above functor
\[
\F(M) \xto{D} \Des(\F,\cU)
\]
gives an equivalence of categories.
\end{prop}
\begin{proof}
Def.\ \ref{stack_def} implies that the objects $x_{i} \in \F(U_{i})$
given in the descent data can be glued together into a global object
which is unique up to isomorphism.
This implies that $D$ is essentially surjective, and hence an equivalence.
\end{proof}

Let $\F$ be a stack over $M$, and $U \ss M$ an open set. It is easy to
see that we can construct a new stack $\F \vert_{U}$ on $U$ which assigns to the open set $V \ss U$,
the category $\F(V)$. We say $F\vert_{U}$ is the stack $\F$ {\bf
  restricted} to $U$. The following theorem describes
how stacks themselves glue together.
\begin{theorem}[\cite{Schapira:2006}]\label{gluing_thm}
Let $\{U_{i}\}$ be a cover of $M$. Given the following data:
\begin{enumerate}
\item{for each $U_{i}$, a stack $\S_{i}$,}
\item{for each $U_{ij}=U_{i} \cap U_{j}$, an equivalence of stacks $\varphi_{ij} \maps \S_{j}
    \vert_{U_{ij}} \iso \S_{i} \vert_{U_{ij}}$, }
\item{for each $U_{ijk}$, a fibered isomorphism $\mu_{ijk} \maps
    \varphi_{ij} \circ \varphi_{jk} \iso \varphi_{ik}$,
such that, for each $U_{ijk}=U_{i} \cap U_{j} \cap U_{k}$, the diagram
\begin{equation}\label{gluing_thm_eq}
\xymatrix{
\varphi_{ij} \circ \varphi_{jk} \circ \varphi_{kl} \ar[d]_{\mu_{ijk}}
\ar[r]^-{\mu_{jkl}} & \varphi_{ij} \circ \varphi_{jl} \ar[d]^{\mu_{ijl}} \\
\varphi_{ik} \circ \varphi_{kl} \ar[r]^{\mu_{ikl}} & \varphi_{il}
}
\end{equation}
commutes,}
\end{enumerate}
there exists a stack $\S$ on $M$, equivalences of
stacks $\varphi_{i} \maps \S \vert_{U_{i}} \iso \S_{i}$, and fibered
isomorphisms $\eta_{ij} \maps \varphi_{ij} \iso \varphi_{i} \circ
\varphi_{j}^{-1}$ satisfying 
\begin{equation}\label{gluing_thm_eq2}
\xymatrix{
\varphi_{ij} \circ \varphi_{jk} \ar[d]_{\mu_{ijk}} \ar[r]^-{\eta_{jk}}& \varphi_{ij} \circ
(\varphi_{j} \circ \varphi^{-1}_{k}) \ar[d]^{\eta_{ij}} \\
\varphi_{ik} \ar[r]^{\eta_{ik}} & \varphi_{i} \circ \varphi^{-1}_{k}   
}
\end{equation}
The data $(\S,\varphi_{i},\eta_{ij})$ are unique up
to equivalence of stacks. Moreover,  this equivalence is unique up to
unique fibered isomorphism.
\end{theorem}

\section*{Constructing the stack $\HVB^{g}$}
Recall that $\HVB$ is the stack on $M$ which assigns to each open set
$V$, the category of Hermitian vector bundles on $V$. Given an
inclusion $V \to U$, the corresponding functor $\HVB(U) \to \HVB(V)$
is just the pull-back of bundles. The
natural isomorphisms $(ij)^{\ast}\simeq j^{\ast}i^{\ast}$ described in
Def.\ \ref{fibered_cat_def} of fibered catgory are given by the
identity.

Let us now construct the stack $\HVB^{g}$ described in the statement
of Prop.\ \ref{2-line_stack_prop}. Let $g \in C^{2}(\cU,\sh{\U(1)})$
be a 2-cocycle defined on an open cover $\cU=\{U_{i} \}$ of
$M$. For each $i$, let
\[
\HVB_{i} = \HVB \vert_{U_{i}}
\]
be the stack of Hermitian bundles on $M$ restricted to the open set
$U_{i}$. By definition of restriction, we have an equality of stacks
\[
\HVB_{j} \vert_{U_{ij}} = \HVB_{i} \vert_{U_{ij}}
\]
for each $i$ and $j$, and therefore, an identity functor
\[
\HVB_{j} \vert_{U_{ij}} \xto{\varphi_{ij}=\id} \HVB_{i} \vert_{U_{ij}}.
\]
For any open subset $V$ of $U_{ijk}$, we define
a natural transformation between identity functors
\[
\id=\varphi_{ij _{V}} \circ \varphi_{jk_{V}} \xto{\mu_{ijk_{V}}} \varphi_{ik_{V}}=\id
\]
which sends a bundle $E \in \HVB(V)$ to the automorphism
\[
E \xto{g_{ijk} \vert_{V} \cdot} E.
\]
Here, $g_{ijk} \vert_{V} \cdot$ corresponds to multiplying sections of
$E$ by $g_{ijk} \vert_{V} \maps V \cap U_{i} \cap U_{j} \cap U_{k} \to \U(1)$.
It is easy to see that this gives a fibered isomorphism
\[
\varphi_{ij} \circ \varphi_{jk} \xto{\mu_{ijk}} \varphi_{ik}.
\]
The fact that $g$ satisfies the cocycle condition on each $U_{ijkl}$
implies $\mu_{ijk}$ satisfies Eq.\ \ref{gluing_thm_eq}.

Hence, it follows from Theorem \ref{gluing_thm}
%\ref{gluing_thm} 
that there exists a
stack $\HVB^{g}$ on $M$ with equivalences of stacks
\[
\vphi_{i} \maps \HVB^{g} \vert_{U_{i}} \iso \HVB_{i}=\HVB \vert_{U_{i}},
\]
and fibered isomorphisms
\[
\eta_{ij} \maps \vphi_{ij}=\id \iso \vphi_{i} \circ \vphi^{-1}_{j}
\]
satisfying Eq.\ \ref{gluing_thm_eq2}.

\section*{Global sections of $\HVB^{g}$ as twisted bundles}
Now we prove Prop.\ \ref{2-line_stack_prop} by showing that the
category $\HVB^{g}(M)$ of global sections of the stack $\HVB^{g}$ is equivalent to
the category $\mathsf{C}$ of $g$-twisted Hermitian vector bundles over $M$. Since $\HVB^{g}$ is a
stack, Prop.\ \ref{descent_prop} implies $\HVB^{g}(M)$ is equivalent to the
category of descent data $\Des(\HVB^{g},\cU)$, where $\cU$ is the open
cover used in defining the cocycle $g$. Hence, it is sufficent to show
that $\Des(\HVB^{g},\cU)$ is equivalent to $\mathsf{C}$.

We build a functor $\Des(\HVB^{g},\cU) \to \mathsf{C}$ in the following way.
Let $(x_{i},\psi_{ij})$ be an object in the category of descent data.
We use the stack morphisms $\vphi \maps \HVB^{g} \vert_{U_{i}} \iso \HVB_{i}$
to send the objects $x_{i} \in \HVB^{g}(U_{i})$ to Hermitian
vector bundles
\[
E_{i} = \vphi_{i}(x_{i}) \in \HVB(U_{i}).
\]
The fibered isomorphisms $\eta_{ij} \maps \id \iso \vphi_{i} \circ
\vphi^{-1}_j$ assign an isomorphism in $\HVB(U_{ij})$ to every object
in $\HVB(U_{ij})$. Given the objects $\vphi_{j}(x_{j}),\vphi_{j}(x_{i}) \in \HVB(U_{j})$, let the
corresponding isomorphisms be denoted
\begin{align*}
E_{j}=\vphi_{j}(x_{j}) &\xto{\eta_{ij}(x_{j})} \vphi_{i}(x_{j})\\
\vphi_{j}(x_{i}) &\xto{\eta_{ij}(x_{i})} \vphi_{i}(x_{i})=E_{i}.
\end{align*}
We have suppressed the restrictions to keep the notation under
control. This will not cause any problems, since the morphisms and
fibered transformations we are considering commute with the
restriction functors ``on the nose''.
We define isomorphisms
\[
\phi_{ij} \maps E_{j} \iso E_{i}, \quad \text{in $\HVB(U_{ij})$},
\]
by using the descent data $\psi_{ij} \maps x_{j} \vert_{U_{ij}} \iso x_{i} \vert_{U_{ij}}$,
and the commutative diagram
\[
\xymatrix{
E_{j} \ar[d]_{\vphi_{j}(\psi_{ij})} \ar @{-->}[drr]^{\phi_{ij}}
\ar[rr]^{\eta_{ij}(x_{j})} & &\vphi_{i}(x_{j}) \ar[d]^{\vphi_{i}(\psi_{ij})}\\
\vphi_{j}(x_{i}) \ar[rr]^{\eta_{ij}(x_{i})} && E_{i}
}
\]
in $\HVB(U_{ij})$, which is given by the naturality of $\eta_{ij}$.

We claim the isomorphisms of bundles $\phi_{ij}$ satisfy
\[
\phi^{-1}_{ik}\circ \phi_{ij}\circ \phi_{jk}= g_{i j k} \cdot
\]
on $U_{ijk}$. To show this, we write out 
convenient expressions for $\phi_{ij}$, $\phi_{jk}$, and $\phi_{ik}$:
\begin{align*}
\phi_{ij} &= \eta_{ij}(x_{i})\vphi_{j}(\psi_{ij})\\
\phi_{jk} &= \vphi_{j}(\psi_{jk}) \eta_{jk}(x_{k})\\
\phi_{ik} &=\vphi_{i}(\psi_{ik})\eta_{ik}(x_{k}).
\end{align*}
We then consider the following commutative diagram of bundle isomorphisms:
\[
\xymatrix{
E_{k} \ar[d]_{\mu_{ijk}} \ar[rr]^{\eta_{jk}(x_{k})} &&
\vphi_{j}(x_{k}) \ar[d]^{\eta_{ij}} \ar[rr]^{\vphi_{j}(\psi_{ij})} &&
\vphi_{j}(x_{i}) \ar[d]^{\eta_{ij}(x_{i})} \\
E_{k} \ar[rr]^{\eta_{ik}(x_{k})} && \vphi_{i}(x_{k})
\ar[rr]^{\vphi_{i}(\psi_{ij})} && E_{i}.
}
\]
The first square on the left-hand side follows from the fact that
$\eta_{ij}$ satisfies Eq.\ \ref{gluing_thm_eq2}, while the second square
follows from naturality. The commutativity of the diagram, combined
with the equality $\vphi_{j}(\psi_{ij} \circ \psi_{jk}) = \vphi_{j}(\psi_{ik})$ given by
Def.\ \ref{descent_def}, implies
\begin{align*}
\phi_{ij} \circ \phi_{jk} &=
\eta_{ij}(x_{i})\vphi_{j}(\psi_{ik})\eta_{jk}(x_{k})\\
&= \vphi_{i}(\psi_{ik})\eta_{ik}(x_{k}) \mu_{ijk}\\
&=\phi_{ik} g_{ijk}\cdot,
\end{align*}
where the last line follows by definition of $\phi_{ik}$ and
$\mu_{ijk}$.

Hence, we have a functor
\[
\Des(\HVB^{g},\cU) \xto{F} \mathsf{C},
\]
which sends an object $(x_{i},\psi_{ij})$ to the $g$-twisted bundle $(E_{i},\phi_{ij})$, as
defined above. On morphisms, $F$ sends
\[
(x_{i},\psi_{ij}) \xto{f} (x_{i}',\psi'_{ij})
\]
to
\[
(E_{i},\phi_{ij}) \xto{\vphi_{i}(f)} (E_{i}',\phi'_{ij}).
\]
The fact that $\vphi_{i}(f)$ satisfies the axioms for a twisted bundle
morphism follow from the naturality of $\eta_{ij}$. Finally, it is
easy to see that $F$ gives an equivalence of categories, since each
$\vphi_{i}$ is an equivalence of stacks. 

This completes the proof of Prop.\ \ref{2-line_stack_prop}.

\end{document}